\documentclass[prd,twocolumn,10pt,aps,longbibliography,nofootinbib,superscriptaddress]{revtex4-2}

\usepackage{amsmath,amssymb,amsthm}
\usepackage{graphicx}
\usepackage{tikz}

\usepackage{color}
\definecolor{darkblue}{rgb}{0,0,.5}
\definecolor{BLUE}{rgb}{0,0,1}
\definecolor{BLACK}{rgb}{0,0,0}
\usepackage[linktocpage, colorlinks=true, linkcolor=darkblue, citecolor=darkblue]{hyperref}
\usepackage{xr}

\makeatletter
\let\old@makecaption=\@makecaption
\usepackage{subcaption}
\let\@makecaption=\old@makecaption
\makeatother

\newcommand{\C}[1]{{\cal{#1}}}
\newcommand{\bb}[1]{\textbf{#1}}
\newcommand{\bs}[1]{\boldsymbol{#1}}

\newcommand{\mf}[1]{{\mathfrak{#1}}}

\newcommand{\lr}[1]{{\langle {#1}\rangle}}

\newcommand{\rl}[0]{{\rangle\langle}}

\newcommand{\proj}[1]{| #1\rangle\!\langle #1 |}

\usepackage{scalerel}
\usetikzlibrary{calc}
\usetikzlibrary{patterns}
\usetikzlibrary{svg.path}
\definecolor{orcidlogocol}{HTML}{A6CE39}
\tikzset{
	orcidlogo/.pic={
		\fill[orcidlogocol] svg{M256,128c0,70.7-57.3,128-128,128C57.3,256,0,198.7,0,128C0,57.3,57.3,0,128,0C198.7,0,256,57.3,256,128z};
		\fill[white] svg{M86.3,186.2H70.9V79.1h15.4v48.4V186.2z}
		svg{M108.9,79.1h41.6c39.6,0,57,28.3,57,53.6c0,27.5-21.5,53.6-56.8,53.6h-41.8V79.1z M124.3,172.4h24.5c34.9,0,42.9-26.5,42.9-39.7c0-21.5-13.7-39.7-43.7-39.7h-23.7V172.4z}
		svg{M88.7,56.8c0,5.5-4.5,10.1-10.1,10.1c-5.6,0-10.1-4.6-10.1-10.1c0-5.6,4.5-10.1,10.1-10.1C84.2,46.7,88.7,51.3,88.7,56.8z};
	}
}
\newcommand\orcid[1]{\!%
	\href{https://orcid.org/#1}{%
		\mbox{%
			\scaleto{%
				\begin{tikzpicture}[yscale=-1,transform shape]
				\pic{orcidlogo};
				\end{tikzpicture}
			}{8pt}%
		}%
	}%
}

\begin{document}

\newtheorem{mydef}{Definition}[section]
\newtheorem{lemma}{Lemma}%[section]
\newtheorem{conj}{Conjecture}
\newtheorem{thm}{Theorem}[section]
\newtheorem{res}{Result}[section]
\theoremstyle{remark}
\newtheorem{rmrk}{Remark}[section]

\title{Approximate Decoherence, Recoherence and Records in Isolated Quantum Systems}

\author{Philipp Strasberg~\orcid{0000-0001-5053-2214}}
\affiliation{Instituto de F\'isica de Cantabria (IFCA), Universidad de Cantabria--CSIC, 39005 Santander, Spain}
\affiliation{F\'isica Te\`orica: Informaci\'o i Fen\`omens Qu\`antics, Departament de F\'isica, Universitat Aut\`onoma de Barcelona, 08193 Bellaterra (Barcelona), Spain}
\author{Joseph Schindler~\orcid{0000-0002-8799-9800}}
\affiliation{F\'isica Te\`orica: Informaci\'o i Fen\`omens Qu\`antics, Departament de F\'isica, Universitat Aut\`onoma de Barcelona, 08193 Bellaterra (Barcelona), Spain}
\author{Jiaozi Wang~\orcid{0000-0001-6308-1950}}
\affiliation{Department of Mathematics/Computer Science/Physics, University of Osnabr\"uck, D-49076 Osnabr\"uck, Germany}
\author{Andreas Winter~\orcid{0000-0001-6344-4870}}
\affiliation{F\'isica Te\`orica: Informaci\'o i Fen\`omens Qu\`antics, Departament de F\'isica, Universitat Aut\`onoma de Barcelona, 08193 Bellaterra (Barcelona), Spain}
\affiliation{Department Mathematik/Informatik–Abteilung Informatik, Universit\"at zu K\"oln, Albertus-Magnus-Platz, 50923 K\"oln, Germany}
\affiliation{ICREA -- Instituci\'o Catalana de Recerca i Estudis Avan{\c{c}}ats, Pg. Lluis Companys 23, 08010 Barcelona, Spain}

\begin{abstract}
 Using the framework of decoherent histories, we study which past events leave detectable records in isolated quantum systems under the realistic assumption that decoherence is approximate and not perfect. Using tools from quantum state discrimination, we establish asymptotically for a large class of (pseudo-)random histories how the strength of decoherence changes the number of reliable records, i.e., we quantitatively answer the question: when can approximate decoherence be approximated by exact decoherence? In addition, we reveal a clear decoherence structure for long histories based on a numerically exact solution of a random matrix model that, as we argue, captures generic aspects of decoherence. In particular, we observe strong recoherences for maverick histories that are statistical outliers with respect to Born's rule. From the perspective of the Many Worlds Interpretation, our results reveal a ``branch selection problem'' and shed light on the emergence of Born's rule and the theory confirmation problem.
\end{abstract}

\maketitle

%%%%%%%%%%%%%%%%%%%%%%%%%%%%%%%%%%%%%%%%%%%%%%%%%%%%%%%%%%%%%%%%%%%%%%%%%%%%%%%%%%%%%%%%%%%%%%%%%%%%%%%%%%%%%%%%%%%%%%%%
\section{Introduction}
%%%%%%%%%%%%%%%%%%%%%%%%%%%%%%%%%%%%%%%%%%%%%%%%%%%%%%%%%%%%%%%%%%%%%%%%%%%%%%%%%%%%%%%%%%%%%%%%%%%%%%%%%%%%%%%%%%%%%%%%

Decoherent histories are a tool developed in quantum cosmology to answer two fundamental questions about isolated quantum systems~\cite{GellMannHartleInBook1990, HartleLecture1992, HalliwellANY1995, DowkerKentJSP1996}: which quantum processes can be simulated by a classical stochastic process? and which past events can be unambiguously identified (i.e., leave \emph{records}) in the current quantum state? Apart from a few exceptions~\cite{DowkerHalliwellPRD1992, DowkerKentPRL1995, DowkerKentJSP1996, McElwainePRA1996, HalliwellPRD2001, HalliwellPRA2005}, it is (often tacitly) assumed that decoherence between histories is ideal or at least so strong such that it does not cause problems for the emergence of classicality and the formation of records. However, approximate (instead of exact) decoherence is generic, and, motivated by recent results quantifying the strength of decoherence in quantum many-body systems~\cite{StrasbergReinhardSchindlerPRX2024, WangStrasbergPRL2025}, we address the following two questions.

First, we ask how many approximately decoherent histories are potentially out there and how many can be reliably distinguished by records? Asymptotically, and for a large class of (pseudo-)random histories, we find that there can be many more histories than records. Intuitively, while decoherence is strong enough to guarantee a negligible discrimination error for any \emph{specific} pair of histories, discrimination errors can accumulate when considering the totality of histories.

Second, we investigate the structure of decoherence for \emph{long} histories, i.e., when we ``run out'' of Hilbert space to accommodate all histories in an approximately decoherent way. This is done by numerical exact integration of the Schr\"odinger equation for a random matrix model, which we argue captures generic features of decoherence. We find that decoherence develops a clear structure: a specific set of histories ``recoheres'' while the others remain decoherent.\footnote{Our notion of recoherence based on histories should be, at least \emph{a priori}, distinguished from the recoherence of cosmological fluctuations in context of inflation, as studied through the revival of purity in open quantum systems in Refs.~\cite{ColasGrainVenninEPL2023, KranasGrainVenninPRD2025}.} This specific set is, among other properties, characterized by histories that are atypical according to Born's rule. Our results extends previous unpublished findings about the structure of decoherence~\cite{StrasbergSchindlerArXiv2023}.

Since the superposition principle is confirmed at increasing scales (e.g., in particle interference experiments~\cite{JoenssonZfP1961, JoenssonAJP1974, ArndtEtAlNat1999, GerlichEtAlNC2011, FeinEtAlNP2019}), our investigation about the structure of decoherence and the (non-)existence of reliable records has relevance for the study of isolated quantum \emph{sub}systems, which arise, e.g., in the fields of quantum computation~\cite{NielsenChuangBook2000, PreskillQuantum2018}, quantum stochastic processes and higher-order operations~\cite{MilzModiPRXQ2021, TarantoEtAlArXiv2025}, quantum state discrimination (in particular for random states)~\cite{BarnettCrokeAOP2009, BaeKwekJPA2015, WatrousBook2018, CieslinskiEtAlPR2024}, etc.

When applied to the quantum Universe as a whole, decoherent histories are closely related to the Many Worlds Interpretation (MWI)~\cite{EverettRMP1957, DeWittPT1970, SaundersEtAlBook2010, WallaceBook2012, Vaidman2021}. The MWI attempts to solve the quantum measurement problem by positing that the superposition principle applies without limits, similar to relational quantum mechanics~\cite{RovelliIJTP1996} but in contrast to, e.g., Griffith's related consistent history interpretation (which rejects the universal validity of unitary time evolution)~\cite{GriffithsJSP1984, GriffithsBook2002, Griffiths2019} or collapse models~\cite{BassiEtAlRMP2013, GhirardiBassi2024}. The potential (in)validity of the MWI clearly influences questions such as: Can we (and do we need to) quantize gravity~\cite{KieferInBook2007, OppenheimIJMP2023}? On which ensemble of universes should we base anthropic reasoning~\cite{Carr2007}? Which quantum effects exist in the early universe~\cite{MukhanovBook2005}?

Owing to the success of decoherence theory~\cite{ZurekRMP2003, JoosEtAlBook2003, SchlosshauerPR2019}, the MWI is often portrayed as leading to a structureless multiverse in which non-interfering parallel branches (or histories in our context) coexist in a peaceful way, describing a multiverse in which ``everything everywhere all at once'' happens without anything particularly quantum left. Indeed, Tegmark proclaims that the quantum multiverse ``adds nothing new''~\cite{TegmarkSA2003} on top of a purely classical multiverse formed by the many cosmic horizon marbles in an infinitely extended homogenous universe. Yet, many quantitative questions related to the MWI are still awaiting answer, and our results suggest that the picture of peacefully coexisting and eternally splitting branches has its limits.

In short, the goal of this paper is to take approximate decoherence seriously and to study the resulting limitations and structure.

\emph{Outline.} We first establish notation and review the main consequences of \emph{exact} decoherence in Sec.~\ref{sec:decoherent histories} before formulating a general ansatz to study approximate decoherence in Sec.~\ref{sec:approximate decoherence}. Section~\ref{sec:records} then considers the problem of identifying records, treated as a coherent generalization of the quantum state discrimnation (QSD) task, and also estimates the total number of histories for a given degree of decoherence. While the results of Sec.~\ref{sec:records} are mostly general and analytic, Sec.~\ref{sec:structure} numerically studies the structure of decoherence for long histories, which reveals a correlation between Born's rule and decoherence. Finally, Sec.~\ref{sec:conclusions} concludes with a particular emphasis on the MWI. A detailed supplemental material (SM) lists mathematical derivations related to Sec.~\ref{sec:records} and many additional numerical results for Sec.~\ref{sec:structure}.

%%%%%%%%%%%%%%%%%%%%%%%%%%%%%%%%%%%%%%%%%%%%%%%%%%%%%%%%%%%%%%%%%%%%%%%%%%%%%%%%%%%%%%%%%%%%%%%%%%%%%%%%%%%%%%%%%%%%%%%%
\section{Decoherent Histories Formalism}\label{sec:decoherent histories}
%%%%%%%%%%%%%%%%%%%%%%%%%%%%%%%%%%%%%%%%%%%%%%%%%%%%%%%%%%%%%%%%%%%%%%%%%%%%%%%%%%%%%%%%%%%%%%%%%%%%%%%%%%%%%%%%%%%%%%%%

We review the decoherent histories formalism for the non-relativistic quantum mechanics of an isolated system with a given absolute time $t$, as used in the numerics of Sec.~\ref{sec:structure}. The results of Sec.~\ref{sec:records} are valid in a broader context, relying only on the idea that history states live in some abstract Hilbert space. For further extensions see, e.g., Refs.~\cite{HartleLecture1992, BlencoweAP1991, HalliwellThorwartPRD2002, ChristodoulakisWalldenJPCS2011}.

The state $|\Psi\rangle\in\C H$ of an isolated quantum system is propagated with a unitary time evolution operator $U_{t} = e^{-iHt}$ with Hamiltonian $H$ ($\hbar \equiv 1$). Physical properties are described by a complete set of orthogonal projectors $\{\Pi_x\}_{x=1}^M$ satisfying $\Pi_x\Pi_{x'} = \delta_{x,x'}\Pi_x$ (with $\delta_{x,x'}$ the Kronecker delta) and $\sum_x\Pi_x = I$ (with $I$ the identity on $\C H$). Those projectors partition the Hilbert space into orthogonal subspaces $\C H = \bigoplus_x\C H_x$ with dimension $D_x = \dim\C H_x = \mbox{tr}\{\Pi_x\}$. The total dimension is $D = \dim\C H = \sum_x D_x$.\footnote{Throughout we assume that $\C H$ has finite $D$ except in Sec.~\ref{sec:sub-remarks} where we address the case $D=\infty$. A discussion of whether our (observable) universe could be describable by a finite-dimensional Hilbert space can be found in Sec.~\ref{sec:conclusions}.} Borrowing terminology from statistical mechanics, we call $\{\Pi_x\}$ a coarse-graining. The terminology ``coarse'' is indeed justified because non-trivial, approximate decoherence requires $M\ll D$.

The histories formalism is introduced using a set of (for simplicity of presentation) equidistant times $t_k = k\Delta t$ with $k\in\{1,2,\dots,L\}$ with $t_0$ the initial and $t_L$ the final or current time. We then write the time evolution of the wave function as a ``sum over histories'' in spirit of (in fact, as a generalization of) Feynman's path integral:
\begin{equation}\label{eq:histories}
 |\Psi_L\rangle = U_{t_L-t_0} |\Psi_0\rangle = \sum_{\bs x} |\psi'(\bs x)\rangle.
\end{equation}
Here, $\bs x = (x_L,\dots,x_1)$ is called a history of length $L$ and the associated history state is
\begin{equation}
 |\psi'(\bs x)\rangle \equiv \Pi_{x_L} U_{\Delta t} \cdots U_{\Delta t}\Pi_{x_1} U_{\Delta t}|\Psi_0\rangle
\end{equation}
with $|\Psi_0\rangle$ the initial state. Here and throughout the paper, we use a prime to denote non-normalized states, such as $|\psi'(\bs x)\rangle$, whereas normalized states are unprimed, such as $|\Psi_0\rangle$. Note that there are $N=M^L$ many history states for a given $M$ and $L$. Furthermore, the $\{\Pi_{x_k}\}$ can describe different coarse-grainings at different times, but this aspect is not further emphasized here.

As far as single-time observables are concerned, eqn~(\ref{eq:histories}) offers no computational advantage (rather the opposite) over, e.g., the Schr\"odinger or Heisenberg picture. However, the conditioning on intermediate events $x_k$ at times $t_k$ causes the $|\psi'(\bs x)\rangle$ to have additional information. Of particular relevance is the condition that the decoherence functional (DF) $\mf D_{\bs x,\bs x'} \equiv \lr{\psi'(\bs x)|\psi'(\bs x')}$ is diagonal,
\begin{equation}\label{eq:DHC}
 \mf D_{\bs x,\bs x'} = 0 ~~~ \forall \bs x\neq\bs x',
\end{equation}
which is known as the decoherent histories condition (DHC). It is worth highligthing that the DF appears in many contexts, but often without explicitly noting it. The DF is a specific Kirkwood-Dirac quasiprobability~\cite{GherardiniDeChiaraPRXQ2024, ArvidssonEtAlNJP2024}, a Gram matrix in quantum measurement theory~\cite{BarnettCrokeAOP2009, BaeKwekJPA2015, WatrousBook2018, CieslinskiEtAlPR2024}, and a sample covariance matrix in statistics. Moreover, the DF is formally a $M^L\times M^L$ density matrix: it has unit trace, it is Hermitian and positive. As such it is connected to superdensity operators defined over a history Hilbert space~\cite{IshamJMP1994, IshamLindenSchreckenbergJMP2004, CotlerEtAlJHEP2018} (see also Refs.~\cite{MilzModiPRXQ2021, TarantoEtAlArXiv2025}), and it has been studied in the context of dynamical entropy and quantum chaos~\cite{AlickiFannesLMP1994, DowlingModiPRXQ2024, PolyakovArefyevaPRA2024, ODonovanEtAlArXiv2025}. For this paper, however, the following two consequences of the DHC are the most relevant.

First, the DHC implies that the multi-time Born probabilities $q(\bs x) \equiv \lr{\psi'(\bs x)|\psi'(\bs x)}$ obey the probability sum rule~\cite{GriffithsJSP1984}. For an outside (Copenhagen) observer this means that there is no measurement disturbance at the coarse-grained level: averaging over measurement results gives the same statistics as not measuring. The probability sum rule also implies that the histories $\bs x$ can be simulated by a \emph{classical} stochastic process~\cite{BittnerRosskyJCP1997, SmirneEtAlQST2018, StrasbergDiazPRA2019, MilzEtAlQuantum2020, MilzEtAlPRX2020, SzankowskiCywinskiQuantum2024}. Conversely, if one wants to find detectable dynamical quantum effects, one has to look for conditions violating the DHC. Notice that the DHC implies the validity of Leggett-Garg inequalities~\cite{EmaryLambertNoriRPP2014}.

Second, the DHC is \emph{equivalent} to the existence of records about the history $\bs x$. Here, ``existence of records'' means that the DHC holds if and only if the $|\psi'(\bs x)\rangle$ live in orthogonal subspaces at the final time $t_L$, which implies that the history $\bs x$ can be uniquely identified by a final projective measurement of some abstract record observable $R$. Note that existence of records in the sense of decoherent histories does not imply redundant records, which is central to quantum Darwinism~\cite{ZurekNP2009, KorbiczQuantum2021}. Conversely, quantum Darwinism does not determine which past events $\bs x$ leave records, i.e., it does not imply the DHC. An attempt to unify both aspects can be found in Refs.~\cite{FinkelsteinPRD1993, RiedelZurekZwolakPRA2016}.

The question of which events leave in principle detectable records in a unitarily evolving quantum system is profound.\footnote{In classical mechanics the situation is trivial from a fundamental point of view. Every pure state (point in phase space) has a unique history of events. Every possible $x_j$ happens either with probability zero or one.} For instance, there are events in quantum mechanics, such as $x_j =~$``the particle passed the left slit'' in a double slit experiment, that do not leave any record at later times. In contrast, we strongly believe that events such as $x_j =~$``in a Stern-Gerlach type experiment we found a silver atom with spin up'' or $x_k =~$``there were dinosaurs on Earth 100 million years ago'' faithfully represent objective facts about the past, even though the global wave function $|\Psi\rangle$ might have been in a superposition of different events.\footnote{Including perhaps $x'_k =$~``the asteroid didn't hit Earth 65 million years ago'', so dinosaurs still rule our world in another branch.}

It is instructive to rewrite eqn~(\ref{eq:histories}) as
\begin{equation}\label{eq:histories Heisenberg}
 |\Psi_L\rangle = \left(\sum_{\bs x} \Pi_{x_L}^{(0)}\Pi_{x_{L-1}}^{(1)} \cdots \Pi_{x_1}^{(L-1)}\right)|\Psi_L\rangle
\end{equation}
with $\Pi_x^{(k)} = U_{-k\Delta t}^\dagger\Pi_x U_{-k\Delta t}$ denoting the projector in the Heisenberg picture with respect to the current time $t_L$ (note that the state on the right hand side is $|\Psi_L\rangle$). Thus, we see that a pairwise commutation of all the Heisenberg projectors implies the DHC. This is the well-known textbook statement that simultaneously measurable observables must commute in quantum mechanics---with the pecularity that the observables are here related to events in the past. The converse direction requires the additional assumption that the DHC holds for an \emph{informationally complete set} of initial states $|\Psi_0\rangle$. This is proven in the SM, see also Ref.~\cite{SchererSoklakovJMP2005}. In light of this result, it is intruiging to ask whether there are records (for instance, about the early universe) that are only a consequence of a specifically fine tuned initial state $|\Psi_0\rangle$, or whether all our records result from commuting projectors, but answering this question goes beyond the scope of the present paper.

Finally, we connect the framework to the preferred basis problem of the MWI: from the infinitely many different ways in which the universal wave function $|\Psi\rangle = \sum_\alpha |\phi'_\alpha\rangle$ can be decomposed into orthogonal states $\{|\phi'_\alpha\rangle\}$, which one is acceptable to define the branches of the quantum multiverse? Unfortunately, there is no agreed on definition of ``acceptable branch decomposition''~\cite{Vaidman2021, RiedelQuantum2025}, but it seems desirable that $\alpha$ allows us to identify the past events $\bs x$ that define our history in the universe (the act of finding your branch in the multiverse is also known as \emph{self-location}). Thus, there should be a relation $\bs x = f(\alpha)$ that allows to identify the history $\bs x$ of a given branch $\alpha$. Defining $|\psi'(\bs x)\rangle \equiv \sum_{\alpha\in f^{-1}(\bs x)}|\phi'_\alpha\rangle$, we see that the requirement that $|\Psi\rangle = \sum_\alpha |\phi'_\alpha\rangle$ solves the preferred basis problem then implies that the $|\psi'(\bs x)\rangle$ obey the DHC. Conversely, any set of histories satisfying the DHC is a candidate solution to the preferred basis problem. It is known as the set selection problem~\cite{PazZurekPRD1993, DowkerKentPRL1995, DowkerKentJSP1996, RiedelZurekZwolakPRA2016} that the DHC alone leaves a lot of freedom in defining a decomposition of $|\Psi\rangle$, some of which are utterly complicated with no similarity to our experienced world. However, this does not conflict with the view taken here that the DHC is a \emph{minimal} or \emph{necessary}, not an exhaustive or sufficient, requirement. For different recent approaches see Refs.~\cite{RiedelPRL2017, WeingartenFP2022, OllivierEnt2022, TouilEtAlQuantum2024, TaylorMcCullochQuantum2025, AdilEtAlPRD2026}.

To summarize, whether it concerns a unitarily evolving quantum subsystem or the quantum Universe as a whole, the DHC identifies which coarse-grainings give rise to a process that is indistinguishable from a classical stochastic process and leave records about said process inside the isolated quantum system.

%%%%%%%%%%%%%%%%%%%%%%%%%%%%%%%%%%%%%%%%%%%%%%%%%%%%%%%%%%%%%%%%%%%%%%%%%%%%%%%%%%%%%%%%%%%%%%%%%%%%%%%%%%%%%%%%%%%%%%%%
\section{Approximate Decoherence}\label{sec:approximate decoherence}
%%%%%%%%%%%%%%%%%%%%%%%%%%%%%%%%%%%%%%%%%%%%%%%%%%%%%%%%%%%%%%%%%%%%%%%%%%%%%%%%%%%%%%%%%%%%%%%%%%%%%%%%%%%%%%%%%%%%%%%%

In reality the best one can hope for is approximate decoherence. Exact decoherence occurs only for exactly conserved quantities, fine-tuned examples, or perhaps in the strict thermodynamic limit when a finite-dimensional Hilbert space (no matter how large) no longer adequately describes the situation. Unfortunately, the implications of approximate decoherence are much harder to understand, not the least because there is no general theory quantifying the strength of decoherence. We here summarize what is known and formulate a useful ansatz for approximate decoherence on which Sec.~\ref{sec:records} is based.

Instead of considering the DF $\mf D_{\bs x,\bs x'} = \lr{\psi'(\bs x)|\psi'(\bs x')}$, approximate decoherence is better understood using the Gram matrix of the normalized history states $|\psi(\bs x)\rangle = |\psi'(\bs x)\rangle/\|\psi'(\bs x)\|$~\cite{DowkerHalliwellPRD1992}:
\begin{equation}\label{eq:NDF}
 G_{\bs x,\bs x'} \equiv \lr{\psi(\bs x)|\psi(\bs x')}.
\end{equation}
We call $G$ the normalized decoherence functional (NDF) and the Cauchy-Schwarz inequality implies $|G_{\bs x,\bs x'}|\le1$. A useful notion of approximate decoherence then arises for $|G_{\bs x,\bs x'}| \ll 1$, and answering the question \emph{how small} $|G_{\bs x,\bs x'}|$ needs to be to count as ``useful'' is the content of Sec.~\ref{sec:records}. Note that by construction we have $G_{\bs x,\bs x'} \sim \delta_{x_L,x'_L}$, i.e., orthogonality with respect to the final projection is guaranteed. It is thus useful to consider histories $\tilde{\bs x} = (x_{L+1},x_L,\dots,x_1)$ of length $L+1$ such that
\begin{equation}\label{eq:DF constraint}
 G_{\tilde{\bs x},\tilde{\bs x}'} = \delta_{x_{L+1},x'_{L+1}} G_{\bs x,\bs x'}(x_{L+1})
\end{equation}
with the \emph{conditional} NDF $G_{\bs x,\bs x'}(x_{L+1})$. It is the Gram matrix of all history states conditioned on ending in subspace $x_{L+1}$ at time $(L+1)\Delta t$.

How does $G(x)$ look like? Clearly, its diagonal entries are one by construction, and its off-diagonal entries describe the overlap of two quantum states that evolved through at least one different subspace $x_k\neq x'_k$ at time $t_k$. Let us now focus on \emph{coarse} projectors $\Pi_x$ describing distinct \emph{macroscopic} subspaces. If the dynamics is sufficiently irregular, the two history states $|\psi(\bs x)\rangle$ and $|\psi(\bs x')\rangle$ will contain many miroscopic phases dependent on different macroscopic events. Hence, to a first order approximation one might assume that they behave like independent Haar random states\footnote{The Haar measure $\mu[|\psi\rangle]$ is the unique measure of quantum states that is invariant under \emph{any} unitary rotation $U$ of the state, i.e., $\mu[|\psi\rangle] = \mu[U|\psi\rangle]$. In this sense, a Haar random state is maximally random. Haar random states can be generated by drawing a vector with independent and identically distributed complex Gaussian entries, followed by normalization. They play an important role in this paper. For further information see the SM.}, whose overlap is statistically equal to $r/\sqrt{D_x}$ with $r$ a complex random number of zero mean and unit variance~\cite{KusMostowskiHaakeJPA1988}. In reality, deviations from this ideal behavior can be captured by an ansatz of the form
\begin{equation}\label{eq:DF scaling form}
 G(x) = I + \frac{R}{D^{\alpha}_x}.
\end{equation}
Besides the $N\times N$ identity matrix $I$, a scaling exponent $\alpha\neq1/2$ describes deviations from the Haar random case. Moreover, $R$ is a Hermitian \emph{pseudo-random} matrix with diagonal zero and off-diagonal elements $R_{\bs x,\bs x'}\in\mathbb{C}$ modeled by zero-mean unit-variance pseudo-random numbers: they are determined by $U_{\Delta t}$, $\{\Pi_x\}$ and $|\Psi_0\rangle$ in such a complicated way such that they \emph{appear} erratically varying without any structure. Even though the exact nature of the (pseudo-)randomness of $R$ is currently unknown, the ansatz~(\ref{eq:DF scaling form}) is in unison with all the numerical evidence that we have~\cite{StrasbergReinhardSchindlerPRX2024, WangStrasbergPRL2025, GemmerSteinigewegPRE2014, SchmidtkeGemmerPRE2016, NationPorrasPRE2020, StrasbergEtAlPRA2023, StrasbergSP2023}. It describes an exponential (as a function of the particle number $N_\text{particle}\sim\log D$) suppression of coherences for coarse observables, which is \emph{robust} in the sense that almost all initial states, Hamiltonians and coarse and slow coarse-grainings give rise to eqn~(\ref{eq:DF scaling form})~\cite{StrasbergReinhardSchindlerPRX2024, WangStrasbergPRL2025}. Furthermore, Refs.~\cite{StrasbergReinhardSchindlerPRX2024, WangStrasbergPRL2025} found that the scaling exponent obeys $\alpha\lesssim 1/2$ over a large range of time scales $\Delta t$, including relevant nonequilibrium time scales. At equilibrium time scales one finds $\alpha\approx1/2$, whereas for shorter time scales one usually gets $\alpha<1/2$ owing to the fact that $\alpha\rightarrow0$ for $\Delta t\rightarrow0$. Exceptions with $\alpha>1/2$ are, however, possible~\cite{WangStrasbergPRL2025}.

The validity of eqn~(\ref{eq:DF scaling form}) requires three assumptions. First, we need to restrict attention to coarse-grainings that are sufficiently coarse and slow (or quasi-conserved). Only those are perceptible to us humans and relevant for the quantum-to-classical transition. Note that this also includes modern quantum experiments: while the quantum systems we measure are small and evolve fast, the measurement outcomes---which are really what matters---are stored in coarse and stable, long-lived (i.e., slow) memories. Second, the system needs to be non-integrable or chaotic, otherwise decoherence is qualitatively weaker~\cite{WangStrasbergPRL2025}. Third, the number of histories must obey $N \lesssim D$. What happens for very long histories when $N\gg D$ is the content of Sec.~\ref{sec:structure}, and our findings indicate that eqn~(\ref{eq:DF scaling form}) remains valid for a subset of histories.

It is worth to slightly generalize and rewrite eqn~(\ref{eq:DF scaling form}). Introducing the Boltzmann entropy $S(x) = \ln D_x$, we write
\begin{equation}\label{eq:DHH}
 G_{\bs x,\bs x'}(x) = \delta_{\bs x,\bs x'} + e^{-\alpha S(x)} F(\Delta_{\bs x,\bs x'}) R_{\bs x,\bs x'}.
\end{equation}
Here, the $R_{\bs x,\bs x'}$ are as above, but $F$ describes a smooth, slowly varying envelope function depending on the distance $\Delta_{\bs x,\bs x'}$ between two histories. We will find numerical evidence for such a function with $\Delta_{\bs x,\bs x'}$ the Hamming distance in Sec.~\ref{sec:structure}. The beauty of the ansatz~(\ref{eq:DHH}) is that it is structurally identical to the so-called eigenstate thermalization hypothesis (ETH), an ansatz for the matrix elements $\lr{E_k|A|E_\ell}$ of physically relevant observables $A$ in the energy eigenbasis $\{|E_k\rangle\}$ of chaotic quantum systems~\cite{SrednickiJPA1999, DAlessioEtAlAP2016, DeutschRPP2018}. While describing two completely different situations in physics, the common feature is that in both cases one is interested in coarse properties of sufficiently complex, non-integrable quantum systems with the result that these properties can be described using some core structural properties (such as $S$ and $F$) plus (pseudo-)random matrix theory. Of course, this echos Wigner's famous finding that random matrices capture generic aspects of complex quantum systems in a wide variety of situations~\cite{Wigner1967, BrodyEtAlRMP1981, BeenakkerRMP1997, GuhrMuellerGroelingWeidenmuellerPR1998, HaakeBook2010, BorgonoviEtAlPR2016}. Indeed, random matrix theory can be used to derive the ETH~\cite{DeutschPRA1991, ReimannDabelowPRE2021}, but in general it remains an unproven hypothesis (albeit backed up by an enormous amount of evidence). Analogously, we might put forward the ansatz~(\ref{eq:DHH}) as the DHH --- the decoherent histories hypothesis.

The ansatz~(\ref{eq:DF scaling form}) and its slight generalization~(\ref{eq:DHH}) is the first main insight that we wish to emphasize here. Together with the history weights $q(\bs x)$, the Gram matrices (or conditional NDFs) $G(x)$ completely describe the decoherence properties of the system. Clearly, it would be desirable to also have a general ansatz for the weights $q(\bs x)$, but we are currently not able to formulate one and believe $q(\bs x)$ strongly depends on the physical situation.\footnote{For instance, when describing the planets of the solar system $q(\bs x)$ will be essentially peaked around one deterministic trajectory $\bs x_\text{Newton}$. In contrast, $q(\bs x)$ tends to be spread out for systems where fluctuations, noise and diffusion dominate.} Below we are therefore often interested in a worst case scenario where $q(\bs x)\approx1/N$ is approximately uniform.

%%%%%%%%%%%%%%%%%%%%%%%%%%%%%%%%%%%%%%%%%%%%%%%%%%%%%%%%%%%%%%%%%%%%%%%%%%%%%%%%%%%%%%%%%%%%%%%%%%%%%%%%%%%%%%%%%%%%%%%%
\section{Identifying Records}\label{sec:records}
%%%%%%%%%%%%%%%%%%%%%%%%%%%%%%%%%%%%%%%%%%%%%%%%%%%%%%%%%%%%%%%%%%%%%%%%%%%%%%%%%%%%%%%%%%%%%%%%%%%%%%%%%%%%%%%%%%%%%%%%

In this section we switch to a lighter notation and write $i,|\psi_i\rangle$ instead of $\bs x,|\psi(\bs x)\rangle$. This is also done to emphasize that the results below have a wider relevance for the geometry of quantum states and quantum state discrimination, independent of their relation to the decoherent histories framework. There are two subsections. In Sec.~\ref{sec:end} we ask how many states $|\psi_i\rangle$ can satisfy the ansatz~(\ref{eq:DF scaling form}), which---geometrically speaking---quantifies the degree of approximate orthogonality between the states. In Sec.~\ref{sec:SLP} we ask how well approximately orthogonal states can be distinguished.

%%%%%%%%%%%%%%%%%%%%%%%%%%%%%%%%%%%%%%%%%%%%%%%%%%%%%%%%%%%%%%%%%%%%%%%%%%%%%%%%%%%%%%%%%%%%%%%%%%%%%%%%%%%%%%%%%%%%%%%%
\subsection{The End of Decoherences}\label{sec:end}
%%%%%%%%%%%%%%%%%%%%%%%%%%%%%%%%%%%%%%%%%%%%%%%%%%%%%%%%%%%%%%%%%%%%%%%%%%%%%%%%%%%%%%%%%%%%%%%%%%%%%%%%%%%%%%%%%%%%%%%%

In a Hilbert space of dimension $D$ there can be at most $D$ pairwise orthogonal vectors or, in our context, at most $D$ exactly decoherent histories.\footnote{When counting the number of histories we have to exclude trivial ``null histories'' $|\psi'(\bs x)\rangle = 0$. For instance, if $[\Pi_x,H] = 0$ for all $x$, then all $\bs x = (x_L,\dots,x_1)$ for any $L$ are formally decoherent, but $|\psi'(\bs x)\rangle = 0$ for all $\bs x$ that do not satisfy $x_L = \dots = x_1$.} If one relaxes the constraint of exact orthogonality, there might be more than $D$ vectors. Since there are infinitely many conceivable history decompositions $|\Psi\rangle = \sum_{i=1}^N|\psi'_i\rangle$, it becomes relevant to ask how large can $N$ be without jeopardizing the approximate decoherence of eqn~(\ref{eq:DF scaling form}).

Let $N_{\mathbb{C}}(D,\epsilon)$ be the maximum number of vectors $|\psi_i\rangle \in \mathbb{C}^D$ that satisfy $|\lr{\psi_i|\psi_j}|\le\epsilon$ for all $i\neq j$. Then, the central insight of this section is
\begin{equation}\label{eq:N bounds}
 e^{(D-1)\epsilon^2} \lesssim N_{\mathbb{C}}(D,\epsilon) \lesssim e^{\C O(D\epsilon^2\ln \epsilon^{-1})},
\end{equation}
where the symbol $\lesssim$ indicates asymptotic (large $D$) approximations and $\C O[f(x)]$ is used to denote an asymptotic proportionality with $f(x)$. Equation~(\ref{eq:N bounds}) shows that for $\epsilon \sim D^{-\alpha}$ with $\alpha < 1/2$ the possible number of approximately orthogonal vectors growth \emph{faster} than any polynomial in $D$ (though the growth is slower than exponential). For $\alpha=1/2$ the lower bound in eqn~(\ref{eq:N bounds}) is useless, and $N_{\mathbb{C}}(D,\epsilon)$ scales polynomially with $D$ (see the references below for further details).

In context of decoherent histories, we need to consider $N_\mathbb{C}(D_x,\epsilon_x)$ for each final subspace $x$, but the conclusion remains the same. Instead of having $D_x$ many orthogonal histories per subspace, there are $\gtrsim \exp(D_x^{1-2\alpha})$ many more for $\alpha<1/2$. Of course, whether a given physical system generates history states compatible with eqn~(\ref{eq:N bounds}) is \emph{a priori} unclear. For instance, in Sec.~\ref{sec:structure} we numerically find that some histories become recoherent such that no bound $|\lr{\psi_i|\psi_j}|\le\epsilon\ll1$ exists for all $i\neq j$. However, recoherent states do not negatively impact the lower bound in eqn~(\ref{eq:N bounds}) and the overall conclusion remains the same: the number of approximately decoherent histories is potentially tremendously larger than $D$.

McElwaine was the first to consider the above problem in the context of decoherent histories~\cite{McElwainePRA1996}, but he mostly specialized to the strongly decoherent cases with $|\Re G_{\bs x,\bs x'}| \le 1/\sqrt{2D+2}$ or $|\Re G_{\bs x,\bs x'}| \le 1/(2D)$. Also in other physical contexts the above problem received attention~\cite{ChaoEtAl2017, ChakravartyJHEP2021, SoulasFP2024}. Mathematically, finding $N_{\mathbb{C}}(D,\epsilon)$ is known as the spherical packing, spherical code or kissing problem. Equation~(\ref{eq:N bounds}) is essentially the Johnson-Lindenstrauss lemma~\cite{JohnsonLindenstraussAMS1984, DasguptaGuptaRSA2003}, which states that any subset of $N$ points in an Euclidean space can be embedded into $\mathbb{R}^D$ with $D = O(\epsilon^{-2}\log N)$ while distorting the pairwise distance between the points by at most $1\pm\epsilon$. Two proofs by geometric (using results from Ref.~\cite{BengtssonZyczkowskiBook2017}) and probabilistic means for the lower bound (which for us is the more interesting bound) are given in the SM. The upper bound can be obtained from Theorem 9.3 in Ref.~\cite{AlonDM2003}. Moreover, for $\epsilon \sim 1/\sqrt{D}$ the growth is bounded by a polynomial, which becomes linear when $\epsilon < 1/\sqrt{D}$~\cite{TaoBlog2013, HaasHammenMixon2017}, which follows from the Welch bound~\cite{WelchIEEE1974} or the Kabatjanskii-Levenstein bound~\cite{KabatianskyLevenshtein1978}.

%%%%%%%%%%%%%%%%%%%%%%%%%%%%%%%%%%%%%%%%%%%%%%%%%%%%%%%%%%%%%%%%%%%%%%%%%%%%%%%%%%%%%%%%%%%%%%%%%%%%%%%%%%%%%%%%%%%%%%%%
\subsection{Self-location and Records}\label{sec:SLP}
%%%%%%%%%%%%%%%%%%%%%%%%%%%%%%%%%%%%%%%%%%%%%%%%%%%%%%%%%%%%%%%%%%%%%%%%%%%%%%%%%%%%%%%%%%%%%%%%%%%%%%%%%%%%%%%%%%%%%%%%

In case of exact decoherence we could find record states $|r_j\rangle\sim|\psi_j\rangle$ at the final time that uniquely identify the history $j$ without errors (see Sec.~\ref{sec:decoherent histories}). One would hope that this is also approximately true for approximately decoherent histories---a tacit assumption underlying much work on decoherent histories, which, to the best of our knowledge, has never been scrutinized. As we will see below, this hope is shattered when decoherence scales with exponent $\alpha\le1/2$ and the number of histories $N$ becomes comparable with $d\equiv D^{2\alpha}$.\footnote{``Comparable'' here means that $N = \beta d$ for some possibly very small but $d$-independent proportionality constant $\beta$. How small $\beta$ has to be depends, of course, on the identification error that we judge tolerable.}

Note that we here dropped already the label $x$ in the notation of $N_x, D_x, G(x), \dots$ because the discussion below is identical for each final subspace $x$. We will also make the simplifying assumption for now that $N\le D<\infty$. The cases $N>D$ or $D=\infty$ will be discussed later. In short, for $N>D$ the identification of records can only become worse, and for $D=\infty$ we can restrict the discussion to the finite dimensional subspace $\mbox{span}\{|\psi_i\rangle\}$.

In our lighter notation the decomposition~(\ref{eq:histories}) of the global wave function $|\Psi\rangle$ (suppressing any time dependence) into histories reads
\begin{equation}\label{eq:self-loc 1}
	|\Psi\rangle = \sum_{j=1}^N |\psi'_j\rangle = \sum_{j=1}^N \sqrt{q_j}|\psi_j\rangle.
\end{equation}
Importantly, the $q_j$ are positive and normalized by construction of the history states, even when the $|\psi_j\rangle$ are not orthogonal (this is a non-trivial geometric feature of the histories framework: a linear combination $|\Psi\rangle = \sum_j c_j |j\rangle$ with non-orthogonal $\{|j\rangle\}$ gives $\sum_j |c_j|^2 \neq 1$ in general). Now, the goal is to find a set of orthonormal record states $\{|r_j\rangle\}$ that decomposes the global wave function such that
\begin{equation}\label{eq:self-loc 2}
	|\Psi\rangle = \sum_{j=1}^N \sqrt{q_j}|r_j\rangle
\end{equation}
with $|r_j\rangle$ as close as possible to $|\psi_j\rangle$. We call this the \emph{self-location problem} (SLP) as this is already common terminology in the philosophy of the MWI to describe the problem of how an inside observer in the multiverse identifies their branch or history. However, to the best of our knowledge, no formal and quantitative discussion of the SLP has been given so far. If instead $|\Psi\rangle$ describes the state of a quantum subsystem, the SLP remains meaningful as a way to judge how well we can distinguish the histories $|\psi_j\rangle$ by a final projective measurement of some record observable $R = \sum_j\lambda_j |r_j\rl r_j|$.

A few comments are immediately in order. (1) While we demanded in words that ``$|r_j\rangle$ [is] as close as possible to $|\psi_j\rangle$'', there is no unique mathematical figure of merit that measures closeness: we will have to make a choice later on. (2) Any figure of merit contains quantities such as $q_j$ and $\lr{\psi_j|r_k}$. Those are often interpreted probabilitistically, but the very notion of probability is a debated concept within the MWI~\cite{SaundersEtAlBook2010}. Since the figures of merit below also have a purely geometric meaning of ``closeness'' we will not further emphasize this debate here (but we return to a broader discussion of it in Sec.~\ref{sec:conclusions}). (3) In eqn~(\ref{eq:self-loc 2}) we demanded that the probability amplitudes $\sqrt{q_j}$ of the record states are the same as those of the histories, but this could be relaxed. However, we prefer to find an orthogonal decomposition of $|\Psi\rangle$ that is structurally as close as possible to the original history decomposition.

The rest of this section is structured as follows. First, we establish in Sec.~\ref{sec:sub-QSD} that the SLP can be seen as a coherent version of the quantum state discrimination (QSD) task~\cite{BarnettCrokeAOP2009, BaeKwekJPA2015, WatrousBook2018, CieslinskiEtAlPR2024}, and that it reduces to it at least asymptotically for $N\rightarrow\infty$. In Sec.~\ref{sec:sub-Haar} we then solve analytically (again for $N\rightarrow\infty$) the QSD task for a simplified model of Haar random history states. Surprisingly, the analytical results remain valid for a much broader class of non-Haar-random and pseudo-random states as we show numerically in Sec.~\ref{sec:sub-beyond}. Some final remarks are presented in Sec.~\ref{sec:sub-remarks}. To keep the presentation focused, all derivations as well as some additional analytical and numerical results are relegated to the SM. Moreover, we emphasize here that, even though many analytical results only have asymptotic validity, convergence to the asymptotic results happens very fast: already around $N\gtrsim100$ the true (numerical) results are practically indistinguishable from the asymptotic result.

%%%%%%%%%%%%%%%%%%%%%%%%%%%%%%%%%%%%%%%%%%%%%%%%%%%%%%%%%%%%%%%%%%%%%%%%%%%%%%%%%%%%%%%%%%%%%%%%%%%%%%%%%%%%%%%%%%%%%%%%
\subsubsection{Connection to QSD}\label{sec:sub-QSD}
%%%%%%%%%%%%%%%%%%%%%%%%%%%%%%%%%%%%%%%%%%%%%%%%%%%%%%%%%%%%%%%%%%%%%%%%%%%%%%%%%%%%%%%%%%%%%%%%%%%%%%%%%%%%%%%%%%%%%%%%

An elementary task of QSD starts from an ensemble of probabilities and states $\{q_j,|\psi_j\rangle\}$ and asks for the orthogonal measurement basis that optimally discriminates the states. The difference between SLP and QSD comes from the constraints that there is a coherent superposition $\sum_j \sqrt{q_j}|\psi_j\rangle$ instead of an ensemble and that the superposition of record states has to add up to $|\Psi\rangle$ as in eqn~(\ref{eq:self-loc 2}). Thus, SLP is a coherent version of QSD with an \emph{additional constraint}. Importantly, this implies that any optimized figure of merit for QSD will upper bound the same figure of merit for SLP.

Also in QSD the precise meaning of ``optimally distinguish'' depends on the chosen figure of merit. One important figure of merit is the average success probability
\begin{equation}
 \overline P_S(\{|s_j\rangle\}) \equiv \sum_j q_j |\lr{\psi_j|s_j}|^2,
\end{equation}
where $\{|s_j\rangle\}$ is an orthonormal basis used to distinguish the states $\{|\psi_j\rangle\}$ (in principle more general measurements, called POVMs, are conceivable, but for an ensemble of pure states it turns out that rank one projective measurements are optimal). Maximizing over all bases we obtain the optimal average success probability
\begin{equation}
 \overline P_S^\text{opt} \equiv \max_{\{|s_j\rangle\}} \overline P_S(\{|s_j\rangle\}).
\end{equation}
Clearly, every choice of $\{|s_j\rangle\}$ gives a lower bound on $\overline P_S^\text{opt}$, but an upper bound can be inferred from the square root (sqrt) or ``pretty good'' measurement that is constructed by setting~\cite{Holevo1979, HausladenWoottersJMO1994}
\begin{equation}
 |s_j\rangle \equiv \sum_{k=1}^N \left(\sqrt{G}^{-1}\right)_{kj}|\psi_k\rangle, ~~~ j\in\{1,2,\dots,N\}.
\end{equation}
The average success probability of the sqrt measurement $\overline P_S^\text{sqrt}$ upper bounds the optimal average success probability as $\overline P_S^\text{opt} \le(\overline P^\text{sqrt}_S)^{1/2}$~\cite{BarnumKnillJMP2002}. This upper bound (which can \emph{not} always be satisfied) provides the justification to focus below only on the sqrt measurement.

Now, the analogy between QSD and SLP suggests a strategy to solve the SLP. First, solve the QSD problem to find a good (ideally optimal) measurement basis $\{|s_j\rangle\}$. Then, set $|\Phi\rangle \equiv \sum_j\sqrt{q_j}|s_j\rangle$ and define a unitary $V$ such that $|\Psi\rangle = V|\Phi\rangle$. The sought-after record states are then $|r_j\rangle = V|s_j\rangle$. Note that there is a phase freedom in choosing the $|s_j\rangle$, and our original hope was that under phase alignment $\lr{s_j|\psi_j} \ge 0$ the fidelity $F=|\lr{\Phi|\Psi}|$ is close to one for $\overline P_S$ close to one, but in the SM we present a counterexample showing that $F=0$ is possible while having $\overline P_S\rightarrow1$ for $N\rightarrow\infty$. Since this counterexample is contrived, we conjecture that $\overline P_S \lesssim 1$ implies $F\lesssim1$ for a generic choice where the $|\psi_i\rangle$ overlap in an essentially random way as in eqn~(\ref{eq:DF scaling form}). First steps in this direction are taken in the SM. In any case, it turns out that for large $N$ and a spread out distribution $q_j$ such a result is also not necessary as the corresponding average success probability for the SLP,
\begin{equation}\label{eq:average success prob SLP}
 \overline Q_S \equiv \sum_j q_j |\lr{\psi_j|r_j}|^2,
\end{equation}
is always close to $\overline P_S$ even when $F=0$. Specifically, we could establish
\begin{equation}\label{eq:QSD SLP bound}
 \left|\overline P_S - \overline Q_S\right| \le B \sim \sqrt{(1-F)/N}.
\end{equation}
The derivation and precise form of the upper bound $B$ is given in the SM---the insight that matters here is that even for $F=0$ we find $\overline Q_S \rightarrow \overline P_S$ for $N\rightarrow\infty$.

The reason for this behavior can be understood geometrically quite easily and this argument applies to any global figure of merit. By construction the unitary rotation $V$ mapping $|\Phi\rangle$ to $|\Psi\rangle$ needs to act non-trivially only on a two-dimensional subspace $\mbox{span}\{|\Phi\rangle,|\Psi\rangle\}$. However, the subspace $\mbox{span}\{|s_j\rangle\}$ has dimension $N$, so for $N\gg2$ the unitary $V$ can affect only slightly the basis $\{|s_j\rangle\}$, thus giving a similar figure of merit for QSD and SLP.

While below we will focus exclusively on the average success probability for the sqrt measurement, we have also evaluated the mutual information between the history states $|\psi_j\rangle$ and record states $|r_j\rangle$ in the SM (which, overall, shows a similar behavior as the average success probability). Furthermore, we find it worthwhile here to communicate the following bounds in terms of the determinant of the Gram matrix, valid for uniform weights $q_j = 1/N$:
\begin{equation}\label{eq:result Andreas}
 P^\text{opt}_\text{UN} \leq \sqrt[N]{\det G} \leq \overline P_S^\text{opt}.
\end{equation}
Here, $P^\text{opt}_\text{UN}$ is the optimal worst-case success probability of unambiguous detection~\cite{CheflesPLA1998}, a variant of the QSD task above that we define precisely in the SM (where we also prove eqn~(\ref{eq:result Andreas}) using results form Refs.~\cite{SutterTomamichelHarrowIEEE2016, HoroshkoEskandariKilinPLA2019}). The point to note here is that $\det G$ equals the volume of the parallelepiped spanned by the history states $|\psi_j\rangle$. The interesting geometric insight is then that, even though the $|\psi_j\rangle$ might be pairwise almost orthogonal, the volume spanned by all of them can very rapidly shrink to zero. Since the SLP is a \emph{global} task involving the totality of all states, eqn~(\ref{eq:result Andreas}) allows one to anticipate the fact that the SLP becomes increasingly difficult for large $N$ even under conditions of approximate decoherence or almost pairwise orthogonality.

%%%%%%%%%%%%%%%%%%%%%%%%%%%%%%%%%%%%%%%%%%%%%%%%%%%%%%%%%%%%%%%%%%%%%%%%%%%%%%%%%%%%%%%%%%%%%%%%%%%%%%%%%%%%%%%%%%%%%%%%
\subsubsection{Haar random approximation}\label{sec:sub-Haar}
%%%%%%%%%%%%%%%%%%%%%%%%%%%%%%%%%%%%%%%%%%%%%%%%%%%%%%%%%%%%%%%%%%%%%%%%%%%%%%%%%%%%%%%%%%%%%%%%%%%%%%%%%%%%%%%%%%%%%%%%

Since the full statistical properties of the pseudo-random matrix $R$ in eqn~(\ref{eq:DF scaling form}) are not known, and since the evaluation of $\overline P_S$ does not only depend on the first two moments of $R$, we have to come up with an approximate but reasonable model for the history states $|\psi_j\rangle$ that statistically reproduces what is known.

In this section we consider the most random model of $N$ history states $|\psi_j\rangle$ Haar randomly drawn in $\mathbb{C}^d$ with $d\equiv D^{2\alpha}$ (which can be embedded in $\C H$ for $\alpha\le1/2$). This implies $\lr{\psi_j|\psi_k} = r/\sqrt{d}$ with a zero-mean-unit-variance complex random number $r$. We thus reproduce the correct behaviour of eqn~(\ref{eq:DF scaling form}) for the first two moments as long as $N\le d$.\footnote{For $N>d$ the set $\{|\psi_j\rangle\}$ is linearly dependent. While the pairwise overlap $\lr{\psi_j|\psi_k}$ still captures the right statistical property, the totality of vectors $\{|\psi_j\rangle\}$ might not. While our results below suggest that the average success probability will not get better in a refined model of linearly independent states $\{|\psi_j\rangle\}$, we can not strictly exclude this possibility.} Unfortunately, independently drawn Haar random states do not contain enough structure to guarantee that $|\Psi'\rangle \equiv \sum_j \sqrt{q_j}|\psi_j\rangle$ is normalized, but in the SM we will show that $\|\Psi'\|$ concentrates around the typical value $\|\Psi'\|\simeq1$.\footnote{We use $X\simeq m$ to denote concentration of a random variable $X$ around its mean value $m$ meaning that for most samples $x\sim X$ the value of $x/m$ is close to one.} Thus, only a slight modification is necessary to guarantee that the $|\psi_j\rangle$ are valid history states. Moreover, the QSD task does not care about the correct normalization of $|\Psi'\rangle$ and it appears reasonable to assume that $\overline P_S$ correctly captures the leading order behavior of the average success probability of the SLP.

Based on these assumptions we are then able to evaluate the average success probability $\overline P_S$ of the sqrt measurement analytically for the asymptotic limit $N,d\rightarrow\infty$ with the ratio $\gamma \equiv N/d$ fixed. The key insight is that the assumptions above allow us to model the Gram matrix $G$ as a member of the complex Wishart ensemble, a well-studied random matrix ensemble used to model sample covariance matrices~\cite{GoodmanAMSa1963, GoodmanAMSb1963, MarchenkoPastur1967, BaiSilversteinBook, LytovaPasturAP2009, PasturShcherbinaBook2010}. In particular, the above assumptions imply that the eigenvector distribution of $G$ is Haar random\footnote{This property is quite easy to see since in our model $G$ and $SGS^\dagger$ for an arbitary $N\times N$ unitary matrix $S$ are identically distributed. Hence, the eigenvector distribution must be invariant under any unitary transformation $S$.}, but note that the eigenvectors of $G$ (which live in $\mathbb{C}^N$) have nothing to do with the Haar random history states (which live in $\mathbb{C}^d\subset\C H$) or the actual history states (which live in $\C H = \mathbb{C}^D$).

As detailed in the SM (which makes also use of Refs.~\cite{TalagrandAP1996, HaydenShorWinterOSID2008, CollinsMatsumotoNovak2022}), the result is that the individual success probabilities concentrate around the mean value $\mu_{\text{sqrt}}(\gamma)^2$ with a standard deviation proportional to $1/\sqrt{N}$:
\begin{equation}\label{eq:PS-individual}
 |(\sqrt{G})_{jj}|^2 \simeq \mu_{\text{sqrt}}(\gamma)^2 \equiv \left(\int \sqrt{\lambda} d\rho_\text{MP}(\lambda)\right)^2.
\end{equation}
Here, $\rho_\text{MP}$ is the distribution of the asymptotic eigenvalue density of the Wishart ensemble, which is known as the Marchenko-Pastur distribution. We give a precise definition in the SM and only note that the Marchenko-Pastur distribution is analogous to Wigner's semicircle distribution, which describes the eigenvalue density for random Hermitian matrices. Moreover, the average of eqn~(\ref{eq:PS-individual}) over the weights $q_j$ will further smooth out fluctuations unless the $q_j$ are peaked around a few histories only (effectively corresponding to a smaller $N$) or the $q_j$ are maliciously correlated with the $|(\sqrt{G})_{jj}|^2$.\footnote{For instance, for $q_j=1/N$ the conventional CLT decreases the standard deviation of $\overline P_S$ by another factor $1/\sqrt{N}$.} Thus, we arrive at the central result of this section:
\begin{equation}\label{eq:p success result}
 \overline P_S^\text{opt} \le \sqrt{\overline P_S^\text{srqt}} \simeq \mu_{\text{sqrt}}(\gamma).
\end{equation}
In words, the optimal average success probability is upper bounded by a function that strongly concentrates around $\mu_{\text{sqrt}}(\gamma)$. This result is plotted and numerically confirmed for finite $N,d$ in Fig.~\ref{fig sqrt meas}. In particular, the function $\mu_{\text{sqrt}}(\gamma)^2$ interpolates almost linearly between $\mu_{\text{sqrt}}(0)^2 = 1$ and $\mu_{\text{sqrt}}(1)^2 = [8/(3\pi)]^2 \approx 0.72$.

The conclusion we can draw from this section is the following. Even when the Gram matrix in eqn~(\ref{eq:DF scaling form}) becomes asymptotically for $D\rightarrow\infty$ diagonal (corresponding to exact asymptotic decoherence), the globally optimized record basis can identify the histories \emph{at most} with success probability $8/(3\pi) \approx 0.85$ if there are $N = d = D^{2\alpha}$ many histories (which for any $\alpha<1/2$ is negligible compared to  $D$). While it might be possible to find that globally optimized record basis for a small quantum subsystem, it is likely an overly optimistic assumption when applied to the MWI where observers can only access coarse and local properties of the Universe. This leads to a branch selection problem as we will discuss in Sec.~\ref{sec:conclusions}.

\begin{figure}[t]
	\centering\includegraphics[width=0.48\textwidth,clip=true]{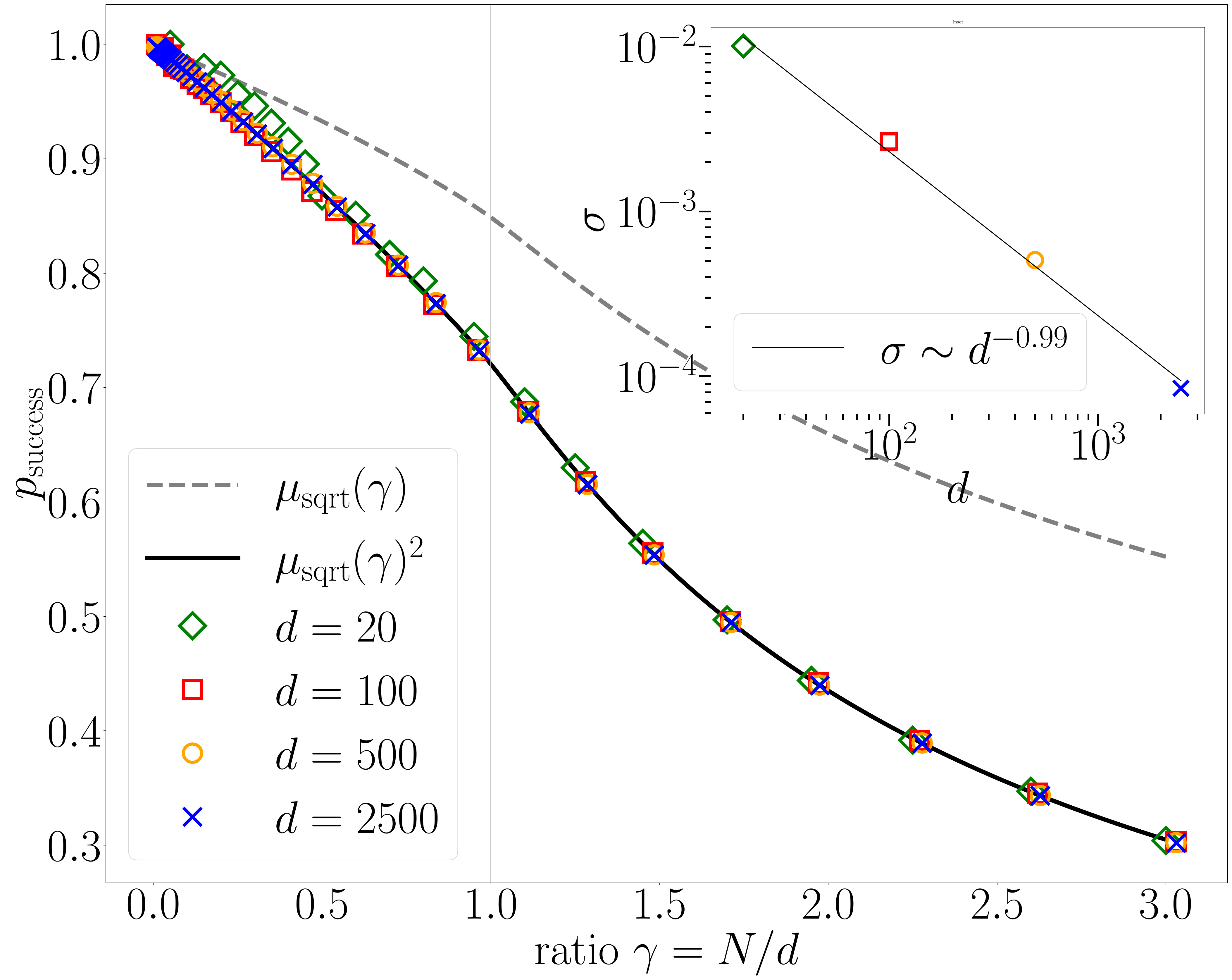}
	\caption{Plot of the success probability for the sqrt measurement as a function of $\gamma$ for $d=20$ (green diamonds), $d=100$ (red squares), $d=500$ (orange circles) and $d=2500$ (blue crosses), numerically evaluated for a single realization of the Gram matrix $G_{jk} = \lr{\psi_j|\psi_k}$ and for weights $q_j=1/N$. The black solid line is the analytical result of eqn~(\ref{eq:PS-individual}), the gray dashed line the upper bound of eqn~(\ref{eq:p success result}). Note that in our setting the plot makes only sense for $\gamma\le1$, but we found it illuminating to show that eqn~(\ref{eq:PS-individual}) remains valid also for $\gamma>1$. Inset: Empirical standard deviation $\sigma$ of the success probability as a function of $d$, averaged over all values of $\gamma$ and obtained from ten realizations of the Gram matrix. We see the expected suppression $\sigma\sim d^{-1} \sim N^{-1}$. }
	\label{fig sqrt meas}
\end{figure}

%%%%%%%%%%%%%%%%%%%%%%%%%%%%%%%%%%%%%%%%%%%%%%%%%%%%%%%%%%%%%%%%%%%%%%%%%%%%%%%%%%%%%%%%%%%%%%%%%%%%%%%%%%%%%%%%%%%%%%%%
\subsubsection{Beyond Haar Randomness}\label{sec:sub-beyond}
%%%%%%%%%%%%%%%%%%%%%%%%%%%%%%%%%%%%%%%%%%%%%%%%%%%%%%%%%%%%%%%%%%%%%%%%%%%%%%%%%%%%%%%%%%%%%%%%%%%%%%%%%%%%%%%%%%%%%%%%

Our choice to study Haar random history states might be conceived as too simplistic and unrealistic. In defense, one should note that we only assumed history states conditioned on interesting nonequilibrium histories $\bs x$ to look Haar random in a (for $\alpha<1/2$) very small subspace of dimension $d=D^{2\alpha}$, which is \emph{not} the same as assuming a Haar random global wave function $|\Psi\rangle$. Be that as it may be, it is certainly true that other much less random history states can also statistically reproduce eqn~(\ref{eq:DF scaling form}), and here we consider three possible alternatives. We first introduce them together with numerical evidence showing that they give the \emph{same} result as before. Then, we discuss why this is the case.

For the first class of states we start with a Haar random state $|\psi_1\rangle\in\mathbb{C}^d$ as a seed and construct $|\psi_i\rangle$ for $i>1$ by randomly permuting the entries of $|\psi_1\rangle$ in some computational basis~\cite{NicaPJP1993, BordenaveCollinsAM2019}. Since there are $d!$ many permutations, this is (for a fixed seed) a discrete distribution of random states instead of a continuous Haar distribution. Now, for a typical Haar random state one expects for large $d$ that $|\lr{\psi_i|\psi_j}|^2 \approx 1/d$ such that we find statistics in unison with eqn~(\ref{eq:DF scaling form}) for $d=D^{2\alpha}$ as before.\footnote{There are $(d-n)!$ many permutations that leave $n$ entries of $|\psi_1\rangle$ fixed. The probability to apply such a permutation is $(d-n)!/d!$, which is negligible small for large $d$. Since the entries of $|\psi_1\rangle$ are independently sampled, apart from the normalization constraint, $\lr{\psi_i|\psi_j}$ will be close to the scalar product of two Haar random states in $\mathbb{C}^d$. Of course, this argument assumes that the initial seed $|\psi_1\rangle$ was not atypical.} Figure~\ref{fig P comp}(a) shows that the average success probability quickly converges to the analytical result for Haar random states [eqn~(\ref{eq:p success result})].

In our second example we consider pseudo-random states as studied in Refs.~\cite{JiLiuSongAC2018, BoulandFeffermanVaziraniArXiv2019, AaronsonEtAlArXiv2022}. Here, for some computational basis $\{|x\rangle\}_d$ one sets $|\psi_j\rangle = \sum_x(-1)^{f_j(x)}|x\rangle/\sqrt{d}$, where $f_j(x)\in\{0,1\}$ is an efficiently computable function that is indistinguishable from a true random number generator to any polynomial-time quantum algorithm. For our example we take $f_j(x)$ to be the digit $jd+x$ of $\pi$ in binary representation.\footnote{Our results do not depend on the choice of $\pi$: the set of normal numbers, which are statistically indistinguishable from a random realization of a fair coin flip experiment, has measure one.} Since $R \equiv \sum_k (-1)^{f_k}$ behaves like the sum of $d$ truly random numbers $\pm1$, we know that $R$ has mean zero and fluctuations of order $\sqrt{d}$, implying that $|\lr{\psi_i|\psi_j}|^2 \approx 1/d$ as required. Of course, in this example we have $\lr{\psi_i|\psi_j}\in\mathbb{R}$ by construction, whereas the Gram matrix in eqn~(\ref{eq:DF scaling form}) has in general complex entries. Nevertheless, as Fig.~\ref{fig P comp}(b) demonstrates, also here we find the same average success probability as for the Haar random case.

Finally, we consider 2-designs. A $2$-design is a set of states $|\psi\rangle$, whose uniform expectation value of any homogenous degree-$2$ polynomial $f(\psi,\psi^*)$ is identical to the Haar average: $\mathbb{E}_\text{2-design}[f(\psi,\psi^*)] = \mathbb{E}_\text{Haar}[f(\psi,\psi^*)]$. This immediately implies that we get the right scaling: $\mathbb{E}_\text{2-design}[\lr{\psi_i|\psi_j}] = 1/\sqrt{d}$. Specifically, to construct a 2-design we consider $d$ to be prime such that we can use mutually unbiased bases~\cite{KlappeneckerRottelerIEEE2005}. Note that there are (up to a global unitary rotation) $d(d+1)$ many 2-design states such that this set is significantly smaller than the first set of randomly permuted states. In particular, one finds $\lr{\psi_j|\psi_k} = e^{i\phi_{jk}}/\sqrt{d}$ where only the phase $\phi_{jk}$ randomly changes for different realizations $j,k$. Nevertheless, Fig.~\ref{fig P comp}(c) again demonstrates that we asymptotically converge to the same average success probability, even though the smallness of the sample space is now clearly visible in form of a slower convergence. Of course, the largest considered dimension of $d=499$ in Fig.~\ref{fig P comp}(c) is minuscule compared to realistic many-body systems.

\begin{figure}[t]
 \centering\includegraphics[width=0.48\textwidth,clip=true]{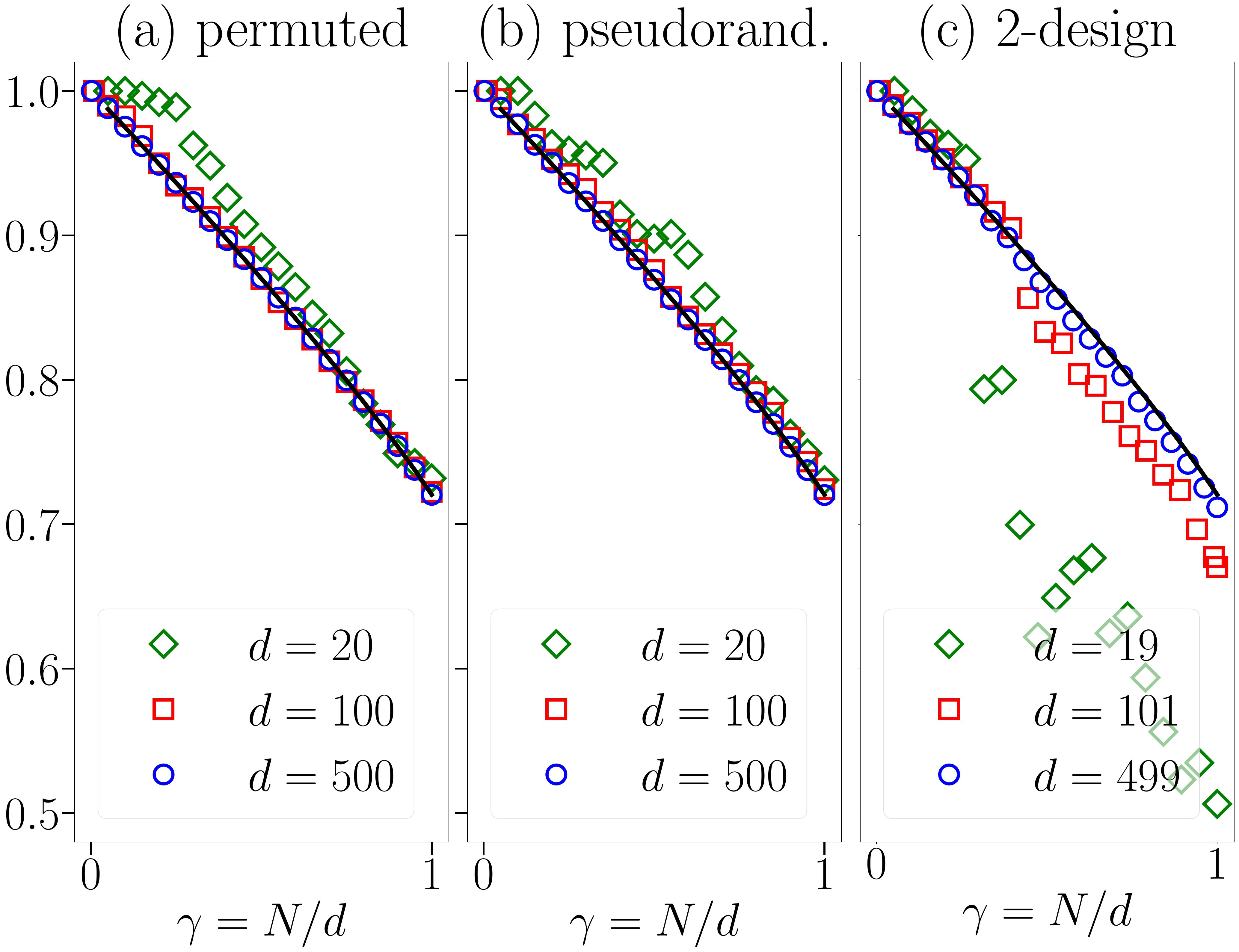}
 \caption{Plot of the average success probability $\overline P_S$ as a function of the ratio $\gamma$ for three different classes of states (a), (b) and (c) as described in the main text. Plots are done for a single realization of the Gram matrix $G$ for different Hilbert space dimensions $d$. As for the Haar random case, we observed measure concentration for large $d$ (not shown here for brevity). The black solid line denotes the asymptotic analytical expression $\mu_\text{sqrt}(\gamma)^2$.}
 \label{fig P comp}
\end{figure}

It remains to explain why the result in eqn~(\ref{eq:p success result}) is so robust, which has also surprised the authors. For instance, the success probability $|\sqrt{G}_{jj}|^2$ can not be well approximated by a degree-2 polynomial such that one would not expect a 2-design to give the same result as the Haar random case. However, recall that a key ingredient to prove eqn~(\ref{eq:PS-individual}) is the property that the \emph{eigenbasis of} $G$ is Haar distributed, but not that the history states $|\psi_j\rangle$ are themselves Haar distributed. Now, recent mathematical results suggest that this property holds asymptotically true under substantially generalized conditions~\cite{SilversteinJMA1989, BruJMA1989, SilversteinAP1990, BaiMiaoPanAP2007, TaoVuRMTA2012, XiaQinBaiAS2013, PillaiYinAAP2014, AlexEtAlEJP2014, BourgadeYauCMP2017, DingAAP2019, XiYangYinAS2020} (for the eigenvalue statistics this was known since the original paper by Marchenko and Pastur~\cite{MarchenkoPastur1967}). Essentially, the required conditions are that the $\{|\psi_j\rangle\}$ are isotropically distributed and that its higher moments are controlled, i.e., the statistics of the $|\psi_j\rangle$ is not dominated by rare events. We give a slightly more precise statement in the SM, but the technical details of Refs.~\cite{SilversteinJMA1989, BruJMA1989, SilversteinAP1990, BaiMiaoPanAP2007, TaoVuRMTA2012, XiaQinBaiAS2013, PillaiYinAAP2014, AlexEtAlEJP2014, BourgadeYauCMP2017, DingAAP2019, XiYangYinAS2020} are involved and differ substantially. Thus, we keep it with the numerical evidence presented here, which clearly demonstrates that the key result~(\ref{eq:p success result}) is more universal than the idealized Haar random model of Sec.~\ref{sec:sub-Haar} \emph{prima facie} suggests.

%%%%%%%%%%%%%%%%%%%%%%%%%%%%%%%%%%%%%%%%%%%%%%%%%%%%%%%%%%%%%%%%%%%%%%%%%%%%%%%%%%%%%%%%%%%%%%%%%%%%%%%%%%%%%%%%%%%%%%%%
\subsubsection{Final remarks}\label{sec:sub-remarks}
%%%%%%%%%%%%%%%%%%%%%%%%%%%%%%%%%%%%%%%%%%%%%%%%%%%%%%%%%%%%%%%%%%%%%%%%%%%%%%%%%%%%%%%%%%%%%%%%%%%%%%%%%%%%%%%%%%%%%%%%

We briefly comment on the cases $N>D$ and $D=\infty$.

The first case is interesting because we have seen in Sec.~\ref{sec:end} that $N\gg D$ is possible without jeopardizing approximate decoherence. However, since there can be at most $D$ many orthogonal record states, it is clear that we can no longer distinguish all histories, not even approximately. Any figure of merit can thus only decrease. How it decreases quantitatively is a complicated question and we consider some examples in the SM in context of the mutual information between history and record states.

The previous problem can, of course, not happen when $D=\infty$ and many of the above results remain valid. To be in unison with eqn~(\ref{eq:DF scaling form}), we can use an ansatz of the form $G = I + R/\sqrt{d}$ to \emph{define} an effective, finite dimension $d$. Furthermore, the subspace $\mbox{span}\{|\psi_i\rangle\}_N$ spanned by all history states is always finite for any $N\in\mathbb{N}$. The above results thus apply unchanged to the case $\gamma = N/d \le1$. For $\gamma>1$ we can not exclude that some figure of merit might increase, but given the previous results we would find this counterintuitive.

%%%%%%%%%%%%%%%%%%%%%%%%%%%%%%%%%%%%%%%%%%%%%%%%%%%%%%%%%%%%%%%%%%%%%%%%%%%%%%%%%%%%%%%%%%%%%%%%%%%%%%%%%%%%%%%%%%%%%%%%
\section{Structure of Decoherence}\label{sec:structure}
%%%%%%%%%%%%%%%%%%%%%%%%%%%%%%%%%%%%%%%%%%%%%%%%%%%%%%%%%%%%%%%%%%%%%%%%%%%%%%%%%%%%%%%%%%%%%%%%%%%%%%%%%%%%%%%%%%%%%%%%

We now turn to a numerical study of very long histories such that the number of histories $N = M^L$ greatly exceeds the Hilbert space dimension $D$. We will see that this causes a strong deviation from the ansatz in eqn~(\ref{eq:DF scaling form}), where each history pair was approximately equally decoherent apart from statistical fluctuatons. Instead, a clear decoherence \emph{structure} develops, with some histories remaining decoherent while others become \emph{recoherent}.

We proceed as follows. In Sec.~\ref{sec:model} we specify the model, fix the parameters used in the numerical simulation, and justify it (for similar models see Refs.~\cite{PereyraJSP1991, EspositoGaspardPRE2003b, LebowitzPasturJPA2004, GorinEtAlNJP2008, BartschSteinigewegGemmerPRE2008, GenwayHoLeePRL2013, RieraCampenySanperaStrasbergPRXQ2021, AlbrechtBaunachArrasmithPRD2022, YanZurekNJP2022, DasGhoshJSM2022, StrasbergSP2023, StrasbergReinhardSchindlerPRX2024, StrasbergSchindlerSP2024}). Afterwards, we only highlight two numerical insights, many more, also with respect to the localization of histories and the purity of the Petz recovered state~\cite{PetzCMP1986, PetzQJM1988, JungeEtAlAHP2018, ParzygnatBuscemiQuantum2023, BaiBuscemiScaraniPRL2025}, are presented in the SM. First, in Sec.~\ref{sec:Hamming} we present evidence for the slowly varying envelope function $F$ in the ansatz~(\ref{eq:DHH}) (not to be confused with the fidelity $F$ from the previous section). Second, we show that the structure of decoherence correlates with Born's rule in Sec.~\ref{sec:Born}.

%%%%%%%%%%%%%%%%%%%%%%%%%%%%%%%%%%%%%%%%%%%%%%%%%%%%%%%%%%%%%%%%%%%%%%%%%%%%%%%%%%%%%%%%%%%%%%%%%%%%%%%%%%%%%%%%%%%%%%%%
\subsection{Model and Justification}\label{sec:model}
%%%%%%%%%%%%%%%%%%%%%%%%%%%%%%%%%%%%%%%%%%%%%%%%%%%%%%%%%%%%%%%%%%%%%%%%%%%%%%%%%%%%%%%%%%%%%%%%%%%%%%%%%%%%%%%%%%%%%%%%

We consider a coarse-graining $\{\Pi_0,\Pi_1\}$ with two projectors of rank $D_0$ and $D_1$, respectively, implying a total Hilbert space dimension $D=D_0+D_1$. The total Hamiltonian $H$ is decomposed into blocks as
\begin{equation}\label{eq H blocks}
 H = \sum_{i,j} \Pi_i H\Pi_j
   = \left(\begin{array}{cc}
            H_{00} & H_{01} \\
            H_{10} & H_{11} \\
           \end{array}\right).
\end{equation}
We set $H_{01} = H_{10}^\dagger = \lambda R$, where $\lambda$ is a coupling strength and $R$ a $D_0\times D_1$  random matrix with entries $\{\pm 1\}$ chosen with equal probability. Furthermore, $H_{00}$ ($H_{11}$) are taken to be diagonal with $D_0$ ($D_1$) evenly spaced entries in an interval of size $\delta\epsilon$.

Numerically, we set $\delta\epsilon=0.5$ and $\lambda = \delta\epsilon/(15\sqrt{D})$, which guarantees that we are in the weak coupling regime~\cite{BartschSteinigewegGemmerPRE2008}. It is convenient to measure time in units of a characteristic nonequilibrium time scale $\tau = \delta\epsilon/(2\pi\lambda^2D)$ (see Fig.~\ref{fig averages}). We choose $D_1 = 2D_0$ throughout. Since the NDF $G$ is diagonal with respect to the last projection [see eqn~(\ref{eq:DF constraint})], we require that all histories end up in the same final subspace $x_L = 0$ and hence there are $2^{L-1}$ different histories $\bs x$ for a given length $L$ of the histories.

We have studied two cases corresponding to different initial states and time scales. The first case considers a Haar random initial state confined to the subspace $\C H_1$ and a long time scale $\Delta t = 8\tau$ giving rise to ``equilibrium histories''. The second case considers a randomly chosen initial energy eigenstate $|\Psi_0\rangle = |E_n\rangle$ for a short time scale $\Delta t=\tau/2$ giving rise to ``nonequilibrium histories''.

\begin{figure}[t]
 \centering\includegraphics[width=0.45\textwidth,clip=true]{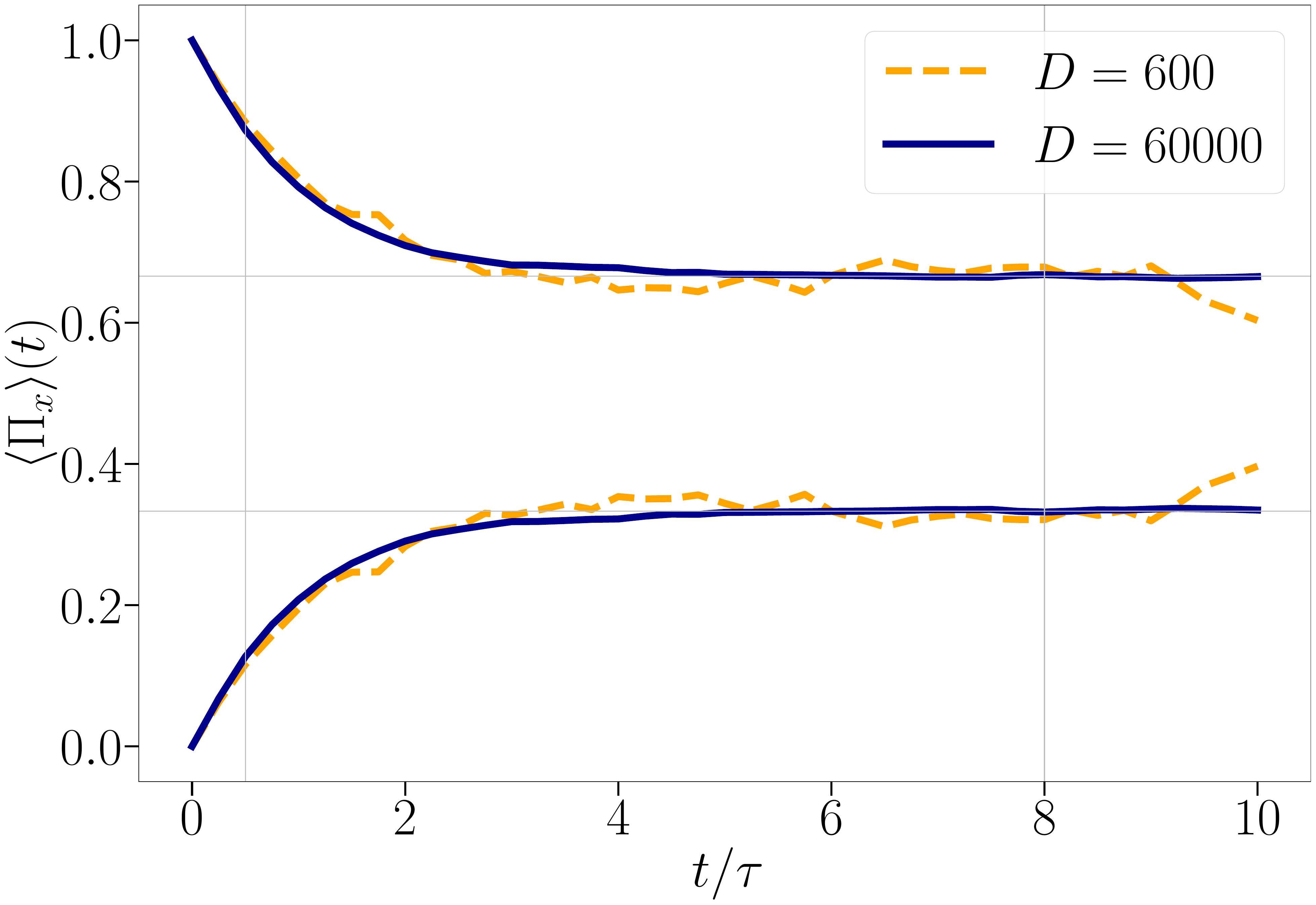}
 \caption{Time evolution (in units of $\tau$) of the probabilities to find the system in subspace $\C H_0$ (lower curves) or $\C H_1$ (upper curves) for $D\in\{600,60000\}$ and a Haar random initial state confined to subspace $\C H_1$. An exponential relaxation to the equilibrium probabilities (horizontal lines at $D_0/D = 1/3$ and $D_1/D = 2/3)$ is visible with deviations that are suppressed with increasing $D$. The vertical lines indicate the considered nonequilibrium and equilibrium time scales $\Delta t = \tau/2$ and $\Delta t = 8\tau$, respectively. }
 \label{fig averages}
\end{figure}

Since our model is seldom used in the decoherence literature, we want to briefly justify it. The starting point is any Hamiltonian $H$ restricted to a microcanonical subspace, which is spanned by all energy eigenstates lying in a sufficiently small but finite energy interval (and which, besides energy, might also include other globally conserved quantities such as particle number). Since histories belonging to different conserved quantities never interfere, i.e., they are exactly decoherent, this can be done without loss of generality. Then, surely, the simplest non-trivial histories require at least two different projectors, $\Pi_0$ and $\Pi_1$, whatever their physical meaning is, in that subspace. The decomposition~(\ref{eq H blocks}) then entails no further assumption. Moreover, for a suitable choice of basis $H_{00}$ and $H_{11}$ can be always chosen diagonal.

To justify the choice of random matrices for the off-diagonal interaction terms $H_{01} = H_{10}^\dagger = \lambda R$, we recall that, since the seminal discovery of Wigner, random matrices were successfully used to capture generic aspects of complex quantum systems in a wide variety of situations~\cite{Wigner1967, BrodyEtAlRMP1981, BeenakkerRMP1997, GuhrMuellerGroelingWeidenmuellerPR1998, HaakeBook2010, BorgonoviEtAlPR2016, DAlessioEtAlAP2016, DeutschRPP2018}. In particular, the ETH (see also Sec.~\ref{sec:approximate decoherence}) suggests that crucial aspects of generic non-integrable quantum many-body systems can be understood with random matrix theory~\cite{BorgonoviEtAlPR2016, DAlessioEtAlAP2016, DeutschRPP2018, DeutschPRA1991, ReimannDabelowPRE2021}. While using a random matrix $R$ without any further structure (which encodes, e.g., locality of interactions) is unrealistic, the surprisingly often justified hope is that it nevertheless captures essential aspects of the problem. Indeed, for short histories (up to $L=5$) numerical calculations for quantum spin chains in Ref.~\cite{WangStrasbergPRL2025} have confirmed the conclusions obtained from random matrix theory in a similar model like the one considered here~\cite{StrasbergReinhardSchindlerPRX2024}. Thus, apart from considering a non-relativistic model, our key assumption is that random matrix theory captures generic effects of decoherence also for long histories.

It should be also noted that, as long as the off-diagonal interactions are modeled by random matrix theory, the precise eigenvalue distribution of $H_{00}$ and $H_{11}$ is irrelevant as long as the eigenvalues are chosen independently from $R$ (no fine-tuning) and describe roughly the same macroscopic energy. We further note that the open quantum system Hamiltonian $H = \omega\sigma_z/2 + \lambda\sigma_x B + H_B$, describing a spin coupled to a bath, can be mapped to the above model for $\omega=0$. This is accomplished by taking $\Pi_0 = |0\rl0|$ and $\Pi_1 = |1\rl1|$ to be projectors onto the eigenbasis of $\sigma_z$, and by identifying $\Pi_0H_B\Pi_0 \simeq H_{00}$, $\Pi_1H_B\Pi_1 \simeq H_{11}$ and $\Pi_0B\Pi_1 \simeq R$.

Next, we justify the focus on the weak coupling regime of small $\lambda$. This choice ensures that the dynamics of $\lr{\Pi_x}(t) = \lr{\psi_t|\Pi_x|\psi_t}$ is slow compared to microscopic time scales of the quantum system, which---for the simple setup considered here---leads to the typical Markovian exponential relaxation to equilibrium as seen in Fig.~\ref{fig averages}. Since human perception is limited to slowly evolving quantities, this ensures that we look at the right physical regime. Note that such slow processes are commonly well described by master equations and related tools of statistical physics~\cite{VanKampenPhys1954, StrasbergEtAlPRA2023}, and in the decoherent histories literature such quantities are also known as ``quasi-conserved'' because they almost commute with the Hamiltonian. For those quantities we can hope that the here observed behavior is generic. For short histories with $L=10$ Fig.~\ref{fig:dec intro} displays the decoherence properties of the system by computing $G(x_L=0)$. The numerics confirm well the ansatz~(\ref{eq:DF scaling form}) of Sec.~\ref{sec:approximate decoherence} for short histories.

\begin{figure}[t]
	\centering\includegraphics[width=0.40\textwidth,clip=true]{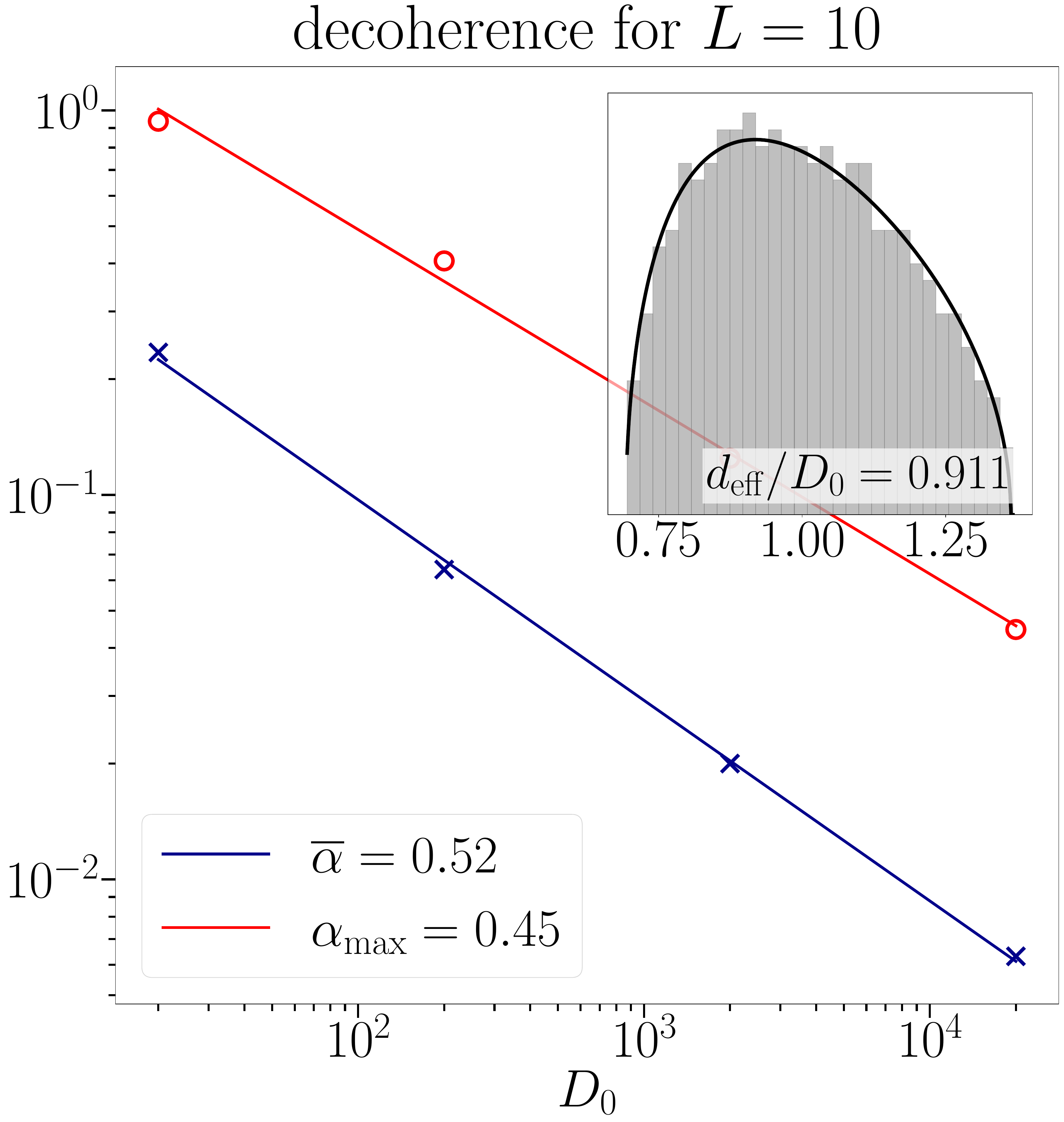}
	\caption{For a Haar random initial state confined to $\C H_1$ and for long equilibrium time steps $\Delta t=8\tau$, we show on a double logarithmic scale the average decoherence [blue crosses, obtained as a mean of the non-zero off-diagonal terms of $G_{\bs x,\bs x'}(x_L=0)$, see SM for more details], and the maximum coherence $\max_{\bs x\neq\bs x'}G_{\bs x,\bs x'}(x_L=0)$ (red crosses) as a function of the final subspace dimension $D_0$. The solid blue (red) line fits a scaling law of the form $D_0^{-\overline\alpha}$ ($D_0^{-\alpha_\text{max}}$) with scaling exponents displayed in the legend. The inset shows for $D=60000$ a histogram of the eigenvalue distribution of $G(x_L=0)$ and the black solid line fits the Marchenko-Pastur distribution, which is the eigenvalue distribution of a random, Wishart distributed DF (see SM for details). The fit is obtained through a least square optimization with the ratio $\gamma=N/d_\text{eff}$ as a free parameter. The extracted effective dimension $d_\text{eff}$ is quite close to the actual dimension $D_0$. Overall, this confirms well our ansatz in eqn~(\ref{eq:DF scaling form}).}
	\label{fig:dec intro}
\end{figure}

We also motivate the choice of initial states. First, a Haar random initial state corresponds to the most unbiased choice akin to a maximum entropy principle for pure states: it does not contain any further assumptions apart from those contained in the history. For instance, in a more realistic cosmological context (which is not at all captured in the present model) this corresponds to regarding all wave functions giving rise to a Big Bang as currently observed as equally likely. Instead, choosing an initial energy eigenstate mimics the boundary condition of the Wheeler-DeWitt equation $H|\Psi\rangle = 0$~\cite{DeWittPR1967}.
%No other feature of the Wheeler-DeWitt equation is captured in our non-relativistic model, but the hope is that it still illustrates some ideas in a minimal way (similar, e.g., to the idea that the non-relativistic Page-Wootters quantum clock~\cite{PageWoottersPRD1983} captures some aspects of the emergence of time in the semiclassical regime of the full Wheeler-DeWitt equation~\cite{KieferInBook2005}).

We conclude with some remarks on the numerical evaluation. First, numerical results are reported for a \emph{single} realization of the Haar random initial state and the random matrix interaction (no ensemble averages). We checked, however, that the reported results are typical: different runs give rise to similar quantitative results. Furthermore, we explicitly checked that several assumptions can be relaxed without changing our main results: (1) The blocks $H_{00}$ and $H_{11}$ do not require evenly spaced entries, but any distribution sufficiently smeared out over $\delta\epsilon$ suffices. (2) Choosing Gaussian real or complex random matrices $R$ also does not change the results. (3) We checked for various $O(1)$ changes of $D_0/D_1$, $\lambda$ or $\Delta t$ that our basic conclusions are robust. Finally, we relied on exact numerical diagonalization with Mathematica~\cite{Mathematica14}, which ensures high accuracy and guarantees no violation of unitarity up to numerical precision. Calculations were performed on a computer with 256 GB of memory and 16 physical kernels. They took around a week to finish to generate all the results (including those in the SM).

%%%%%%%%%%%%%%%%%%%%%%%%%%%%%%%%%%%%%%%%%%%%%%%%%%%%%%%%%%%%%%%%%%%%%%%%%%%%%%%%%%%%%%%%%%%%%%%%%%%%%%%%%%%%%%%%%%%%%%%%
\subsection{Decoherence and Hamming Distance}\label{sec:Hamming}
%%%%%%%%%%%%%%%%%%%%%%%%%%%%%%%%%%%%%%%%%%%%%%%%%%%%%%%%%%%%%%%%%%%%%%%%%%%%%%%%%%%%%%%%%%%%%%%%%%%%%%%%%%%%%%%%%%%%%%%%

The basic ansatz~(\ref{eq:DF scaling form}) describes structureless decoherence with random fluctuations, but we remarked in Sec.~\ref{sec:approximate decoherence} that, akin to the ETH, some structure could be included using the refined ansatz~(\ref{eq:DHH}). This refinement describes the idea that there is some kind of metric on the space of histories $\bs x$ that correlates with the amount of (de)coherence between the states $|\psi(\bs x)\rangle$. Yet, no such metric has ever been postulated or found.

Numerically, we found that the Hamming distance
\begin{equation}
 \Delta_{\bs x,\bs x'} = \sum_{i=1}^{L-1} |x_i-x'_i|
\end{equation}
provides a reasonable candidate for such a metric: the closer two histories are in terms of the Hamming distance, the more likely it is that they are coherent. This is illustrated in Fig.~\ref{fig dec Ham} (preliminary results pointing into this direction can be also found in Fig.~6 of Ref.~\cite{StrasbergReinhardSchindlerPRX2024}) for $L=25$ and $L=100$.  Since the number of histories $N=2^L$ is way too large to store all states $|\psi(\bs x)\rangle$, we only consider a randomly chosen subset $\C S$ of 1000 histories $\bs x$ (how $\C S$ is precisely sampled is discussed in the SM). We then consider the restricted NDF $G_{\bs x,\bs x'}(\C S) = \lr{\psi(\bs x)|\psi(\bs x')}$ with $\bs x,\bs x' \in \C S$ only.

The red line connecting the white disks in Fig.~\ref{fig dec Ham} then describes the smooth envelope function $F(\Delta_{\bs x,\bs x'})$ of eqn~(\ref{eq:DHH}). Importantly, $F$ describes an average tendency: the correlation between decoherence and the Hamming distance remains statistical in nature and a short Hamming distance does not necessarily imply coherence. Nevertheless, we observe a clear tendency (with $F$ varying more than one order of magnitude), which, surprisingly, becomes even more pronounced for larger $D$.

\begin{figure}[tb]
	\centering\includegraphics[width=0.48\textwidth,clip=true]{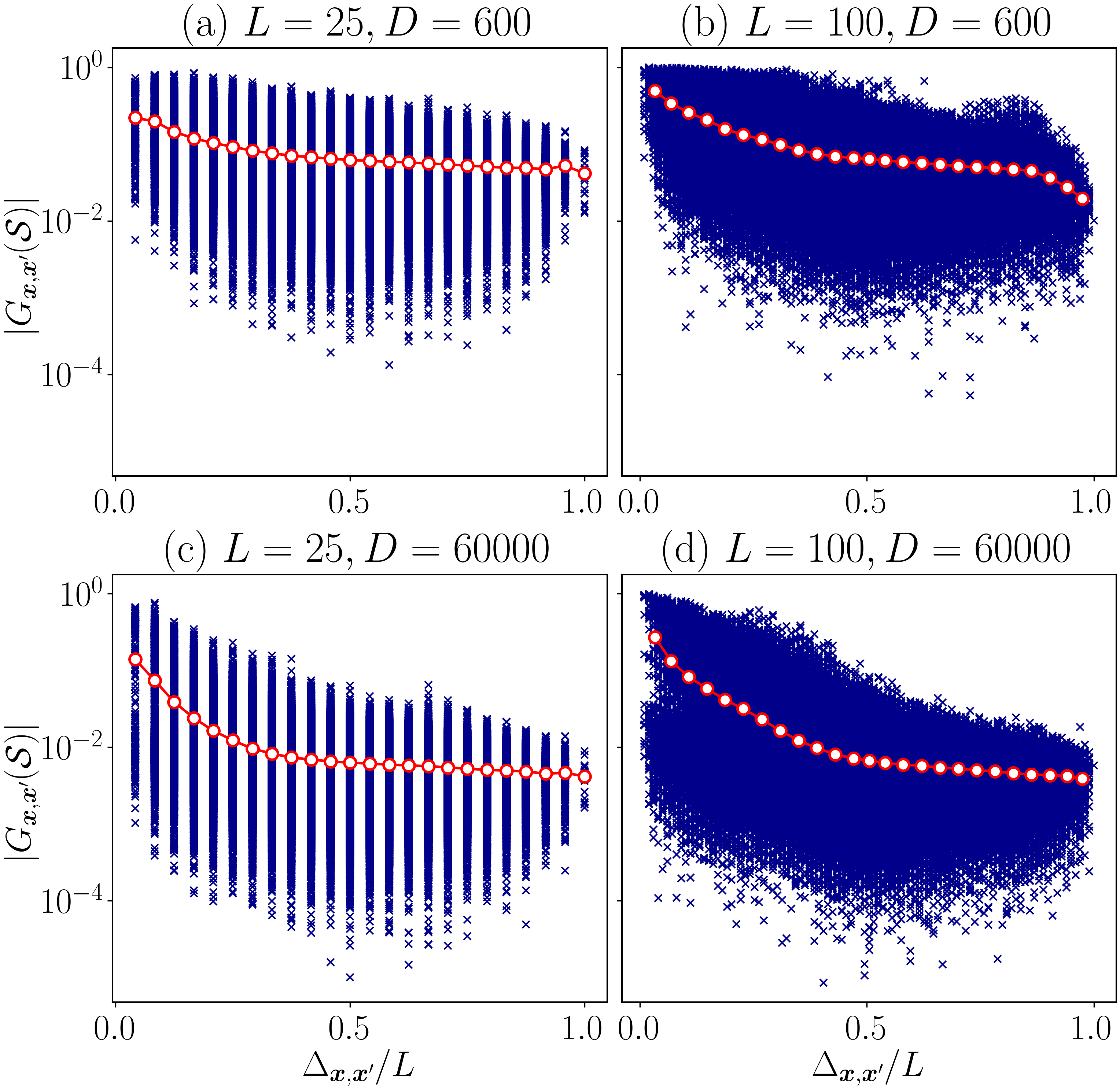}
	\caption{Plot of the off-diagonal elements of the NDF (on a shared logarithmic scale) as a function of the (normalized) Hamming distance for different $L$ and $D$. The white disks correspond to a sample average by dividing the $x$-axis into 25 bins (red lines are a guide for the eye). Results are displayed for a Haar random initial state confined to subspace $\C H_1$ and $\Delta t=8\tau$.}
	\label{fig dec Ham}
\end{figure}

The results in Fig.~\ref{fig dec Ham} were obtained for a Haar random initial state confined to subspace $\C H_1$ and for a long equilibrium time step $\Delta t=8\tau$. The case of a global energy eigenstate and a nonequilibrium time step $\Delta t=\tau/2$ is diplayed in the SM and shows weaker signs. We still lack a full understanding of what drives the correlations between (de)coherence of histories, and it would be interesting to obtain results from other models.

It is also interesting to speculate whether it is the full Hamming distance that best characterizes the decoherence among states, or whether, perhaps, only the last few labels of $\bs x$ matter (say, $x_{L-1},\dots,x_k$ for some $k>1$), somewhat echoing an idea of a Markovian decoherence structure. However, we did not find any strong evidence for this conjecture (results not shown here for brevity).

%%%%%%%%%%%%%%%%%%%%%%%%%%%%%%%%%%%%%%%%%%%%%%%%%%%%%%%%%%%%%%%%%%%%%%%%%%%%%%%%%%%%%%%%%%%%%%%%%%%%%%%%%%%%%%%%%%%%%%%%
\subsection{Emergence of Born's rule}\label{sec:Born}
%%%%%%%%%%%%%%%%%%%%%%%%%%%%%%%%%%%%%%%%%%%%%%%%%%%%%%%%%%%%%%%%%%%%%%%%%%%%%%%%%%%%%%%%%%%%%%%%%%%%%%%%%%%%%%%%%%%%%%%%

A second, central numerical finding concerns \emph{inhomogeneous} histories, which arise from further coarse-graining or lumping of the elementary or homogeneous history states $|\psi'(\bs x)\rangle$. Specifically, we consider
\begin{equation}\label{eq inhomo history}
	|\psi'(n)\rangle = \sum_{\bs x} \delta_{n,\#_1(\bs x)} |\psi'(\bs x)\rangle,
\end{equation}
where $\#_1(\bs x)$ denotes the number of ``1s'' in the history $\bs x$, e.g., $n = \#_1(\bs x) = 3$ for $\bs x = 0110010$. Thus, $|\psi'(n)\rangle$ is a coherent superposition of all histories $\bs x$ that have $n$ ``1s''. Inhomogeneous histories play an important role in quantum cosmology since there is no absolute scale of, e.g., space and time. Without external reference system, it does not make sense to ask whether a particle passes through the space interval $[x_a,x_b]$, but rather whether it passes through any space interval of length $L=x_b-x_a$. For further discussion see Ref.~\cite{HartlePRD1991}. Obviously, $|\psi'(n)\rangle$ is an abstract object, not related to a particle passing through some spacetime region, but we will nevertheless find an intriguing relation to Born's rule. Moreover, to the best of our knowledge, inhomogeneous histories have never been evaluated before apart from the unpublished Ref.~\cite{StrasbergSchindlerArXiv2023}, whose results we extend here.

We briefly summarize some properties of inhomogeneous histories. We have $|\Psi_L\rangle = \sum_{n=0}^L |\psi'(n)\rangle$ and the weights $q(n) \equiv \lr{\psi'(n)|\psi'(n)}$ add up to one. Orthogonality in form of $\lr{\psi'(n)|\psi'(n')} = 0$ for $n\neq n'$ implies again the existence of records, but it does not imply decoherence of the elementary histories $|\psi'(\bs x)\rangle$. However, if the $|\psi'(\bs x)\rangle$ are decoherent, then so are the $|\psi'(n)\rangle$. As before, we consider the NDF or Gram matrix $G_{n,n'} = \lr{\psi(n)|\psi(n')}$ of the normalized inhomogeneous history states $|\psi(n)\rangle$. Furthermore, there is an advantage for inhomogeneous histories: since there are only $L+1$ many of them, we can track \emph{all} of them for very large $L$ (in contrast, in Sec.~\ref{sec:Hamming} we were restricted to an exponentially small sample set $\C S$ of histories $\bs x$).

To add some value to the already existing results in Ref.~\cite{StrasbergSchindlerArXiv2023}, which considered Haar random initial states for long equilibrium time steps $\Delta t$ and a slightly different model, we here consider histories defined for a short nonequilibrium time scale $\Delta t=\tau/2$ and a random initial energy eigenstate $|E_i\rangle$ (results for the equilibrium and Haar random setting are shown in the SM). Furthermore, we also investigate here questions related to the Markovianity of the histories.

\begin{figure*}[tb]
	\centering\includegraphics[width=0.99\textwidth,clip=true]{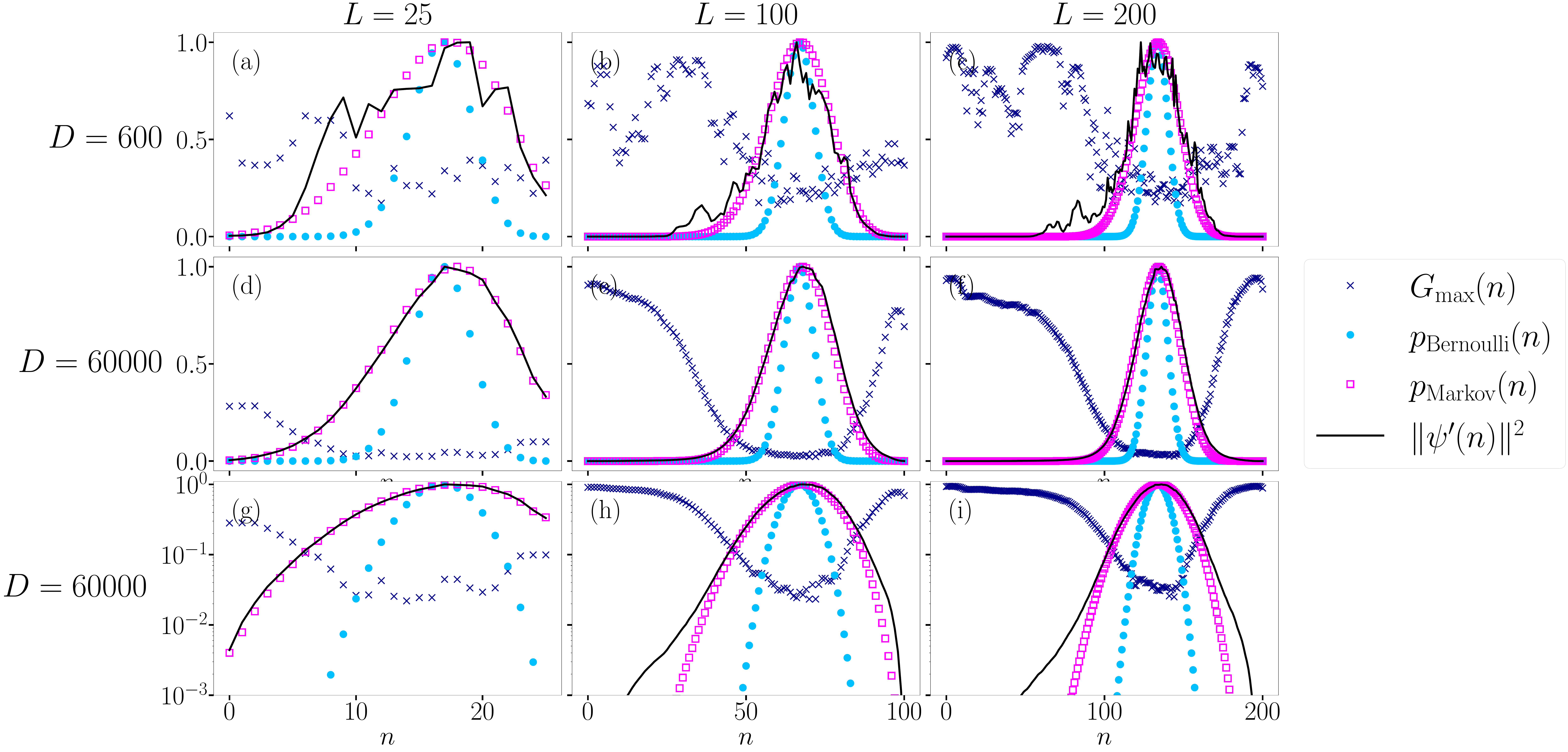}
	\caption{Decoherence, probabilities and weights of inhomogeneous histories as a function of $n\in\{0,1,\dots,L\}$ for $L\in\{25,100,200\}$ and $D\in\{600,60000\}$ (the third row is identical to the second but displayed on a logarithmic scale). We plot the maximum coherence (dark blue crosses), the weights of the history states (displayed for better visibility as a black solid line instead of discrete points), the probabilities resulting from an idealized Markov process (magenta squares), and the probabilities of an idealized Bernoulli trial with the same average (sky blue disks). To plot the probabilistic quantities on the same scale as the (de)coherence, all of them got rescaled such that their maximum value equals one. In contrast to before, we here do not demand that the histories end in subspace $\C H_0$: we keep both outcomes $x_L\in\{0,1\}$ to construct the $|\psi(n)\rangle$. }
	\label{fig Born neq}
\end{figure*}

Figure~\ref{fig Born neq} displays four different quantities for various $L$ and $D$. First, we consider the $n$-resolved maximum (de)coherence
\begin{equation}
	G_\text{max}(n) = \max_{n'\neq n} |G_{n,n'}|.
\end{equation}
If $G_\text{max}(n)$ is very small (ideally zero), this implies that $|\psi(n)\rangle$ is decohered from all other inhomogeneous histories and an associated record exists (and, ideally, identifies $n$ with certainty). These are the dark blue crosses in Fig.~\ref{fig Born neq}, which give rise to the following observation. First, there is a clear pattern emerging for large $L$, which seems even more pronounced for larger $D$. Looking at the pattern we see that some histories retain the same amount of minimum decoherence for all $L$ (as it becomes particularly clear from the logarithmic scale of the last row in Fig.~\ref{fig Born neq}), whereas the other histories start to recohere, partially even close to the maximum value $G_\text{max}(n)=1$. Now, the truly interesting observation is related to the question \emph{where} this recoherence happens.

To this end, we consider three quantities related to the probabilistic interpretation of the process. First of all, there are the weights $q(n)$ of the wave function themselves (black solid line), which are conventionally identified with the probabilities of Born's rule. We can see that they clearly concentrate around the value of strongest decoherence. Second, we consider the probabilities $p_\text{Markov}(n)$ that would result from assuming that the underlying elementary histories $\bs x$ obey a \emph{classical time-homogeneous Markov process}, meaning that
\begin{equation}
	p(\bs x) = T_{x_L|x_{L-1}}\cdots T_{x_1|x_0}p(\bs x_0).
\end{equation}
Here, $T$ is a transition (or stochastic) matrix, whose elements $T_{x'|x}$ can be identified with the conditional probability to transition into subspace $\C H_{x'}$ starting from $\C H_x$ within time step $\Delta t$. Numerically, we extract $T_{x'|x}$ by averaging $\lr{\psi'(x',x)|\psi'(x',x)}$ with $|\psi'(x',x)\rangle = \Pi_{x'} U_{\Delta t}\Pi_x|\Psi_0\rangle$ over $100$ Haar random initial states $|\Psi_0\rangle$. Finally, we set $p_\text{Markov}(n) = p_0(n;t_L) + p_1(n;t_L)$, where $p_i(n;t)$ is the probability to find the system in subspace $\C H_i$ at time $t$ after a history with $n$ ``1s'', which is computed as
\begin{align}
	p_0(n;t_{k+1}) &= T_{0|0} p_0(n;t_k) + T_{0|1} p_1(n;t_k), \\
	p_1(n;t_{k+1}) &= T_{1|0} p_0(n-1;t_k) + T_{1|1} p_1(n-1;t_k).
\end{align}
The magenta squares in Fig.~\ref{fig Born neq} display the resulting probability distribution $p_\text{Markov}(n)$ after $L$ time steps. As we can see, for large $D$ they match very well the black line of $q(n)$ in the region of decoherence. Outside it, for the recoherent histories, the logarithmic scale (last row of Fig.~\ref{fig Born neq}) reveals clear deviations as expected for a coherent quantum process. That one single process can be neatly divided into two sub-processes---one classical and Markovian, the other coherent and non-Markovian---is a curious observation in its own right. Finally, the sky blue discs in Fig.~\ref{fig Born neq} display the probability distribution
\begin{equation}
	p_\text{Bernoulli}(n) = \binom{L}{n} (1-p)^{L-n} p^n
\end{equation}
for $p=D_1/D$. This is the Bernoulli distribution for finding $n$ ``heads'' in a coin-flip experiment with a probability $D_1/D$ for heads. We mention this distribution here for two reasons. First, we can see that it does not match the black line (even though the averages match), which shows that we have a process here with correlations in time instead of a memoryless Bernoulli process. Second, at equilibrium time scales we expect the system to relax to subspace $\C H_1$ with probability $p_1 = D_1/D$, such that $p_\text{Bernoulli}(n)$ should match with $q(n)$ in presence of decoherence. This is indeed found for the case displayed in the SM.

To conclude, for large $L$ the inhomogenous histories of eqn~(\ref{eq inhomo history}) show a strong (de)coherence pattern with decoherence persisting on all those branches that one expects to likely occur if one were to perform an ideal frequency experiment within the Copenhagen interpretation of quantum mechanics. Put differently, maverick branches that are highly atypical according to Born's rule are strongly recoherent. We discuss the possible importance of these findings for the MWI below.

%%%%%%%%%%%%%%%%%%%%%%%%%%%%%%%%%%%%%%%%%%%%%%%%%%%%%%%%%%%%%%%%%%%%%%%%%%%%%%%%%%%%%%%%%%%%%%%%%%%%%%%%%%%%%%%%%%%%%%%%
\section{Discussion and outlook}\label{sec:conclusions}
%%%%%%%%%%%%%%%%%%%%%%%%%%%%%%%%%%%%%%%%%%%%%%%%%%%%%%%%%%%%%%%%%%%%%%%%%%%%%%%%%%%%%%%%%%%%%%%%%%%%%%%%%%%%%%%%%%%%%%%%

This paper attempted a rigorous exploration of approximate decoherence and the structure of decoherence for long histories. We here not only summarize our findings, but also discuss their possible consequences for the MWI. This is necessarily done from a speculative angle as we currently do not know the strength of decoherence on a cosmological scale or whether very long histories matter. Our results thus only reveal \emph{potential} consequences for the MWI, but these consequences could be profound and (to the best of our knowledge) have never been discussed before.

Before discussing them, we briefly review arguments in favor or against our assumptions in a cosmological context. Essentially, Sec.~\ref{sec:records} was based on eqn~(\ref{eq:DF scaling form}), an ansatz for the strength of decoherence that currently lacks justification for infinite $D = \dim\C H$ (albeit the derivation in Sec.~\ref{sec:records} only requires approximate decoherence, not finite $D$), and Sec.~\ref{sec:structure} highlighted the fate of decoherence for very long histories for a model of finite $D$.

Clearly, several arguments speak against a finite $D$ on a cosmological scale. For instance, the current cosmological standard (or $\Lambda$CDM) model predicts a flat Universe without boundaries (even though this is compatible with a finite, torus-shape geometry~\cite{GillaniInterview2021}). Furthermore, already certain results in non-relativistic quantum mechanics, such as $[x,p]=i\hbar$, require infinite dimensions (by the Stone-von Neumann theorem), and in quantum field theory the Nielsen-Ninomiya theorem forbids any lattice discretization of electroweak interactions~\cite{NielsenNinomiyaNPB1981a, NielsenNinomiyaNPB1981b, TongNotes2018}.

On the other hand, whatever the $\Lambda$CDM predicts, our \emph{observable} universe is certainly spatially finite. Moreover, finding any observable consequences of an infinite $D$ in quantum mechanics usually requires infinite precision, time or energy---indicating that $D=\infty$ could be an abstract mathematical idealization. Indeed, speculations that spacetime itself is discrete exist, and an argument in favor of finite $D$ is suggested by the Bekenstein-Hawking entropy of a black hole~\cite{BekensteinPT1980}: any spacetime region with a finite mass should have a finite amount of microstates. Indeed, estimates of the (finite) number of qubits in the universe have even been computed~\cite{LloydNature2000, LloydPRL2002}.

Finally, there are also good reasons to study very long histories because, at a fundamental level, time itself could be a redundant concept as evidenced, e.g., by the Wheeler-DeWitt equation $H|\Psi\rangle = 0$~\cite{DeWittPR1967, PriceBook1996, BarbourBook1999, BarbourArXiv2009, KieferInBook2017}. In such a timeless ``block Universe'' all possible histories should \emph{a priori} be treated on an equal footing.

%%%%%%%%%%%%%%%%%%%%%%%%%%%%%%%%%%%%%%%%%%%%%%%%%%%%%%%%%%%%%%%%%%%%%%%%%%%%%%%%%%%%%%%%%%%%%%%%%%%%%%%%%%%%%%%%%%%%%%%%
\subsection{The Branch Selection Problem}\label{sec discussion part 1}
%%%%%%%%%%%%%%%%%%%%%%%%%%%%%%%%%%%%%%%%%%%%%%%%%%%%%%%%%%%%%%%%%%%%%%%%%%%%%%%%%%%%%%%%%%%%%%%%%%%%%%%%%%%%%%%%%%%%%%%%

Geometrically speaking, Sec.~\ref{sec:records} revealed that a set of almost orthogonal (history) vectors $\{|\psi(\bs x)\rangle\}_N$ can not in general be well approximated by a set of orthogonal (record) vectors $\{|r(\bs x)\rangle\}_N$. This is true if (i) the scalar product scales as $\lr{\psi(\bs x)|\psi(\bs x')} \sim D^{-\alpha}$ with $\alpha\le1/2$ and (ii) the $|\psi(\bs x)\rangle$ have some independent (pseudo-)randomness. A large identification error $|\lr{\psi(\bs x)|r(\bs x)}|^2$ then follows when $N\rightarrow D^{2\alpha}$. In addition, $N$ can be much larger than $D$ without loosing almost orthogonality. This finding is quite general and robust: it follows from the geometry of Hilbert spaces and a large class of (pseudo-)random vectors shows the same behavior. Thus, for $\alpha\le1/2$ we find a \emph{branch selection problem}: There are $N$ different histories, which are roughly equally decoherent and thus look equally classical, but potentially only $N_\text{detectable}(\ll N)$ branches contain observes with \emph{reliable} records about their history---hence, the question how does the multiverse ``select'' the histories that are reliably identified by branches?

The branch selection problem persists even if we agree on the right history coarse-graining (the set selection problem) or the record basis to decompose $|\Psi\rangle$ (the preferred basis problem). Furthermore, it would be questionable to simply disregard decompositions $|\Psi\rangle = \sum_{\bs x=1}^N |\psi'(\bs x)\rangle$ that contain too many histories $N$ to be identifiable by records. Everett's idea was to get rid of any privileged role of an observer: the multiverse is a ``democratic'' place in which every decomposition of $|\Psi\rangle = \sum_{\bs x=1}^N |\psi'(\bs x)\rangle$ is \emph{a priori} equally justified. The mechanics of the multiverse alone, perhaps together with a specific choice of boundary condition, must bound sensible history decompositions in a way that they are identifiable by records.

It also seems questionable to reject the branch selection problem by asserting that $N$ will be small enough based on the current records that we have about our universe described by a Big Bang cosmology. First, we might have to share the Hilbert space with other universes different from our Big Bang cosmology. Second, this view seems incompatible with an eternal timeless block Universe. Third, it is an example of observer-centric reasoning: why should \emph{we} be allowed to decide how to decompose the universal wave function $|\Psi\rangle$?

If one does not want to reject the ontology of the MWI, a possible way out is to speculate that the physics of the Universe does not produce decoherent histories that obey the ansatz~(\ref{eq:DF scaling form}), despite the evidence of Refs.~\cite{StrasbergReinhardSchindlerPRX2024, WangStrasbergPRL2025, GemmerSteinigewegPRE2014, SchmidtkeGemmerPRE2016, NationPorrasPRE2020, StrasbergEtAlPRA2023, StrasbergSP2023}. Indeed, the example in Sec.~\ref{sec:structure} shows already that there is additional structure for long histories (albeit it does not disprove that there could be $N\gg D$ many approximately decoherent histories). Other physical constraints, such as quantum Darwinism~\cite{ZurekNP2009, KorbiczQuantum2021}, could impose further constraints: perhaps redundancy implies decoherence with an exponent $\alpha>1/2$? At present, however, there is no rigorous evidence for that, and numerical evidence connecting quantum Darwinism and decoherent histories is scarce~\cite{FerteFarciCaoPRL2026, ArefyevaPolyakovPRA2026}.

%%%%%%%%%%%%%%%%%%%%%%%%%%%%%%%%%%%%%%%%%%%%%%%%%%%%%%%%%%%%%%%%%%%%%%%%%%%%%%%%%%%%%%%%%%%%%%%%%%%%%%%%%%%%%%%%%%%%%%%%
\subsection{Emergence of Born's rule?}\label{sec discussion part 2}
%%%%%%%%%%%%%%%%%%%%%%%%%%%%%%%%%%%%%%%%%%%%%%%%%%%%%%%%%%%%%%%%%%%%%%%%%%%%%%%%%%%%%%%%%%%%%%%%%%%%%%%%%%%%%%%%%%%%%%%%

Section~\ref{sec:structure} concerned a purely numerical analysis of a simple random matrix toy model within nonrelativistic quantum mechanics. Nevertheless, it was a ``first principles'' analysis based on a numerically exact integration of Schr\"odinger's equation (no assumptions violating unitarity entered) for the time evolution of pure states (no classical ensemble averages). Moreover, backed up by a large literature in statistical mechanics, we argued that results found for such random matrix models are generic.\footnote{We checked that also the non-integrable spin chain model of Ref.~\cite{WangStrasbergPRL2025} shows a behavior similar to Fig.~\ref{fig Born neq} (not shown here for brevity).}

We then found that histories do not decohere equally, but that some histories start to become \emph{recoherent} while others remain decoherent. Perhaps the most remarkable observation---confirming and extending the preliminary results of Ref.~\cite{StrasbergSchindlerArXiv2023}---was that strong recoherences happen for inhomogeneous history states $|\psi(n)\rangle$ whose frequencies $n=\#_1(\bs x)$ (where $\#_1(\bs x)$ denotes the number of ``1s'' in an elementary history $\bs x\in\{0,1\}^L$) deviate significantly from the expectation value $\overline n = \sum_{n=0}^L nq(n)$ with $q(n) = \lr{\psi(n)|\psi(n)}$.

Since the derivation of Born's rule is an essential and unsolved problem within the MWI, it is worth exploring the topic further (for an introduction see Refs.~\cite{SaundersEtAlBook2010, MaudlinBook2019, Vaidman2020, RidleyEJPS2025}, various attempts to solve it can be found in Refs.~\cite{EverettRMP1957, HartleAJP1968, DeWittGrahamBook1973, VaidmanISPS1998, DeutschPRSCA1999, WallaceSHPSB2003, HansonFP2003, GreavesSHPSB2004, ZurekPRA2005, BuniyHsuZeePLB2006, SaundersInBook2010, AguirreTegmarkPRD2011, WallaceBook2012, SebensCarrollBJP2018, MasanesGalleyMuellerNC2019, McQueenVaidmanSHPSB2019, Vaidman2020, SaundersPRSA2021, HsuMPLA2021, ShortQuantum2023, LazaroviciQR2023, RidleyQR2023, RidleyEJPS2025, WeidnerArXiv2025}\footnote{The fact that even the proponents of the MWI offer vastly different ``derivations'' of Born's rule hints at the fact that this is an unsolved problem.}). In particular, the theory confirmation problem asks why do we find ourselves in a world where we confirm Born's rule if the \emph{vast majority} of histories are maverick branches, i.e., histories with frequencies that are atypical with respect to Born's rule? This becomes clear in the example we considered: the \emph{number} of histories $\bs x$ with $n=\#_1(\bs x)$ many ``1s'' is \emph{always} given by the binomial coefficient $\binom{L}{n}$, independent of the details of the dynamics. Thus, $n\approx L/2$ for most histories $\bs x$ independent of their statistical weight $q(\bs x)$. The often considered solution to cut away branches with low weight will not work because we observe in the SM that the individual weights $q(\bs x)$ are \emph{not} concentrated on histories $\bs x$ that are Born-typical (though the inhomogeneous histories $|\psi(n)\rangle$ are) and, more generally, as Kent emphasizes~\cite{KentInBook2010}: \emph{``It's no more scientifically respectable to declare that we can [...] confirm Everettian quantum theory by neglecting the observations made on selected low Born weight branches [...] unless we add further structure to the theory [...].''}

Here, we have evidence for such ``further structure'': histories with frequencies deviating significantly from the expectation of Born's rule are strongly \emph{recoherent}. Since only decoherent histories leave detectable records in the multiverse, this argument suggests that branches violating Born's rule exist but are not detectable.

This argumentation, which is based solely on numerical observations, is clearly speculative. It also might be criticized as being circular: after all, does not the interpretation of decoherence, i.e., the meaning of the overlaps $\lr{\psi(\bs x)|\psi(\bs x')}$, presupposes a notion of Born's rule? While that is clearly true in the conventional Copenhagen interpretation, it is not clear to us whether one could not use the geometrical meaning of $\lr{\psi(\bs x)|\psi(\bs x')}$ to circumvent the problem. It is also interesting to note that the here seen structure of decoherence is close to the ``mangled worlds'' imagined in Ref.~\cite{HansonFP2003}, which provides an interesting opportunity for future work.

To conclude, the current numerical observations reveals an interesting correlation between Born's rule and decoherence, but the derivation of Born's rule remains an open problem.

%%%%%%%%%%%%%%%%%%%%%%%%%%%%%%%%%%%%%%%%%%%%%%%%%%%%%%%%%%%%%%%%%%%%%%%%%%%%%%%%%%%%%%%%%%%%%%%%%%%%%%%%%%%%%%%%%%%%%%%%
\subsection*{Acknowledgements}
%%%%%%%%%%%%%%%%%%%%%%%%%%%%%%%%%%%%%%%%%%%%%%%%%%%%%%%%%%%%%%%%%%%%%%%%%%%%%%%%%%%%%%%%%%%%%%%%%%%%%%%%%%%%%%%%%%%%%%%%

We gratefully acknowledge discussions with Teresa E.~Reinhard, Giulio Gasbarri, Marco Fanizza, and the StackExchange community, in particular user Dustin G.~Mixon. This project has also benefited from AI support by Lumo and ChatGPT related to random matrix theory, the construction of the counterexample and the presentation of plots.

PS, JS and AW acknowledge financial support from MICINN with funding from European Union NextGenerationEU (PRTR-C17.I1), the Generalitat de Catalunya, Spanish MICIN (project PID2022-141283NB-I00) with the support of FEDER funds, and the Spanish MTDFP through the QUANTUM ENIA project: Quantum Spain, funded by the European Union NextGenerationEU within the framework of the ``Digital Spain 2026 Agenda''. PS and AW also acknowledge the European Commission QuantERA grant ExTRaQT (Spanish MICIN grant no.~PCI2022-132965). PS was furthermore supported by ``la Caixa'' Foundation (ID 100010434, fellowship code LCF/BQ/PR21/11840014) and is supported by the Ram\'on y Cajal program RYC2022-035908-I. JS is supported by a Beatriu de Pin\'os fellowship of the Catalan Agency for Management of University and Research Grants (AGAUR)	and EU Horizon 2020 MSCA grant 801370. JW acknowledges the Deutsche Forschungsgemeinschaft (DFG), under Grant No.~531128043, as well as under Grants No.~397107022, No.~397067869, and No.~397082825 within the DFG Research Unit FOR 2692, under Grant No.~355031190. AW further acknowledges the Alexander von Humboldt Foundation.

%We refer to Eq.~\ref{S-eq:extra_result} in the supplement.

\bibliography{/home/philipp/current/references/books.bib,/home/philipp/current/references/open_systems.bib,/home/philipp/current/references/thermo.bib,/home/philipp/current/references/info_thermo.bib,/home/philipp/current/references/general_QM.bib,/home/philipp/current/references/math_phys.bib,/home/philipp/current/references/equilibration.bib,/home/philipp/current/references/time.bib,/home/philipp/current/references/cosmology.bib,/home/philipp/current/references/general_refs.bib}

\begin{thebibliography}{195}%
\makeatletter
\providecommand \@ifxundefined [1]{%
 \@ifx{#1\undefined}
}%
\providecommand \@ifnum [1]{%
 \ifnum #1\expandafter \@firstoftwo
 \else \expandafter \@secondoftwo
 \fi
}%
\providecommand \@ifx [1]{%
 \ifx #1\expandafter \@firstoftwo
 \else \expandafter \@secondoftwo
 \fi
}%
\providecommand \natexlab [1]{#1}%
\providecommand \enquote  [1]{``#1''}%
\providecommand \bibnamefont  [1]{#1}%
\providecommand \bibfnamefont [1]{#1}%
\providecommand \citenamefont [1]{#1}%
\providecommand \href@noop [0]{\@secondoftwo}%
\providecommand \href [0]{\begingroup \@sanitize@url \@href}%
\providecommand \@href[1]{\@@startlink{#1}\@@href}%
\providecommand \@@href[1]{\endgroup#1\@@endlink}%
\providecommand \@sanitize@url [0]{\catcode `\\12\catcode `\$12\catcode
  `\&12\catcode `\#12\catcode `\^12\catcode `\_12\catcode `\%12\relax}%
\providecommand \@@startlink[1]{}%
\providecommand \@@endlink[0]{}%
\providecommand \url  [0]{\begingroup\@sanitize@url \@url }%
\providecommand \@url [1]{\endgroup\@href {#1}{\urlprefix }}%
\providecommand \urlprefix  [0]{URL }%
\providecommand \Eprint [0]{\href }%
\providecommand \doibase [0]{https://doi.org/}%
\providecommand \selectlanguage [0]{\@gobble}%
\providecommand \bibinfo  [0]{\@secondoftwo}%
\providecommand \bibfield  [0]{\@secondoftwo}%
\providecommand \translation [1]{[#1]}%
\providecommand \BibitemOpen [0]{}%
\providecommand \bibitemStop [0]{}%
\providecommand \bibitemNoStop [0]{.\EOS\space}%
\providecommand \EOS [0]{\spacefactor3000\relax}%
\providecommand \BibitemShut  [1]{\csname bibitem#1\endcsname}%
\let\auto@bib@innerbib\@empty
%</preamble>
\bibitem [{\citenamefont {Gell-Mann}\ and\ \citenamefont
  {Hartle}(1990)}]{GellMannHartleInBook1990}%
  \BibitemOpen
  \bibfield  {author} {\bibinfo {author} {\bibfnamefont {M.}~\bibnamefont
  {Gell-Mann}}\ and\ \bibinfo {author} {\bibfnamefont {J.~B.}\ \bibnamefont
  {Hartle}},\ }\bibinfo {title} {{Complexity, Entropy and the Physics of
  Information}}\ (\bibinfo  {publisher} {Reading: Addison-Wesley},\ \bibinfo
  {year} {1990})\ Chap.\ \bibinfo {chapter} {Quantum Mechanics in the Light of
  Quantum Cosmology}, pp.\ \bibinfo {pages} {425--459.}\BibitemShut {Stop}%
\bibitem [{\citenamefont {Hartle}(1992)}]{HartleLecture1992}%
  \BibitemOpen
  \bibfield  {author} {\bibinfo {author} {\bibfnamefont {J.~B.}\ \bibnamefont
  {Hartle}},\ }\href {https://arxiv.org/abs/gr-qc/9304006} {\emph {\bibinfo
  {title} {{Spacetime Quantum Mechanics and the Quantum Mechanics of
  Spacetime}}}}\ (\bibinfo  {publisher} {Les Houches Lecture Notes},\ \bibinfo
  {year} {1992})\BibitemShut {NoStop}%
\bibitem [{\citenamefont {Halliwell}(1995)}]{HalliwellANY1995}%
  \BibitemOpen
  \bibfield  {author} {\bibinfo {author} {\bibfnamefont {J.~J.}\ \bibnamefont
  {Halliwell}},\ }\bibfield  {title} {\bibinfo {title} {{A Review of the
  Decoherent Histories Approach to Quantum Mechanics}},\ }\href
  {https://doi.org/https://doi.org/10.1111/j.1749-6632.1995.tb39014.x}
  {\bibfield  {journal} {\bibinfo  {journal} {Ann. (N.Y.) Acad. Sci.}\ }\textbf
  {\bibinfo {volume} {755}},\ \bibinfo {pages} {726} (\bibinfo {year}
  {1995})}\BibitemShut {NoStop}%
\bibitem [{\citenamefont {Dowker}\ and\ \citenamefont
  {Kent}(1996)}]{DowkerKentJSP1996}%
  \BibitemOpen
  \bibfield  {author} {\bibinfo {author} {\bibfnamefont {F.}~\bibnamefont
  {Dowker}}\ and\ \bibinfo {author} {\bibfnamefont {A.}~\bibnamefont {Kent}},\
  }\bibfield  {title} {\bibinfo {title} {On the consistent histories approach
  to quantum mechanics},\ }\bibfield  {journal} {\bibinfo  {journal} {J. Stat.
  Phys.}\ }\textbf {\bibinfo {volume} {82}},\ \href
  {https://doi.org/10.1007/BF02183396} {10.1007/BF02183396} (\bibinfo {year}
  {1996})\BibitemShut {NoStop}%
\bibitem [{\citenamefont {Dowker}\ and\ \citenamefont
  {Halliwell}(1992)}]{DowkerHalliwellPRD1992}%
  \BibitemOpen
  \bibfield  {author} {\bibinfo {author} {\bibfnamefont {H.~F.}\ \bibnamefont
  {Dowker}}\ and\ \bibinfo {author} {\bibfnamefont {J.~J.}\ \bibnamefont
  {Halliwell}},\ }\bibfield  {title} {\bibinfo {title} {{Quantum mechanics of
  history: The decoherence functional in quantum mechanics}},\ }\href
  {https://doi.org/10.1103/PhysRevD.46.1580} {\bibfield  {journal} {\bibinfo
  {journal} {Phys. Rev. D}\ }\textbf {\bibinfo {volume} {46}},\ \bibinfo
  {pages} {1580} (\bibinfo {year} {1992})}\BibitemShut {NoStop}%
\bibitem [{\citenamefont {Dowker}\ and\ \citenamefont
  {Kent}(1995)}]{DowkerKentPRL1995}%
  \BibitemOpen
  \bibfield  {author} {\bibinfo {author} {\bibfnamefont {F.}~\bibnamefont
  {Dowker}}\ and\ \bibinfo {author} {\bibfnamefont {A.}~\bibnamefont {Kent}},\
  }\bibfield  {title} {\bibinfo {title} {{Properties of Consistent
  Histories}},\ }\href {https://doi.org/10.1103/PhysRevLett.75.3038} {\bibfield
   {journal} {\bibinfo  {journal} {Phys. Rev. Lett.}\ }\textbf {\bibinfo
  {volume} {75}},\ \bibinfo {pages} {3038} (\bibinfo {year}
  {1995})}\BibitemShut {NoStop}%
\bibitem [{\citenamefont {McElwaine}(1996)}]{McElwainePRA1996}%
  \BibitemOpen
  \bibfield  {author} {\bibinfo {author} {\bibfnamefont {J.~N.}\ \bibnamefont
  {McElwaine}},\ }\bibfield  {title} {\bibinfo {title} {Approximate and exact
  consistency of histories},\ }\href {https://doi.org/10.1103/PhysRevA.53.2021}
  {\bibfield  {journal} {\bibinfo  {journal} {Phys. Rev. A}\ }\textbf {\bibinfo
  {volume} {53}},\ \bibinfo {pages} {2021} (\bibinfo {year}
  {1996})}\BibitemShut {NoStop}%
\bibitem [{\citenamefont {Halliwell}(2001)}]{HalliwellPRD2001}%
  \BibitemOpen
  \bibfield  {author} {\bibinfo {author} {\bibfnamefont {J.~J.}\ \bibnamefont
  {Halliwell}},\ }\bibfield  {title} {\bibinfo {title} {{Approximate
  decoherence of histories and 't Hooft's deterministic quantum theory}},\
  }\href {https://doi.org/10.1103/PhysRevD.63.085013} {\bibfield  {journal}
  {\bibinfo  {journal} {Phys. Rev. D}\ }\textbf {\bibinfo {volume} {63}},\
  \bibinfo {pages} {085013} (\bibinfo {year} {2001})}\BibitemShut {NoStop}%
\bibitem [{\citenamefont {Halliwell}(2005)}]{HalliwellPRA2005}%
  \BibitemOpen
  \bibfield  {author} {\bibinfo {author} {\bibfnamefont {J.~J.}\ \bibnamefont
  {Halliwell}},\ }\bibfield  {title} {\bibinfo {title} {Commuting position and
  momentum operators, exact decoherence and emergent classicality},\ }\href
  {https://doi.org/10.1103/PhysRevA.72.042109} {\bibfield  {journal} {\bibinfo
  {journal} {Phys. Rev. A}\ }\textbf {\bibinfo {volume} {72}},\ \bibinfo
  {pages} {042109} (\bibinfo {year} {2005})}\BibitemShut {NoStop}%
\bibitem [{\citenamefont {Strasberg}\ \emph {et~al.}(2024)\citenamefont
  {Strasberg}, \citenamefont {Reinhard},\ and\ \citenamefont
  {Schindler}}]{StrasbergReinhardSchindlerPRX2024}%
  \BibitemOpen
  \bibfield  {author} {\bibinfo {author} {\bibfnamefont {P.}~\bibnamefont
  {Strasberg}}, \bibinfo {author} {\bibfnamefont {T.~E.}\ \bibnamefont
  {Reinhard}},\ and\ \bibinfo {author} {\bibfnamefont {J.}~\bibnamefont
  {Schindler}},\ }\bibfield  {title} {\bibinfo {title} {{First Principles
  Numerical Demonstration of Emergent Decoherent Histories}},\ }\href
  {https://doi.org/10.1103/PhysRevX.14.041027} {\bibfield  {journal} {\bibinfo
  {journal} {Phys. Rev. X}\ }\textbf {\bibinfo {volume} {14}},\ \bibinfo
  {pages} {041027} (\bibinfo {year} {2024})}\BibitemShut {NoStop}%
\bibitem [{\citenamefont {Wang}\ and\ \citenamefont
  {Strasberg}(2025)}]{WangStrasbergPRL2025}%
  \BibitemOpen
  \bibfield  {author} {\bibinfo {author} {\bibfnamefont {J.}~\bibnamefont
  {Wang}}\ and\ \bibinfo {author} {\bibfnamefont {P.}~\bibnamefont
  {Strasberg}},\ }\bibfield  {title} {\bibinfo {title} {{Decoherence of
  Histories: Chaotic Versus Integrable Systems}},\ }\href
  {https://doi.org/10.1103/m8vq-l449} {\bibfield  {journal} {\bibinfo
  {journal} {Phys. Rev. Lett.}\ }\textbf {\bibinfo {volume} {134}},\ \bibinfo
  {pages} {220401} (\bibinfo {year} {2025})}\BibitemShut {NoStop}%
\bibitem [{\citenamefont {Colas}\ \emph {et~al.}(2023)\citenamefont {Colas},
  \citenamefont {Grain},\ and\ \citenamefont
  {Vennin}}]{ColasGrainVenninEPL2023}%
  \BibitemOpen
  \bibfield  {author} {\bibinfo {author} {\bibfnamefont {T.}~\bibnamefont
  {Colas}}, \bibinfo {author} {\bibfnamefont {J.}~\bibnamefont {Grain}},\ and\
  \bibinfo {author} {\bibfnamefont {V.}~\bibnamefont {Vennin}},\ }\bibfield
  {title} {\bibinfo {title} {Quantum recoherence in the early universe},\
  }\href {https://doi.org/10.1209/0295-5075/acdd94} {\bibfield  {journal}
  {\bibinfo  {journal} {Europhys. Lett.}\ }\textbf {\bibinfo {volume} {142}},\
  \bibinfo {pages} {69002} (\bibinfo {year} {2023})}\BibitemShut {NoStop}%
\bibitem [{\citenamefont {Kranas}\ \emph {et~al.}(2025)\citenamefont {Kranas},
  \citenamefont {Grain},\ and\ \citenamefont
  {Vennin}}]{KranasGrainVenninPRD2025}%
  \BibitemOpen
  \bibfield  {author} {\bibinfo {author} {\bibfnamefont {D.}~\bibnamefont
  {Kranas}}, \bibinfo {author} {\bibfnamefont {J.}~\bibnamefont {Grain}},\ and\
  \bibinfo {author} {\bibfnamefont {V.}~\bibnamefont {Vennin}},\ }\bibfield
  {title} {\bibinfo {title} {{Recoherence, adiabaticity, and Markovianity in
  Gaussian maps}},\ }\href {https://doi.org/10.1103/r8pq-kms4} {\bibfield
  {journal} {\bibinfo  {journal} {Phys. Rev. D}\ }\textbf {\bibinfo {volume}
  {112}},\ \bibinfo {pages} {076014} (\bibinfo {year} {2025})}\BibitemShut
  {NoStop}%
\bibitem [{\citenamefont {Strasberg}\ and\ \citenamefont
  {Schindler}(2023)}]{StrasbergSchindlerArXiv2023}%
  \BibitemOpen
  \bibfield  {author} {\bibinfo {author} {\bibfnamefont {P.}~\bibnamefont
  {Strasberg}}\ and\ \bibinfo {author} {\bibfnamefont {J.}~\bibnamefont
  {Schindler}},\ }\bibfield  {title} {\bibinfo {title} {{Shearing Off the Tree:
  Emerging Branch Structure and Born's Rule in an Equilibrated Multiverse}},\
  }\href {https://arxiv.org/abs/2310.06755} {\bibfield  {journal} {\bibinfo
  {journal} {arXiv: 2310.06755}\ } (\bibinfo {year} {2023})}\BibitemShut
  {NoStop}%
\bibitem [{\citenamefont {J\"onsson}(1961)}]{JoenssonZfP1961}%
  \BibitemOpen
  \bibfield  {author} {\bibinfo {author} {\bibfnamefont {C.}~\bibnamefont
  {J\"onsson}},\ }\bibfield  {title} {\bibinfo {title}
  {{Elektroneninterferenzen an mehreren k\"unstlich hergestellten
  Feinspalten}},\ }\href {https://doi.org/10.1007/bf01342460} {\bibfield
  {journal} {\bibinfo  {journal} {Zeitschr. f. Phys.}\ }\textbf {\bibinfo
  {volume} {161}},\ \bibinfo {pages} {454} (\bibinfo {year}
  {1961})}\BibitemShut {NoStop}%
\bibitem [{\citenamefont {J\"onsson}(1974)}]{JoenssonAJP1974}%
  \BibitemOpen
  \bibfield  {author} {\bibinfo {author} {\bibfnamefont {C.}~\bibnamefont
  {J\"onsson}},\ }\bibfield  {title} {\bibinfo {title} {{Electron Diffraction
  at Multiple Slits}},\ }\href {https://doi.org/10.1119/1.1987592} {\bibfield
  {journal} {\bibinfo  {journal} {Am. J. Phys.}\ }\textbf {\bibinfo {volume}
  {42}},\ \bibinfo {pages} {4} (\bibinfo {year} {1974})}\BibitemShut {NoStop}%
\bibitem [{\citenamefont {Arndt}\ \emph {et~al.}(1999)\citenamefont {Arndt},
  \citenamefont {Nairz}, \citenamefont {Vos-Andreae}, \citenamefont {Keller},
  \citenamefont {van~der Zouw},\ and\ \citenamefont
  {Zeilinger}}]{ArndtEtAlNat1999}%
  \BibitemOpen
  \bibfield  {author} {\bibinfo {author} {\bibfnamefont {M.}~\bibnamefont
  {Arndt}}, \bibinfo {author} {\bibfnamefont {O.}~\bibnamefont {Nairz}},
  \bibinfo {author} {\bibfnamefont {J.}~\bibnamefont {Vos-Andreae}}, \bibinfo
  {author} {\bibfnamefont {C.}~\bibnamefont {Keller}}, \bibinfo {author}
  {\bibfnamefont {G.}~\bibnamefont {van~der Zouw}},\ and\ \bibinfo {author}
  {\bibfnamefont {A.}~\bibnamefont {Zeilinger}},\ }\bibfield  {title} {\bibinfo
  {title} {Wave–particle duality of {C}$_{60}$ molecules},\ }\href
  {https://doi.org/10.1038/44348} {\bibfield  {journal} {\bibinfo  {journal}
  {Nature}\ }\textbf {\bibinfo {volume} {401}},\ \bibinfo {pages} {680}
  (\bibinfo {year} {1999})}\BibitemShut {NoStop}%
\bibitem [{\citenamefont {Gerlich}\ \emph {et~al.}(2011)\citenamefont
  {Gerlich}, \citenamefont {Eibenberger}, \citenamefont {Tomandl},
  \citenamefont {Nimmrichter}, \citenamefont {Hornberger}, \citenamefont
  {Fagan}, \citenamefont {T\"uxen}, \citenamefont {Mayor},\ and\ \citenamefont
  {Arndt}}]{GerlichEtAlNC2011}%
  \BibitemOpen
  \bibfield  {author} {\bibinfo {author} {\bibfnamefont {S.}~\bibnamefont
  {Gerlich}}, \bibinfo {author} {\bibfnamefont {S.}~\bibnamefont
  {Eibenberger}}, \bibinfo {author} {\bibfnamefont {M.}~\bibnamefont
  {Tomandl}}, \bibinfo {author} {\bibfnamefont {S.}~\bibnamefont
  {Nimmrichter}}, \bibinfo {author} {\bibfnamefont {K.}~\bibnamefont
  {Hornberger}}, \bibinfo {author} {\bibfnamefont {P.~J.}\ \bibnamefont
  {Fagan}}, \bibinfo {author} {\bibfnamefont {J.}~\bibnamefont {T\"uxen}},
  \bibinfo {author} {\bibfnamefont {M.}~\bibnamefont {Mayor}},\ and\ \bibinfo
  {author} {\bibfnamefont {M.}~\bibnamefont {Arndt}},\ }\bibfield  {title}
  {\bibinfo {title} {Quantum interference of large organic molecules},\ }\href
  {https://doi.org/10.1038/ncomms1263} {\bibfield  {journal} {\bibinfo
  {journal} {Nat. Comm.}\ }\textbf {\bibinfo {volume} {2}},\ \bibinfo {pages}
  {263} (\bibinfo {year} {2011})}\BibitemShut {NoStop}%
\bibitem [{\citenamefont {Fein}\ \emph {et~al.}(2019)\citenamefont {Fein},
  \citenamefont {Geyer}, \citenamefont {Zwick}, \citenamefont {Kia{\l}ka},
  \citenamefont {Pedalino}, \citenamefont {Mayor}, \citenamefont {Gerlich},\
  and\ \citenamefont {Arndt}}]{FeinEtAlNP2019}%
  \BibitemOpen
  \bibfield  {author} {\bibinfo {author} {\bibfnamefont {Y.~Y.}\ \bibnamefont
  {Fein}}, \bibinfo {author} {\bibfnamefont {P.}~\bibnamefont {Geyer}},
  \bibinfo {author} {\bibfnamefont {P.}~\bibnamefont {Zwick}}, \bibinfo
  {author} {\bibfnamefont {F.}~\bibnamefont {Kia{\l}ka}}, \bibinfo {author}
  {\bibfnamefont {S.}~\bibnamefont {Pedalino}}, \bibinfo {author}
  {\bibfnamefont {M.}~\bibnamefont {Mayor}}, \bibinfo {author} {\bibfnamefont
  {S.}~\bibnamefont {Gerlich}},\ and\ \bibinfo {author} {\bibfnamefont
  {M.}~\bibnamefont {Arndt}},\ }\bibfield  {title} {\bibinfo {title} {{Quantum
  superposition of molecules beyond 25 kDa}},\ }\href
  {https://doi.org/10.1038/s41567-019-0663-9} {\bibfield  {journal} {\bibinfo
  {journal} {Nat. Phys.}\ }\textbf {\bibinfo {volume} {15}},\ \bibinfo {pages}
  {1242} (\bibinfo {year} {2019})}\BibitemShut {NoStop}%
\bibitem [{\citenamefont {Nielsen}\ and\ \citenamefont
  {Chuang}(2000)}]{NielsenChuangBook2000}%
  \BibitemOpen
  \bibfield  {author} {\bibinfo {author} {\bibfnamefont {M.~A.}\ \bibnamefont
  {Nielsen}}\ and\ \bibinfo {author} {\bibfnamefont {I.~L.}\ \bibnamefont
  {Chuang}},\ }\href@noop {} {\emph {\bibinfo {title} {Quantum Computation and
  Quantum Information}}}\ (\bibinfo  {publisher} {Cambridge University Press},\
  \bibinfo {address} {Cambridge},\ \bibinfo {year} {2000})\BibitemShut
  {NoStop}%
\bibitem [{\citenamefont {Preskill}(2018)}]{PreskillQuantum2018}%
  \BibitemOpen
  \bibfield  {author} {\bibinfo {author} {\bibfnamefont {J.}~\bibnamefont
  {Preskill}},\ }\bibfield  {title} {\bibinfo {title} {Quantum {C}omputing in
  the {NISQ} era and beyond},\ }\href
  {https://doi.org/10.22331/q-2018-08-06-79} {\bibfield  {journal} {\bibinfo
  {journal} {{Quantum}}\ }\textbf {\bibinfo {volume} {2}},\ \bibinfo {pages}
  {79} (\bibinfo {year} {2018})}\BibitemShut {NoStop}%
\bibitem [{\citenamefont {Milz}\ and\ \citenamefont
  {Modi}(2021)}]{MilzModiPRXQ2021}%
  \BibitemOpen
  \bibfield  {author} {\bibinfo {author} {\bibfnamefont {S.}~\bibnamefont
  {Milz}}\ and\ \bibinfo {author} {\bibfnamefont {K.}~\bibnamefont {Modi}},\
  }\bibfield  {title} {\bibinfo {title} {{Quantum Stochastic Processes and
  Quantum non-Markovian Phenomena}},\ }\href
  {https://doi.org/10.1103/PRXQuantum.2.030201} {\bibfield  {journal} {\bibinfo
   {journal} {PRX Quantum}\ }\textbf {\bibinfo {volume} {2}},\ \bibinfo {pages}
  {030201} (\bibinfo {year} {2021})}\BibitemShut {NoStop}%
\bibitem [{\citenamefont {Taranto}\ \emph {et~al.}(2025)\citenamefont
  {Taranto}, \citenamefont {Milz}, \citenamefont {Murao}, \citenamefont
  {Quintino},\ and\ \citenamefont {Modi}}]{TarantoEtAlArXiv2025}%
  \BibitemOpen
  \bibfield  {author} {\bibinfo {author} {\bibfnamefont {P.}~\bibnamefont
  {Taranto}}, \bibinfo {author} {\bibfnamefont {S.}~\bibnamefont {Milz}},
  \bibinfo {author} {\bibfnamefont {M.}~\bibnamefont {Murao}}, \bibinfo
  {author} {\bibfnamefont {M.~T.}\ \bibnamefont {Quintino}},\ and\ \bibinfo
  {author} {\bibfnamefont {K.}~\bibnamefont {Modi}},\ }\bibfield  {title}
  {\bibinfo {title} {{Higher-Order Quantum Operations}},\ }\href
  {https://arxiv.org/abs/2503.09693} {\bibfield  {journal} {\bibinfo  {journal}
  {arXiv 2503.09693}\ } (\bibinfo {year} {2025})}\BibitemShut {NoStop}%
\bibitem [{\citenamefont {Barnett}\ and\ \citenamefont
  {Croke}(2009)}]{BarnettCrokeAOP2009}%
  \BibitemOpen
  \bibfield  {author} {\bibinfo {author} {\bibfnamefont {S.~M.}\ \bibnamefont
  {Barnett}}\ and\ \bibinfo {author} {\bibfnamefont {S.}~\bibnamefont
  {Croke}},\ }\bibfield  {title} {\bibinfo {title} {Quantum state
  discrimination},\ }\href {https://doi.org/10.1364/AOP.1.000238} {\bibfield
  {journal} {\bibinfo  {journal} {Adv. Opt. Photon.}\ }\textbf {\bibinfo
  {volume} {1}},\ \bibinfo {pages} {238} (\bibinfo {year} {2009})}\BibitemShut
  {NoStop}%
\bibitem [{\citenamefont {Bae}\ and\ \citenamefont
  {Kwek}(2015)}]{BaeKwekJPA2015}%
  \BibitemOpen
  \bibfield  {author} {\bibinfo {author} {\bibfnamefont {J.}~\bibnamefont
  {Bae}}\ and\ \bibinfo {author} {\bibfnamefont {L.-C.}\ \bibnamefont {Kwek}},\
  }\bibfield  {title} {\bibinfo {title} {Quantum state discrimination and its
  applications},\ }\href {https://doi.org/10.1088/1751-8113/48/8/083001}
  {\bibfield  {journal} {\bibinfo  {journal} {J. Phys. A: Math. Theor.}\
  }\textbf {\bibinfo {volume} {48}},\ \bibinfo {pages} {083001} (\bibinfo
  {year} {2015})}\BibitemShut {NoStop}%
\bibitem [{\citenamefont {Watrous}(2018)}]{WatrousBook2018}%
  \BibitemOpen
  \bibfield  {author} {\bibinfo {author} {\bibfnamefont {J.}~\bibnamefont
  {Watrous}},\ }\href@noop {} {\emph {\bibinfo {title} {The Theory of Quantum
  Information}}}\ (\bibinfo  {publisher} {Cambridge University Press},\
  \bibinfo {year} {2018})\BibitemShut {NoStop}%
\bibitem [{\citenamefont {Cie\'sli\'nski}\ \emph {et~al.}(2024)\citenamefont
  {Cie\'sli\'nski}, \citenamefont {Imai}, \citenamefont {Dziewior},
  \citenamefont {G\"uhne}, \citenamefont {Knips}, \citenamefont {Laskowski},
  \citenamefont {Meinecke}, \citenamefont {Paterek},\ and\ \citenamefont
  {V\'ertesi}}]{CieslinskiEtAlPR2024}%
  \BibitemOpen
  \bibfield  {author} {\bibinfo {author} {\bibfnamefont {P.}~\bibnamefont
  {Cie\'sli\'nski}}, \bibinfo {author} {\bibfnamefont {S.}~\bibnamefont
  {Imai}}, \bibinfo {author} {\bibfnamefont {J.}~\bibnamefont {Dziewior}},
  \bibinfo {author} {\bibfnamefont {O.}~\bibnamefont {G\"uhne}}, \bibinfo
  {author} {\bibfnamefont {L.}~\bibnamefont {Knips}}, \bibinfo {author}
  {\bibfnamefont {W.}~\bibnamefont {Laskowski}}, \bibinfo {author}
  {\bibfnamefont {J.}~\bibnamefont {Meinecke}}, \bibinfo {author}
  {\bibfnamefont {T.}~\bibnamefont {Paterek}},\ and\ \bibinfo {author}
  {\bibfnamefont {T.}~\bibnamefont {V\'ertesi}},\ }\bibfield  {title} {\bibinfo
  {title} {Analysing quantum systems with randomised measurements},\ }\href
  {https://doi.org/https://doi.org/10.1016/j.physrep.2024.09.009} {\bibfield
  {journal} {\bibinfo  {journal} {Phys. Rep.}\ }\textbf {\bibinfo {volume}
  {1095}},\ \bibinfo {pages} {1} (\bibinfo {year} {2024})}\BibitemShut
  {NoStop}%
\bibitem [{\citenamefont {Everett}(1957)}]{EverettRMP1957}%
  \BibitemOpen
  \bibfield  {author} {\bibinfo {author} {\bibfnamefont {H.}~\bibnamefont
  {Everett}},\ }\bibfield  {title} {\bibinfo {title} {{"Relative State"
  Formulation of Quantum Mechanics}},\ }\href
  {https://doi.org/10.1103/RevModPhys.29.454} {\bibfield  {journal} {\bibinfo
  {journal} {Rev. Mod. Phys.}\ }\textbf {\bibinfo {volume} {29}},\ \bibinfo
  {pages} {454} (\bibinfo {year} {1957})}\BibitemShut {NoStop}%
\bibitem [{\citenamefont {DeWitt}(1970)}]{DeWittPT1970}%
  \BibitemOpen
  \bibfield  {author} {\bibinfo {author} {\bibfnamefont {B.~S.}\ \bibnamefont
  {DeWitt}},\ }\bibfield  {title} {\bibinfo {title} {{Quantum Mechanics and
  Reality}},\ }\href {https://doi.org/10.1063/1.3022331} {\bibfield  {journal}
  {\bibinfo  {journal} {Phys. Today}\ }\textbf {\bibinfo {volume} {23}},\
  \bibinfo {pages} {30} (\bibinfo {year} {1970})}\BibitemShut {NoStop}%
\bibitem [{\citenamefont {Saunders}\ \emph {et~al.}(2010)\citenamefont
  {Saunders}, \citenamefont {Barrett}, \citenamefont {Kent},\ and\
  \citenamefont {Wallace}}]{SaundersEtAlBook2010}%
  \BibitemOpen
  \bibinfo {editor} {\bibfnamefont {S.}~\bibnamefont {Saunders}}, \bibinfo
  {editor} {\bibfnamefont {J.}~\bibnamefont {Barrett}}, \bibinfo {editor}
  {\bibfnamefont {A.}~\bibnamefont {Kent}},\ and\ \bibinfo {editor}
  {\bibfnamefont {D.}~\bibnamefont {Wallace}},\ eds.,\ \href@noop {} {\emph
  {\bibinfo {title} {{Many Worlds? Everett, Quantum Theory and Reality}}}}\
  (\bibinfo  {publisher} {Oxford University Press},\ \bibinfo {address}
  {Oxford},\ \bibinfo {year} {2010})\BibitemShut {NoStop}%
\bibitem [{\citenamefont {Wallace}(2012)}]{WallaceBook2012}%
  \BibitemOpen
  \bibfield  {author} {\bibinfo {author} {\bibfnamefont {D.}~\bibnamefont
  {Wallace}},\ }\href@noop {} {\emph {\bibinfo {title} {{The Emergent
  Multiverse: Quantum Theory According to the Everett Interpretation}}}}\
  (\bibinfo  {publisher} {Oxford University Press},\ \bibinfo {address}
  {Oxford},\ \bibinfo {year} {2012})\BibitemShut {NoStop}%
\bibitem [{\citenamefont {Vaidman}(2021)}]{Vaidman2021}%
  \BibitemOpen
  \bibfield  {author} {\bibinfo {author} {\bibfnamefont {L.}~\bibnamefont
  {Vaidman}},\ }\bibinfo {title} {{Stanford Encyclopedia of Philosophy}}\
  (\bibinfo {year} {2021})\ Chap.\ \bibinfo {chapter} {{Many-Worlds
  Interpretation of Quantum Mechanics}},\ \bibinfo {edition} {fall 2021}\
  ed.\BibitemShut {Stop}%
\bibitem [{\citenamefont {Rovelli}(1996)}]{RovelliIJTP1996}%
  \BibitemOpen
  \bibfield  {author} {\bibinfo {author} {\bibfnamefont {C.}~\bibnamefont
  {Rovelli}},\ }\bibfield  {title} {\bibinfo {title} {Relational quantum
  mechanics},\ }\href {https://doi.org/10.1007/BF02302261} {\bibfield
  {journal} {\bibinfo  {journal} {Int. J. Theor. Phys.}\ }\textbf {\bibinfo
  {volume} {35}},\ \bibinfo {pages} {1637} (\bibinfo {year}
  {1996})}\BibitemShut {NoStop}%
\bibitem [{\citenamefont {Griffiths}(1984)}]{GriffithsJSP1984}%
  \BibitemOpen
  \bibfield  {author} {\bibinfo {author} {\bibfnamefont {R.~B.}\ \bibnamefont
  {Griffiths}},\ }\bibfield  {title} {\bibinfo {title} {Consistent histories
  and the interpretation of quantum mechanics},\ }\href
  {https://doi.org/10.1007/BF01015734} {\bibfield  {journal} {\bibinfo
  {journal} {J. Stat. Phys.}\ }\textbf {\bibinfo {volume} {36}},\ \bibinfo
  {pages} {219} (\bibinfo {year} {1984})}\BibitemShut {NoStop}%
\bibitem [{\citenamefont {Griffiths}(2002)}]{GriffithsBook2002}%
  \BibitemOpen
  \bibfield  {author} {\bibinfo {author} {\bibfnamefont {R.~B.}\ \bibnamefont
  {Griffiths}},\ }\href {https://doi.org/10.1017/CBO9780511606052} {\emph
  {\bibinfo {title} {{Consistent Quantum Theory}}}}\ (\bibinfo  {publisher}
  {Cambridge University Press},\ \bibinfo {address} {Cambridge},\ \bibinfo
  {year} {2002})\BibitemShut {NoStop}%
\bibitem [{\citenamefont {Griffiths}(2019)}]{Griffiths2019}%
  \BibitemOpen
  \bibfield  {author} {\bibinfo {author} {\bibfnamefont {R.~B.}\ \bibnamefont
  {Griffiths}},\ }\bibinfo {title} {{Stanford Encyclopedia of Philosophy}}\
  (\bibinfo {year} {2019})\ Chap.\ \bibinfo {chapter} {{The Consistent
  Histories Approach to Quantum Mechanics}},\ \bibinfo {edition} {summer 2019}\
  ed.\BibitemShut {Stop}%
\bibitem [{\citenamefont {Bassi}\ \emph {et~al.}(2013)\citenamefont {Bassi},
  \citenamefont {Lochan}, \citenamefont {Satin}, \citenamefont {Singh},\ and\
  \citenamefont {Ulbricht}}]{BassiEtAlRMP2013}%
  \BibitemOpen
  \bibfield  {author} {\bibinfo {author} {\bibfnamefont {A.}~\bibnamefont
  {Bassi}}, \bibinfo {author} {\bibfnamefont {K.}~\bibnamefont {Lochan}},
  \bibinfo {author} {\bibfnamefont {S.}~\bibnamefont {Satin}}, \bibinfo
  {author} {\bibfnamefont {T.~P.}\ \bibnamefont {Singh}},\ and\ \bibinfo
  {author} {\bibfnamefont {H.}~\bibnamefont {Ulbricht}},\ }\bibfield  {title}
  {\bibinfo {title} {Models of wave-function collapse, underlying theories and
  experimental tests},\ }\href {https://doi.org/10.1103/RevModPhys.85.471}
  {\bibfield  {journal} {\bibinfo  {journal} {Rev. Mod. Phys.}\ }\textbf
  {\bibinfo {volume} {85}},\ \bibinfo {pages} {471} (\bibinfo {year}
  {2013})}\BibitemShut {NoStop}%
\bibitem [{\citenamefont {Ghirardi}\ and\ \citenamefont
  {Bassi}(2024)}]{GhirardiBassi2024}%
  \BibitemOpen
  \bibfield  {author} {\bibinfo {author} {\bibfnamefont {G.}~\bibnamefont
  {Ghirardi}}\ and\ \bibinfo {author} {\bibfnamefont {A.}~\bibnamefont
  {Bassi}},\ }\bibinfo {title} {{The Stanford Encyclopedia of Philosophy}}\
  (\bibinfo {year} {2024})\ Chap.\ \bibinfo {chapter} {{Collapse Theories}},\
  \bibinfo {edition} {fall 2024}\ ed.\BibitemShut {Stop}%
\bibitem [{\citenamefont {Kiefer}(2007)}]{KieferInBook2007}%
  \BibitemOpen
  \bibfield  {author} {\bibinfo {author} {\bibfnamefont {C.}~\bibnamefont
  {Kiefer}},\ }\bibinfo {title} {Why quantum gravity?},\ in\ \href
  {https://doi.org/10.1007/978-3-540-71117-9_7} {\emph {\bibinfo {booktitle}
  {Approaches to Fundamental Physics}}},\ \bibinfo {editor} {edited by\
  \bibinfo {editor} {\bibfnamefont {S.~E.}\ \bibnamefont {Stamatescu},
  \bibfnamefont {I.~O.}}}\ (\bibinfo  {publisher} {Springer Berlin
  Heidelberg},\ \bibinfo {year} {2007})\ pp.\ \bibinfo {pages}
  {123--130}\BibitemShut {NoStop}%
\bibitem [{\citenamefont {Oppenheim}(2023)}]{OppenheimIJMP2023}%
  \BibitemOpen
  \bibfield  {author} {\bibinfo {author} {\bibfnamefont {J.}~\bibnamefont
  {Oppenheim}},\ }\bibfield  {title} {\bibinfo {title} {Is it time to rethink
  quantum gravity?},\ }\href {https://doi.org/10.1142/S0218271823420245}
  {\bibfield  {journal} {\bibinfo  {journal} {Int. J. Mod. Phys. D}\ }\textbf
  {\bibinfo {volume} {32}},\ \bibinfo {pages} {2342024} (\bibinfo {year}
  {2023})}\BibitemShut {NoStop}%
\bibitem [{\citenamefont {Carr}(2007)}]{Carr2007}%
  \BibitemOpen
  \bibinfo {editor} {\bibfnamefont {B.}~\bibnamefont {Carr}},\ ed.,\ \href@noop
  {} {\emph {\bibinfo {title} {{Universe or Multiverse?}}}}\ (\bibinfo
  {publisher} {Cambridge University Press},\ \bibinfo {address} {Cambridge},\
  \bibinfo {year} {2007})\BibitemShut {NoStop}%
\bibitem [{\citenamefont {Mukhanov}(2005)}]{MukhanovBook2005}%
  \BibitemOpen
  \bibfield  {author} {\bibinfo {author} {\bibfnamefont {V.}~\bibnamefont
  {Mukhanov}},\ }\href@noop {} {\emph {\bibinfo {title} {{Physical Foundations
  of Cosmology}}}}\ (\bibinfo  {publisher} {Cambridge University Press},\
  \bibinfo {address} {Cambridge},\ \bibinfo {year} {2005})\BibitemShut
  {NoStop}%
\bibitem [{\citenamefont {Zurek}(2003)}]{ZurekRMP2003}%
  \BibitemOpen
  \bibfield  {author} {\bibinfo {author} {\bibfnamefont {W.~H.}\ \bibnamefont
  {Zurek}},\ }\bibfield  {title} {\bibinfo {title} {Decoherence, einselection
  and the quantum origins of the classical},\ }\href
  {https://doi.org/10.1103/RevModPhys.75.715} {\bibfield  {journal} {\bibinfo
  {journal} {Rev. Mod. Phys.}\ }\textbf {\bibinfo {volume} {75}},\ \bibinfo
  {pages} {715} (\bibinfo {year} {2003})}\BibitemShut {NoStop}%
\bibitem [{\citenamefont {Joos}\ \emph {et~al.}(2003)\citenamefont {Joos},
  \citenamefont {Zeh}, \citenamefont {Kiefer}, \citenamefont {Giulini},
  \citenamefont {Kupsch},\ and\ \citenamefont {Stamatescu}}]{JoosEtAlBook2003}%
  \BibitemOpen
  \bibfield  {author} {\bibinfo {author} {\bibfnamefont {E.}~\bibnamefont
  {Joos}}, \bibinfo {author} {\bibfnamefont {H.~D.}\ \bibnamefont {Zeh}},
  \bibinfo {author} {\bibfnamefont {C.}~\bibnamefont {Kiefer}}, \bibinfo
  {author} {\bibfnamefont {D.}~\bibnamefont {Giulini}}, \bibinfo {author}
  {\bibfnamefont {J.}~\bibnamefont {Kupsch}},\ and\ \bibinfo {author}
  {\bibfnamefont {I.-O.}\ \bibnamefont {Stamatescu}},\ }\href@noop {} {\emph
  {\bibinfo {title} {Decoherence and the Appearance of a Classical World in
  Quantum Theory}}}\ (\bibinfo  {publisher} {Springer},\ \bibinfo {address}
  {Berlin Heidelberg},\ \bibinfo {year} {2003})\BibitemShut {NoStop}%
\bibitem [{\citenamefont {Schlosshauer}(2019)}]{SchlosshauerPR2019}%
  \BibitemOpen
  \bibfield  {author} {\bibinfo {author} {\bibfnamefont {M.}~\bibnamefont
  {Schlosshauer}},\ }\bibfield  {title} {\bibinfo {title} {Quantum
  decoherence},\ }\href
  {https://doi.org/https://doi.org/10.1016/j.physrep.2019.10.001} {\bibfield
  {journal} {\bibinfo  {journal} {Phys. Rep.}\ }\textbf {\bibinfo {volume}
  {831}},\ \bibinfo {pages} {1} (\bibinfo {year} {2019})}\BibitemShut {NoStop}%
\bibitem [{\citenamefont {Tegmark}(2003)}]{TegmarkSA2003}%
  \BibitemOpen
  \bibfield  {author} {\bibinfo {author} {\bibfnamefont {M.}~\bibnamefont
  {Tegmark}},\ }\bibfield  {title} {\bibinfo {title} {Parallel universes},\
  }\href {https://doi.org/10.1038/scientificamerican0503-40} {\bibfield
  {journal} {\bibinfo  {journal} {Sci. Am.}\ }\textbf {\bibinfo {volume}
  {288}},\ \bibinfo {pages} {40} (\bibinfo {year} {2003})}\BibitemShut
  {NoStop}%
\bibitem [{\citenamefont {Blencowe}(1991)}]{BlencoweAP1991}%
  \BibitemOpen
  \bibfield  {author} {\bibinfo {author} {\bibfnamefont {M.}~\bibnamefont
  {Blencowe}},\ }\bibfield  {title} {\bibinfo {title} {The consistent histories
  interpretation of quantum fields in curved spacetime},\ }\href
  {https://doi.org/https://doi.org/10.1016/0003-4916(91)90193-C} {\bibfield
  {journal} {\bibinfo  {journal} {Ann. Phys.}\ }\textbf {\bibinfo {volume}
  {211}},\ \bibinfo {pages} {87} (\bibinfo {year} {1991})}\BibitemShut
  {NoStop}%
\bibitem [{\citenamefont {Halliwell}\ and\ \citenamefont
  {Thorwart}(2002)}]{HalliwellThorwartPRD2002}%
  \BibitemOpen
  \bibfield  {author} {\bibinfo {author} {\bibfnamefont {J.~J.}\ \bibnamefont
  {Halliwell}}\ and\ \bibinfo {author} {\bibfnamefont {J.}~\bibnamefont
  {Thorwart}},\ }\bibfield  {title} {\bibinfo {title} {{Life in an energy
  eigenstate: Decoherent histories analysis of a model timeless universe}},\
  }\href {https://doi.org/10.1103/PhysRevD.65.104009} {\bibfield  {journal}
  {\bibinfo  {journal} {Phys. Rev. D}\ }\textbf {\bibinfo {volume} {65}},\
  \bibinfo {pages} {104009} (\bibinfo {year} {2002})}\BibitemShut {NoStop}%
\bibitem [{\citenamefont {Christodoulakis}\ and\ \citenamefont
  {Wallden}(2011)}]{ChristodoulakisWalldenJPCS2011}%
  \BibitemOpen
  \bibfield  {author} {\bibinfo {author} {\bibfnamefont {T.}~\bibnamefont
  {Christodoulakis}}\ and\ \bibinfo {author} {\bibfnamefont {P.}~\bibnamefont
  {Wallden}},\ }\bibfield  {title} {\bibinfo {title} {{The Problem of Time in
  Quantum Cosmology: A Decoherent Histories View}},\ }\href
  {https://doi.org/10.1088/1742-6596/283/1/012041} {\bibfield  {journal}
  {\bibinfo  {journal} {J. Phys. Conf. Series}\ }\textbf {\bibinfo {volume}
  {283}},\ \bibinfo {pages} {012041} (\bibinfo {year} {2011})}\BibitemShut
  {NoStop}%
\bibitem [{\citenamefont {Gherardini}\ and\ \citenamefont
  {De~Chiara}(2024)}]{GherardiniDeChiaraPRXQ2024}%
  \BibitemOpen
  \bibfield  {author} {\bibinfo {author} {\bibfnamefont {S.}~\bibnamefont
  {Gherardini}}\ and\ \bibinfo {author} {\bibfnamefont {G.}~\bibnamefont
  {De~Chiara}},\ }\bibfield  {title} {\bibinfo {title} {Quasiprobabilities in
  quantum thermodynamics and many-body systems},\ }\href
  {https://doi.org/10.1103/PRXQuantum.5.030201} {\bibfield  {journal} {\bibinfo
   {journal} {PRX Quantum}\ }\textbf {\bibinfo {volume} {5}},\ \bibinfo {pages}
  {030201} (\bibinfo {year} {2024})}\BibitemShut {NoStop}%
\bibitem [{\citenamefont {Arvidsson-Shukur}\ \emph {et~al.}(2024)\citenamefont
  {Arvidsson-Shukur}, \citenamefont {Braasch~Jr}, \citenamefont {De~Bi\`evre},
  \citenamefont {Dressel}, \citenamefont {Jordan}, \citenamefont {Langrenez},
  \citenamefont {Lostaglio}, \citenamefont {Lundeen},\ and\ \citenamefont
  {Halpern}}]{ArvidssonEtAlNJP2024}%
  \BibitemOpen
  \bibfield  {author} {\bibinfo {author} {\bibfnamefont {D.~R.~M.}\
  \bibnamefont {Arvidsson-Shukur}}, \bibinfo {author} {\bibfnamefont {W.~F.}\
  \bibnamefont {Braasch~Jr}}, \bibinfo {author} {\bibfnamefont
  {S.}~\bibnamefont {De~Bi\`evre}}, \bibinfo {author} {\bibfnamefont
  {J.}~\bibnamefont {Dressel}}, \bibinfo {author} {\bibfnamefont {A.~N.}\
  \bibnamefont {Jordan}}, \bibinfo {author} {\bibfnamefont {C.}~\bibnamefont
  {Langrenez}}, \bibinfo {author} {\bibfnamefont {M.}~\bibnamefont
  {Lostaglio}}, \bibinfo {author} {\bibfnamefont {J.~S.}\ \bibnamefont
  {Lundeen}},\ and\ \bibinfo {author} {\bibfnamefont {N.~Y.}\ \bibnamefont
  {Halpern}},\ }\bibfield  {title} {\bibinfo {title} {{Properties and
  applications of the Kirkwood–Dirac distribution}},\ }\href
  {https://doi.org/10.1088/1367-2630/ada05d} {\bibfield  {journal} {\bibinfo
  {journal} {New Journal of Physics}\ }\textbf {\bibinfo {volume} {26}},\
  \bibinfo {pages} {121201} (\bibinfo {year} {2024})}\BibitemShut {NoStop}%
\bibitem [{\citenamefont {Isham}(1994)}]{IshamJMP1994}%
  \BibitemOpen
  \bibfield  {author} {\bibinfo {author} {\bibfnamefont {C.~J.}\ \bibnamefont
  {Isham}},\ }\bibfield  {title} {\bibinfo {title} {{Quantum logic and the
  histories approach to quantum theory}},\ }\href
  {https://doi.org/10.1063/1.530544} {\bibfield  {journal} {\bibinfo  {journal}
  {J. Math. Phys.}\ }\textbf {\bibinfo {volume} {35}},\ \bibinfo {pages} {2157}
  (\bibinfo {year} {1994})}\BibitemShut {NoStop}%
\bibitem [{\citenamefont {Isham}\ \emph {et~al.}(1994)\citenamefont {Isham},
  \citenamefont {Linden},\ and\ \citenamefont
  {Schreckenberg}}]{IshamLindenSchreckenbergJMP2004}%
  \BibitemOpen
  \bibfield  {author} {\bibinfo {author} {\bibfnamefont {C.~J.}\ \bibnamefont
  {Isham}}, \bibinfo {author} {\bibfnamefont {N.}~\bibnamefont {Linden}},\ and\
  \bibinfo {author} {\bibfnamefont {S.}~\bibnamefont {Schreckenberg}},\
  }\bibfield  {title} {\bibinfo {title} {{The classification of decoherence
  functionals: An analog of Gleason’s theorem}},\ }\href
  {https://doi.org/10.1063/1.530679} {\bibfield  {journal} {\bibinfo  {journal}
  {J. Math. Phys.}\ }\textbf {\bibinfo {volume} {35}},\ \bibinfo {pages} {6360}
  (\bibinfo {year} {1994})}\BibitemShut {NoStop}%
\bibitem [{\citenamefont {Cotler}\ \emph {et~al.}(2018)\citenamefont {Cotler},
  \citenamefont {Jian}, \citenamefont {Qi},\ and\ \citenamefont
  {Wilczek}}]{CotlerEtAlJHEP2018}%
  \BibitemOpen
  \bibfield  {author} {\bibinfo {author} {\bibfnamefont {J.}~\bibnamefont
  {Cotler}}, \bibinfo {author} {\bibfnamefont {C.-M.}\ \bibnamefont {Jian}},
  \bibinfo {author} {\bibfnamefont {X.-L.}\ \bibnamefont {Qi}},\ and\ \bibinfo
  {author} {\bibfnamefont {F.}~\bibnamefont {Wilczek}},\ }\bibfield  {title}
  {\bibinfo {title} {Superdensity operators for spacetime quantum mechanics},\
  }\href {https://doi.org/10.1007/JHEP09(2018)093} {\bibfield  {journal}
  {\bibinfo  {journal} {J. High Energ. Phys.}\ }\textbf {\bibinfo {volume}
  {2018}},\ \bibinfo {pages} {93}}\BibitemShut {NoStop}%
\bibitem [{\citenamefont {Alicki}\ and\ \citenamefont
  {Fannes}(1994)}]{AlickiFannesLMP1994}%
  \BibitemOpen
  \bibfield  {author} {\bibinfo {author} {\bibfnamefont {R.}~\bibnamefont
  {Alicki}}\ and\ \bibinfo {author} {\bibfnamefont {M.}~\bibnamefont
  {Fannes}},\ }\bibfield  {title} {\bibinfo {title} {Defining quantum dynamical
  entropy},\ }\href {https://doi.org/10.1007/BF00761125} {\bibfield  {journal}
  {\bibinfo  {journal} {Lett. Math. Phys.}\ }\textbf {\bibinfo {volume} {32}},\
  \bibinfo {pages} {75} (\bibinfo {year} {1994})}\BibitemShut {NoStop}%
\bibitem [{\citenamefont {Dowling}\ and\ \citenamefont
  {Modi}(2024)}]{DowlingModiPRXQ2024}%
  \BibitemOpen
  \bibfield  {author} {\bibinfo {author} {\bibfnamefont {N.}~\bibnamefont
  {Dowling}}\ and\ \bibinfo {author} {\bibfnamefont {K.}~\bibnamefont {Modi}},\
  }\bibfield  {title} {\bibinfo {title} {Operational metric for quantum chaos
  and the corresponding spatiotemporal-entanglement structure},\ }\href
  {https://doi.org/10.1103/PRXQuantum.5.010314} {\bibfield  {journal} {\bibinfo
   {journal} {PRX Quantum}\ }\textbf {\bibinfo {volume} {5}},\ \bibinfo {pages}
  {010314} (\bibinfo {year} {2024})}\BibitemShut {NoStop}%
\bibitem [{\citenamefont {Polyakov}\ and\ \citenamefont
  {Arefyeva}(2024)}]{PolyakovArefyevaPRA2024}%
  \BibitemOpen
  \bibfield  {author} {\bibinfo {author} {\bibfnamefont {E.}~\bibnamefont
  {Polyakov}}\ and\ \bibinfo {author} {\bibfnamefont {N.}~\bibnamefont
  {Arefyeva}},\ }\bibfield  {title} {\bibinfo {title} {Probing quantum chaos
  with the entropy of decoherent histories},\ }\href
  {https://doi.org/10.1103/PhysRevA.109.062204} {\bibfield  {journal} {\bibinfo
   {journal} {Phys. Rev. A}\ }\textbf {\bibinfo {volume} {109}},\ \bibinfo
  {pages} {062204} (\bibinfo {year} {2024})}\BibitemShut {NoStop}%
\bibitem [{\citenamefont {O'Donovan}\ \emph {et~al.}(2025)\citenamefont
  {O'Donovan}, \citenamefont {Dowling}, \citenamefont {Modi},\ and\
  \citenamefont {Mitchison}}]{ODonovanEtAlArXiv2025}%
  \BibitemOpen
  \bibfield  {author} {\bibinfo {author} {\bibfnamefont {P.}~\bibnamefont
  {O'Donovan}}, \bibinfo {author} {\bibfnamefont {N.}~\bibnamefont {Dowling}},
  \bibinfo {author} {\bibfnamefont {K.}~\bibnamefont {Modi}},\ and\ \bibinfo
  {author} {\bibfnamefont {M.~T.}\ \bibnamefont {Mitchison}},\ }\bibfield
  {title} {\bibinfo {title} {Diagnosing chaos with projected ensembles of
  process tensors},\ }\href {https://arxiv.org/abs/2502.13930} {\bibfield
  {journal} {\bibinfo  {journal} {arXiv 2502.13930}\ } (\bibinfo {year}
  {2025})}\BibitemShut {NoStop}%
\bibitem [{\citenamefont {Bittner}\ and\ \citenamefont
  {Rossky}(1997)}]{BittnerRosskyJCP1997}%
  \BibitemOpen
  \bibfield  {author} {\bibinfo {author} {\bibfnamefont {E.~R.}\ \bibnamefont
  {Bittner}}\ and\ \bibinfo {author} {\bibfnamefont {P.~J.}\ \bibnamefont
  {Rossky}},\ }\bibfield  {title} {\bibinfo {title} {Decoherent histories and
  nonadiabatic quantum molecular dynamics simulations},\ }\href
  {https://doi.org/10.1063/1.475013} {\bibfield  {journal} {\bibinfo  {journal}
  {J. Chem. Phys.}\ }\textbf {\bibinfo {volume} {107}},\ \bibinfo {pages}
  {8611} (\bibinfo {year} {1997})}\BibitemShut {NoStop}%
\bibitem [{\citenamefont {Smirne}\ \emph {et~al.}(2018)\citenamefont {Smirne},
  \citenamefont {Egloff}, \citenamefont {D\'iaz}, \citenamefont {Plenio},\ and\
  \citenamefont {Huelga}}]{SmirneEtAlQST2018}%
  \BibitemOpen
  \bibfield  {author} {\bibinfo {author} {\bibfnamefont {A.}~\bibnamefont
  {Smirne}}, \bibinfo {author} {\bibfnamefont {D.}~\bibnamefont {Egloff}},
  \bibinfo {author} {\bibfnamefont {M.~G.}\ \bibnamefont {D\'iaz}}, \bibinfo
  {author} {\bibfnamefont {M.~B.}\ \bibnamefont {Plenio}},\ and\ \bibinfo
  {author} {\bibfnamefont {S.~F.}\ \bibnamefont {Huelga}},\ }\bibfield  {title}
  {\bibinfo {title} {Coherence and non-classicality of quantum {M}arkov
  processes},\ }\href {https://doi.org/10.1088/2058-9565/aaebd5} {\bibfield
  {journal} {\bibinfo  {journal} {Quantum Sci. Technol.}\ }\textbf {\bibinfo
  {volume} {4}},\ \bibinfo {pages} {01LT01} (\bibinfo {year}
  {2018})}\BibitemShut {NoStop}%
\bibitem [{\citenamefont {Strasberg}\ and\ \citenamefont
  {D\'{\i}az}(2019)}]{StrasbergDiazPRA2019}%
  \BibitemOpen
  \bibfield  {author} {\bibinfo {author} {\bibfnamefont {P.}~\bibnamefont
  {Strasberg}}\ and\ \bibinfo {author} {\bibfnamefont {M.~G.}\ \bibnamefont
  {D\'{\i}az}},\ }\bibfield  {title} {\bibinfo {title} {Classical quantum
  stochastic processes},\ }\href {https://doi.org/10.1103/PhysRevA.100.022120}
  {\bibfield  {journal} {\bibinfo  {journal} {Phys. Rev. A}\ }\textbf {\bibinfo
  {volume} {100}},\ \bibinfo {pages} {022120} (\bibinfo {year}
  {2019})}\BibitemShut {NoStop}%
\bibitem [{\citenamefont {Milz}\ \emph
  {et~al.}(2020{\natexlab{a}})\citenamefont {Milz}, \citenamefont {Sakuldee},
  \citenamefont {Pollock},\ and\ \citenamefont {Modi}}]{MilzEtAlQuantum2020}%
  \BibitemOpen
  \bibfield  {author} {\bibinfo {author} {\bibfnamefont {S.}~\bibnamefont
  {Milz}}, \bibinfo {author} {\bibfnamefont {F.}~\bibnamefont {Sakuldee}},
  \bibinfo {author} {\bibfnamefont {F.~A.}\ \bibnamefont {Pollock}},\ and\
  \bibinfo {author} {\bibfnamefont {K.}~\bibnamefont {Modi}},\ }\bibfield
  {title} {\bibinfo {title} {Kolmogorov extension theorem for (quantum) causal
  modelling and general probabilistic theories},\ }\href
  {https://doi.org/10.22331/q-2020-04-20-255} {\bibfield  {journal} {\bibinfo
  {journal} {{Quantum}}\ }\textbf {\bibinfo {volume} {4}},\ \bibinfo {pages}
  {255} (\bibinfo {year} {2020}{\natexlab{a}})}\BibitemShut {NoStop}%
\bibitem [{\citenamefont {Milz}\ \emph
  {et~al.}(2020{\natexlab{b}})\citenamefont {Milz}, \citenamefont {Egloff},
  \citenamefont {Taranto}, \citenamefont {Theurer}, \citenamefont {Plenio},
  \citenamefont {Smirne},\ and\ \citenamefont {Huelga}}]{MilzEtAlPRX2020}%
  \BibitemOpen
  \bibfield  {author} {\bibinfo {author} {\bibfnamefont {S.}~\bibnamefont
  {Milz}}, \bibinfo {author} {\bibfnamefont {D.}~\bibnamefont {Egloff}},
  \bibinfo {author} {\bibfnamefont {P.}~\bibnamefont {Taranto}}, \bibinfo
  {author} {\bibfnamefont {T.}~\bibnamefont {Theurer}}, \bibinfo {author}
  {\bibfnamefont {M.~B.}\ \bibnamefont {Plenio}}, \bibinfo {author}
  {\bibfnamefont {A.}~\bibnamefont {Smirne}},\ and\ \bibinfo {author}
  {\bibfnamefont {S.~F.}\ \bibnamefont {Huelga}},\ }\bibfield  {title}
  {\bibinfo {title} {{When Is a Non-Markovian Quantum Process Classical?}},\
  }\href {https://doi.org/10.1103/PhysRevX.10.041049} {\bibfield  {journal}
  {\bibinfo  {journal} {Phys. Rev. X}\ }\textbf {\bibinfo {volume} {10}},\
  \bibinfo {pages} {041049} (\bibinfo {year} {2020}{\natexlab{b}})}\BibitemShut
  {NoStop}%
\bibitem [{\citenamefont {Sza{\'{n}}kowski}\ and\ \citenamefont
  {Cywi{\'{n}}ski}(2024)}]{SzankowskiCywinskiQuantum2024}%
  \BibitemOpen
  \bibfield  {author} {\bibinfo {author} {\bibfnamefont {P.}~\bibnamefont
  {Sza{\'{n}}kowski}}\ and\ \bibinfo {author} {\bibfnamefont
  {{\L{}}.}~\bibnamefont {Cywi{\'{n}}ski}},\ }\bibfield  {title} {\bibinfo
  {title} {Objectivity of classical quantum stochastic processes},\ }\href
  {https://doi.org/10.22331/q-2024-06-27-1390} {\bibfield  {journal} {\bibinfo
  {journal} {{Quantum}}\ }\textbf {\bibinfo {volume} {8}},\ \bibinfo {pages}
  {1390} (\bibinfo {year} {2024})}\BibitemShut {NoStop}%
\bibitem [{\citenamefont {Emary}\ \emph {et~al.}(2014)\citenamefont {Emary},
  \citenamefont {Lambert},\ and\ \citenamefont
  {Nori}}]{EmaryLambertNoriRPP2014}%
  \BibitemOpen
  \bibfield  {author} {\bibinfo {author} {\bibfnamefont {C.}~\bibnamefont
  {Emary}}, \bibinfo {author} {\bibfnamefont {N.}~\bibnamefont {Lambert}},\
  and\ \bibinfo {author} {\bibfnamefont {F.}~\bibnamefont {Nori}},\ }\bibfield
  {title} {\bibinfo {title} {Leggett-{G}arg inequalities},\ }\href
  {https://doi.org/10.1088/0034-4885/77/1/016001} {\bibfield  {journal}
  {\bibinfo  {journal} {Rep. Prog. Phys.}\ }\textbf {\bibinfo {volume} {77}},\
  \bibinfo {pages} {039501} (\bibinfo {year} {2014})}\BibitemShut {NoStop}%
\bibitem [{\citenamefont {Zurek}(2009)}]{ZurekNP2009}%
  \BibitemOpen
  \bibfield  {author} {\bibinfo {author} {\bibfnamefont {W.~H.}\ \bibnamefont
  {Zurek}},\ }\bibfield  {title} {\bibinfo {title} {{Quantum Darwinism}},\
  }\bibfield  {journal} {\bibinfo  {journal} {Nature Phys.}\ }\textbf {\bibinfo
  {volume} {5}},\ \href {https://doi.org/doi.org/10.1038/nphys1202}
  {doi.org/10.1038/nphys1202} (\bibinfo {year} {2009})\BibitemShut {NoStop}%
\bibitem [{\citenamefont {Korbicz}(2021)}]{KorbiczQuantum2021}%
  \BibitemOpen
  \bibfield  {author} {\bibinfo {author} {\bibfnamefont {J.~K.}\ \bibnamefont
  {Korbicz}},\ }\bibfield  {title} {\bibinfo {title} {Roads to objectivity:
  {Q}uantum {D}arwinism, {S}pectrum {B}roadcast {S}tructures and {S}trong
  quantum {D}arwinism – a review},\ }\href
  {https://doi.org/10.22331/q-2021-11-08-571} {\bibfield  {journal} {\bibinfo
  {journal} {{Quantum}}\ }\textbf {\bibinfo {volume} {5}},\ \bibinfo {pages}
  {571} (\bibinfo {year} {2021})}\BibitemShut {NoStop}%
\bibitem [{\citenamefont {Finkelstein}(1993)}]{FinkelsteinPRD1993}%
  \BibitemOpen
  \bibfield  {author} {\bibinfo {author} {\bibfnamefont {J.}~\bibnamefont
  {Finkelstein}},\ }\bibfield  {title} {\bibinfo {title} {Definition of
  decoherence},\ }\href {https://doi.org/10.1103/PhysRevD.47.5430} {\bibfield
  {journal} {\bibinfo  {journal} {Phys. Rev. D}\ }\textbf {\bibinfo {volume}
  {47}},\ \bibinfo {pages} {5430} (\bibinfo {year} {1993})}\BibitemShut
  {NoStop}%
\bibitem [{\citenamefont {Riedel}\ \emph {et~al.}(2016)\citenamefont {Riedel},
  \citenamefont {Zurek},\ and\ \citenamefont
  {Zwolak}}]{RiedelZurekZwolakPRA2016}%
  \BibitemOpen
  \bibfield  {author} {\bibinfo {author} {\bibfnamefont {C.~J.}\ \bibnamefont
  {Riedel}}, \bibinfo {author} {\bibfnamefont {W.~H.}\ \bibnamefont {Zurek}},\
  and\ \bibinfo {author} {\bibfnamefont {M.}~\bibnamefont {Zwolak}},\
  }\bibfield  {title} {\bibinfo {title} {{Objective past of a quantum universe:
  Redundant records of consistent histories}},\ }\href
  {https://doi.org/10.1103/PhysRevA.93.032126} {\bibfield  {journal} {\bibinfo
  {journal} {Phys. Rev. A}\ }\textbf {\bibinfo {volume} {93}},\ \bibinfo
  {pages} {032126} (\bibinfo {year} {2016})}\BibitemShut {NoStop}%
\bibitem [{\citenamefont {Scherer}\ and\ \citenamefont
  {Soklakov}(2005)}]{SchererSoklakovJMP2005}%
  \BibitemOpen
  \bibfield  {author} {\bibinfo {author} {\bibfnamefont {A.}~\bibnamefont
  {Scherer}}\ and\ \bibinfo {author} {\bibfnamefont {A.~N.}\ \bibnamefont
  {Soklakov}},\ }\bibfield  {title} {\bibinfo {title} {Initial states and
  decoherence of histories},\ }\href {https://doi.org/10.1063/1.1888030}
  {\bibfield  {journal} {\bibinfo  {journal} {J. Math. Phys.}\ }\textbf
  {\bibinfo {volume} {46}},\ \bibinfo {pages} {042108} (\bibinfo {year}
  {2005})}\BibitemShut {NoStop}%
\bibitem [{\citenamefont {Riedel}(2025)}]{RiedelQuantum2025}%
  \BibitemOpen
  \bibfield  {author} {\bibinfo {author} {\bibfnamefont {C.~J.}\ \bibnamefont
  {Riedel}},\ }\bibfield  {title} {\bibinfo {title} {Wavefunction branches
  demand a definition!},\ }\href {https://doi.org/10.22331/qv-2025-06-16-85}
  {\bibfield  {journal} {\bibinfo  {journal} {{Quantum Views}}\ }\textbf
  {\bibinfo {volume} {9}},\ \bibinfo {pages} {85} (\bibinfo {year}
  {2025})}\BibitemShut {NoStop}%
\bibitem [{\citenamefont {Paz}\ and\ \citenamefont
  {Zurek}(1993)}]{PazZurekPRD1993}%
  \BibitemOpen
  \bibfield  {author} {\bibinfo {author} {\bibfnamefont {J.~P.}\ \bibnamefont
  {Paz}}\ and\ \bibinfo {author} {\bibfnamefont {W.~H.}\ \bibnamefont
  {Zurek}},\ }\bibfield  {title} {\bibinfo {title} {Environment-induced
  decoherence, classicality and consistency of quantum histories},\ }\href
  {https://doi.org/10.1103/PhysRevD.48.2728} {\bibfield  {journal} {\bibinfo
  {journal} {Phys. Rev. D}\ }\textbf {\bibinfo {volume} {48}},\ \bibinfo
  {pages} {2728} (\bibinfo {year} {1993})}\BibitemShut {NoStop}%
\bibitem [{\citenamefont {Riedel}(2017)}]{RiedelPRL2017}%
  \BibitemOpen
  \bibfield  {author} {\bibinfo {author} {\bibfnamefont {C.~J.}\ \bibnamefont
  {Riedel}},\ }\bibfield  {title} {\bibinfo {title} {{Classical Branch
  Structure from Spatial Redundancy in a Many-Body Wave Function}},\ }\href
  {https://doi.org/10.1103/PhysRevLett.118.120402} {\bibfield  {journal}
  {\bibinfo  {journal} {Phys. Rev. Lett.}\ }\textbf {\bibinfo {volume} {118}},\
  \bibinfo {pages} {120402} (\bibinfo {year} {2017})}\BibitemShut {NoStop}%
\bibitem [{\citenamefont {Weingarten}(2022)}]{WeingartenFP2022}%
  \BibitemOpen
  \bibfield  {author} {\bibinfo {author} {\bibfnamefont {D.}~\bibnamefont
  {Weingarten}},\ }\bibfield  {title} {\bibinfo {title} {{Macroscopic Reality
  from Quantum Complexity}},\ }\href
  {https://doi.org/https://doi.org/10.1007/s10701-022-00554-0} {\bibfield
  {journal} {\bibinfo  {journal} {Found. Phys.}\ }\textbf {\bibinfo {volume}
  {52}},\ \bibinfo {pages} {45} (\bibinfo {year} {2022})}\BibitemShut {NoStop}%
\bibitem [{\citenamefont {Ollivier}(2022)}]{OllivierEnt2022}%
  \BibitemOpen
  \bibfield  {author} {\bibinfo {author} {\bibfnamefont {H.}~\bibnamefont
  {Ollivier}},\ }\bibfield  {title} {\bibinfo {title} {{Emergence of
  Objectivity for Quantum Many-Body Systems}},\ }\href
  {https://doi.org/10.3390/e24020277} {\bibfield  {journal} {\bibinfo
  {journal} {Entropy}\ }\textbf {\bibinfo {volume} {24}},\ \bibinfo {pages}
  {277} (\bibinfo {year} {2022})}\BibitemShut {NoStop}%
\bibitem [{\citenamefont {Touil}\ \emph {et~al.}(2024)\citenamefont {Touil},
  \citenamefont {Anza}, \citenamefont {Deffner},\ and\ \citenamefont
  {Crutchfield}}]{TouilEtAlQuantum2024}%
  \BibitemOpen
  \bibfield  {author} {\bibinfo {author} {\bibfnamefont {A.}~\bibnamefont
  {Touil}}, \bibinfo {author} {\bibfnamefont {F.}~\bibnamefont {Anza}},
  \bibinfo {author} {\bibfnamefont {S.}~\bibnamefont {Deffner}},\ and\ \bibinfo
  {author} {\bibfnamefont {J.~P.}\ \bibnamefont {Crutchfield}},\ }\bibfield
  {title} {\bibinfo {title} {{Branching States as The Emergent Structure of a
  Quantum Universe}},\ }\href {https://doi.org/10.22331/q-2024-10-10-1494}
  {\bibfield  {journal} {\bibinfo  {journal} {Quantum}\ }\textbf {\bibinfo
  {volume} {8}},\ \bibinfo {pages} {1494} (\bibinfo {year} {2024})}\BibitemShut
  {NoStop}%
\bibitem [{\citenamefont {Taylor}\ and\ \citenamefont
  {McCulloch}(2025)}]{TaylorMcCullochQuantum2025}%
  \BibitemOpen
  \bibfield  {author} {\bibinfo {author} {\bibfnamefont {J.~K.}\ \bibnamefont
  {Taylor}}\ and\ \bibinfo {author} {\bibfnamefont {I.~P.}\ \bibnamefont
  {McCulloch}},\ }\bibfield  {title} {\bibinfo {title} {Wavefunction branching:
  when you can't tell pure states from mixed states},\ }\href
  {https://doi.org/10.22331/q-2025-03-25-1670} {\bibfield  {journal} {\bibinfo
  {journal} {{Quantum}}\ }\textbf {\bibinfo {volume} {9}},\ \bibinfo {pages}
  {1670} (\bibinfo {year} {2025})}\BibitemShut {NoStop}%
\bibitem [{\citenamefont {Adil}\ \emph {et~al.}(2026)\citenamefont {Adil},
  \citenamefont {Rudolph}, \citenamefont {Arrasmith}, \citenamefont {Holmes},
  \citenamefont {Albrecht},\ and\ \citenamefont
  {Sornborger}}]{AdilEtAlPRD2026}%
  \BibitemOpen
  \bibfield  {author} {\bibinfo {author} {\bibfnamefont {A.}~\bibnamefont
  {Adil}}, \bibinfo {author} {\bibfnamefont {M.~S.}\ \bibnamefont {Rudolph}},
  \bibinfo {author} {\bibfnamefont {A.}~\bibnamefont {Arrasmith}}, \bibinfo
  {author} {\bibfnamefont {Z.}~\bibnamefont {Holmes}}, \bibinfo {author}
  {\bibfnamefont {A.}~\bibnamefont {Albrecht}},\ and\ \bibinfo {author}
  {\bibfnamefont {A.}~\bibnamefont {Sornborger}},\ }\bibfield  {title}
  {\bibinfo {title} {Search for classical subsystems in quantum worlds},\
  }\href {https://doi.org/10.1103/vj4x-96fk} {\bibfield  {journal} {\bibinfo
  {journal} {Phys. Rev. D}\ }\textbf {\bibinfo {volume} {113}},\ \bibinfo
  {pages} {103535} (\bibinfo {year} {2026})}\BibitemShut {NoStop}%
\bibitem [{\citenamefont {Kus}\ \emph {et~al.}(1988)\citenamefont {Kus},
  \citenamefont {Mostowski},\ and\ \citenamefont
  {Haake}}]{KusMostowskiHaakeJPA1988}%
  \BibitemOpen
  \bibfield  {author} {\bibinfo {author} {\bibfnamefont {M.}~\bibnamefont
  {Kus}}, \bibinfo {author} {\bibfnamefont {J.}~\bibnamefont {Mostowski}},\
  and\ \bibinfo {author} {\bibfnamefont {F.}~\bibnamefont {Haake}},\ }\bibfield
   {title} {\bibinfo {title} {Universality of eigenvector statistics of kicked
  tops of different symmetries},\ }\href
  {https://doi.org/10.1088/0305-4470/21/22/006} {\bibfield  {journal} {\bibinfo
   {journal} {J. Phys. A}\ }\textbf {\bibinfo {volume} {21}},\ \bibinfo {pages}
  {L1073} (\bibinfo {year} {1988})}\BibitemShut {NoStop}%
\bibitem [{\citenamefont {Gemmer}\ and\ \citenamefont
  {Steinigeweg}(2014)}]{GemmerSteinigewegPRE2014}%
  \BibitemOpen
  \bibfield  {author} {\bibinfo {author} {\bibfnamefont {J.}~\bibnamefont
  {Gemmer}}\ and\ \bibinfo {author} {\bibfnamefont {R.}~\bibnamefont
  {Steinigeweg}},\ }\bibfield  {title} {\bibinfo {title} {Entropy increase in
  $k$-step {M}arkovian and consistent dynamics of closed quantum systems},\
  }\href {https://doi.org/10.1103/PhysRevE.89.042113} {\bibfield  {journal}
  {\bibinfo  {journal} {Phys. Rev. E}\ }\textbf {\bibinfo {volume} {89}},\
  \bibinfo {pages} {042113} (\bibinfo {year} {2014})}\BibitemShut {NoStop}%
\bibitem [{\citenamefont {Schmidtke}\ and\ \citenamefont
  {Gemmer}(2016)}]{SchmidtkeGemmerPRE2016}%
  \BibitemOpen
  \bibfield  {author} {\bibinfo {author} {\bibfnamefont {D.}~\bibnamefont
  {Schmidtke}}\ and\ \bibinfo {author} {\bibfnamefont {J.}~\bibnamefont
  {Gemmer}},\ }\bibfield  {title} {\bibinfo {title} {{Numerical evidence for
  approximate consistency and Markovianity of some quantum histories in a class
  of finite closed spin systems}},\ }\href
  {https://doi.org/10.1103/PhysRevE.93.012125} {\bibfield  {journal} {\bibinfo
  {journal} {Phys. Rev. E}\ }\textbf {\bibinfo {volume} {93}},\ \bibinfo
  {pages} {012125} (\bibinfo {year} {2016})}\BibitemShut {NoStop}%
\bibitem [{\citenamefont {Nation}\ and\ \citenamefont
  {Porras}(2020)}]{NationPorrasPRE2020}%
  \BibitemOpen
  \bibfield  {author} {\bibinfo {author} {\bibfnamefont {C.}~\bibnamefont
  {Nation}}\ and\ \bibinfo {author} {\bibfnamefont {D.}~\bibnamefont
  {Porras}},\ }\bibfield  {title} {\bibinfo {title} {{Taking snapshots of a
  quantum thermalization process: Emergent classicality in quantum jump
  trajectories}},\ }\href {https://doi.org/10.1103/PhysRevE.102.042115}
  {\bibfield  {journal} {\bibinfo  {journal} {Phys. Rev. E}\ }\textbf {\bibinfo
  {volume} {102}},\ \bibinfo {pages} {042115} (\bibinfo {year}
  {2020})}\BibitemShut {NoStop}%
\bibitem [{\citenamefont {Strasberg}\ \emph {et~al.}(2023)\citenamefont
  {Strasberg}, \citenamefont {Winter}, \citenamefont {Gemmer},\ and\
  \citenamefont {Wang}}]{StrasbergEtAlPRA2023}%
  \BibitemOpen
  \bibfield  {author} {\bibinfo {author} {\bibfnamefont {P.}~\bibnamefont
  {Strasberg}}, \bibinfo {author} {\bibfnamefont {A.}~\bibnamefont {Winter}},
  \bibinfo {author} {\bibfnamefont {J.}~\bibnamefont {Gemmer}},\ and\ \bibinfo
  {author} {\bibfnamefont {J.}~\bibnamefont {Wang}},\ }\bibfield  {title}
  {\bibinfo {title} {{Classicality, Markovianity and local detailed balance
  from pure-state dynamics}},\ }\href
  {https://doi.org/10.1103/PhysRevA.108.012225} {\bibfield  {journal} {\bibinfo
   {journal} {Phys. Rev. A}\ }\textbf {\bibinfo {volume} {108}},\ \bibinfo
  {pages} {012225} (\bibinfo {year} {2023})}\BibitemShut {NoStop}%
\bibitem [{\citenamefont {Strasberg}(2023)}]{StrasbergSP2023}%
  \BibitemOpen
  \bibfield  {author} {\bibinfo {author} {\bibfnamefont {P.}~\bibnamefont
  {Strasberg}},\ }\bibfield  {title} {\bibinfo {title} {{Classicality with(out)
  decoherence: Concepts, relation to Markovianity and a random matrix theory
  approach}},\ }\href {https://doi.org/10.21468/SciPostPhys.15.1.024}
  {\bibfield  {journal} {\bibinfo  {journal} {SciPost Phys.}\ }\textbf
  {\bibinfo {volume} {15}},\ \bibinfo {pages} {024} (\bibinfo {year}
  {2023})}\BibitemShut {NoStop}%
\bibitem [{\citenamefont {Srednicki}(1999)}]{SrednickiJPA1999}%
  \BibitemOpen
  \bibfield  {author} {\bibinfo {author} {\bibfnamefont {M.}~\bibnamefont
  {Srednicki}},\ }\bibfield  {title} {\bibinfo {title} {The approach to thermal
  equilibrium in quantized chaotic systems},\ }\href
  {https://doi.org/10.1088/0305-4470/32/7/007} {\bibfield  {journal} {\bibinfo
  {journal} {J. Phys. A}\ }\textbf {\bibinfo {volume} {32}},\ \bibinfo {pages}
  {1163} (\bibinfo {year} {1999})}\BibitemShut {NoStop}%
\bibitem [{\citenamefont {D'Alessio}\ \emph {et~al.}(2016)\citenamefont
  {D'Alessio}, \citenamefont {Kafri}, \citenamefont {Polkovnikov},\ and\
  \citenamefont {Rigol}}]{DAlessioEtAlAP2016}%
  \BibitemOpen
  \bibfield  {author} {\bibinfo {author} {\bibfnamefont {L.}~\bibnamefont
  {D'Alessio}}, \bibinfo {author} {\bibfnamefont {Y.}~\bibnamefont {Kafri}},
  \bibinfo {author} {\bibfnamefont {A.}~\bibnamefont {Polkovnikov}},\ and\
  \bibinfo {author} {\bibfnamefont {M.}~\bibnamefont {Rigol}},\ }\bibfield
  {title} {\bibinfo {title} {From quantum chaos and eigenstate thermalization
  to statistical mechanics and thermodynamics},\ }\href
  {https://doi.org/10.1080/00018732.2016.1198134} {\bibfield  {journal}
  {\bibinfo  {journal} {Adv. Phys.}\ }\textbf {\bibinfo {volume} {65}},\
  \bibinfo {pages} {239} (\bibinfo {year} {2016})}\BibitemShut {NoStop}%
\bibitem [{\citenamefont {Deutsch}(2018)}]{DeutschRPP2018}%
  \BibitemOpen
  \bibfield  {author} {\bibinfo {author} {\bibfnamefont {J.~M.}\ \bibnamefont
  {Deutsch}},\ }\bibfield  {title} {\bibinfo {title} {Eigenstate thermalization
  hypothesis},\ }\href {https://doi.org/10.1088/1361-6633/aac9f1} {\bibfield
  {journal} {\bibinfo  {journal} {Rep. Prog. Phys.}\ }\textbf {\bibinfo
  {volume} {81}},\ \bibinfo {pages} {082001} (\bibinfo {year}
  {2018})}\BibitemShut {NoStop}%
\bibitem [{\citenamefont {Wigner}(1967)}]{Wigner1967}%
  \BibitemOpen
  \bibfield  {author} {\bibinfo {author} {\bibfnamefont {E.}~\bibnamefont
  {Wigner}},\ }\bibfield  {title} {\bibinfo {title} {Random matrices in
  physics},\ }\href {https://doi.org/10.1137/1009001} {\bibfield  {journal}
  {\bibinfo  {journal} {SIAM Reviews}\ }\textbf {\bibinfo {volume} {9}},\
  \bibinfo {pages} {1} (\bibinfo {year} {1967})}\BibitemShut {NoStop}%
\bibitem [{\citenamefont {Brody}\ \emph {et~al.}(1981)\citenamefont {Brody},
  \citenamefont {Flores}, \citenamefont {French}, \citenamefont {Mello},
  \citenamefont {Pandey},\ and\ \citenamefont {Wong}}]{BrodyEtAlRMP1981}%
  \BibitemOpen
  \bibfield  {author} {\bibinfo {author} {\bibfnamefont {T.~A.}\ \bibnamefont
  {Brody}}, \bibinfo {author} {\bibfnamefont {J.}~\bibnamefont {Flores}},
  \bibinfo {author} {\bibfnamefont {J.~B.}\ \bibnamefont {French}}, \bibinfo
  {author} {\bibfnamefont {P.~A.}\ \bibnamefont {Mello}}, \bibinfo {author}
  {\bibfnamefont {A.}~\bibnamefont {Pandey}},\ and\ \bibinfo {author}
  {\bibfnamefont {S.~S.~M.}\ \bibnamefont {Wong}},\ }\bibfield  {title}
  {\bibinfo {title} {Random-matrix physics: spectrum and strength
  fluctuations},\ }\href {https://doi.org/10.1103/RevModPhys.53.385} {\bibfield
   {journal} {\bibinfo  {journal} {Rev. Mod. Phys.}\ }\textbf {\bibinfo
  {volume} {53}},\ \bibinfo {pages} {385} (\bibinfo {year} {1981})}\BibitemShut
  {NoStop}%
\bibitem [{\citenamefont {Beenakker}(1997)}]{BeenakkerRMP1997}%
  \BibitemOpen
  \bibfield  {author} {\bibinfo {author} {\bibfnamefont {C.~W.~J.}\
  \bibnamefont {Beenakker}},\ }\bibfield  {title} {\bibinfo {title}
  {Random-matrix theory of quantum transport},\ }\href
  {https://doi.org/10.1103/RevModPhys.69.731} {\bibfield  {journal} {\bibinfo
  {journal} {Rev. Mod. Phys.}\ }\textbf {\bibinfo {volume} {69}},\ \bibinfo
  {pages} {731} (\bibinfo {year} {1997})}\BibitemShut {NoStop}%
\bibitem [{\citenamefont {Guhr}\ \emph {et~al.}(1998)\citenamefont {Guhr},
  \citenamefont {A},\ and\ \citenamefont
  {A.Weidenm\"uller}}]{GuhrMuellerGroelingWeidenmuellerPR1998}%
  \BibitemOpen
  \bibfield  {author} {\bibinfo {author} {\bibfnamefont {T.}~\bibnamefont
  {Guhr}}, \bibinfo {author} {\bibfnamefont {M.-G.}\ \bibnamefont {A}},\ and\
  \bibinfo {author} {\bibfnamefont {H.}~\bibnamefont {A.Weidenm\"uller}},\
  }\bibfield  {title} {\bibinfo {title} {Random-matrix theories in quantum
  physics: common concepts},\ }\href
  {https://doi.org/10.1016/S0370-1573(97)00088-4} {\bibfield  {journal}
  {\bibinfo  {journal} {Phys. Rep.}\ }\textbf {\bibinfo {volume} {299}},\
  \bibinfo {pages} {189} (\bibinfo {year} {1998})}\BibitemShut {NoStop}%
\bibitem [{\citenamefont {Haake}(2010)}]{HaakeBook2010}%
  \BibitemOpen
  \bibfield  {author} {\bibinfo {author} {\bibfnamefont {F.}~\bibnamefont
  {Haake}},\ }\href {https://doi.org/10.1007/978-3-642-05428-0} {\emph
  {\bibinfo {title} {Quantum Signatures of Chaos}}}\ (\bibinfo  {publisher}
  {Springer-Verlag},\ \bibinfo {address} {Berlin Heidelberg},\ \bibinfo {year}
  {2010})\BibitemShut {NoStop}%
\bibitem [{\citenamefont {Borgonovi}\ \emph {et~al.}(2016)\citenamefont
  {Borgonovi}, \citenamefont {Izrailev}, \citenamefont {Santos},\ and\
  \citenamefont {Zelevinsky}}]{BorgonoviEtAlPR2016}%
  \BibitemOpen
  \bibfield  {author} {\bibinfo {author} {\bibfnamefont {F.}~\bibnamefont
  {Borgonovi}}, \bibinfo {author} {\bibfnamefont {F.}~\bibnamefont {Izrailev}},
  \bibinfo {author} {\bibfnamefont {L.~F.}\ \bibnamefont {Santos}},\ and\
  \bibinfo {author} {\bibfnamefont {V.~G.}\ \bibnamefont {Zelevinsky}},\
  }\bibfield  {title} {\bibinfo {title} {Quantum chaos and thermalization in
  isolated systems of interacting particles},\ }\href
  {https://www.sciencedirect.com/science/article/pii/S0370157316000831?casa_token=ah3werHg1Q8AAAAA:XfkLZFNbO7u3o7UVgUslfAWKu6D-yOy0AkPPTxN4oK7YDmWnsUO0GGnAsLL8ovL-ep3mlmA6Kg}
  {\bibfield  {journal} {\bibinfo  {journal} {Phys. Rep.}\ }\textbf {\bibinfo
  {volume} {626}},\ \bibinfo {pages} {1} (\bibinfo {year} {2016})}\BibitemShut
  {NoStop}%
\bibitem [{\citenamefont {Deutsch}(1991)}]{DeutschPRA1991}%
  \BibitemOpen
  \bibfield  {author} {\bibinfo {author} {\bibfnamefont {J.~M.}\ \bibnamefont
  {Deutsch}},\ }\bibfield  {title} {\bibinfo {title} {Quantum statistical
  mechanics in a closed system},\ }\href
  {https://doi.org/10.1103/PhysRevA.43.2046} {\bibfield  {journal} {\bibinfo
  {journal} {Phys. Rev. A}\ }\textbf {\bibinfo {volume} {43}},\ \bibinfo
  {pages} {2046} (\bibinfo {year} {1991})}\BibitemShut {NoStop}%
\bibitem [{\citenamefont {Reimann}\ and\ \citenamefont
  {Dabelow}(2021)}]{ReimannDabelowPRE2021}%
  \BibitemOpen
  \bibfield  {author} {\bibinfo {author} {\bibfnamefont {P.}~\bibnamefont
  {Reimann}}\ and\ \bibinfo {author} {\bibfnamefont {L.}~\bibnamefont
  {Dabelow}},\ }\bibfield  {title} {\bibinfo {title} {{Refining Deutsch's
  approach to thermalization}},\ }\href
  {https://doi.org/10.1103/PhysRevE.103.022119} {\bibfield  {journal} {\bibinfo
   {journal} {Phys. Rev. E}\ }\textbf {\bibinfo {volume} {103}},\ \bibinfo
  {pages} {022119} (\bibinfo {year} {2021})}\BibitemShut {NoStop}%
\bibitem [{\citenamefont {Chao}\ \emph {et~al.}(2017)\citenamefont {Chao},
  \citenamefont {Reichardt}, \citenamefont {Sutherland},\ and\ \citenamefont
  {Vidick}}]{ChaoEtAl2017}%
  \BibitemOpen
  \bibfield  {author} {\bibinfo {author} {\bibfnamefont {R.}~\bibnamefont
  {Chao}}, \bibinfo {author} {\bibfnamefont {B.~W.}\ \bibnamefont {Reichardt}},
  \bibinfo {author} {\bibfnamefont {C.}~\bibnamefont {Sutherland}},\ and\
  \bibinfo {author} {\bibfnamefont {T.}~\bibnamefont {Vidick}},\ }\bibfield
  {title} {\bibinfo {title} {{Overlapping Qubits}},\ }in\ \href
  {https://doi.org/10.4230/LIPIcs.ITCS.2017.48} {\emph {\bibinfo {booktitle}
  {8th Innovations in Theoretical Computer Science Conference (ITCS 2017)}}},\
  Vol.~\bibinfo {volume} {67}\ (\bibinfo  {publisher} {Leibniz International
  Proceedings in Informatics (LIPIcs)},\ \bibinfo {address} {Schloss Dagstuhl
  – Leibniz-Zentrum für Informatik},\ \bibinfo {year} {2017})\ p.\ \bibinfo
  {pages} {48:1}\BibitemShut {NoStop}%
\bibitem [{\citenamefont {Chakravarty}(2021)}]{ChakravartyJHEP2021}%
  \BibitemOpen
  \bibfield  {author} {\bibinfo {author} {\bibfnamefont {J.}~\bibnamefont
  {Chakravarty}},\ }\bibfield  {title} {\bibinfo {title} {{Overcounting of
  interior excitations: a resolution to the bags of gold paradox in AdS}},\
  }\href {https://doi.org/10.1007/JHEP02(2021)027} {\bibfield  {journal}
  {\bibinfo  {journal} {J. High Energ. Phys.}\ }\textbf {\bibinfo {volume}
  {2021}},\ \bibinfo {pages} {27}}\BibitemShut {NoStop}%
\bibitem [{\citenamefont {Soulas}(2024)}]{SoulasFP2024}%
  \BibitemOpen
  \bibfield  {author} {\bibinfo {author} {\bibfnamefont {A.}~\bibnamefont
  {Soulas}},\ }\bibfield  {title} {\bibinfo {title} {{Decoherence as a
  High-Dimensional Geometrical Phenomenon}},\ }\href
  {https://doi.org/10.1007/s10701-023-00740-8} {\bibfield  {journal} {\bibinfo
  {journal} {Found. Phys.}\ }\textbf {\bibinfo {volume} {54}},\ \bibinfo
  {pages} {11} (\bibinfo {year} {2024})}\BibitemShut {NoStop}%
\bibitem [{\citenamefont {Johnson}\ and\ \citenamefont
  {Lindenstrauss}(1984)}]{JohnsonLindenstraussAMS1984}%
  \BibitemOpen
  \bibfield  {author} {\bibinfo {author} {\bibfnamefont {W.~B.}\ \bibnamefont
  {Johnson}}\ and\ \bibinfo {author} {\bibfnamefont {J.}~\bibnamefont
  {Lindenstrauss}},\ }\bibfield  {title} {\bibinfo {title} {Extensions of
  {L}ipschitz mappings into a {H}ilbert space},\ }in\ \href
  {https://doi.org/10.1090/conm/026/737400} {\emph {\bibinfo {booktitle}
  {Conference in modern analysis and probability ({N}ew {H}aven, {C}onn.,
  1982)}}},\ \bibinfo {series} {Contemp. Math.}, Vol.~\bibinfo {volume} {26}\
  (\bibinfo  {publisher} {Amer. Math. Soc., Providence, RI},\ \bibinfo {year}
  {1984})\ pp.\ \bibinfo {pages} {189--206}\BibitemShut {NoStop}%
\bibitem [{\citenamefont {Dasgupta}\ and\ \citenamefont
  {Gupta}(2003)}]{DasguptaGuptaRSA2003}%
  \BibitemOpen
  \bibfield  {author} {\bibinfo {author} {\bibfnamefont {S.}~\bibnamefont
  {Dasgupta}}\ and\ \bibinfo {author} {\bibfnamefont {A.}~\bibnamefont
  {Gupta}},\ }\bibfield  {title} {\bibinfo {title} {{An elementary proof of a
  theorem of Johnson and Lindenstrauss}},\ }\href
  {https://doi.org/https://doi.org/10.1002/rsa.10073} {\bibfield  {journal}
  {\bibinfo  {journal} {Random Struct. Algorithms}\ }\textbf {\bibinfo {volume}
  {22}},\ \bibinfo {pages} {60} (\bibinfo {year} {2003})}\BibitemShut {NoStop}%
\bibitem [{\citenamefont {Bengtsson}\ and\ \citenamefont
  {\.{Z}yczkowski}(2017)}]{BengtssonZyczkowskiBook2017}%
  \BibitemOpen
  \bibfield  {author} {\bibinfo {author} {\bibfnamefont {I.}~\bibnamefont
  {Bengtsson}}\ and\ \bibinfo {author} {\bibfnamefont {K.}~\bibnamefont
  {\.{Z}yczkowski}},\ }\href@noop {} {\emph {\bibinfo {title} {Geometry of
  quantum states: an introduction to quantum entanglement}}},\ \bibinfo
  {edition} {2nd}\ ed.\ (\bibinfo  {publisher} {Cambridge University Press},\
  \bibinfo {address} {Cambridge},\ \bibinfo {year} {2017})\BibitemShut
  {NoStop}%
\bibitem [{\citenamefont {Alon}(2003)}]{AlonDM2003}%
  \BibitemOpen
  \bibfield  {author} {\bibinfo {author} {\bibfnamefont {N.}~\bibnamefont
  {Alon}},\ }\bibfield  {title} {\bibinfo {title} {{Problems and results in
  extremal combinatorics--I}},\ }\href
  {https://doi.org/https://doi.org/10.1016/S0012-365X(03)00227-9} {\bibfield
  {journal} {\bibinfo  {journal} {Discrete Math.}\ }\textbf {\bibinfo {volume}
  {273}},\ \bibinfo {pages} {31} (\bibinfo {year} {2003})}\BibitemShut
  {NoStop}%
\bibitem [{\citenamefont {Tao}(2013)}]{TaoBlog2013}%
  \BibitemOpen
  \bibfield  {author} {\bibinfo {author} {\bibfnamefont {T.}~\bibnamefont
  {Tao}},\ }\bibfield  {title} {\bibinfo {title} {{A cheap version of the
  Kabatjanskii-Levenstein bound for almost orthogonal vectors}},\ }\href
  {https://terrytao.wordpress.com/2013/07/18/a-cheap-version-of-the-kabatjanskii-levenstein-bound-for-almost-orthogonal-vectors/}
  {\bibfield  {journal} {\bibinfo  {journal} {Blog post}\ } (\bibinfo {year}
  {2013})}\BibitemShut {NoStop}%
\bibitem [{\citenamefont {IV}\ \emph {et~al.}(2017)\citenamefont {IV},
  \citenamefont {Hammen},\ and\ \citenamefont {Mixon}}]{HaasHammenMixon2017}%
  \BibitemOpen
  \bibfield  {author} {\bibinfo {author} {\bibfnamefont {J.~I.~H.}\
  \bibnamefont {IV}}, \bibinfo {author} {\bibfnamefont {N.}~\bibnamefont
  {Hammen}},\ and\ \bibinfo {author} {\bibfnamefont {D.~G.}\ \bibnamefont
  {Mixon}},\ }\bibfield  {title} {\bibinfo {title} {{The Levenstein bound for
  packings in projective spaces}},\ }in\ \href
  {https://doi.org/10.1117/12.2275373} {\emph {\bibinfo {booktitle} {Wavelets
  and Sparsity XVII}}},\ Vol.\ \bibinfo {volume} {10394},\ \bibinfo {editor}
  {edited by\ \bibinfo {editor} {\bibfnamefont {Y.~M.}\ \bibnamefont {Lu}},
  \bibinfo {editor} {\bibfnamefont {D.~V.~D.}\ \bibnamefont {Ville}},\ and\
  \bibinfo {editor} {\bibfnamefont {M.}~\bibnamefont {Papadakis}}},\ \bibinfo
  {organization} {International Society for Optics and Photonics}\ (\bibinfo
  {publisher} {SPIE},\ \bibinfo {year} {2017})\ p.\ \bibinfo {pages}
  {103940V}\BibitemShut {NoStop}%
\bibitem [{\citenamefont {Welch}(1974)}]{WelchIEEE1974}%
  \BibitemOpen
  \bibfield  {author} {\bibinfo {author} {\bibfnamefont {L.}~\bibnamefont
  {Welch}},\ }\bibfield  {title} {\bibinfo {title} {Lower bounds on the maximum
  cross correlation of signals (corresp.)},\ }\href
  {https://doi.org/10.1109/TIT.1974.1055219} {\bibfield  {journal} {\bibinfo
  {journal} {IEEE Trans. Inf. Theory}\ }\textbf {\bibinfo {volume} {20}},\
  \bibinfo {pages} {397} (\bibinfo {year} {1974})}\BibitemShut {NoStop}%
\bibitem [{\citenamefont {Kabatiansky}\ and\ \citenamefont
  {Levenshtein}(1978)}]{KabatianskyLevenshtein1978}%
  \BibitemOpen
  \bibfield  {author} {\bibinfo {author} {\bibfnamefont {G.~A.}\ \bibnamefont
  {Kabatiansky}}\ and\ \bibinfo {author} {\bibfnamefont {V.~I.}\ \bibnamefont
  {Levenshtein}},\ }\bibfield  {title} {\bibinfo {title} {{On Bounds for
  Packings on a Sphere and in Space}},\ }\href
  {https://www.mathnet.ru/php/archive.phtml?wshow=paper&jrnid=ppi&paperid=1518&option_lang=eng}
  {\bibfield  {journal} {\bibinfo  {journal} {Probl. Peredachi Inf.}\ }\textbf
  {\bibinfo {volume} {14}},\ \bibinfo {pages} {3} (\bibinfo {year}
  {1978})}\BibitemShut {NoStop}%
\bibitem [{\citenamefont {Holevo}(1979)}]{Holevo1979}%
  \BibitemOpen
  \bibfield  {author} {\bibinfo {author} {\bibfnamefont {A.~S.}\ \bibnamefont
  {Holevo}},\ }\bibfield  {title} {\bibinfo {title} {On asymptotically optimal
  hypotheses testing in quantum statistics},\ }\href
  {https://doi.org/10.1137/1123048} {\bibfield  {journal} {\bibinfo  {journal}
  {Theory of Probability and its Applications}\ }\textbf {\bibinfo {volume}
  {23}},\ \bibinfo {pages} {411} (\bibinfo {year} {1979})}\BibitemShut
  {NoStop}%
\bibitem [{\citenamefont {Hausladen}\ and\ \citenamefont
  {Wootters}(1994)}]{HausladenWoottersJMO1994}%
  \BibitemOpen
  \bibfield  {author} {\bibinfo {author} {\bibfnamefont {P.}~\bibnamefont
  {Hausladen}}\ and\ \bibinfo {author} {\bibfnamefont {W.~K.}\ \bibnamefont
  {Wootters}},\ }\bibfield  {title} {\bibinfo {title} {{A ‘Pretty Good’
  Measurement for Distinguishing Quantum States}},\ }\href
  {https://doi.org/10.1080/09500349414552221} {\bibfield  {journal} {\bibinfo
  {journal} {J. Mod. Opt.}\ }\textbf {\bibinfo {volume} {41}},\ \bibinfo
  {pages} {2385} (\bibinfo {year} {1994})}\BibitemShut {NoStop}%
\bibitem [{\citenamefont {Barnum}\ and\ \citenamefont
  {Knill}(2002)}]{BarnumKnillJMP2002}%
  \BibitemOpen
  \bibfield  {author} {\bibinfo {author} {\bibfnamefont {H.}~\bibnamefont
  {Barnum}}\ and\ \bibinfo {author} {\bibfnamefont {E.}~\bibnamefont {Knill}},\
  }\bibfield  {title} {\bibinfo {title} {{Reversing quantum dynamics with
  near-optimal quantum and classical fidelity}},\ }\href
  {https://doi.org/10.1063/1.1459754} {\bibfield  {journal} {\bibinfo
  {journal} {J. Math. Phys.}\ }\textbf {\bibinfo {volume} {43}},\ \bibinfo
  {pages} {2097} (\bibinfo {year} {2002})}\BibitemShut {NoStop}%
\bibitem [{\citenamefont {Chefles}(1998)}]{CheflesPLA1998}%
  \BibitemOpen
  \bibfield  {author} {\bibinfo {author} {\bibfnamefont {A.}~\bibnamefont
  {Chefles}},\ }\bibfield  {title} {\bibinfo {title} {Unambiguous
  discrimination between linearly independent quantum states},\ }\href
  {https://doi.org/https://doi.org/10.1016/S0375-9601(98)00064-4} {\bibfield
  {journal} {\bibinfo  {journal} {Phys. Lett. A}\ }\textbf {\bibinfo {volume}
  {239}},\ \bibinfo {pages} {339} (\bibinfo {year} {1998})}\BibitemShut
  {NoStop}%
\bibitem [{\citenamefont {Sutter}\ \emph {et~al.}(2016)\citenamefont {Sutter},
  \citenamefont {Tomamichel},\ and\ \citenamefont
  {Harrow}}]{SutterTomamichelHarrowIEEE2016}%
  \BibitemOpen
  \bibfield  {author} {\bibinfo {author} {\bibfnamefont {D.}~\bibnamefont
  {Sutter}}, \bibinfo {author} {\bibfnamefont {M.}~\bibnamefont {Tomamichel}},\
  and\ \bibinfo {author} {\bibfnamefont {A.~W.}\ \bibnamefont {Harrow}},\
  }\bibfield  {title} {\bibinfo {title} {{Strengthened Monotonicity of Relative
  Entropy via Pinched Petz Recovery Map}},\ }\href
  {https://doi.org/10.1109/TIT.2016.2545680} {\bibfield  {journal} {\bibinfo
  {journal} {IEEE Trans. Inf. Theory}\ }\textbf {\bibinfo {volume} {62}},\
  \bibinfo {pages} {2907} (\bibinfo {year} {2016})}\BibitemShut {NoStop}%
\bibitem [{\citenamefont {Horoshko}\ \emph {et~al.}(2019)\citenamefont
  {Horoshko}, \citenamefont {Eskandari},\ and\ \citenamefont
  {Kilin}}]{HoroshkoEskandariKilinPLA2019}%
  \BibitemOpen
  \bibfield  {author} {\bibinfo {author} {\bibfnamefont {D.~B.}\ \bibnamefont
  {Horoshko}}, \bibinfo {author} {\bibfnamefont {M.~M.}\ \bibnamefont
  {Eskandari}},\ and\ \bibinfo {author} {\bibfnamefont {S.~Y.}\ \bibnamefont
  {Kilin}},\ }\bibfield  {title} {\bibinfo {title} {Equiprobable unambiguous
  discrimination of quantum states by symmetric orthogonalisation},\ }\href
  {https://doi.org/https://doi.org/10.1016/j.physleta.2019.03.006} {\bibfield
  {journal} {\bibinfo  {journal} {Phys. Lett. A}\ }\textbf {\bibinfo {volume}
  {383}},\ \bibinfo {pages} {1728} (\bibinfo {year} {2019})}\BibitemShut
  {NoStop}%
\bibitem [{\citenamefont {Goodman}(1963{\natexlab{a}})}]{GoodmanAMSa1963}%
  \BibitemOpen
  \bibfield  {author} {\bibinfo {author} {\bibfnamefont {N.~R.}\ \bibnamefont
  {Goodman}},\ }\bibfield  {title} {\bibinfo {title} {{Statistical Analysis
  Based on a Certain Multivariate Complex Gaussian Distribution (An
  Introduction)}},\ }\href {https://doi.org/10.1214/aoms/1177704250} {\bibfield
   {journal} {\bibinfo  {journal} {Ann. Math. Statist.}\ }\textbf {\bibinfo
  {volume} {34}},\ \bibinfo {pages} {152} (\bibinfo {year}
  {1963}{\natexlab{a}})}\BibitemShut {NoStop}%
\bibitem [{\citenamefont {Goodman}(1963{\natexlab{b}})}]{GoodmanAMSb1963}%
  \BibitemOpen
  \bibfield  {author} {\bibinfo {author} {\bibfnamefont {N.~R.}\ \bibnamefont
  {Goodman}},\ }\bibfield  {title} {\bibinfo {title} {{The Distribution of the
  Determinant of a Complex Wishart Distributed Matrix}},\ }\href
  {https://doi.org/10.1214/aoms/1177704251} {\bibfield  {journal} {\bibinfo
  {journal} {Ann. Math. Statist.}\ }\textbf {\bibinfo {volume} {34}},\ \bibinfo
  {pages} {178} (\bibinfo {year} {1963}{\natexlab{b}})}\BibitemShut {NoStop}%
\bibitem [{\citenamefont {Marchenko}\ and\ \citenamefont
  {Pastur}(1967)}]{MarchenkoPastur1967}%
  \BibitemOpen
  \bibfield  {author} {\bibinfo {author} {\bibfnamefont {V.~A.}\ \bibnamefont
  {Marchenko}}\ and\ \bibinfo {author} {\bibfnamefont {L.~A.}\ \bibnamefont
  {Pastur}},\ }\bibfield  {title} {\bibinfo {title} {{Distribution of
  eigenvalues for some sets of random matrices}},\ }\href
  {https://doi.org/10.1070/SM1967v001n04ABEH001994} {\bibfield  {journal}
  {\bibinfo  {journal} {Math. USSR-Sb.}\ }\textbf {\bibinfo {volume} {1}},\
  \bibinfo {pages} {457} (\bibinfo {year} {1967})}\BibitemShut {NoStop}%
\bibitem [{\citenamefont {Bai}\ and\ \citenamefont
  {Silverstein}(2010)}]{BaiSilversteinBook}%
  \BibitemOpen
  \bibfield  {author} {\bibinfo {author} {\bibfnamefont {Z.}~\bibnamefont
  {Bai}}\ and\ \bibinfo {author} {\bibfnamefont {J.~W.}\ \bibnamefont
  {Silverstein}},\ }\href@noop {} {\emph {\bibinfo {title} {{Spectral Analysis
  of Large Dimensional Random Matrices}}}},\ \bibinfo {edition} {2nd}\ ed.\
  (\bibinfo  {publisher} {Springer},\ \bibinfo {address} {New York},\ \bibinfo
  {year} {2010})\BibitemShut {NoStop}%
\bibitem [{\citenamefont {Lytova}\ and\ \citenamefont
  {Pastur}(2009)}]{LytovaPasturAP2009}%
  \BibitemOpen
  \bibfield  {author} {\bibinfo {author} {\bibfnamefont {A.}~\bibnamefont
  {Lytova}}\ and\ \bibinfo {author} {\bibfnamefont {L.}~\bibnamefont
  {Pastur}},\ }\bibfield  {title} {\bibinfo {title} {Central limit theorem for
  linear eigenvalue statistics of random matrices with independent entries},\
  }\href {https://doi.org/10.1214/09-AOP452} {\bibfield  {journal} {\bibinfo
  {journal} {Ann. Probab.}\ }\textbf {\bibinfo {volume} {37}},\ \bibinfo
  {pages} {1778} (\bibinfo {year} {2009})}\BibitemShut {NoStop}%
\bibitem [{\citenamefont {Pastur}\ and\ \citenamefont
  {Shcherbina}(2010)}]{PasturShcherbinaBook2010}%
  \BibitemOpen
  \bibfield  {author} {\bibinfo {author} {\bibfnamefont {L.}~\bibnamefont
  {Pastur}}\ and\ \bibinfo {author} {\bibfnamefont {M.}~\bibnamefont
  {Shcherbina}},\ }\href@noop {} {\emph {\bibinfo {title} {{Eigenvalue
  Distribution of Large Random Matrices}}}},\ Vol.\ \bibinfo {volume} {171}\
  (\bibinfo  {publisher} {American Mathematical Society},\ \bibinfo {address}
  {Providence, Rhode Island},\ \bibinfo {year} {2010})\BibitemShut {NoStop}%
\bibitem [{\citenamefont {Talagrand}(1996)}]{TalagrandAP1996}%
  \BibitemOpen
  \bibfield  {author} {\bibinfo {author} {\bibfnamefont {M.}~\bibnamefont
  {Talagrand}},\ }\bibfield  {title} {\bibinfo {title} {A new look at
  independence},\ }\href {https://doi.org/10.1214/aop/1042644705} {\bibfield
  {journal} {\bibinfo  {journal} {Ann. Prob.}\ }\textbf {\bibinfo {volume}
  {24}},\ \bibinfo {pages} {1} (\bibinfo {year} {1996})}\BibitemShut {NoStop}%
\bibitem [{\citenamefont {Hayden}\ \emph {et~al.}(2008)\citenamefont {Hayden},
  \citenamefont {Shor},\ and\ \citenamefont
  {Winter}}]{HaydenShorWinterOSID2008}%
  \BibitemOpen
  \bibfield  {author} {\bibinfo {author} {\bibfnamefont {P.}~\bibnamefont
  {Hayden}}, \bibinfo {author} {\bibfnamefont {P.~W.}\ \bibnamefont {Shor}},\
  and\ \bibinfo {author} {\bibfnamefont {A.}~\bibnamefont {Winter}},\
  }\bibfield  {title} {\bibinfo {title} {{Random Quantum Codes from Gaussian
  Ensembles and an Uncertainty Relation}},\ }\href
  {https://doi.org/10.1142/S1230161208000079} {\bibfield  {journal} {\bibinfo
  {journal} {Open Sys. Inf. Dyn.}\ }\textbf {\bibinfo {volume} {15}},\ \bibinfo
  {pages} {71} (\bibinfo {year} {2008})}\BibitemShut {NoStop}%
\bibitem [{\citenamefont {Collins}\ \emph {et~al.}(2022)\citenamefont
  {Collins}, \citenamefont {Matsumoto},\ and\ \citenamefont
  {Novak}}]{CollinsMatsumotoNovak2022}%
  \BibitemOpen
  \bibfield  {author} {\bibinfo {author} {\bibfnamefont {B.}~\bibnamefont
  {Collins}}, \bibinfo {author} {\bibfnamefont {S.}~\bibnamefont {Matsumoto}},\
  and\ \bibinfo {author} {\bibfnamefont {J.}~\bibnamefont {Novak}},\ }\bibfield
   {title} {\bibinfo {title} {{The Weingarten Calculus}},\ }\href
  {https://doi.org/10.1090/noti2474} {\bibfield  {journal} {\bibinfo  {journal}
  {Not. Am. Math. Soc.}\ }\textbf {\bibinfo {volume} {69}},\ \bibinfo {pages}
  {1} (\bibinfo {year} {2022})}\BibitemShut {NoStop}%
\bibitem [{\citenamefont {Nica}(1993)}]{NicaPJP1993}%
  \BibitemOpen
  \bibfield  {author} {\bibinfo {author} {\bibfnamefont {A.}~\bibnamefont
  {Nica}},\ }\bibfield  {title} {\bibinfo {title} {{Asymtotically Free Families
  of Random Unitaries in Symmetric Groups}},\ }\href
  {https://doi.org/10.2140/pjm.1993.157.295} {\bibfield  {journal} {\bibinfo
  {journal} {Pacific J. Math.}\ }\textbf {\bibinfo {volume} {157}},\ \bibinfo
  {pages} {295} (\bibinfo {year} {1993})}\BibitemShut {NoStop}%
\bibitem [{\citenamefont {Bordenave}\ and\ \citenamefont
  {Collins}(2019)}]{BordenaveCollinsAM2019}%
  \BibitemOpen
  \bibfield  {author} {\bibinfo {author} {\bibfnamefont {C.}~\bibnamefont
  {Bordenave}}\ and\ \bibinfo {author} {\bibfnamefont {B.}~\bibnamefont
  {Collins}},\ }\bibfield  {title} {\bibinfo {title} {{Eigenvalues of random
  lifts and polynomials of random permutation matrices}},\ }\href
  {https://doi.org/10.4007/annals.2019.190.3.3} {\bibfield  {journal} {\bibinfo
   {journal} {Ann. Math.}\ }\textbf {\bibinfo {volume} {190}},\ \bibinfo
  {pages} {811} (\bibinfo {year} {2019})}\BibitemShut {NoStop}%
\bibitem [{\citenamefont {Ji}\ \emph {et~al.}(2018)\citenamefont {Ji},
  \citenamefont {Liu},\ and\ \citenamefont {Song}}]{JiLiuSongAC2018}%
  \BibitemOpen
  \bibfield  {author} {\bibinfo {author} {\bibfnamefont {Z.}~\bibnamefont
  {Ji}}, \bibinfo {author} {\bibfnamefont {Y.-K.}\ \bibnamefont {Liu}},\ and\
  \bibinfo {author} {\bibfnamefont {F.}~\bibnamefont {Song}},\ }\bibinfo
  {title} {Advances in cryptology – crypto 2018, lecture notes in computer
  science}\ (\bibinfo  {publisher} {Springer},\ \bibinfo {address} {Cham},\
  \bibinfo {year} {2018})\ Chap.\ \bibinfo {chapter} {{Pseudorandom Quantum
  States}}, p.\ \bibinfo {pages} {126}\BibitemShut {NoStop}%
\bibitem [{\citenamefont {Bouland}\ \emph {et~al.}(2019)\citenamefont
  {Bouland}, \citenamefont {Fefferman},\ and\ \citenamefont
  {Vazirani}}]{BoulandFeffermanVaziraniArXiv2019}%
  \BibitemOpen
  \bibfield  {author} {\bibinfo {author} {\bibfnamefont {A.}~\bibnamefont
  {Bouland}}, \bibinfo {author} {\bibfnamefont {B.}~\bibnamefont {Fefferman}},\
  and\ \bibinfo {author} {\bibfnamefont {U.}~\bibnamefont {Vazirani}},\
  }\bibfield  {title} {\bibinfo {title} {{Computational pseudorandomness, the
  wormhole growth paradox and constraints on the AdS/CFT duality}},\ }\href
  {https://arxiv.org/abs/1910.14646} {\bibfield  {journal} {\bibinfo  {journal}
  {arXiv 1910.14646}\ } (\bibinfo {year} {2019})}\BibitemShut {NoStop}%
\bibitem [{\citenamefont {Aaronson}\ \emph {et~al.}(2022)\citenamefont
  {Aaronson}, \citenamefont {Bouland}, \citenamefont {Fefferman}, \citenamefont
  {Ghosh}, \citenamefont {Vazirani}, \citenamefont {Zhang},\ and\ \citenamefont
  {Zhou}}]{AaronsonEtAlArXiv2022}%
  \BibitemOpen
  \bibfield  {author} {\bibinfo {author} {\bibfnamefont {S.}~\bibnamefont
  {Aaronson}}, \bibinfo {author} {\bibfnamefont {A.}~\bibnamefont {Bouland}},
  \bibinfo {author} {\bibfnamefont {B.}~\bibnamefont {Fefferman}}, \bibinfo
  {author} {\bibfnamefont {S.}~\bibnamefont {Ghosh}}, \bibinfo {author}
  {\bibfnamefont {U.}~\bibnamefont {Vazirani}}, \bibinfo {author}
  {\bibfnamefont {C.}~\bibnamefont {Zhang}},\ and\ \bibinfo {author}
  {\bibfnamefont {Z.}~\bibnamefont {Zhou}},\ }\bibfield  {title} {\bibinfo
  {title} {{Quantum Pseudoentanglement}},\ }\href
  {https://arxiv.org/abs/2211.00747} {\bibfield  {journal} {\bibinfo  {journal}
  {arXiv 2211.00747}\ } (\bibinfo {year} {2022})}\BibitemShut {NoStop}%
\bibitem [{\citenamefont {Klappenecker}\ and\ \citenamefont
  {Rotteler}(2005)}]{KlappeneckerRottelerIEEE2005}%
  \BibitemOpen
  \bibfield  {author} {\bibinfo {author} {\bibfnamefont {A.}~\bibnamefont
  {Klappenecker}}\ and\ \bibinfo {author} {\bibfnamefont {M.}~\bibnamefont
  {Rotteler}},\ }\bibfield  {title} {\bibinfo {title} {Mutually unbiased bases
  are complex projective 2-designs},\ }in\ \href
  {https://doi.org/10.1109/ISIT.2005.1523643} {\emph {\bibinfo {booktitle}
  {Proceedings. International Symposium on Information Theory, 2005. ISIT
  2005.}}}\ (\bibinfo {year} {2005})\ pp.\ \bibinfo {pages}
  {1740--1744}\BibitemShut {NoStop}%
\bibitem [{\citenamefont {Silverstein}(1989)}]{SilversteinJMA1989}%
  \BibitemOpen
  \bibfield  {author} {\bibinfo {author} {\bibfnamefont {J.~W.}\ \bibnamefont
  {Silverstein}},\ }\bibfield  {title} {\bibinfo {title} {On the eigenvectors
  of large dimensional sample covariance matrices},\ }\href
  {https://doi.org/https://doi.org/10.1016/0047-259X(89)90084-5} {\bibfield
  {journal} {\bibinfo  {journal} {J. Multivar. Anal.}\ }\textbf {\bibinfo
  {volume} {30}},\ \bibinfo {pages} {1} (\bibinfo {year} {1989})}\BibitemShut
  {NoStop}%
\bibitem [{\citenamefont {Bru}(1989)}]{BruJMA1989}%
  \BibitemOpen
  \bibfield  {author} {\bibinfo {author} {\bibfnamefont {M.-F.}\ \bibnamefont
  {Bru}},\ }\bibfield  {title} {\bibinfo {title} {Diffusions of perturbed
  principal component analysis},\ }\href
  {https://doi.org/https://doi.org/10.1016/0047-259X(89)90080-8} {\bibfield
  {journal} {\bibinfo  {journal} {J. Multivar. Anal.}\ }\textbf {\bibinfo
  {volume} {29}},\ \bibinfo {pages} {127} (\bibinfo {year} {1989})}\BibitemShut
  {NoStop}%
\bibitem [{\citenamefont {Silverstein}(1990)}]{SilversteinAP1990}%
  \BibitemOpen
  \bibfield  {author} {\bibinfo {author} {\bibfnamefont {J.~W.}\ \bibnamefont
  {Silverstein}},\ }\bibfield  {title} {\bibinfo {title} {{Weak Convergence of
  Random Functions Defined by The Eigenvectors of Sample Covariance
  Matrices}},\ }\href {http://www.jstor.org/stable/2244421} {\bibfield
  {journal} {\bibinfo  {journal} {Ann. Prob.}\ }\textbf {\bibinfo {volume}
  {18}},\ \bibinfo {pages} {1174} (\bibinfo {year} {1990})}\BibitemShut
  {NoStop}%
\bibitem [{\citenamefont {Bai}\ \emph {et~al.}(2007)\citenamefont {Bai},
  \citenamefont {Miao},\ and\ \citenamefont {Pan}}]{BaiMiaoPanAP2007}%
  \BibitemOpen
  \bibfield  {author} {\bibinfo {author} {\bibfnamefont {Z.~D.}\ \bibnamefont
  {Bai}}, \bibinfo {author} {\bibfnamefont {B.~Q.}\ \bibnamefont {Miao}},\ and\
  \bibinfo {author} {\bibfnamefont {G.~M.}\ \bibnamefont {Pan}},\ }\bibfield
  {title} {\bibinfo {title} {{On asymptotics of eigenvectors of large sample
  covariance matrix}},\ }\href {https://doi.org/10.1214/009117906000001079}
  {\bibfield  {journal} {\bibinfo  {journal} {Ann. Prob.}\ }\textbf {\bibinfo
  {volume} {35}},\ \bibinfo {pages} {1532 } (\bibinfo {year}
  {2007})}\BibitemShut {NoStop}%
\bibitem [{\citenamefont {Tao}\ and\ \citenamefont {Vu}(2012)}]{TaoVuRMTA2012}%
  \BibitemOpen
  \bibfield  {author} {\bibinfo {author} {\bibfnamefont {T.}~\bibnamefont
  {Tao}}\ and\ \bibinfo {author} {\bibfnamefont {V.}~\bibnamefont {Vu}},\
  }\bibfield  {title} {\bibinfo {title} {Random matrices: Universal properties
  of eigenvectors},\ }\href {https://doi.org/10.1142/S2010326311500018}
  {\bibfield  {journal} {\bibinfo  {journal} {Rand. Mat. Theory Appl.}\
  }\textbf {\bibinfo {volume} {01}},\ \bibinfo {pages} {1150001} (\bibinfo
  {year} {2012})}\BibitemShut {NoStop}%
\bibitem [{\citenamefont {Xia}\ \emph {et~al.}(2013)\citenamefont {Xia},
  \citenamefont {Qin},\ and\ \citenamefont {Bai}}]{XiaQinBaiAS2013}%
  \BibitemOpen
  \bibfield  {author} {\bibinfo {author} {\bibfnamefont {N.}~\bibnamefont
  {Xia}}, \bibinfo {author} {\bibfnamefont {Y.}~\bibnamefont {Qin}},\ and\
  \bibinfo {author} {\bibfnamefont {Z.}~\bibnamefont {Bai}},\ }\bibfield
  {title} {\bibinfo {title} {Convergence rates of eigenvector empirical
  spectral distribution of large dimensional sample covariance matrix},\
  }\bibfield  {journal} {\bibinfo  {journal} {Ann. Statist.}\ }\textbf
  {\bibinfo {volume} {41}},\ \href {https://doi.org/10.1214/13-aos1154}
  {10.1214/13-aos1154} (\bibinfo {year} {2013})\BibitemShut {NoStop}%
\bibitem [{\citenamefont {Pillai}\ and\ \citenamefont
  {Yin}(2014)}]{PillaiYinAAP2014}%
  \BibitemOpen
  \bibfield  {author} {\bibinfo {author} {\bibfnamefont {N.~S.}\ \bibnamefont
  {Pillai}}\ and\ \bibinfo {author} {\bibfnamefont {J.}~\bibnamefont {Yin}},\
  }\bibfield  {title} {\bibinfo {title} {{Universality of covariance
  matrices}},\ }\href {https://doi.org/10.1214/13-AAP939} {\bibfield  {journal}
  {\bibinfo  {journal} {Ann. Appl. Prob.}\ }\textbf {\bibinfo {volume} {24}},\
  \bibinfo {pages} {935 } (\bibinfo {year} {2014})}\BibitemShut {NoStop}%
\bibitem [{\citenamefont {Alex}\ \emph {et~al.}(2014)\citenamefont {Alex},
  \citenamefont {Erdős}, \citenamefont {Knowles}, \citenamefont {Yau},\ and\
  \citenamefont {Yin}}]{AlexEtAlEJP2014}%
  \BibitemOpen
  \bibfield  {author} {\bibinfo {author} {\bibfnamefont {B.}~\bibnamefont
  {Alex}}, \bibinfo {author} {\bibfnamefont {L.}~\bibnamefont {Erdős}},
  \bibinfo {author} {\bibfnamefont {A.}~\bibnamefont {Knowles}}, \bibinfo
  {author} {\bibfnamefont {H.-T.}\ \bibnamefont {Yau}},\ and\ \bibinfo {author}
  {\bibfnamefont {J.}~\bibnamefont {Yin}},\ }\bibfield  {title} {\bibinfo
  {title} {{Isotropic local laws for sample covariance and generalized Wigner
  matrices}},\ }\href {https://doi.org/10.1214/EJP.v19-3054} {\bibfield
  {journal} {\bibinfo  {journal} {Electr. J. Prob.}\ }\textbf {\bibinfo
  {volume} {19}},\ \bibinfo {pages} {1 } (\bibinfo {year} {2014})}\BibitemShut
  {NoStop}%
\bibitem [{\citenamefont {Bourgade}\ and\ \citenamefont
  {Yau}(2017)}]{BourgadeYauCMP2017}%
  \BibitemOpen
  \bibfield  {author} {\bibinfo {author} {\bibfnamefont {P.}~\bibnamefont
  {Bourgade}}\ and\ \bibinfo {author} {\bibfnamefont {H.~T.}\ \bibnamefont
  {Yau}},\ }\bibfield  {title} {\bibinfo {title} {The eigenvector moment flow
  and local quantum unique ergodicity},\ }\href
  {https://doi.org/10.1007/s00220-016-2627-6} {\bibfield  {journal} {\bibinfo
  {journal} {Commun. Math. Phys.}\ }\textbf {\bibinfo {volume} {350}},\
  \bibinfo {pages} {231} (\bibinfo {year} {2017})}\BibitemShut {NoStop}%
\bibitem [{\citenamefont {Ding}(2019)}]{DingAAP2019}%
  \BibitemOpen
  \bibfield  {author} {\bibinfo {author} {\bibfnamefont {X.}~\bibnamefont
  {Ding}},\ }\bibfield  {title} {\bibinfo {title} {Singular vector distribution
  of sample covariance matrices},\ }\href {https://doi.org/10.1017/apr.2019.10}
  {\bibfield  {journal} {\bibinfo  {journal} {Adv. Appl. Prob.}\ }\textbf
  {\bibinfo {volume} {51}},\ \bibinfo {pages} {236–267} (\bibinfo {year}
  {2019})}\BibitemShut {NoStop}%
\bibitem [{\citenamefont {Xi}\ \emph {et~al.}(2020)\citenamefont {Xi},
  \citenamefont {Yang},\ and\ \citenamefont {Yin}}]{XiYangYinAS2020}%
  \BibitemOpen
  \bibfield  {author} {\bibinfo {author} {\bibfnamefont {H.}~\bibnamefont
  {Xi}}, \bibinfo {author} {\bibfnamefont {F.}~\bibnamefont {Yang}},\ and\
  \bibinfo {author} {\bibfnamefont {J.}~\bibnamefont {Yin}},\ }\bibfield
  {title} {\bibinfo {title} {{Convergence of eigenvector empirical spectral
  distribution of sample covariance matrices}},\ }\href
  {https://doi.org/10.1214/19-AOS1832} {\bibfield  {journal} {\bibinfo
  {journal} {Ann. Stat.}\ }\textbf {\bibinfo {volume} {48}},\ \bibinfo {pages}
  {953 } (\bibinfo {year} {2020})}\BibitemShut {NoStop}%
\bibitem [{\citenamefont {Pereyra}(1991)}]{PereyraJSP1991}%
  \BibitemOpen
  \bibfield  {author} {\bibinfo {author} {\bibfnamefont {P.}~\bibnamefont
  {Pereyra}},\ }\bibfield  {title} {\bibinfo {title} {Random-matrix model for
  dissipative two-level systems},\ }\href
  {https://link.springer.com/article/10.1007/BF01053754} {\bibfield  {journal}
  {\bibinfo  {journal} {J. Stat. Phys.}\ }\textbf {\bibinfo {volume} {65}},\
  \bibinfo {pages} {773} (\bibinfo {year} {1991})}\BibitemShut {NoStop}%
\bibitem [{\citenamefont {Esposito}\ and\ \citenamefont
  {Gaspard}(2003)}]{EspositoGaspardPRE2003b}%
  \BibitemOpen
  \bibfield  {author} {\bibinfo {author} {\bibfnamefont {M.}~\bibnamefont
  {Esposito}}\ and\ \bibinfo {author} {\bibfnamefont {P.}~\bibnamefont
  {Gaspard}},\ }\bibfield  {title} {\bibinfo {title} {Spin relaxation in a
  complex environment},\ }\href {https://doi.org/10.1103/PhysRevE.68.066113}
  {\bibfield  {journal} {\bibinfo  {journal} {Phys. Rev. E}\ }\textbf {\bibinfo
  {volume} {68}},\ \bibinfo {pages} {066113} (\bibinfo {year}
  {2003})}\BibitemShut {NoStop}%
\bibitem [{\citenamefont {Lebowitz}\ and\ \citenamefont
  {Pastur}(2004)}]{LebowitzPasturJPA2004}%
  \BibitemOpen
  \bibfield  {author} {\bibinfo {author} {\bibfnamefont {J.~L.}\ \bibnamefont
  {Lebowitz}}\ and\ \bibinfo {author} {\bibfnamefont {L.}~\bibnamefont
  {Pastur}},\ }\bibfield  {title} {\bibinfo {title} {A random matrix model of
  relaxation},\ }\href {https://doi.org/10.1088/0305-4470/37/5/004} {\bibfield
  {journal} {\bibinfo  {journal} {J. Phys. A}\ }\textbf {\bibinfo {volume}
  {37}},\ \bibinfo {pages} {1517} (\bibinfo {year} {2004})}\BibitemShut
  {NoStop}%
\bibitem [{\citenamefont {Gorin}\ \emph {et~al.}(2008)\citenamefont {Gorin},
  \citenamefont {Pineda}, \citenamefont {Kohler},\ and\ \citenamefont
  {Seligman}}]{GorinEtAlNJP2008}%
  \BibitemOpen
  \bibfield  {author} {\bibinfo {author} {\bibfnamefont {T.}~\bibnamefont
  {Gorin}}, \bibinfo {author} {\bibfnamefont {C.}~\bibnamefont {Pineda}},
  \bibinfo {author} {\bibfnamefont {H.}~\bibnamefont {Kohler}},\ and\ \bibinfo
  {author} {\bibfnamefont {T.~H.}\ \bibnamefont {Seligman}},\ }\bibfield
  {title} {\bibinfo {title} {A random matrix theory of decoherence},\ }\href
  {https://doi.org/10.1088/1367-2630/10/11/115016} {\bibfield  {journal}
  {\bibinfo  {journal} {New Journal of Physics}\ }\textbf {\bibinfo {volume}
  {10}},\ \bibinfo {pages} {115016} (\bibinfo {year} {2008})}\BibitemShut
  {NoStop}%
\bibitem [{\citenamefont {Bartsch}\ \emph {et~al.}(2008)\citenamefont
  {Bartsch}, \citenamefont {Steinigeweg},\ and\ \citenamefont
  {Gemmer}}]{BartschSteinigewegGemmerPRE2008}%
  \BibitemOpen
  \bibfield  {author} {\bibinfo {author} {\bibfnamefont {C.}~\bibnamefont
  {Bartsch}}, \bibinfo {author} {\bibfnamefont {R.}~\bibnamefont
  {Steinigeweg}},\ and\ \bibinfo {author} {\bibfnamefont {J.}~\bibnamefont
  {Gemmer}},\ }\bibfield  {title} {\bibinfo {title} {Occurrence of exponential
  relaxation in closed quantum systems},\ }\href
  {https://doi.org/10.1103/PhysRevE.77.011119} {\bibfield  {journal} {\bibinfo
  {journal} {Phys. Rev. E}\ }\textbf {\bibinfo {volume} {77}},\ \bibinfo
  {pages} {011119} (\bibinfo {year} {2008})}\BibitemShut {NoStop}%
\bibitem [{\citenamefont {Genway}\ \emph {et~al.}(2013)\citenamefont {Genway},
  \citenamefont {Ho},\ and\ \citenamefont {Lee}}]{GenwayHoLeePRL2013}%
  \BibitemOpen
  \bibfield  {author} {\bibinfo {author} {\bibfnamefont {S.}~\bibnamefont
  {Genway}}, \bibinfo {author} {\bibfnamefont {A.~F.}\ \bibnamefont {Ho}},\
  and\ \bibinfo {author} {\bibfnamefont {D.~K.~K.}\ \bibnamefont {Lee}},\
  }\bibfield  {title} {\bibinfo {title} {{Dynamics of Thermalization and
  Decoherence of a Nanoscale System}},\ }\href
  {https://doi.org/10.1103/PhysRevLett.111.130408} {\bibfield  {journal}
  {\bibinfo  {journal} {Phys. Rev. Lett.}\ }\textbf {\bibinfo {volume} {111}},\
  \bibinfo {pages} {130408} (\bibinfo {year} {2013})}\BibitemShut {NoStop}%
\bibitem [{\citenamefont {Riera-Campeny}\ \emph {et~al.}(2021)\citenamefont
  {Riera-Campeny}, \citenamefont {Sanpera},\ and\ \citenamefont
  {Strasberg}}]{RieraCampenySanperaStrasbergPRXQ2021}%
  \BibitemOpen
  \bibfield  {author} {\bibinfo {author} {\bibfnamefont {A.}~\bibnamefont
  {Riera-Campeny}}, \bibinfo {author} {\bibfnamefont {A.}~\bibnamefont
  {Sanpera}},\ and\ \bibinfo {author} {\bibfnamefont {P.}~\bibnamefont
  {Strasberg}},\ }\bibfield  {title} {\bibinfo {title} {{Quantum Systems
  Correlated with a Finite Bath: Nonequilibrium Dynamics and Thermodynamics}},\
  }\href {https://doi.org/10.1103/PRXQuantum.2.010340} {\bibfield  {journal}
  {\bibinfo  {journal} {PRX Quantum}\ }\textbf {\bibinfo {volume} {2}},\
  \bibinfo {pages} {010340} (\bibinfo {year} {2021})}\BibitemShut {NoStop}%
\bibitem [{\citenamefont {Albrecht}\ \emph {et~al.}(2022)\citenamefont
  {Albrecht}, \citenamefont {Baunach},\ and\ \citenamefont
  {Arrasmith}}]{AlbrechtBaunachArrasmithPRD2022}%
  \BibitemOpen
  \bibfield  {author} {\bibinfo {author} {\bibfnamefont {A.}~\bibnamefont
  {Albrecht}}, \bibinfo {author} {\bibfnamefont {R.}~\bibnamefont {Baunach}},\
  and\ \bibinfo {author} {\bibfnamefont {A.}~\bibnamefont {Arrasmith}},\
  }\bibfield  {title} {\bibinfo {title} {Einselection, equilibrium and
  cosmology},\ }\href {https://doi.org/10.1103/PhysRevD.106.123507} {\bibfield
  {journal} {\bibinfo  {journal} {Phys. Rev. D}\ }\textbf {\bibinfo {volume}
  {106}},\ \bibinfo {pages} {123507} (\bibinfo {year} {2022})}\BibitemShut
  {NoStop}%
\bibitem [{\citenamefont {Yan}\ and\ \citenamefont
  {Zurek}(2022)}]{YanZurekNJP2022}%
  \BibitemOpen
  \bibfield  {author} {\bibinfo {author} {\bibfnamefont {B.}~\bibnamefont
  {Yan}}\ and\ \bibinfo {author} {\bibfnamefont {W.~H.}\ \bibnamefont
  {Zurek}},\ }\bibfield  {title} {\bibinfo {title} {{Decoherence factor as a
  convolution: an interplay between a Gaussian and an exponential coherence
  loss}},\ }\href {https://doi.org/10.1088/1367-2630/ac9fe8} {\bibfield
  {journal} {\bibinfo  {journal} {New J. Phys.}\ }\textbf {\bibinfo {volume}
  {24}},\ \bibinfo {pages} {113029} (\bibinfo {year} {2022})}\BibitemShut
  {NoStop}%
\bibitem [{\citenamefont {Das}\ and\ \citenamefont
  {Ghosh}(2022)}]{DasGhoshJSM2022}%
  \BibitemOpen
  \bibfield  {author} {\bibinfo {author} {\bibfnamefont {A.~K.}\ \bibnamefont
  {Das}}\ and\ \bibinfo {author} {\bibfnamefont {A.}~\bibnamefont {Ghosh}},\
  }\bibfield  {title} {\bibinfo {title} {{Chaos due to symmetry-breaking in
  deformed Poisson ensemble}},\ }\href
  {https://doi.org/10.1088/1742-5468/ac70dd} {\bibfield  {journal} {\bibinfo
  {journal} {J. Stat. Mech.}\ }\textbf {\bibinfo {volume} {2022}},\ \bibinfo
  {pages} {063101} (\bibinfo {year} {2022})}\BibitemShut {NoStop}%
\bibitem [{\citenamefont {Strasberg}\ and\ \citenamefont
  {Schindler}(2024)}]{StrasbergSchindlerSP2024}%
  \BibitemOpen
  \bibfield  {author} {\bibinfo {author} {\bibfnamefont {P.}~\bibnamefont
  {Strasberg}}\ and\ \bibinfo {author} {\bibfnamefont {J.}~\bibnamefont
  {Schindler}},\ }\bibfield  {title} {\bibinfo {title} {{Comparative
  Microscopic Study of Entropies and their Production}},\ }\href
  {https://doi.org/10.21468/SciPostPhys.17.5.143} {\bibfield  {journal}
  {\bibinfo  {journal} {SciPost Phys.}\ }\textbf {\bibinfo {volume} {17}},\
  \bibinfo {pages} {143} (\bibinfo {year} {2024})}\BibitemShut {NoStop}%
\bibitem [{\citenamefont {Petz}(1986)}]{PetzCMP1986}%
  \BibitemOpen
  \bibfield  {author} {\bibinfo {author} {\bibfnamefont {D.}~\bibnamefont
  {Petz}},\ }\bibfield  {title} {\bibinfo {title} {{Sufficient subalgebras and
  the relative entropy of states of a von Neumann algebra}},\ }\href
  {https://doi.org/10.1007/BF01212345} {\bibfield  {journal} {\bibinfo
  {journal} {Commun. Math. Phys.}\ }\textbf {\bibinfo {volume} {105}},\
  \bibinfo {pages} {123} (\bibinfo {year} {1986})}\BibitemShut {NoStop}%
\bibitem [{\citenamefont {Petz}(1988)}]{PetzQJM1988}%
  \BibitemOpen
  \bibfield  {author} {\bibinfo {author} {\bibfnamefont {D.}~\bibnamefont
  {Petz}},\ }\bibfield  {title} {\bibinfo {title} {{Sufficiency of channels
  over von Neumann algebras}},\ }\href {https://doi.org/10.1093/qmath/39.1.97}
  {\bibfield  {journal} {\bibinfo  {journal} {Q. J. Math.}\ }\textbf {\bibinfo
  {volume} {39}},\ \bibinfo {pages} {97} (\bibinfo {year} {1988})}\BibitemShut
  {NoStop}%
\bibitem [{\citenamefont {Junge}\ \emph {et~al.}(2018)\citenamefont {Junge},
  \citenamefont {Renner}, \citenamefont {Sutter}, \citenamefont {Wilde},\ and\
  \citenamefont {Winter}}]{JungeEtAlAHP2018}%
  \BibitemOpen
  \bibfield  {author} {\bibinfo {author} {\bibfnamefont {M.}~\bibnamefont
  {Junge}}, \bibinfo {author} {\bibfnamefont {R.}~\bibnamefont {Renner}},
  \bibinfo {author} {\bibfnamefont {D.}~\bibnamefont {Sutter}}, \bibinfo
  {author} {\bibfnamefont {M.~M.}\ \bibnamefont {Wilde}},\ and\ \bibinfo
  {author} {\bibfnamefont {A.}~\bibnamefont {Winter}},\ }\bibfield  {title}
  {\bibinfo {title} {Universal recovery maps and approximate sufficiency of
  quantum relative entropy},\ }\href
  {https://doi.org/10.1007/s00023-018-0716-0} {\bibfield  {journal} {\bibinfo
  {journal} {Ann. Henri Poincare}\ }\textbf {\bibinfo {volume} {19}},\ \bibinfo
  {pages} {2955–78} (\bibinfo {year} {2018})}\BibitemShut {NoStop}%
\bibitem [{\citenamefont {Parzygnat}\ and\ \citenamefont
  {Buscemi}(2023)}]{ParzygnatBuscemiQuantum2023}%
  \BibitemOpen
  \bibfield  {author} {\bibinfo {author} {\bibfnamefont {A.~J.}\ \bibnamefont
  {Parzygnat}}\ and\ \bibinfo {author} {\bibfnamefont {F.}~\bibnamefont
  {Buscemi}},\ }\bibfield  {title} {\bibinfo {title} {Axioms for retrodiction:
  achieving time-reversal symmetry with a prior},\ }\href
  {https://doi.org/10.22331/q-2023-05-23-1013} {\bibfield  {journal} {\bibinfo
  {journal} {{Quantum}}\ }\textbf {\bibinfo {volume} {7}},\ \bibinfo {pages}
  {1013} (\bibinfo {year} {2023})}\BibitemShut {NoStop}%
\bibitem [{\citenamefont {Bai}\ \emph {et~al.}(2025)\citenamefont {Bai},
  \citenamefont {Buscemi},\ and\ \citenamefont
  {Scarani}}]{BaiBuscemiScaraniPRL2025}%
  \BibitemOpen
  \bibfield  {author} {\bibinfo {author} {\bibfnamefont {G.}~\bibnamefont
  {Bai}}, \bibinfo {author} {\bibfnamefont {F.}~\bibnamefont {Buscemi}},\ and\
  \bibinfo {author} {\bibfnamefont {V.}~\bibnamefont {Scarani}},\ }\bibfield
  {title} {\bibinfo {title} {Quantum bayes' rule and petz transpose map from
  the minimum change principle},\ }\href {https://doi.org/10.1103/5n4p-bxhm}
  {\bibfield  {journal} {\bibinfo  {journal} {Phys. Rev. Lett.}\ }\textbf
  {\bibinfo {volume} {135}},\ \bibinfo {pages} {090203} (\bibinfo {year}
  {2025})}\BibitemShut {NoStop}%
\bibitem [{\citenamefont {{Van Kampen}}(1954)}]{VanKampenPhys1954}%
  \BibitemOpen
  \bibfield  {author} {\bibinfo {author} {\bibfnamefont {N.}~\bibnamefont {{Van
  Kampen}}},\ }\bibfield  {title} {\bibinfo {title} {Quantum statistics of
  irreversible processes},\ }\href
  {https://doi.org/10.1016/S0031-8914(54)80074-7} {\bibfield  {journal}
  {\bibinfo  {journal} {Physica}\ }\textbf {\bibinfo {volume} {20}},\ \bibinfo
  {pages} {603} (\bibinfo {year} {1954})}\BibitemShut {NoStop}%
\bibitem [{\citenamefont {DeWitt}(1967)}]{DeWittPR1967}%
  \BibitemOpen
  \bibfield  {author} {\bibinfo {author} {\bibfnamefont {B.~S.}\ \bibnamefont
  {DeWitt}},\ }\bibfield  {title} {\bibinfo {title} {{Quantum Theory of
  Gravity. I. The Canonical Theory}},\ }\href
  {https://doi.org/10.1103/PhysRev.160.1113} {\bibfield  {journal} {\bibinfo
  {journal} {Phys. Rev.}\ }\textbf {\bibinfo {volume} {160}},\ \bibinfo {pages}
  {1113} (\bibinfo {year} {1967})}\BibitemShut {NoStop}%
\bibitem [{\citenamefont {{Wolfram Research, Inc.}}(2025)}]{Mathematica14}%
  \BibitemOpen
  \bibfield  {author} {\bibinfo {author} {\bibnamefont {{Wolfram Research,
  Inc.}}},\ }\href {https://www.wolfram.com/mathematica} {\emph {\bibinfo
  {title} {Mathematica}}},\ \bibinfo {edition} {version 14}\ ed.\ (\bibinfo
  {publisher} {Wolfram Research, Inc.},\ \bibinfo {address} {Champaign,
  Illinois},\ \bibinfo {year} {2025})\BibitemShut {NoStop}%
\bibitem [{\citenamefont {Hartle}(1991)}]{HartlePRD1991}%
  \BibitemOpen
  \bibfield  {author} {\bibinfo {author} {\bibfnamefont {J.~B.}\ \bibnamefont
  {Hartle}},\ }\bibfield  {title} {\bibinfo {title} {Spacetime coarse grainings
  in nonrelativistic quantum mechanics},\ }\href
  {https://doi.org/10.1103/PhysRevD.44.3173} {\bibfield  {journal} {\bibinfo
  {journal} {Phys. Rev. D}\ }\textbf {\bibinfo {volume} {44}},\ \bibinfo
  {pages} {3173} (\bibinfo {year} {1991})}\BibitemShut {NoStop}%
\bibitem [{\citenamefont {Gillani}(2021)}]{GillaniInterview2021}%
  \BibitemOpen
  \bibfield  {author} {\bibinfo {author} {\bibfnamefont {N.}~\bibnamefont
  {Gillani}},\ }\href
  {https://theconversation.com/is-space-infinite-we-asked-5-experts-165742}
  {\bibinfo {title} {{Is space infinite? We asked 5 experts}}},\ \bibinfo
  {howpublished} {The Conversation} (\bibinfo {year} {2021})\BibitemShut
  {NoStop}%
\bibitem [{\citenamefont {Nielsen}\ and\ \citenamefont
  {Ninomiya}(1981{\natexlab{a}})}]{NielsenNinomiyaNPB1981a}%
  \BibitemOpen
  \bibfield  {author} {\bibinfo {author} {\bibfnamefont {H.}~\bibnamefont
  {Nielsen}}\ and\ \bibinfo {author} {\bibfnamefont {M.}~\bibnamefont
  {Ninomiya}},\ }\bibfield  {title} {\bibinfo {title} {{Absence of neutrinos on
  a lattice: (I). Proof by homotopy theory}},\ }\href
  {https://doi.org/https://doi.org/10.1016/0550-3213(81)90361-8} {\bibfield
  {journal} {\bibinfo  {journal} {Nucl. Phys. B}\ }\textbf {\bibinfo {volume}
  {185}},\ \bibinfo {pages} {20} (\bibinfo {year}
  {1981}{\natexlab{a}})}\BibitemShut {NoStop}%
\bibitem [{\citenamefont {Nielsen}\ and\ \citenamefont
  {Ninomiya}(1981{\natexlab{b}})}]{NielsenNinomiyaNPB1981b}%
  \BibitemOpen
  \bibfield  {author} {\bibinfo {author} {\bibfnamefont {H.}~\bibnamefont
  {Nielsen}}\ and\ \bibinfo {author} {\bibfnamefont {M.}~\bibnamefont
  {Ninomiya}},\ }\bibfield  {title} {\bibinfo {title} {{Absence of neutrinos on
  a lattice: (II). Intuitive topological proof}},\ }\href
  {https://doi.org/https://doi.org/10.1016/0550-3213(81)90524-1} {\bibfield
  {journal} {\bibinfo  {journal} {Nucl. Phys. B}\ }\textbf {\bibinfo {volume}
  {193}},\ \bibinfo {pages} {173} (\bibinfo {year}
  {1981}{\natexlab{b}})}\BibitemShut {NoStop}%
\bibitem [{\citenamefont {D.Tong}(2018)}]{TongNotes2018}%
  \BibitemOpen
  \bibfield  {author} {\bibinfo {author} {\bibnamefont {D.Tong}},\ }\href
  {https://www.damtp.cam.ac.uk/user/tong/gaugetheory.html} {\emph {\bibinfo
  {title} {{Lectures on Gauge Theory}}}},\ {Lecture Notes}\ (\bibinfo
  {publisher} {University of Cambridge},\ \bibinfo {year} {2018})\BibitemShut
  {NoStop}%
\bibitem [{\citenamefont {Bekenstein}(1980)}]{BekensteinPT1980}%
  \BibitemOpen
  \bibfield  {author} {\bibinfo {author} {\bibfnamefont {J.~D.}\ \bibnamefont
  {Bekenstein}},\ }\bibfield  {title} {\bibinfo {title} {Black‐hole
  thermodynamics},\ }\href {https://doi.org/10.1063/1.2913906} {\bibfield
  {journal} {\bibinfo  {journal} {Phys. Today}\ }\textbf {\bibinfo {volume}
  {33}},\ \bibinfo {pages} {24} (\bibinfo {year} {1980})}\BibitemShut {NoStop}%
\bibitem [{\citenamefont {Lloyd}(2000)}]{LloydNature2000}%
  \BibitemOpen
  \bibfield  {author} {\bibinfo {author} {\bibfnamefont {S.}~\bibnamefont
  {Lloyd}},\ }\bibfield  {title} {\bibinfo {title} {{Ultimate physical limits
  to computation}},\ }\href {https://doi.org/10.1038/35023282} {\bibfield
  {journal} {\bibinfo  {journal} {Nature}\ }\textbf {\bibinfo {volume} {406}},\
  \bibinfo {pages} {1047} (\bibinfo {year} {2000})}\BibitemShut {NoStop}%
\bibitem [{\citenamefont {Lloyd}(2002)}]{LloydPRL2002}%
  \BibitemOpen
  \bibfield  {author} {\bibinfo {author} {\bibfnamefont {S.}~\bibnamefont
  {Lloyd}},\ }\bibfield  {title} {\bibinfo {title} {Computational capacity of
  the universe},\ }\href {https://doi.org/10.1103/PhysRevLett.88.237901}
  {\bibfield  {journal} {\bibinfo  {journal} {Phys. Rev. Lett.}\ }\textbf
  {\bibinfo {volume} {88}},\ \bibinfo {pages} {237901} (\bibinfo {year}
  {2002})}\BibitemShut {NoStop}%
\bibitem [{\citenamefont {Price}(1996)}]{PriceBook1996}%
  \BibitemOpen
  \bibfield  {author} {\bibinfo {author} {\bibfnamefont {H.}~\bibnamefont
  {Price}},\ }\href@noop {} {\emph {\bibinfo {title} {{Time's Arrow and
  Archimedes' Point}}}}\ (\bibinfo  {publisher} {Oxford University Press},\
  \bibinfo {address} {New York},\ \bibinfo {year} {1996})\BibitemShut {NoStop}%
\bibitem [{\citenamefont {Barbour}(1999)}]{BarbourBook1999}%
  \BibitemOpen
  \bibfield  {author} {\bibinfo {author} {\bibfnamefont {J.}~\bibnamefont
  {Barbour}},\ }\href@noop {} {\emph {\bibinfo {title} {{The End of Time: The
  Next Revolution in Physics}}}}\ (\bibinfo  {publisher} {Oxford University
  Press},\ \bibinfo {address} {Oxford},\ \bibinfo {year} {1999})\BibitemShut
  {NoStop}%
\bibitem [{\citenamefont {Barbour}(2009)}]{BarbourArXiv2009}%
  \BibitemOpen
  \bibfield  {author} {\bibinfo {author} {\bibfnamefont {J.}~\bibnamefont
  {Barbour}},\ }\bibfield  {title} {\bibinfo {title} {{The Nature of Time}},\
  }\href {https://arxiv.org/abs/0903.3489} {\bibfield  {journal} {\bibinfo
  {journal} {arXiv: 0903.3489}\ } (\bibinfo {year} {2009})}\BibitemShut
  {NoStop}%
\bibitem [{\citenamefont {Kiefer}(2017)}]{KieferInBook2017}%
  \BibitemOpen
  \bibfield  {author} {\bibinfo {author} {\bibfnamefont {C.}~\bibnamefont
  {Kiefer}},\ }\bibinfo {title} {{Towards a Theory of Spacetime Theories.
  Einstein Studies}}\ (\bibinfo  {publisher} {Birkhäuser},\ \bibinfo {address}
  {New York},\ \bibinfo {year} {2017})\ Chap.\ \bibinfo {chapter} {{Does Time
  Exist in Quantum Gravity?}}\BibitemShut {Stop}%
\bibitem [{\citenamefont {Fert\'e}\ \emph {et~al.}(2026)\citenamefont
  {Fert\'e}, \citenamefont {Farci},\ and\ \citenamefont
  {Cao}}]{FerteFarciCaoPRL2026}%
  \BibitemOpen
  \bibfield  {author} {\bibinfo {author} {\bibfnamefont {B.}~\bibnamefont
  {Fert\'e}}, \bibinfo {author} {\bibfnamefont {D.}~\bibnamefont {Farci}},\
  and\ \bibinfo {author} {\bibfnamefont {X.}~\bibnamefont {Cao}},\ }\bibfield
  {title} {\bibinfo {title} {{Decoherent Histories with(out) Objectivity in a
  (Broken) Apparatus}},\ }\href {https://doi.org/10.1103/8qzx-xpfz} {\bibfield
  {journal} {\bibinfo  {journal} {Phys. Rev. Lett.}\ }\textbf {\bibinfo
  {volume} {136}},\ \bibinfo {pages} {090404} (\bibinfo {year}
  {2026})}\BibitemShut {NoStop}%
\bibitem [{\citenamefont {Arefyeva}\ and\ \citenamefont
  {Polyakov}(2026)}]{ArefyevaPolyakovPRA2026}%
  \BibitemOpen
  \bibfield  {author} {\bibinfo {author} {\bibfnamefont {N.}~\bibnamefont
  {Arefyeva}}\ and\ \bibinfo {author} {\bibfnamefont {E.}~\bibnamefont
  {Polyakov}},\ }\bibfield  {title} {\bibinfo {title} {{Emergence of
  non-Markovian decoherent histories in an integrable environment: A
  tape-recorder model for local quantum observables}},\ }\href
  {https://doi.org/10.1103/3ct7-y9tn} {\bibfield  {journal} {\bibinfo
  {journal} {Phys. Rev. A}\ }\textbf {\bibinfo {volume} {113}},\ \bibinfo
  {pages} {062225} (\bibinfo {year} {2026})}\BibitemShut {NoStop}%
\bibitem [{\citenamefont {Maudlin}(2019)}]{MaudlinBook2019}%
  \BibitemOpen
  \bibfield  {author} {\bibinfo {author} {\bibfnamefont {T.}~\bibnamefont
  {Maudlin}},\ }\href@noop {} {\emph {\bibinfo {title} {Philosophy of Physics:
  Quantum Theory}}}\ (\bibinfo  {publisher} {Princeton University Press},\
  \bibinfo {address} {Princeton},\ \bibinfo {year} {2019})\BibitemShut
  {NoStop}%
\bibitem [{\citenamefont {Vaidman}(2020)}]{Vaidman2020}%
  \BibitemOpen
  \bibfield  {author} {\bibinfo {author} {\bibfnamefont {L.}~\bibnamefont
  {Vaidman}},\ }\bibinfo {title} {{Derivations of the Born Rule}},\ in\ \href
  {https://doi.org/10.1007/978-3-030-34316-3_26} {\emph {\bibinfo {booktitle}
  {Quantum, Probability, Logic: The Work and Influence of Itamar Pitowsky}}},\
  \bibinfo {editor} {edited by\ \bibinfo {editor} {\bibfnamefont
  {M.}~\bibnamefont {Hemmo}}\ and\ \bibinfo {editor} {\bibfnamefont
  {O.}~\bibnamefont {Shenker}}}\ (\bibinfo  {publisher} {Springer International
  Publishing},\ \bibinfo {address} {Cham},\ \bibinfo {year} {2020})\ Chap.\
  \bibinfo {chapter} {{Derivations of the Born Rule}}, pp.\ \bibinfo {pages}
  {567--584}\BibitemShut {NoStop}%
\bibitem [{\citenamefont {Ridley}(2025)}]{RidleyEJPS2025}%
  \BibitemOpen
  \bibfield  {author} {\bibinfo {author} {\bibfnamefont {M.}~\bibnamefont
  {Ridley}},\ }\bibfield  {title} {\bibinfo {title} {{Many retrocausal worlds:
  A foundation for quantum probability}},\ }\href
  {https://doi.org/10.1007/s13194-025-00696-8} {\bibfield  {journal} {\bibinfo
  {journal} {Euro. J. Phil. Sci.}\ }\textbf {\bibinfo {volume} {15}},\ \bibinfo
  {pages} {70} (\bibinfo {year} {2025})}\BibitemShut {NoStop}%
\bibitem [{\citenamefont {Hartle}(1968)}]{HartleAJP1968}%
  \BibitemOpen
  \bibfield  {author} {\bibinfo {author} {\bibfnamefont {J.~B.}\ \bibnamefont
  {Hartle}},\ }\bibfield  {title} {\bibinfo {title} {{Quantum Mechanics of
  Individual Systems}},\ }\href {https://doi.org/10.1119/1.1975096} {\bibfield
  {journal} {\bibinfo  {journal} {Am. J. Phys.}\ }\textbf {\bibinfo {volume}
  {36}},\ \bibinfo {pages} {704} (\bibinfo {year} {1968})}\BibitemShut
  {NoStop}%
\bibitem [{\citenamefont {DeWitt}\ and\ \citenamefont
  {Graham}(1973)}]{DeWittGrahamBook1973}%
  \BibitemOpen
  \bibinfo {editor} {\bibfnamefont {B.~S.}\ \bibnamefont {DeWitt}}\ and\
  \bibinfo {editor} {\bibfnamefont {N.}~\bibnamefont {Graham}},\ eds.,\
  \href@noop {} {\emph {\bibinfo {title} {The many-worlds interpretation of
  quantum mechanics}}},\ Vol.~\bibinfo {volume} {63}\ (\bibinfo  {publisher}
  {Princeton University Press},\ \bibinfo {address} {Princeton},\ \bibinfo
  {year} {1973})\BibitemShut {NoStop}%
\bibitem [{\citenamefont {Vaidman}(1998)}]{VaidmanISPS1998}%
  \BibitemOpen
  \bibfield  {author} {\bibinfo {author} {\bibfnamefont {L.}~\bibnamefont
  {Vaidman}},\ }\bibfield  {title} {\bibinfo {title} {{On schizophrenic
  experiences of the neutron or why we should believe in the many‐worlds
  interpretation of quantum theory}},\ }\href
  {https://doi.org/10.1080/02698599808573600} {\bibfield  {journal} {\bibinfo
  {journal} {Int. Stud. Phil. Sci.}\ }\textbf {\bibinfo {volume} {12}},\
  \bibinfo {pages} {245} (\bibinfo {year} {1998})}\BibitemShut {NoStop}%
\bibitem [{\citenamefont {Deutsch}(1999)}]{DeutschPRSCA1999}%
  \BibitemOpen
  \bibfield  {author} {\bibinfo {author} {\bibfnamefont {D.}~\bibnamefont
  {Deutsch}},\ }\bibfield  {title} {\bibinfo {title} {Quantum theory of
  probability and decisions},\ }\href {https://doi.org/10.1098/rspa.1999.0443}
  {\bibfield  {journal} {\bibinfo  {journal} {Proc. R. Soc. London, Ser. A}\
  }\textbf {\bibinfo {volume} {455}},\ \bibinfo {pages} {3129} (\bibinfo {year}
  {1999})}\BibitemShut {NoStop}%
\bibitem [{\citenamefont {Wallace}(2003)}]{WallaceSHPSB2003}%
  \BibitemOpen
  \bibfield  {author} {\bibinfo {author} {\bibfnamefont {D.}~\bibnamefont
  {Wallace}},\ }\bibfield  {title} {\bibinfo {title} {{Everettian rationality:
  defending Deutsch's approach to probability in the Everett interpretation}},\
  }\href {https://doi.org/https://doi.org/10.1016/S1355-2198(03)00036-4}
  {\bibfield  {journal} {\bibinfo  {journal} {Stud. Hist. Phil. Sci. B}\
  }\textbf {\bibinfo {volume} {34}},\ \bibinfo {pages} {415} (\bibinfo {year}
  {2003})}\BibitemShut {NoStop}%
\bibitem [{\citenamefont {Hanson}(2003)}]{HansonFP2003}%
  \BibitemOpen
  \bibfield  {author} {\bibinfo {author} {\bibfnamefont {R.}~\bibnamefont
  {Hanson}},\ }\bibfield  {title} {\bibinfo {title} {{When Worlds Collide:
  Quantum Probability from Observer Selection?}},\ }\href
  {https://doi.org/10.1023/A:1025642019178} {\bibfield  {journal} {\bibinfo
  {journal} {Found. Phys.}\ }\textbf {\bibinfo {volume} {33}},\ \bibinfo
  {pages} {1129} (\bibinfo {year} {2003})}\BibitemShut {NoStop}%
\bibitem [{\citenamefont {Greaves}(2004)}]{GreavesSHPSB2004}%
  \BibitemOpen
  \bibfield  {author} {\bibinfo {author} {\bibfnamefont {H.}~\bibnamefont
  {Greaves}},\ }\bibfield  {title} {\bibinfo {title} {{Understanding Deutsch's
  probability in a deterministic multiverse}},\ }\href
  {https://doi.org/https://doi.org/10.1016/j.shpsb.2004.04.006} {\bibfield
  {journal} {\bibinfo  {journal} {Stud. Hist. Phil. Sci. B}\ }\textbf {\bibinfo
  {volume} {35}},\ \bibinfo {pages} {423} (\bibinfo {year} {2004})}\BibitemShut
  {NoStop}%
\bibitem [{\citenamefont {Zurek}(2005)}]{ZurekPRA2005}%
  \BibitemOpen
  \bibfield  {author} {\bibinfo {author} {\bibfnamefont {W.~H.}\ \bibnamefont
  {Zurek}},\ }\bibfield  {title} {\bibinfo {title} {{Probabilities from
  entanglement, Born's rule
  ${p}_{k}={\ensuremath{\mid}{\ensuremath{\psi}}_{k}\ensuremath{\mid}}^{2}$
  from envariance}},\ }\href {https://doi.org/10.1103/PhysRevA.71.052105}
  {\bibfield  {journal} {\bibinfo  {journal} {Phys. Rev. A}\ }\textbf {\bibinfo
  {volume} {71}},\ \bibinfo {pages} {052105} (\bibinfo {year}
  {2005})}\BibitemShut {NoStop}%
\bibitem [{\citenamefont {Buniy}\ \emph {et~al.}(2006)\citenamefont {Buniy},
  \citenamefont {Hsu},\ and\ \citenamefont {Zee}}]{BuniyHsuZeePLB2006}%
  \BibitemOpen
  \bibfield  {author} {\bibinfo {author} {\bibfnamefont {R.~V.}\ \bibnamefont
  {Buniy}}, \bibinfo {author} {\bibfnamefont {S.~D.}\ \bibnamefont {Hsu}},\
  and\ \bibinfo {author} {\bibfnamefont {A.}~\bibnamefont {Zee}},\ }\bibfield
  {title} {\bibinfo {title} {Discreteness and the origin of probability in
  quantum mechanics},\ }\href
  {https://doi.org/https://doi.org/10.1016/j.physletb.2006.07.050} {\bibfield
  {journal} {\bibinfo  {journal} {Phys. Lett. B}\ }\textbf {\bibinfo {volume}
  {640}},\ \bibinfo {pages} {219} (\bibinfo {year} {2006})}\BibitemShut
  {NoStop}%
\bibitem [{\citenamefont {Saunders}(2010)}]{SaundersInBook2010}%
  \BibitemOpen
  \bibfield  {author} {\bibinfo {author} {\bibfnamefont {S.}~\bibnamefont
  {Saunders}},\ }\bibinfo {title} {{Many Worlds? Everett, Quantum Theory and
  Reality}}\ (\bibinfo  {publisher} {Oxford University Press},\ \bibinfo
  {address} {Oxford},\ \bibinfo {year} {2010})\ Chap.\ \bibinfo {chapter}
  {{Chance in the Everett Interpretation}}, pp.\ \bibinfo {pages}
  {181--205}\BibitemShut {NoStop}%
\bibitem [{\citenamefont {Aguirre}\ and\ \citenamefont
  {Tegmark}(2011)}]{AguirreTegmarkPRD2011}%
  \BibitemOpen
  \bibfield  {author} {\bibinfo {author} {\bibfnamefont {A.}~\bibnamefont
  {Aguirre}}\ and\ \bibinfo {author} {\bibfnamefont {M.}~\bibnamefont
  {Tegmark}},\ }\bibfield  {title} {\bibinfo {title} {Born in an infinite
  universe: A cosmological interpretation of quantum mechanics},\ }\href
  {https://doi.org/10.1103/PhysRevD.84.105002} {\bibfield  {journal} {\bibinfo
  {journal} {Phys. Rev. D}\ }\textbf {\bibinfo {volume} {84}},\ \bibinfo
  {pages} {105002} (\bibinfo {year} {2011})}\BibitemShut {NoStop}%
\bibitem [{\citenamefont {Sebens}\ and\ \citenamefont
  {Carroll}(2018)}]{SebensCarrollBJP2018}%
  \BibitemOpen
  \bibfield  {author} {\bibinfo {author} {\bibfnamefont {C.~T.}\ \bibnamefont
  {Sebens}}\ and\ \bibinfo {author} {\bibfnamefont {S.~M.}\ \bibnamefont
  {Carroll}},\ }\bibfield  {title} {\bibinfo {title} {{Self-locating
  Uncertainty and the Origin of Probability in Everettian Quantum Mechanics}},\
  }\href {https://doi.org/10.1093/bjps/axw004} {\bibfield  {journal} {\bibinfo
  {journal} {Brit. J. Phil. Sci.}\ }\textbf {\bibinfo {volume} {69}},\ \bibinfo
  {pages} {25} (\bibinfo {year} {2018})}\BibitemShut {NoStop}%
\bibitem [{\citenamefont {Masanes}\ \emph {et~al.}(2019)\citenamefont
  {Masanes}, \citenamefont {Galley},\ and\ \citenamefont
  {M\"uller}}]{MasanesGalleyMuellerNC2019}%
  \BibitemOpen
  \bibfield  {author} {\bibinfo {author} {\bibfnamefont {L.}~\bibnamefont
  {Masanes}}, \bibinfo {author} {\bibfnamefont {T.~D.}\ \bibnamefont
  {Galley}},\ and\ \bibinfo {author} {\bibfnamefont {M.~P.}\ \bibnamefont
  {M\"uller}},\ }\bibfield  {title} {\bibinfo {title} {The measurement
  postulates of quantum mechanics are operationally redundant},\ }\href
  {https://doi.org/10.1038/s41467-019-09348-x} {\bibfield  {journal} {\bibinfo
  {journal} {Nat. Comm.}\ }\textbf {\bibinfo {volume} {10}},\ \bibinfo {pages}
  {1361} (\bibinfo {year} {2019})}\BibitemShut {NoStop}%
\bibitem [{\citenamefont {McQueen}\ and\ \citenamefont
  {Vaidman}(2019)}]{McQueenVaidmanSHPSB2019}%
  \BibitemOpen
  \bibfield  {author} {\bibinfo {author} {\bibfnamefont {K.~J.}\ \bibnamefont
  {McQueen}}\ and\ \bibinfo {author} {\bibfnamefont {L.}~\bibnamefont
  {Vaidman}},\ }\bibfield  {title} {\bibinfo {title} {In defence of the
  self-location uncertainty account of probability in the many-worlds
  interpretation},\ }\href
  {https://doi.org/https://doi.org/10.1016/j.shpsb.2018.10.003} {\bibfield
  {journal} {\bibinfo  {journal} {Stud. His. Phil. Sci. B}\ }\textbf {\bibinfo
  {volume} {66}},\ \bibinfo {pages} {14} (\bibinfo {year} {2019})}\BibitemShut
  {NoStop}%
\bibitem [{\citenamefont {Saunders}(2021)}]{SaundersPRSA2021}%
  \BibitemOpen
  \bibfield  {author} {\bibinfo {author} {\bibfnamefont {S.}~\bibnamefont
  {Saunders}},\ }\bibfield  {title} {\bibinfo {title} {{Branch-counting in the
  Everett interpretation of quantum mechanics}},\ }\href
  {https://doi.org/10.1098/rspa.2021.0600} {\bibfield  {journal} {\bibinfo
  {journal} {Proc. R. Soc. A}\ }\textbf {\bibinfo {volume} {477}},\ \bibinfo
  {pages} {20210600} (\bibinfo {year} {2021})}\BibitemShut {NoStop}%
\bibitem [{\citenamefont {Hsu}(2021)}]{HsuMPLA2021}%
  \BibitemOpen
  \bibfield  {author} {\bibinfo {author} {\bibfnamefont {S.~D.~H.}\
  \bibnamefont {Hsu}},\ }\bibfield  {title} {\bibinfo {title} {{Discrete
  Hilbert space, the Born Rule, and quantum gravity}},\ }\href
  {https://doi.org/10.1142/S0217732321500139} {\bibfield  {journal} {\bibinfo
  {journal} {Mod. Phys. Lett. A}\ }\textbf {\bibinfo {volume} {36}},\ \bibinfo
  {pages} {2150013} (\bibinfo {year} {2021})}\BibitemShut {NoStop}%
\bibitem [{\citenamefont {Short}(2023)}]{ShortQuantum2023}%
  \BibitemOpen
  \bibfield  {author} {\bibinfo {author} {\bibfnamefont {A.~J.}\ \bibnamefont
  {Short}},\ }\bibfield  {title} {\bibinfo {title} {Probability in many-worlds
  theories},\ }\href {https://doi.org/10.22331/q-2023-04-06-971} {\bibfield
  {journal} {\bibinfo  {journal} {Quantum}\ }\textbf {\bibinfo {volume} {7}},\
  \bibinfo {pages} {971} (\bibinfo {year} {2023})}\BibitemShut {NoStop}%
\bibitem [{\citenamefont {Lazarovici}(2023)}]{LazaroviciQR2023}%
  \BibitemOpen
  \bibfield  {author} {\bibinfo {author} {\bibfnamefont {D.}~\bibnamefont
  {Lazarovici}},\ }\bibfield  {title} {\bibinfo {title} {{How Everett Solved
  the Probability Problem in Everettian Quantum Mechanics}},\ }\href
  {https://doi.org/10.3390/quantum5020026} {\bibfield  {journal} {\bibinfo
  {journal} {Quantum Rep.}\ }\textbf {\bibinfo {volume} {5}},\ \bibinfo {pages}
  {407} (\bibinfo {year} {2023})}\BibitemShut {NoStop}%
\bibitem [{\citenamefont {Ridley}(2023)}]{RidleyQR2023}%
  \BibitemOpen
  \bibfield  {author} {\bibinfo {author} {\bibfnamefont {M.}~\bibnamefont
  {Ridley}},\ }\bibfield  {title} {\bibinfo {title} {{Quantum Probability from
  Temporal Structure}},\ }\href {https://doi.org/10.3390/quantum5020033}
  {\bibfield  {journal} {\bibinfo  {journal} {Quantum Rep.}\ }\textbf {\bibinfo
  {volume} {5}},\ \bibinfo {pages} {496} (\bibinfo {year} {2023})}\BibitemShut
  {NoStop}%
\bibitem [{\citenamefont {Weidner}(2025)}]{WeidnerArXiv2025}%
  \BibitemOpen
  \bibfield  {author} {\bibinfo {author} {\bibfnamefont {M.}~\bibnamefont
  {Weidner}},\ }\href {https://arxiv.org/abs/2504.06495} {\bibinfo {title}
  {{Unified Quantum Dynamics: The Emergence of the Born Rule}}} (\bibinfo
  {year} {2025}),\ \Eprint {https://arxiv.org/abs/2504.06495} {arXiv:2504.06495
  [quant-ph]} \BibitemShut {NoStop}%
\bibitem [{\citenamefont {Kent}(2010)}]{KentInBook2010}%
  \BibitemOpen
  \bibfield  {author} {\bibinfo {author} {\bibfnamefont {A.}~\bibnamefont
  {Kent}},\ }\bibinfo {title} {{Many Worlds? Everett, Quantum Theory and
  Reality}}\ (\bibinfo  {publisher} {Oxford University Press},\ \bibinfo
  {address} {Oxford},\ \bibinfo {year} {2010})\ Chap.\ \bibinfo {chapter} {{One
  World Versus Many: The Inadequacy of Everettian Accounts of Evolution,
  Probability and Scientific Confirmation}}, pp.\ \bibinfo {pages}
  {307--354}\BibitemShut {NoStop}%
\end{thebibliography}%

%%%%%%%%%%%%%%%%%%%%%%%%%%%%%%%%%%%%%%%%%%%%%%%%%%%%%%%%%%%%%%%%%%%%%%%%%%%%%%%%%%%%%%%%%%%%%%%%%%%%%%%%%%%%%%%%%%%%%%%%
%\newpage\clearpage
\appendix
%%%%%%%%%%%%%%%%%%%%%%%%%%%%%%%%%%%%%%%%%%%%%%%%%%%%%%%%%%%%%%%%%%%%%%%%%%%%%%%%%%%%%%%%%%%%%%%%%%%%%%%%%%%%%%%%%%%%%%%%
\begin{center}
 \emph{\bb{SUPPLEMENTAL MATERIAL}}
\end{center}
%%%%%%%%%%%%%%%%%%%%%%%%%%%%%%%%%%%%%%%%%%%%%%%%%%%%%%%%%%%%%%%%%%%%%%%%%%%%%%%%%%%%%%%%%%%%%%%%%%%%%%%%%%%%%%%%%%%%%%%%

This supplemental material (SM) contains technical derivations, a few additional analytical results and many additional numerical insights. Analytically, we show that various expressions (such as the average success probability, fidelity, mutual information) concentrate around the mean value, for which we derive closed form expressions asymptotically. These calculations could be also of broader interest for the quantum state discrimination (QSD) or self-location problem (SLP). We also numerically verify the analytical results, which show a very quick convergence to the asymptotic regime, and which also includes a plot of the maximum past-present correlations in the multiverse. We then study numerically the structure of decoherence for our random matrix model from Sec.~\ref{sec:structure} in detail. In particular, we study the structure of decoherence as a function of (a) the localization or inverse participation ratio in Hilbert space, (b) the purity of the Petz recovered state, (c) the Hamming distance and (d) Born's rule. This is done for two different initial states and two different time steps.

%%%%%%%%%%%%%%%%%%%%%%%%%%%%%%%%%%%%%%%%%%%%%%%%%%%%%%%%%%%%%%%%%%%%%%%%%%%%%%%%%%%%%%%%%%%%%%%%%%%%%%%%%%%%%%%%%%%%%%%%
\section{Supplement to Sec.~\ref{sec:decoherent histories}}
%%%%%%%%%%%%%%%%%%%%%%%%%%%%%%%%%%%%%%%%%%%%%%%%%%%%%%%%%%%%%%%%%%%%%%%%%%%%%%%%%%%%%%%%%%%%%%%%%%%%%%%%%%%%%%%%%%%%%%%%

%%%%%%%%%%%%%%%%%%%%%%%%%%%%%%%%%%%%%%%%%%%%%%%%%%%%%%%%%%%%%%%%%%%%%%%%%%%%%%%%%%%%%%%%%%%%%%%%%%%%%%%%%%%%%%%%%%%%%%%%
\subsection{DHC for an informationally complete set implies pairwise commutativity}
%%%%%%%%%%%%%%%%%%%%%%%%%%%%%%%%%%%%%%%%%%%%%%%%%%%%%%%%%%%%%%%%%%%%%%%%%%%%%%%%%%%%%%%%%%%%%%%%%%%%%%%%%%%%%%%%%%%%%%%%

Before we prove the statement, we give a counterexample that shows that the DHC for a specific initial state is not sufficient to establish pairwise commutativity.

The model results from a unitary dilation of a repeatedly measured quantum system using additional apparatus subsystems that keeps track of the measurement results, an extension of a setting originally conceived by von Neumann. Specifically, we consider a ``system'' $S$ and $L$ ancilla systems $A_1, \dots, A_L$ for the apparatus. The unitary evolution $U_S$ of $S$ alone is interrupted at times $t_k = k\Delta t$ by a unitary $V_k$ acting on $S$ and $A_k$ such that $|\Psi_L\rangle = V_L U_S \cdots V_2 U_S V_1 U_S |\Psi_0\rangle$. Here, $V_k = \sum_{x=0}^{M-1} \Pi^S_x\otimes S_x^{(k)}$ is a unitary ``measurement'' operator, where $\{\Pi_x^S\}$ is a complete set of orthogonal system projectors and $S_x^{(k)}$ the unitary shift operator acting on $A_k$, defined by $S_x|j\rangle = |j+x \mod M\rangle$ where $(|j\rangle)$ is some orthonormal basis in the ancilla space.

Now, let us consider the specific initial state $|\Psi_0\rangle = |\psi_0\rangle_S\otimes |0\rangle_{A_1}\otimes\cdots\otimes|0\rangle_{A_L}$. Then, it follows that history states $|\psi'(\bs x)\rangle$ defined by projectors $\Pi_{x_k}^S$ at times $t_k$ for $k\in\{1,\dots,L\}$ obey the DHC. Note that $[V_k,\Pi_{x_k}^S] = 0$, so the order of the $V_k$ and the history projectors does not matter. Yet, for system evolutions $U_S$ that do not commute with $\Pi_x^S$ one finds that the resulting Heisenberg picture projectors do not commute in general. Furthermore, note that the orthogonal record states are proportional to $|\bs x\rangle_{\bs A} \equiv |x_1\rangle_{A_1}\otimes\cdots\otimes|x_L\rangle_{A_L}$ such that one possible preferred decomposition of the final wave function is $|\Psi_L\rangle = \sum_{\bs x} |\psi'(\bs x)\rangle_S\otimes|\bs x\rangle_{\bs A}$ with $|\psi'(\bs x)\rangle_S = \lr{\bs x|_{\bs A}|\Psi_L}$. \\

\begin{center}
 \emph{Proof}
\end{center}

We start by rewriting the DHC using the trace and eqn~(\ref{eq:histories Heisenberg}):
\begin{widetext}
\begin{equation}
	\lr{\psi'(\bs x)|\psi'(\bs x')}
	= \mbox{tr}\left\{ \Pi_{x_1}^{(L-1)} \cdots \Pi_{x_{L-1}}^{(1)}\Pi_{x_L}^{(0)} \Pi_{x'_L}^{(0)}\Pi_{x'_{L-1}}^{(1)} \cdots \Pi_{x'_1}^{(L-1)}U_{t_L-t_0}|\Psi_0\rl\Psi_0|U_{t_L-t_0}^\dagger \right\}
	= 0 ~~~ \forall \bs x\neq\bs x'.
\end{equation}
We use that decoherence of all histories of length $L$ implies decoherence of any history defined on any subset of times. Hence, for $t_k > t_j$ we find for $x_k=x'_k$ and $x_j\neq x'_j$
\begin{equation}
	\lr{\psi'(x_k,x_j)|\psi'(x_k,x'_j)}
	= \mbox{tr}\left\{ \Pi_{x_j}^{(L-j)} \Pi_{x_k}^{(L-k)} \Pi_{x'_j}^{(L-j)}U_{t_L-t_0}|\Psi_0\rl\Psi_0|U_{t_L-t_0}^\dagger \right\} = 0.
\end{equation}
\end{widetext}
Since $\mbox{tr}\{A\rho\} = 0$ for a basis of density matrices $\rho$ implies $A=0$, we find that the DHC for an informationally complete set $\{|\Psi_0\rangle\}$, by which we mean that $\{|\Psi_0\rl\Psi_0|\}$ is a basis of the space of $D\times D$ complex matrices, implies
\begin{equation}
	\Pi_{x_j}^{(L-j)} \Pi_{x_k}^{(L-k)} \Pi_{x'_j}^{(L-j)} = 0.
\end{equation}
For simplicity of notation, we now label $P_j \equiv \Pi_{x_j}^{(L-j)}$ and $Q_k \equiv \Pi_{x_k}^{(L-k)}$, where $\{P_j\}$ and $\{Q_k\}$ are a complete set of orthogonal projectors. We now claim that $P_jQ_kP_\ell = 0$ for all $k$ and all $j\neq\ell$ implies $[P_j,Q_k] = 0$ for all $j$ and $k$.

To prove it, we set $R_{jk} \equiv [P_j,Q_k]$. Multiplying $R_{jk}$ by $P_\ell$ from the right and using $P_jQ_kP_\ell = 0$, we find $R_{jk}P_\ell = 0$ for all $k$ and all $j\neq\ell$. This implies $R_{jk}P_j = R_{jk}$ owing to the completeness of $\{P_j\}$. We then find $P_j R_{jk} = P_j R_{jk}P_j = 0$. Finally, note that $R_{jk}$ is skew Hermitian: $R_{jk}^\dagger = -R_{jk}$. Thus, we finally find $R_{jk}^\dagger = (R_{jk}P_j)^\dagger = -P_j R_{jk} = 0$ as desired.

%%%%%%%%%%%%%%%%%%%%%%%%%%%%%%%%%%%%%%%%%%%%%%%%%%%%%%%%%%%%%%%%%%%%%%%%%%%%%%%%%%%%%%%%%%%%%%%%%%%%%%%%%%%%%%%%%%%%%%%%
\section{Supplement to Sec.~\ref{sec:records}}
%%%%%%%%%%%%%%%%%%%%%%%%%%%%%%%%%%%%%%%%%%%%%%%%%%%%%%%%%%%%%%%%%%%%%%%%%%%%%%%%%%%%%%%%%%%%%%%%%%%%%%%%%%%%%%%%%%%%%%%%

\emph{Outline.} We first briefly review some useful mathematical tools (Sec.~\ref{S:sec:preliminaries}). While all the remaining subsections can be read independently, we recommend taking a look at Sec.~\ref{S:sec:preliminaries} first. The remaining subsections mostly focus on proving statements of the main text, but some also contain additional results (in particular Sec.~\ref{S:sec:mutual info}). We also present additional numerical results for our Haar random histories model from Sec.~\ref{sec:sub-Haar}. Remember that the states are Haar randomly drawn from $\mathbb{C}^d$ with $d = D^{2\alpha}$ and $\alpha$ the scaling exponent of eqn~(\ref{eq:DF scaling form}). The (effective) Hilbert space dimension is then denoted by $d$ instead of $D$ in the following. In all numerical calculations we will assume uniform weights $q_j=1/N$, but most analytical results are valid for arbitrary $q_j$.

%%%%%%%%%%%%%%%%%%%%%%%%%%%%%%%%%%%%%%%%%%%%%%%%%%%%%%%%%%%%%%%%%%%%%%%%%%%%%%%%%%%%%%%%%%%%%%%%%%%%%%%%%%%%%%%%%%%%%%%%
\subsection{Mathematical Preliminaries}\label{S:sec:preliminaries}
%%%%%%%%%%%%%%%%%%%%%%%%%%%%%%%%%%%%%%%%%%%%%%%%%%%%%%%%%%%%%%%%%%%%%%%%%%%%%%%%%%%%%%%%%%%%%%%%%%%%%%%%%%%%%%%%%%%%%%%%

We recall some useful mathematical preliminaries and briefly repeat essential notation.

We consider Hilbert spaces $\C H$ of dimension $D$ (and later of dimension $d$ when we consider the proofs of Sec.~\ref{sec:sub-Haar}). State vectors without a prime, such as $|\psi\rangle, |\phi\rangle, \dots$, are always \emph{normalized} (i.e., $\|\psi\| = \sqrt{\lr{\psi|\psi}} = 1$), whereas unnormalized states get a prime: $|\psi'\rangle, |\phi'\rangle, \dots$ Clearly $|\psi\rangle = |\psi'\rangle/\|\psi'\|$ (except for the zero vector). Collections of vectors $\{|\psi'_1\rangle, |\psi'_2\rangle, \dots, |\psi'_N\rangle\}$ are denoted by $\{\psi'_i\}_N$. If the quantum context does not matter, we also use a boldface notation ($\bs a, \bs b, \bs c, \dots$) for vectors and sequences.

Below, it is particularly important to understand the \emph{typical} behavior of various random quantities. In essence, typicality (also known as concentration of measure) arises whenever there is a low-dimensional and sufficiently smooth function $f(\bs a)$ that depends on a high-dimensional vector $\bs a$ of sufficiently independent random variables~\cite{TalagrandAP1996}. It then turns out that for most $\bs a$, $f(\bs a)$ is very close to its ensemble average $\mu_f \equiv \mathbb{E}[f(\bs a)]$, i.e., fluctuations around its mean are strongly suppressed relative to the mean itself. How strong the suppression is depends on the context, but we will generically denote this phenomenology by $f(\bs a)\simeq \mu_f$.

We start with considerations about Gaussian vectors. For any fixed orthonormal basis $\{e_i\}_D$ they are defined as $|\psi'\rangle = \sum_{i=1}^D c_i|e_i\rangle$ with independent random complex numbers such that $\mbox{Re}(c_i) \sim \C N(0,1/2D)$ and $\mbox{Im}(c_i) \sim \C N(0,1/2D)$, where $\sim\C N(\mu,\sigma^2)$ means drawn from a Gaussian distribution with mean $\mu$ and variance $\sigma^2$. The distribution of $|\psi'\rangle$ is independent of the orthonormal basis $\{|e_i\rangle\}_D$ above and invariant under unitary rotations. Gaussian vectors are normalized on average: $\mathbb{E}[\lr{\psi'|\psi'}] = 1$.
In fact, we have $\lr{\psi'|\psi'} \simeq 1$ because (Lemma 1 in Ref.~\cite{HaydenShorWinterOSID2008}):
\begin{equation}\label{S:eq:norm concentration}
 \mbox{Pr}\left[\big|\|\psi'\|^2-1\big|>\epsilon\right] \le 2\exp(-\epsilon^2D/6).
\end{equation}
Since we are later interested only in rough estimates in high dimensional spaces, eqn~(\ref{S:eq:norm concentration}) permits the assumption that results derived for $|\psi'\rangle$ are also (approximately) true for $|\psi\rangle = |\psi'\rangle/\|\psi'\|$ and vice versa. This can be also easily verified numerically for the results reviewed below. Note that normalized Gaussian vectors are distributed according to the Haar measure.

We call $F = |\lr{\psi|\chi}|$ the fidelity between $|\psi\rangle$ and $|\chi\rangle$, which can be viewed as an inverse distance measure: unit/zero fidelity corresponds to parallel/orthogonal states~\cite{NielsenChuangBook2000}. The probability density for the squared fidelity $F^2$ between an arbitrary fixed vector $|\chi\rangle$ and a Haar random vector $|\psi\rangle$ is~\cite{KusMostowskiHaakeJPA1988}
\begin{equation}\label{S:eq:fidelity pd}
	\rho(F^2) = (D-1)(1-F^2)^{D-2},
\end{equation}
which is also known as the beta distribution $B(1,D-1)$. It follows that we can write
\begin{equation}\label{S:eq:Haar overlap}
 \lr{\psi|\chi} = \frac{r}{\sqrt{D}}
\end{equation}
with a complex random number $r\in\mathbb{C}$ with expectation $\mathbb{E}[r] = 0$ and variance $\mathbb{V}(r) = \mathbb{E}[|r|^2] = 1$.

Gram matrices $G$ are particularly important in this paper. Their elements $G_{ij} = \lr{\psi'_i|\psi'_j}$ are given by the mutual scalar products or overlaps of a collection of vectors $\{\psi'_i\}_N$. Note that we can also write $G = X^\dagger X$ with $X = [|\psi'_1\rangle\cdots|\psi'_N\rangle]$ the $D\times N$ matrix of $N$ column vectors $|\psi'_i\rangle$. We have $\det(G) = 0$ if and only if the vectors are linearly dependent. Random Gram matrices of $N$ independent Gaussian vectors $|\psi'_i\rangle\in\mathbb{C}^D$ are called (complex) Wishart distributed matrices $W$~\cite{GoodmanAMSa1963, GoodmanAMSb1963} or, more generally, sample covariance matrices. An important parameter of the Wishart ensemble is the aspect ratio $\gamma \equiv N/D$. Whenever we talk about an asymptotic result below we mean the limit $N,D\rightarrow\infty$ with $\gamma$ fixed. We have~\cite{BaiSilversteinBook}:

\begin{lemma}\label{lemma Wishart}
	For $\gamma\le1$ the eigenvalues of the Wishart ensemble are asymptotically distributed according to the Marchenko-Pastur measure~\cite{MarchenkoPastur1967}
	\begin{equation}\label{S:eq:MP measure}
		d\rho_\text{MP}(\lambda) = \frac{\sqrt{(\lambda_+-\lambda)(\lambda-\lambda_-)}}{2\pi\gamma\lambda} d\lambda
	\end{equation}
	for $\lambda_-\le\lambda\le\lambda_+$ with $\lambda_\pm = (1\pm\sqrt{\gamma})^2$ (and zero outside). Moreover, the eigenbasis is independent of the eigenvalues $\{\lambda_i\}$ and distributed like a random orthonormal basis drawn from the Haar measure.
\end{lemma}

For $\gamma>1$ the Marchenko-Pastur measure has an additional delta peak at $\lambda=0$, but we are mostly interested in the case $\gamma\le1$ in this paper. Moreover, the Haar distribution of the eigenvectors requires no asymptotic limit, as it follows from the fact that $W$ and $VWV^\dagger$ for any $N\times N$ unitary matrix $V$ are identically distributed.

The eigenvalues statistics satisfy a (modified) central limit theorem (CLT) of the following form:

\begin{lemma}[Theorem 4.2 in Ref.~\cite{LytovaPasturAP2009}]\label{lemma CLT}
 Let $L_f \equiv \sum_{j=1}^N f(\lambda_j)$ with $f$ a bounded function with bounded derivative. Then, asymptotically we have almost surely
\begin{equation}\label{S:eq:mu def}
 \frac{L_f}{N} \longrightarrow \int_{\lambda_-}^{\lambda_+} f(\lambda) d\rho_{MP}(\lambda) \equiv \mu_f(\gamma).
\end{equation}
Moreover, $L_f-N\mu_f(\gamma)$ asymptotically fluctuates like a Gaussian with zero mean and a variance, which depends on $N$ and $D$ only through $\gamma$.\footnote{See eqn~(4.28) in Ref.~\cite{LytovaPasturAP2009}. Note that Ref.~\cite{LytovaPasturAP2009} assumes $\gamma\ge 1$, which is related to the case $\gamma\le1$ by realizing that $W = X^\dagger X$ and $W' = XX^\dagger$ have the same non-zero eigenvalues but a reciprocal aspect ratio. }
\end{lemma}

Lemma~\ref{lemma CLT} should be contrasted with the standard CLT where the variance of a sum of iid variables scales with $N$. The remarkable strong concentration of the eigenvalue statistics holds more widely in random matrix theory~\cite{BaiSilversteinBook, PasturShcherbinaBook2010}.

Finally, let $\bs\psi\in\mathbb{C}^N$ be a Haar random vector and consider the $N$-dimensional vector $\bs w$ of weights $w_i = |\psi_i|^2$. These weights are uniformally distributed over the $(N-1)$-probability simplex (also known as the Dirichlet distribution with parameters $\alpha_1 = \dots = \alpha_N = 1$). Its moments are given by
\begin{equation}\label{S:eq:Dirichlet}
 \mathbb{E}\left[\prod_{i=1}^N w_i^{r_i}\right] = \frac{(N-1)!}{(N+R-1)!} \prod_{i=1}^N r_i!
\end{equation}
with $R = \sum_i r_i$. Thus, for example, $\mathbb{E}[w_i] = 1/N$ or $\mathbb{E}[w_iw_j] = (1+\delta_{ij})/[N(N+1)]$.

%%%%%%%%%%%%%%%%%%%%%%%%%%%%%%%%%%%%%%%%%%%%%%%%%%%%%%%%%%%%%%%%%%%%%%%%%%%%%%%%%%%%%%%%%%%%%%%%%%%%%%%%%%%%%%%%%%%%%%%%
\subsection{Two proofs of the lower bound of eqn~(\ref{eq:N bounds})}
%%%%%%%%%%%%%%%%%%%%%%%%%%%%%%%%%%%%%%%%%%%%%%%%%%%%%%%%%%%%%%%%%%%%%%%%%%%%%%%%%%%%%%%%%%%%%%%%%%%%%%%%%%%%%%%%%%%%%%%%

We want to prove $e^{(D-1)\epsilon^2} \lesssim N_{\mathbb{C}}(D,\epsilon)$, i.e., the potential number of approximately decoherent histories is much larger than the Hilbert space dimension as long as $\epsilon\sim D^{-\alpha}$ with $\alpha<1/2$.

Before we present two different proofs and for easier comparison with related results in the literature (which often consider a real vector space $\mathbb{R}^D$), it is useful to note
\begin{equation}\label{S:eq:N relations}
	N_\mathbb{R}(D,\epsilon) \le N_{\mathbb{C}}(D,\epsilon)
	\le N_{\mathbb{R}}(2D,\epsilon),
\end{equation}
where $N_\mathbb{R}(D,\epsilon)$ denotes the maximum number of vectors $|\psi_i\rangle\in\mathbb{R}^D$, whose mutual scalar product satisfies $|\lr{\psi_i|\psi_j}_{\mathbb{R}^D}| \le \epsilon$ for $i\neq j$. The first inequality follows from the fact that $\mathbb{R}^D\subset\mathbb{C}^D$. For the second inequality we use the embedding $f(\psi) = [\Re|\psi\rangle,\Im|\psi\rangle]^T\in\mathbb{R}^{2D}$ for $|\psi\rangle\in\mathbb{C}^D$ together with the facts that $\lr{f(\psi)|f(\phi)}_{\mathbb{R}^{2D}} = \Re\lr{\psi|\phi}_{\mathbb{C}^D}$ and $|\Re\lr{\psi|\phi}_{\mathbb{C}^D}| \le |\lr{\psi|\phi}_{\mathbb{C}^D}|$.

\begin{center}
 \emph{Geometric proof}
\end{center}

The idea of this proof is best illustrated by first assuming that we are in $\mathbb{R}^D$. The set of normalized states then equals $\bb S^{D-1}$ (the $D-1$ dimensional unit sphere) and all states $|\psi\rangle$ with fidelity $|\lr{\chi|\psi}|^2 \le \cos^2(\theta)$ define a spherical cap $\bb C_\chi(\theta)$ with polar angle $\theta$ around a vector $|\chi\rangle$. If $A(\bb M)$ denotes the area (or volume) of some manifold $\bb M$, then $A(\bb S^{D-1})/A[\bb C(\theta)]$ is the number of spherical cap areas that fit into the $D-1$ dimensional unit sphere (since this number is independent of $|\chi\rangle$, we dropped it in the notation). Then, starting with an arbitrary $|\psi_1\rangle$, we successively place vectors $|\psi_k\rangle$ at the boundary of the previous caps $\bigcup_{j<k}\bb C_{\psi_j}(\theta)$ choosing $\theta = \arccos(\epsilon)$. This ensures that all the $|\psi_j\rangle$ are at least an angle $\theta$ apart or, equivalently, that $|\lr{\psi_j|\psi_k}|\le\epsilon$ for all $j<k$. This explicit construction can be repeated $A(\bb S^{D-1})/A[\bb C(\arccos \epsilon)]$ many times until we run out of free volume, thereby providing a lower bound on $N_\mathbb{R}(D,\epsilon)$. Note that the caps themselves might get deformed or ripped apart while fitting their areas onto the unit sphere.

The same argument works in the complex case, but, owing to an overall phase redundancy in quantum mechanics, the manifold of states is the complex projective plane $\mathbb{C}\bb{P}^{D-1} = \bb{S}^{2D-1}/\bb{S}^1$ instead of a unit sphere (see, e.g., Ref.~\cite{BengtssonZyczkowskiBook2017} for further details and an explicit coordinate system to compute the areas below). The area of this manifold is
\begin{equation}
	A(\mathbb{C}\bb{P}^{D-1}) = \frac{A(\bb{S}^{2D-1})}{A(\bb{S}^1)} = \frac{\pi^{D-1}}{(D-1)!}.
\end{equation}
The area of a complex spherical cap $\bb C_\chi(\theta)$ defined by all vectors $|\psi\rangle$ satisfying $|\lr{\chi|\psi}|^2 \le \cos^2(\theta)$ for some fixed reference vector $|\chi\rangle$ is
\begin{equation}
	A[\bb C(\theta)] = \frac{\pi^{D-1}}{(D-1)!} \sin^{2D-2}(\theta)
\end{equation}
such that
\begin{equation}
	\frac{A(\mathbb{C}\bb{P}^{D-1})}{A[\bb C(\theta)]} = \frac{1}{\sin^{2D-2}(\theta)}
\end{equation}
many cap areas fit into $\mathbb{C}\bb{P}^{D-1}$. Thus, we find the lower bound
\begin{equation}\label{S:eq:N lbound 1}
 \begin{split}
  N_{\mathbb{C}}(D,\epsilon) &\ge \frac{1}{\sin^{2D-2}[\arccos(\epsilon)]} \\
  &\approx \frac{1}{(1-\epsilon^2/2)^{2D-2}} \approx e^{(D-1)\epsilon^2}.
 \end{split}
\end{equation}
In the first approximation we used $\arccos(x) = \pi/2-x+\C O(x^3)$, $\sin(\pi/2-x) = \cos x$ and $\cos x = 1 - x^2/2 + \C O(x^4)$. In the second approximation we used $(1+x)^n \approx e^{nx}$, which works well for $x\approx0$.

\begin{center}
 \emph{Probabilistic proof (communicated to us by Marco Fanizza)}
\end{center}

The idea of this proof is to ask: When is the probability to draw $N$ Haar random vectors, whose pairwise fidelity is smaller than $\epsilon$, strictly greater than zero?

We start from
\begin{equation}
 \begin{split}
  \mbox{Pr}&\left[|\lr{\psi_i|\psi_j}|^2\le\epsilon^2 ~\forall i\neq j\right] \\
  &= 1 - \mbox{Pr}\left[\exists i\neq j \text{ s.t. } |\lr{\psi_i|\psi_j}|^2>\epsilon^2\right]
 \end{split}
\end{equation}
and use the union bound
\begin{equation}
 \begin{split}
  \mbox{Pr}&\left[\exists i\neq j \text{ s.t. } |\lr{\psi_i|\psi_j}|^2>\epsilon^2\right] \\
	&\le \sum_{i>j} \mbox{Pr}\left[|\lr{\psi_i|\psi_j}|^2>\epsilon^2\right] \\
	&= \frac{N(N-1)}{2} \mbox{Pr}\left[F^2 > \epsilon^2\right].
 \end{split}
\end{equation}
where $F = |\lr{\psi_i|\psi_j}|$ is the fidelity. From eqn~(\ref{S:eq:fidelity pd}) we thus find
\begin{equation}
	\mbox{Pr}\left[|\lr{\psi_i|\psi_j}|^2\le\epsilon^2 ~\forall i\neq j\right]
	\ge 1 - \frac{N(N-1)}{2} (1-\epsilon^2)^{D-1},
\end{equation}
which we want to be positive. This yields after rearrangement $N(N-1) \ge 2 (1-\epsilon^2)^{-(D-1)}$. Since $N$ will be large, we do not loose much by using $N^2>N(N-1)$, which finally gives the lower bound
\begin{equation}\label{S:eq:N lbound Marco}
	N_{\mathbb{C}}(D,\epsilon) \ge \sqrt{2}e^{(D-1)\epsilon^2/2}.
\end{equation}
This bound is looser than eqn~(\ref{S:eq:N lbound 1}) for $\epsilon^2 \ge \ln(2)/(D-1)$, but it does not involve any approximations and the scaling is the same.

%%%%%%%%%%%%%%%%%%%%%%%%%%%%%%%%%%%%%%%%%%%%%%%%%%%%%%%%%%%%%%%%%%%%%%%%%%%%%%%%%%%%%%%%%%%%%%%%%%%%%%%%%%%%%%%%%%%%%%%%
\subsection{Proof and illustration of the bound for $|\overline P_S-\overline Q_S|$ [eqn~(\ref{eq:QSD SLP bound})]}\label{S:sec:bound PS}
%%%%%%%%%%%%%%%%%%%%%%%%%%%%%%%%%%%%%%%%%%%%%%%%%%%%%%%%%%%%%%%%%%%%%%%%%%%%%%%%%%%%%%%%%%%%%%%%%%%%%%%%%%%%%%%%%%%%%%%%

We first explicitly construct the unitary $V$ rotating $|\Phi\rangle$ to $|\Psi\rangle$, assuming we have chosen the global phase of $|\Phi\rangle$ such that the fidelity $F=\lr{\Phi|\Psi}\in\mathbb{R}_+$ is real-valued and positive. Applying a Gram-Schmidt process to $\C K \equiv \mbox{span}\{|\Phi\rangle,|\Psi\rangle\}$ gives the orthonormal basis vectors $|u_1\rangle = |\Phi\rangle$ and $|u_2\rangle = (|\Psi\rangle - F|\Phi\rangle)/\sqrt{1-F^2}$. In $\C K$ we have
\begin{equation}
 V_{\C K} = \left(\begin{array}{cc}
                            F	&-\sqrt{1-F^2} \\
                            \sqrt{1-F^2} & F \\
                           \end{array}\right).
\end{equation}
The unitary in the full space $\C H = \C K\oplus\C K_\perp$ is then $V = V_{\C K}\oplus I_{\C K_\perp}$.

Next, using $|r_j\rangle = V|s_j\rangle$ and defining $Y\equiv V-I_{\C H} = (V_{\C K} - I_{\C K})\oplus0_{\C K_\perp}$, we write $|\lr{\psi_j|r_j}|^2 - |\lr{\psi_j|s_j}|^2 = 2\Re[\lr{\psi_j|Y|s_j}\lr{s_j|\psi_j}] + |\lr{\psi_j|Y|s_j}|^2$. Averaging over $q_j$ and applying first the triangle and then the Cauchy-Schwarz inequality, yields
\begin{widetext}
\begin{equation}\label{S:eq:help QSD SLP bound}
 \left|\overline P_S-\overline Q_S\right|
 \le \left|2\sum_jq_j\Re[\lr{\psi_j|Y|r_j}\lr{s_j|\psi_j}]\right| + \sum_j q_j|\lr{\psi_j|Y|s_j}|^2
 \le 2\sqrt{\overline P_S B} + B.
\end{equation}
Here, we introduced the term $B\equiv\sum_j q_j|\lr{\psi_j|Y|s_j}|^2$. We now note that $Y$ acts non-trivially only on $\C K$ and has squared operator norm $\|Y\|^2 = 2(1-F)$. Using the Cauchy-Schwarz inequality again, we find two possible bounds
\begin{equation}
 |\lr{\psi_j|Y|s_j}|^2 =
 \left\{\begin{array}{ccccc}
  |\lr{s_j|Y^\dagger\Pi|\psi_j}|^2 &\le& \lr{s_j|Y^\dagger Y|s_j} \lr{\psi_j|\Pi|\psi_j} &\le& \|Y\|^2\lr{\psi_j|\Pi|\psi_j} \\
  |\lr{\psi_j|Y\Pi|s_j}|^2 &\le& \lr{\psi_j|Y Y^\dagger|\psi_j} \lr{s_j|\Pi|s_j} &\le& \|Y\|^2\lr{s_j|\Pi|s_j} \\
 \end{array}\right\},
\end{equation}
\end{widetext}
where $\Pi$ is the projector onto $\C K$. We thus find $B\le 2(1-F)T$ with either $T = \sum_j q_j\lr{\psi_j|\Pi|\psi_j}$ or $T = \sum_j q_j\lr{s_j|\Pi|s_j}$. In the main text we only consider the second option as it is easier to handle owing to the exact orthogonality of the $|s_j\rangle$.

\begin{figure}[t]
	\centering\includegraphics[width=0.48\textwidth,clip=true]{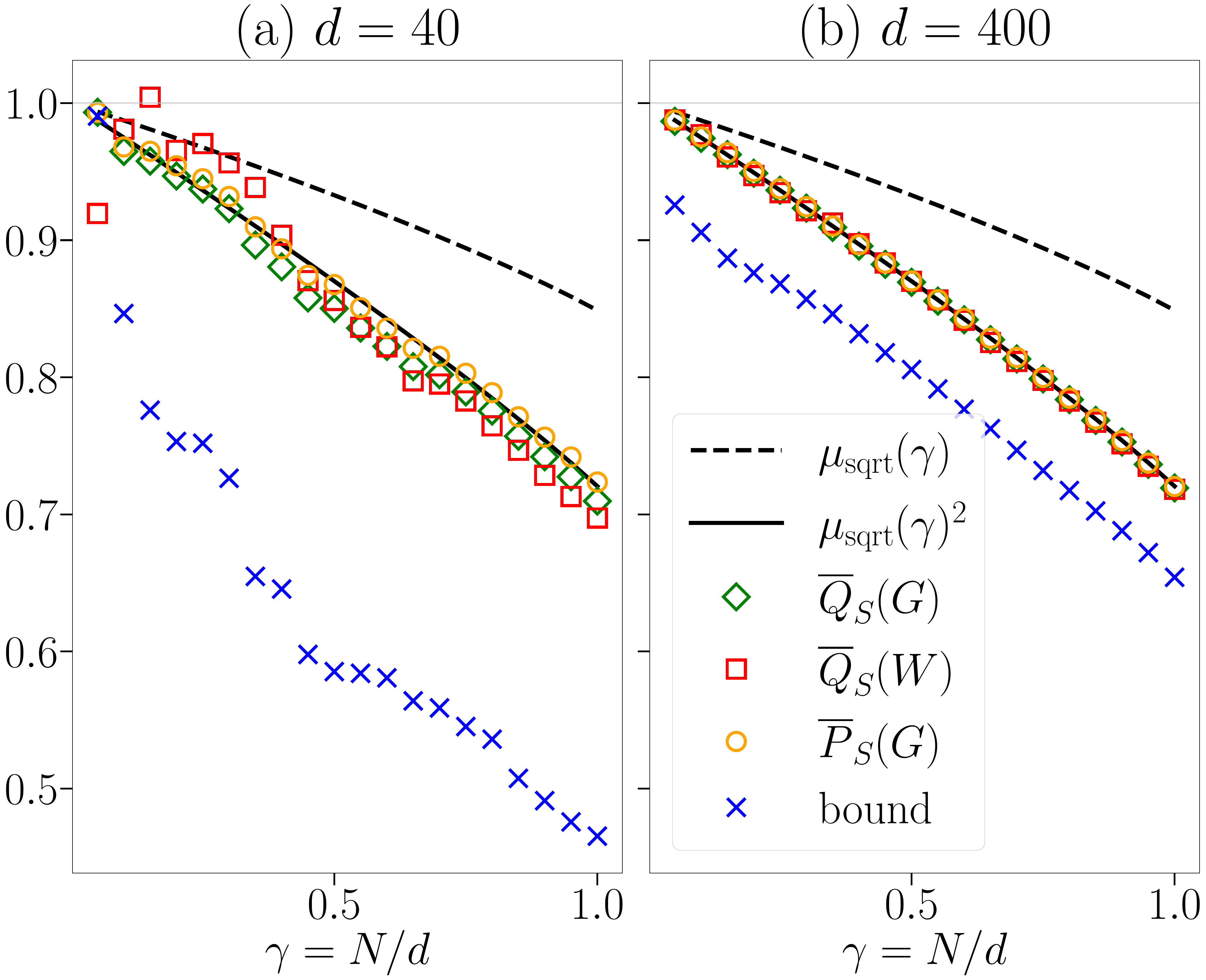}
	\caption{Plot of various central quantities as a function of the ratio $\gamma = N/d$. We evaluate numerically for $d=40$ (a) and $d=400$ (b) the average success probability of QSD (yellow circles) and SLP (green diamonds) for a single random realization of the Gram matrix $G$. The red squares show what happens to $\overline Q_S$ if one replaces $G$ by $W$. The blue crosses show the bound of eqn~(\ref{eq:QSD SLP bound}). Furthermore, we show the asymptotic analytical result for $\overline P^\text{sqrt}_S$ [eqn~(\ref{eq:p success result}), black solid line] and the bound $\mu_\text{sqrt}(\gamma)\ge\overline P_S^\text{opt}$ (black dashed line).}
	\label{fig P final}
\end{figure}

To briefly illuminate the difference between QSD and SLP, we evaluate all quantities numerically in Fig.~\ref{fig P final} for the model of Haar random history states $|\psi_j\rangle$ in $\mathbb{C}^d$ with $d=D^{2\alpha}$ and with uniform weights $q_j=1/N$ (as introduced in Sec.~\ref{sec:sub-Haar}). The plot reveals that the bound of eqn~(\ref{eq:QSD SLP bound}) is rather loose. This is expected because eqn~(\ref{eq:QSD SLP bound}) is general and even works for fine-tuned history decompositions: adding some realistic (pseudo)randomness likely allows to substantially tighten the bound. Moreover, Fig.~\ref{fig P final} also confirms once again that for large $d$ it is justified to replace $G$ by $W$ and discrepancies from the asymptotic behaviour [eqn~(\ref{eq:p success result})] become invisible to the naked eye already for $d=400$. For smaller $d$, of course, this does need to be the case and replacing $G$ by $W$ can cause spurious results such as $\overline Q_S$ seemingly not obeying the bound of eqn~(\ref{eq:QSD SLP bound}) or probabilities greater than one.

%%%%%%%%%%%%%%%%%%%%%%%%%%%%%%%%%%%%%%%%%%%%%%%%%%%%%%%%%%%%%%%%%%%%%%%%%%%%%%%%%%%%%%%%%%%%%%%%%%%%%%%%%%%%%%%%%%%%%%%%
\subsection{Counterexample: Large average success probability does not imply large fidelity}\label{s:sec:counterexample}
%%%%%%%%%%%%%%%%%%%%%%%%%%%%%%%%%%%%%%%%%%%%%%%%%%%%%%%%%%%%%%%%%%%%%%%%%%%%%%%%%%%%%%%%%%%%%%%%%%%%%%%%%%%%%%%%%%%%%%%%

We show that for $N\ge4$ one can construct history states with an average success probability scaling as $\overline P_S \sim 1-1/N$ while the fidelity is always zero. We consider
\begin{equation}
 |\Psi\rangle = \frac{1}{\sqrt{N}}\sum_{m=1}^N |\psi_m\rangle, ~~~ |\Phi\rangle = \frac{1}{\sqrt{N}}\sum_{m=1}^N |j\rangle
\end{equation}
with $\{|j\rangle\}$ some ONB. We decompose the history states as
\begin{equation}
 |\psi_m\rangle = \sum_j \left(a_j^{(m)}+ib_j^{(m)}\right)|j\rangle,
\end{equation}
where $\bs a^{(m)}, \bs b^{(m)} \in\mathbb{R}^N$ are $2N$ real-valued vectors for $m=1,2,\dots,N$. In the following, $\bs v\cdot\bs w$ and $\|\bs v\|$ denotes the Euclidean scalar product and norm in $\mathbb{R}^N$. Let $N$ be even. The vectors $\bs a^{(m)}$ are chosen as
\begin{align}
 \bs a^{(m)} &= \bs u + \hat{\bs a}^{(m)} - \lambda_m\bs u, \\
 \bs u &= \frac{1}{N} (\underbrace{+1,\dots,+1}_{N/2 \text{ many}},\underbrace{-1,\dots,-1}_{N/2 \text{ many}}), \\
 \hat{a}_j^{(m)} &= \alpha\delta_{jm} - \frac{\alpha}{N-1}(1-\delta_{jm}), \\
 \lambda_m &= \frac{\bs u\cdot \hat{\bs a}^{(m)}}{\|\bs u\|^2}
\end{align}
with $\alpha\in(0,1)$ to be determined later. For later use, we take note of the following facts:
\begin{align}
 & \text{$\bs u$ is orthogonal to $\bs a^{(m)} - \bs u = \hat{\bs a}^{(m)} - \lambda_m\bs u$}, \\
 & \|\bs u\|^2 = 1/N, \\
 & \bs u\cdot \hat{\bs a}^{(m)} = +\frac{\alpha}{N}\left(1+\frac{1}{N-1}\right) \text{ if $m\le N/2$}, \\
 & \bs u\cdot \hat{\bs a}^{(m)} = -\frac{\alpha}{N}\left(1+\frac{1}{N-1}\right) \text{ if $m> N/2$}, \\
 & \|\bs a^{(m)}\|^2 = \alpha^2 + \frac{1}{N} + \C O(\alpha^2N^{-1}), \label{s:eq:a-norm} \\
 & \sum_j a^{(m)}_j = 0.
\end{align}
Furthermore, the vectors $\bs b^{(m)}$ are chosen to satisfy
\begin{equation}\label{s:eq:b-conditions}
 b^{(m)}_m = 0, ~~~ \sum_m \bs b^{(m)} = \bs 0, ~~~ \|\bs a^{(m)}\|^2 + \|\bs b^{(m)}\|^2 = 1.
\end{equation}
The last equation can be satisfied because of eqn~(\ref{s:eq:a-norm}) for a suitable choice of $\alpha<1$. Thus, eqn~(\ref{s:eq:b-conditions}) are $3N$ equations for $N^2$ free real parameters in $\{\bs b^{(m)}\}$ such that the construction will work whenever $N\ge3$. We now check that those choices satisfy all the necessary properties.

\bb{Normalization.} Follows immediately from the last condition of eqn~(\ref{s:eq:b-conditions}):
\begin{equation}
 \lr{\psi_m|\psi_m} = \|\bs a^{(m)}\|^2 + \|\bs b^{(m)}\|^2 = 1.
\end{equation}

\bb{Phase-alignment.} Since $b_m^{(m)} = 0$, we find
\begin{equation}\label{s:eq:phase-align}
 \lr{m|\psi_m} = a_m^{(m)} = \alpha + \C O(N^{-1},\alpha N^{-1}),
\end{equation}
which is positive for a large enough $\alpha$ (for instance, $\alpha>2/N$).

\bb{Large $\overline P_S$.} Since $b^{(m)}_m = 0$, we find for the average success probability
\begin{equation}
 \overline P_S = \frac{1}{N}\sum_m \left(a_m^{(m)}\right)^2 = \alpha^2 + \C O(N^{-2},\alpha^2 N^{-2}),
\end{equation}
where we used eqn~(\ref{s:eq:phase-align}). Now, we can choose $\alpha$ as large as possible as long as $\|\bs a^{(m)}\|^2<1$ is satisfied, see eqn~(\ref{s:eq:a-norm}). Certainly, for $N\rightarrow\infty$ we can let $\alpha\rightarrow1$ such that $\overline P_S\rightarrow1$.

\bb{Zero fidelity.} The scalar product of $|\Psi\rangle$ and $|\Phi\rangle$ is
\begin{equation}
 \lr{\Phi|\Psi} = \frac{1}{N}\sum_{j,m} \lr{j|\psi_m} = \frac{1}{N} \sum_{j,m} \left(a_j^{(m)} + b_j^{(m)}\right).
\end{equation}
But our construction ensures that $\sum_j a_j^{(m)} = 0$ and $\sum_m b_j^{(m)} = 0$, consequently $F=|\lr{\Phi|\Psi}|=0$.

We remark that the above construction would also work if we set $\bs b^{(m)} = \bs 0$ for all $m$. Then, we would just need to ensure that $\|\bs a^{(m)}\| = 1$ for all $m$, which would explicitly determine $\alpha$ as a function of $N$. The result would be the same since we would find $\alpha\rightarrow1$ for $N\rightarrow\infty$. By having non-zero $\bs b^{(m)}$ in the above construction we can hide part of the norm in the $\bs b^{(m)}$ component, so we have a greater freedom to choose $\alpha$.

%%%%%%%%%%%%%%%%%%%%%%%%%%%%%%%%%%%%%%%%%%%%%%%%%%%%%%%%%%%%%%%%%%%%%%%%%%%%%%%%%%%%%%%%%%%%%%%%%%%%%%%%%%%%%%%%%%%%%%%%
\subsection{Random phase argument}\label{S:sec:random phase}
%%%%%%%%%%%%%%%%%%%%%%%%%%%%%%%%%%%%%%%%%%%%%%%%%%%%%%%%%%%%%%%%%%%%%%%%%%%%%%%%%%%%%%%%%%%%%%%%%%%%%%%%%%%%%%%%%%%%%%%%

The counterexample of Sec.~\ref{s:sec:counterexample} seems contrived as each history vector is distored by the \emph{same} zero-sum vector $\bs u$. Their overlap would not randomly fluctuate as our ansatz in eqn~(\ref{eq:DF scaling form}) asserts, and we thus believe that a suitable randomization argument implies high fidelity for a high average success probability.

To this end, we start from a given history decomposition $|\Psi\rangle = \sum_j \sqrt{q_j}|\psi_j\rangle$ and modify it to
\begin{equation}
 |\Psi_{\bs\alpha}\rangle = \sum_j \sqrt{q_j}e^{i\alpha_j}|\psi_j\rangle,
\end{equation}
where $\bs\alpha = (\alpha_1,\dots,\alpha_N)$ are random iid distributed phases with $\mathbb{E}[e^{i\alpha_j}] = 0$ for all $j$. Their overlap with $|\Phi_{\bs\alpha}\rangle = \sum_j \sqrt{q_j}e^{i\alpha_j}|s_j\rangle$, where we modified the states $|s_j\rangle$ obtained by solving the QSD task by the same phases to ensure phase alignment $\lr{s_j|\psi_k}\ge 0$, is
\begin{equation}
 \lr{\Phi_{\bs\alpha}|\Psi_{\bs\alpha}} = \sum_{i,j} \sqrt{q_iq_j}e^{i(\alpha_i-\alpha_j)}\lr{s_j|\psi_i}.
\end{equation}
The expectation value of the fidelity is
\begin{equation}
 \mathbb{E}[\lr{\Phi_{\bs\alpha}|\Psi_{\bs\alpha}}] = \sum_j q_j \lr{s_j|\psi_j} \ge \overline P_S
\end{equation}
where we used $\lr{s_j|\psi_j} \ge \lr{s_j|\psi_j}^2$. Note that the expectation value is real-valued.

The squared modulus of the fidelity is
\begin{equation}
 \begin{split}
  &|\lr{\Phi_{\bs\alpha}|\Psi_{\bs\alpha}}|^2 \\
  &= \sum_{ijkl} \sqrt{q_iq_jq_kq_l} e^{i(\alpha_i-\alpha_j-\alpha_k+\alpha_l)}\lr{s_j|\psi_i}\lr{\psi_k|s_l}.
 \end{split}
\end{equation}
Using
\begin{equation}
 \begin{split}
  \mathbb{E}[e^{i(\alpha_i-\alpha_j-\alpha_k+\alpha_l)}] &= \delta_{ij}\delta_{kl}(1-\delta_{jk})(1-\delta_{il}) \\
  &+ \delta_{ik}\delta_{jl}(1-\delta_{ij})(1-\delta_{kl}) \\
  &+ \delta_{ij}\delta_{jk}\delta_{kl},
 \end{split}
\end{equation}
one finds for the variance
\begin{equation}
 \mathbb{E}\left[|\lr{\Phi_{\bs\alpha}|\Psi_{\bs\alpha}}|^2\right] - \mathbb{E}[\lr{\Phi_{\bs\alpha}|\Psi_{\bs\alpha}}]^2 = \sum_{i\neq j} q_i q_j |\lr{s_j|\psi_i}|^2.
\end{equation}
Next, write $|\psi_i\rangle = \sqrt{p_i}|s_i\rangle + \sqrt{1-p_i}|\tilde s_i\rangle$ with $p_i = |\lr{s_i|\psi_i}|^2$ the individual success probability and $|\tilde s_i\rangle$ orthogonal to $|s_i\rangle$. We can then upper bound
\begin{equation}
 \begin{split}
  \sum_{i\neq j} q_i q_j |\lr{s_j|\psi_i}|^2 &\le \sum_{i\neq j} q_i q_j (1-p_i) \\
  & \le 1 - \overline P_S - \sum_i q_i^2(1-p_i) \\
  &\le 1 - \overline P_S.
 \end{split}
\end{equation}
Here, we used that $\overline P_S = \sum_j q_j p_j$.

Thus, we found that the expectation value of the fidelity is lower bounded by $\overline P_S$ and the variance is upper bounded by $1-\overline P_S$. Unfortunately, the problem of the construction above is that $|\Psi_{\bs\alpha}\rangle$ is no longer normalized even though $|\Psi\rangle$ is, so $|\Psi_{\bs\alpha}\rangle$ is no longer a valid history decomposition. It is, however, at least possible to show that also the norm concentrates around one. To this end, we first note that
\begin{equation}
 \mathbb{E}[\lr{\Psi_{\bs\alpha}|\Psi_{\bs\alpha}}] = 1.
\end{equation}
The variance equals
\begin{equation}
 \mathbb{E}\left[|\lr{\Psi_{\bs\alpha}|\Psi_{\bs\alpha}}|^2\right] - \mathbb{E}[\lr{\Psi_{\bs\alpha}|\Psi_{\bs\alpha}}]^2 = \sum_{j\neq k} q_j q_k|\lr{\psi_j|\psi_k}|^2.
\end{equation}
We now use again the decomposition $|\psi_i\rangle = \sqrt{p_i}|s_i\rangle + \sqrt{1-p_i}|\tilde s_i\rangle$, which (after applying triangle and Cauchy-Schwarz inequalities) allows us to bound
\begin{equation}
 |\lr{\psi_j|\psi_k}|^2 \le 1-p_jp_k.
\end{equation}
This allows us to express the variance as
\begin{equation}
 \begin{split}
  &\mathbb{E}\left[|\lr{\Psi_{\bs\alpha}|\Psi_{\bs\alpha}}|^2\right] - \mathbb{E}[\lr{\Psi_{\bs\alpha}|\Psi_{\bs\alpha}}]^2 \\
  &\le \sum_{jk}q_jq_k(1-p_jp_k) - \sum_j q_j^2(1-p_j^2) \\
  &\le 1 - \overline P_S^2.
 \end{split}
\end{equation}
Thus, for large $\overline P_S$ we find that $\|\Psi_{\bs\alpha}\|$ concentrates around one.

Unfortunately, we were not able to find a rigorous lower bound on the fidelity for a distribution of $\bs\alpha$ such that $\|\Psi_{\bs\alpha}\| = 1$ is satisfied \emph{exactly}. We believe, however, that the results above clearly point in the direction that $\overline P_S\lesssim1$ implies $F\gtrsim1$ for a suitable random choice of histories, as required by the ansatz~(\ref{eq:DF scaling form}).

%%%%%%%%%%%%%%%%%%%%%%%%%%%%%%%%%%%%%%%%%%%%%%%%%%%%%%%%%%%%%%%%%%%%%%%%%%%%%%%%%%%%%%%%%%%%%%%%%%%%%%%%%%%%%%%%%%%%%%%%
\subsection{Derivation of the general bound in terms of $\det G$ [eqn~(\ref{eq:result Andreas})]}
%%%%%%%%%%%%%%%%%%%%%%%%%%%%%%%%%%%%%%%%%%%%%%%%%%%%%%%%%%%%%%%%%%%%%%%%%%%%%%%%%%%%%%%%%%%%%%%%%%%%%%%%%%%%%%%%%%%%%%%%

\begin{figure}[t]
	\centering\includegraphics[width=0.48\textwidth,clip=true]{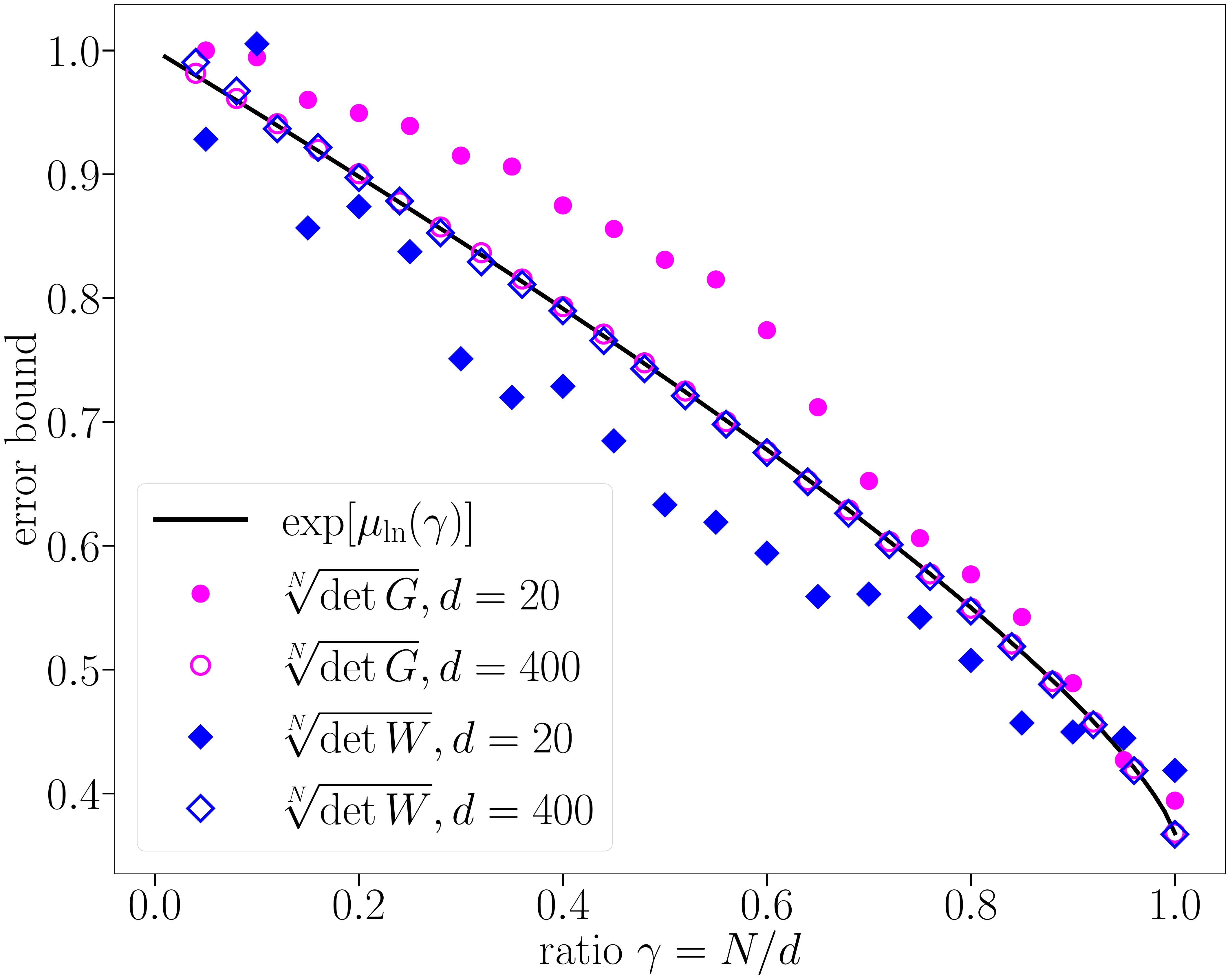}
	\caption{We plot $\sqrt[N]{\det G}$ for $d\in\{20,400\}$ (filled, empty magenta circles) and $\sqrt[N]{\det W}$ for $d\in\{20,400\}$ (filled, empty blue diamonds) for a single realization of the Haar random or Gaussian vectors as a function of $\gamma=N/d$. The solid black line displays the asymptotic, analytical result of eqn~(\ref{S:eq:analytical Andreas}). Discrepancies between $G$, $W$ and the asymptotic result are almost invisible for $d=400$. }
	\label{fig Andreas bound}
\end{figure}

Before we prove the bound~(\ref{eq:result Andreas}) in full generality, we first evaluate and numerically illustrate it for the model of Haar random history states $|\psi_j\rangle$ in $\mathbb{C}^d$ with $d=D^{2\alpha}$ and with uniform weights $q_j=1/N$ (as introduced in Sec.~\ref{sec:sub-Haar}). A key idea is to approximate $|\psi_j\rangle\approx|\psi'_j\rangle$ where $|\psi'_j\rangle$ are almost normalized Gaussian vectors. Thus, we approximate $G\approx W$, where $W$ is a member of the complex Wishart ensemble. Then, after applying the CLT (Lemma~\ref{lemma CLT}) to $\ln(\sqrt[N]{\det W}) = N^{-1}\sum_j\ln\lambda_j$, we find the analytical expression
\begin{equation}\label{S:eq:analytical Andreas}
 \sqrt[N]{\det G} \simeq \exp[\mu_\text{ln}(\gamma)],
\end{equation}
where---in unison with the notation of eqn~(\ref{S:eq:mu def})---$\mu_\text{ln}(\gamma)$ is the expectation value of $\ln(\lambda)$ with respect to the Marchenko-Pastur distribution as a function of the ratio $\gamma=N/d$. The function $\exp[\mu_\text{ln}(\gamma)]$ decays monotonically and almost linearly from $\exp[\mu_\text{ln}(0)] = 1$ to $\exp[\mu_\text{ln}(1)] = e^{-1} \approx 0.37$. Thus, for $\gamma\rightarrow1$ the optimal worst-case success probability for unambiguous detection becomes small and the inconclusive 0-result dominates the statistics. A numerical evaluation shown in Fig.~\ref{fig Andreas bound} not only confirms eqn~(\ref{S:eq:analytical Andreas}), but also shows that the replacement $G\approx W$ is justified.

We now turn to the general proof for which we need the notion of a positive operator valued measure (POVM): a collection $(M_x)_x$ is called a POVM if and only if $M_x\ge 0$ for all $x$ and $\sum_x M_x = I$. Furthermore, $S(A)_\rho$ is the von Neumann entropy of some (sub)system $A$ in state $\rho^A$, $I(A:B)_\rho$ the mutual information between subsystems $A$ and $B$ in a joint state $\rho^{AB}$, and $H(p)$ is the Shannon entropy of some probability distribution $(p_x)_x$.

\bb{Proof of the upper bound.} We start with a general lemma relating near-maximum Holevo information of an ensemble $\{p_x,\rho_x\}$ of arbitrary (pure or mixed) states. To state it, it is convenient to introduce the bipartite cq-state $\rho^{XA} = \sum_x p_x \proj{x}^X \otimes \rho_x^A$:

\begin{lemma}
	\label{lemma:recoverability}
	There exists a POVM $(M_x)_x$ with $P_{\text{S}}(x) = \mbox{\normalfont tr}\{\rho_x M_x\}$ and $\overline{P}_{\text{S}} = \sum_x p_x P_{\text{S}}(x)$, such that
	\begin{equation}
		\begin{split}
			S(X|A)_\rho &\equiv S(XA)_\rho - S(A)_\rho \\
			&= H(p)-I(X:A)_\rho \\
			&\geq -\sum_x p_x \log P_{\text{S}}(x) \geq -\log \overline{P}_{\text{S}}.
		\end{split}
	\end{equation}
\end{lemma}
\begin{proof}
	We use the recovery map result \cite[Cor.~8]{SutterTomamichelHarrowIEEE2016} for the state $\gamma^{XX'} = \sum_x p_x \proj{x}^{X} \otimes \proj{x}^{X'}$ and the cq-channel $\C N:X'\rightarrow A$ mapping $\proj{x}$ to $\rho_x$. Then, we have $H(p)-I(X:A)_\rho = I(X:X')_\gamma-I(X:A)_\rho$ and the cited result lower-bounds this by $D_{\mathbb{M}}(\gamma^{XX'}|(I_X\otimes\C R)\rho^{XA})$ for a suitable recovery map $\C R:A\rightarrow X'$, where $D_{\mathbb{M}}$ is the measured relative entropy: $D_\mathbb{M}(\rho|\sigma) = \sup_{M} D(\Phi_M\rho|\Phi_M\sigma)$, where $M=(M_x)_x$ is a POVM and $\Phi_M$ the qc-channel implementing $M$, i.e., $\Phi_M\rho = \sum_x \mbox{tr}\{M_x\rho\}|x\rl x|$. Thus, we have
	\begin{equation}
		\begin{split}
			H(p)-I(X:A)_\rho &\ge \sup_{M} D\left.\left[\Phi_M\gamma^{XX'}\right|\Phi_M(I_X\otimes\C R)\rho^{XA}\right] \\
			& \ge D[\gamma^{XX'}|\widetilde{\gamma}^{XX'}],
		\end{split}
	\end{equation}
	where for the second inequality we simply chose a specific qc-channel, namely the one which dephases the state in the computational basis of $X'$. This leaves $\gamma^{XX'}$ invariant, but turns $(I_X\otimes\C R)\rho^{XA}$ into the following classical state
	\begin{equation}
		\widetilde{\gamma}^{XX'} = \sum_{xx'} p_x \proj{x}^X \otimes \proj{x'}^{X'} \mbox{tr}\{\rho_x R_{x'}\}
	\end{equation}
	for a suitable POVM $(R_x)_x$. Writing $P_{\text{S}}(x) = \mbox{tr}\{R_x\rho_x\}$ for the individual success probability, the relative entropy $D[\gamma^{XX'}|\widetilde{\gamma}^{XX'}]$ is now easily evaluated to be
	\begin{equation}
		D(\gamma\|\widetilde{\gamma}) = -\sum_x p_x \log P_{\text{S}}(x),
	\end{equation}
	and the latter is $\geq -\log \overline{P}_{\text{S}}$ by the convexity of the negative logarithm.
\end{proof}

\begin{rmrk}
	Note that this is a kind of reverse Fano inequality, because it says that, if the equivocation (conditional entropy of the input given the output of the channel) is small, then a good decoder exists. Fano's inequality (equivalently the
	Audenaert-Fannes-Csisz\'ar inequality) instead states that if $\overline{P}_{\text{S}} \geq 1-\epsilon$, then, using the data processing inequality for the conditional entropy (aka strong subadditivity), $S(X|A) \leq H(X|\widehat{X}) \leq \epsilon\log(n-1)+H(\epsilon,1-\epsilon)$.
\end{rmrk}

Next, we need another auxiliary result.

\begin{lemma}
	\label{lemma:D-maxmixed}
	For an arbitrary state $\rho$ and the maximally mixed state $\tau = I/d$
	in dimension $d$,
	\begin{equation}
		D(\rho\|\tau) \leq D(\tau\|\rho),
	\end{equation}
	with equality iff $\rho=\tau$.
\end{lemma}
\begin{proof}
	This follows by finding the critical points of $D(\tau\|\rho) - D(\rho\|\tau)$
	as a function of $\rho$. First of all, this function depends only on the
	spectrum of $\rho$, $\operatorname{spec}\rho = (q_1,\ldots,q_d)$, and it is
	$+\infty$ on the boundary of the state space. Thus, all, potential minima
	lie in the interior of the probability simplex.

	To do that, we differentiate it with respect to the $q_x$ and use Lagrange
	multipliers to represent the normalisation condition, which yields the only
	candidate $p_x = \frac1d$ for all $x$, at which point the difference is $0$.
\end{proof}

Now, using the fact that we have the uniform distribution over pure states, we have in dimension $d=N$
\begin{equation}
	\begin{split}
		-\log \overline{P}_{\text{S}} &\leq \log N - S(\rho^A) = D(\rho^A\|\tau^A) \\
		&\leq D(\tau^A\|\rho^A) = -\log N - \frac1N \mbox{tr}\log\rho^A.
	\end{split}
\end{equation}
Here, the inequalities follow from Lemma~\ref{lemma:recoverability} (which lower bounds $\overline P_S$ in terms of one concrete measurement) and~\ref{lemma:D-maxmixed}, respectively, and the steps in between are just elementary algebra. Finally, we notice that $N\rho^A = \sum_x |\psi_x\rl\psi_x|$ is isospectral with the Gram matrix $G$, except for potentially different multiplicity of the eigenvalue $0$. Thus,
\begin{equation}
	\begin{split}
		-\log \overline{P}_{\text{S}} &\leq -\frac1N \mbox{tr}\log(N\rho^A) \\
		&= -\frac1N \mbox{tr}\log G = -\frac1N \log\det G,
	\end{split}
\end{equation}
and we are done.
\qed

\medskip
\noindent\bb{Proof of the lower bound.}
We first note that there cannot be a universal lower bound on $\det G$ in terms of $\overline{P}_{\text{S}}$. To see this, consider the example that $|\psi_x\rangle$ are mutually orthogonal states for $x=1,\ldots,N-1$, and an arbitrary state for $x=N$. Irrespective of the last state, $\overline{P}_{\text{S}} \geq 1-1/N$, whereas $\det G$ can vary arbitrarily between $0$ and $1$, the extreme cases attained when $|\psi_N\rangle$ is in the subspace $\mbox{span}\{|\psi_x\rangle:x=1,\ldots,N-1\}$ or orthogonal to it, respectively.

We now turn to \emph{unambiguous state discrimination} as a variant of QSD in which one allows for a set of generalized measurements (or POVMs) $\{M_i\}_{i=0}^N$ that unambiguously detect the right state $|\psi_i\rangle$, i.e., the probability to confuse it with another state $j$ is zero: $\lr{\psi_j|M_i|\psi_j}\sim\delta_{ij}$. However, this is only possible at the expense of introducing one additional inconclusive measurement $M_0$ that does not reveal anything. Obviously, an important goal of unambiguous detection is to upper bound the probability $P_0$ of obtaining the inconclusive 0-result. Therefore, one defines the optimal worst-case success probability of unambiguous state discrimnation as
\begin{equation}\label{S:eq:PUN}
	P^\text{opt}_\text{UN} = \max \pi \text{ s.t. } 0\le\pi\le\lr{\psi_i|M_i|\psi_i}~\forall i\ge1,
\end{equation}
which implies the bound $P_0\le1-P_\text{UN}$.

That we can find a a lower bound for unambiguous state discrimination comes from the fact that it is known to be sensitive to the linear dependence of the state vectors. Namely, every state $|\psi_x\rangle$ that is a linear combination of the other $|\psi_y\rangle$, $y\neq x$, cannot ever be detected, i.e.~$\mbox{tr}\{\psi_x M_x\} = 0$ in any admissible unambiguous detection protocol~\cite{CheflesPLA1998}. Indeed, to prove the lower bound , we can quote the main result of \cite{HoroshkoEskandariKilinPLA2019}, which states that
\begin{equation}
	P_{\text{UN}} = \min\operatorname{spec}(G) \leq \sqrt[N]{\det G},
\end{equation}
and we are done.
\qed

%%%%%%%%%%%%%%%%%%%%%%%%%%%%%%%%%%%%%%%%%%%%%%%%%%%%%%%%%%%%%%%%%%%%%%%%%%%%%%%%%%%%%%%%%%%%%%%%%%%%%%%%%%%%%%%%%%%%%%%%
\subsection{Derivation of $\overline P_S$ for the sqrt measurement}\label{S:sec:av succ prob sqrt meas}
%%%%%%%%%%%%%%%%%%%%%%%%%%%%%%%%%%%%%%%%%%%%%%%%%%%%%%%%%%%%%%%%%%%%%%%%%%%%%%%%%%%%%%%%%%%%%%%%%%%%%%%%%%%%%%%%%%%%%%%%

We want to compute the individual success probability $|\lr{s_j|\psi_j}|^2 = |\sqrt{G}_{jj}|^2$. Our key insight is the norm concentration result of eqn~(\ref{S:eq:norm concentration}), which justifies to replace asymptotically for large $N$ the normalized Haar random states $|\psi_i\rangle$ by the almost normalized Gaussian vectors $|\psi'_i\rangle$. In turn, this allows us to replace the Gram matrix $G$ by a member $W$ of the Wishart ensemble such that we can use the results reviewed in Sec.~\ref{S:sec:preliminaries}. Thus, we want to compute $|\lr{s_j|\psi_j}|^2 \approx |\sqrt{W}_{jj}|^2$ asymptotically for large $N,d$ with the ratio $\gamma = N/d$ fixed.

Writing $W = U\mbox{diag}(\lambda_1,\dots,\lambda_N)U^\dagger$ we find
\begin{equation}
 \begin{split}
  |\sqrt{W}_{jj}|^2 &= \sum_{kl} \sqrt{\lambda_k\lambda_\ell} U_{jk} U_{jk}^* U_{j\ell} U_{j\ell}^* \\
  &= \sum_{kl} \sqrt{\lambda_k\lambda_\ell} w_k^{(j)} w_\ell^{(j)}.
 \end{split}
\end{equation}
In the last step we introduced the weight vector $w^{(j)}$ with components $w_k^{(j)} \equiv |U_{jk}|^2$, which is uniformly distributed over the $(N-1)$ probability simplex. Averaging over the probability simplex (which we indicate by $\mathbb{E}_w[\dots]$) with the help of eqn~(\ref{S:eq:Dirichlet}), we find
\begin{equation}\label{S:eq:average Wjj}
 \begin{split}
  \mathbb{E}_w\left[|\sqrt{W}_{jj}|^2\right] &= \frac{N}{N+1} \left(\frac{\sum_k \sqrt{\lambda_k}}{N}\right)^2 + \frac{1}{N+1}\frac{\sum_k\lambda_k}{N} \\
  &\longrightarrow \mu_\text{sqrt}(\gamma)^2.
 \end{split}
\end{equation}
Here, $\longrightarrow$ indicates the asymptotic behavior for $N,d\rightarrow\infty$ while keeping their ratio $\gamma = N/d$ fixed, and we further used the CLT (Lemma~\ref{lemma CLT}).

Next, we want to show that fluctuations around the average~(\ref{S:eq:average Wjj}) are small by considering the variance of $|\sqrt{W}_{jj}|^2$. We use that the variance of a function $f(x,y)$ of two independent random variables $x$ and $y$ is given by $\mathbb{V}_{xy}(f) = \mathbb{E}_x[\mathbb{V}_y(f)] + \mathbb{V}_x[\mathbb{E}_y(f)]$, hence
\begin{equation}
 \begin{split}
  \mathbb{V}&\left[|\sqrt{W}_{jj}|^2\right] =\mathbb{E}_\lambda\left[\mathbb{V}_w\left[|\sqrt{W}_{jj}|^2\right]\right] \\
  &+ \mathbb{V}_\lambda\left[\frac{N}{N+1} \left(\frac{\sum_k \sqrt{\lambda_k}}{N}\right)^2 + \frac{1}{N+1}\frac{\sum_k\lambda_k}{N}\right].
 \end{split}
\end{equation}
From the CLT we know that the variance of the second term is of order $\C O(N^{-2})$, and hence negligible. Therefore, it remains to compute the first term. To this end, we note that
\begin{equation}
 |\sqrt{W}_{jj}|^4 = \sum_{klmn} \sqrt{\lambda_k\lambda_\ell\lambda_m\lambda_n} w_k^{(j)} w_\ell^{(j)}w_m^{(j)} w_n^{(j)}.
\end{equation}
Taking again the average over the probability simplex, we find
\begin{widetext}
\begin{equation}
 \begin{split}
  \mathbb{E}_w\left[|\sqrt{W}_{jj}|^4\right]
  &= \sum_{k\neq\ell\neq m\neq n} \frac{\sqrt{\lambda_k\lambda_\ell\lambda_m\lambda_n}}{N(N+1)(N+2)(N+3)}
  + 12\sum_{k\neq\ell\neq m} \frac{\lambda_k\sqrt{\lambda_\ell\lambda_m}}{N(N+1)(N+2)(N+3)} + \dots \\
  &= \frac{N^3}{(N+1)(N+2)(N+3)} \left(\frac{\sum_k\sqrt{\lambda_k}}{N}\right)^4
  +  \frac{6N^2}{(N+1)(N+2)(N+3)} \frac{\sum_k \lambda_k}{N}\left(\frac{\sum_k \sqrt{\lambda_k}}{N}\right)^2 + \dots
 \end{split}
\end{equation}
\end{widetext}
Here, a sum running over, e.g., $k\neq\ell\neq m\neq n$ shall indicate that $k, \ell, m, n$ are \emph{all different}. Moreover, terms that we neglected contain sums over fewer indices and are therefore of subleading order $\C O(N^{-2})$. Now, subtracting $\mathbb{E}_w[|\sqrt{W}_{jj}|^2]^2$ we find to leading order
\begin{equation}
 \begin{split}
  \mathbb{V}_w & \left[|\sqrt{W}_{jj}|^2\right] = \\
  &\frac{4}{N}\left\{\frac{\sum_k \lambda_k}{N}\left(\frac{\sum_k \sqrt{\lambda_k}}{N}\right)^2 - \left(\frac{\sum_k\sqrt{\lambda_k}}{N}\right)^4\right\} + \dots
 \end{split}
\end{equation}
From the CLT we know that the terms inside the curly bracket are strongly concentrated around their $\C O(1)$ mean value, hence we find $\mathbb{V}_w[|\sqrt{W}_{jj}|^2] \sim N^{-1}$.

To conclude, for large $N$ the individual success probability $|\sqrt{W}_{jj}|^2$ fluctuates around an $\C O(1)$ mean value $\mu_\text{sqrt}(\gamma)^2$ with fluctuations of size $1/\sqrt{N}$ (measured in terms of the standard deviation). Since $\mu_\text{sqrt}(\gamma)^2$ does not depend on $j$,
\begin{equation}
 \overline{P}_S = \sum_j q_j |\sqrt{W}_{jj}|^2 \simeq \mu_\text{sqrt}(\gamma)^2.
\end{equation}
Note that the conventional CLT implies that fluctuations of $\overline{P}_S$ around its mean are further suppressed compared to the fluctuations of $|\sqrt{W}_{jj}|^2$ by an additional factor $1/\sqrt{N}$ if $q_j\approx 1/N$ is sufficiently smeared out.

%%%%%%%%%%%%%%%%%%%%%%%%%%%%%%%%%%%%%%%%%%%%%%%%%%%%%%%%%%%%%%%%%%%%%%%%%%%%%%%%%%%%%%%%%%%%%%%%%%%%%%%%%%%%%%%%%%%%%%%%
\subsection{Norm and fidelity}\label{S:sec:fidelity sqrt meas}
%%%%%%%%%%%%%%%%%%%%%%%%%%%%%%%%%%%%%%%%%%%%%%%%%%%%%%%%%%%%%%%%%%%%%%%%%%%%%%%%%%%%%%%%%%%%%%%%%%%%%%%%%%%%%%%%%%%%%%%%

By construction, a history decomposition $|\Psi\rangle = \sum_j \sqrt{q_j}|\psi_j\rangle$ is characterized by coefficients $q_j$ that are a probability distribution. If we now replace in our model the $|\psi_j\rangle$ by Haar randomly drawn states, the normalization of $|\Psi\rangle$ will be violated in general. Even worse, in our calculations below we like to replace the Haar random states $|\psi_j\rangle$ by Gaussian vectors $|\psi'_j\rangle$ that are almost normalized [eqn~(\ref{S:eq:norm concentration})] and consider $|\Psi'\rangle = \sum_j \sqrt{q_j}|\psi'_j\rangle$ instead of $|\Psi\rangle$. Luckily, as we show now, the norm of $|\Psi'\rangle$ is strongly concentrated around one, which motivates the conjecture that our histories toy model captures some aspects of real decoherence.\footnote{Indeed, a more precise conjecture would be that higher order correlations among the $|\psi_j\rangle$, which are not detectable in the first two cumulants, ensure the normalization. However, the first two cumulants seem to dominate global properties such as the average success probability, as long as higher order correlations are sufficiently weak. This idea seems quite plausible and it suggests that our Haar random histories toy model captures correctly the leading order behavior.}

To this end, we notice that we can write
\begin{equation}
	\|\Psi'\|^2 = \sum_{a=1}^d |z_a|^2
\end{equation}
with $z_a = \sum_j \sqrt{q_j} \lr{a|\psi'_j}$ where $\{|a\rangle\}$ is some ONB in $\mathbb{C}^d$. Now, recall that the coefficients $\lr{a|\psi'_j}$ are independently drawn from a Gaussian with mean zero and variance $1/2d$ for the real and imaginary parts. It follows that $z_a$, as a linear combination of iid Gaussian variables, is also Gaussian with mean zero and variance $\sigma^2_z = \sum_k q_k/2d = 1/2d$ for the real and imaginary part. Thus, we immediately get the expectation value and variance
\begin{equation}
 \mathbb{E}\big[\|\Psi'\|^2\big] = 1, ~~~ \mathbb{V}\big[\|\Psi'\|^2\big] = \frac{2}{d}.
\end{equation}
An application of Chebyshev's inequality
\begin{equation}
	\mbox{Pr}\left[\big|\|\Psi'\|-1\big|>\frac{1}{d^{1/3}}\right] \le \frac{2}{d^{1/3}}
\end{equation}
shows concentration of measure.

Next, we turn to the fidelity. Starting from $|\Phi\rangle = \sum_j \sqrt{q_j}|s_j\rangle$ with $|s_j\rangle = \sum_{k=1}^N (\sqrt{G}^{-1})_{kj}|\psi_k\rangle$, we find
\begin{equation}
 \begin{split}
  \lr{\Psi|\Phi} = \frac{\lr{\Psi'|\Phi}}{\|\Psi'\|} &= \frac{1}{\|\Psi'\|} \sum_{i,j} \sqrt{q_iq_j}\sqrt{G}_{ij} \\
  &\approx \frac{1}{\|\Psi'\|} \sum_{i,j} \sqrt{q_iq_j}\sqrt{W}_{ij},
 \end{split}
\end{equation}
where we replaced $G$ by $W$. We now separately estimate $S \equiv \sum_{i,j}\sqrt{q_iq_j}\sqrt{W}_{ij}$, even though $\|\Psi'\|$ and $S$ both depend on the same random variables $W_{jk}$. This is justified \emph{a posteriori} because we will find concentration of measure also for $S$, which implies concentration of measure for $S/\|\Psi'\|$. The following auxiliary result justifies this strategy.

\begin{lemma}
 Suppose $f(x)\simeq\mu_f$ and $g(x)\simeq\mu_g$. Then, $f(x)g(x)\simeq\mu_f\mu_g$.
\end{lemma}

\begin{proof}
 We define the symbol $\simeq$ (concentration of measure) by the following criteria for $f$:
 \begin{equation}
  \mbox{Pr}\left[\left|\frac{f(x)}{\mu_f}-1\right|>\delta\right] \le \epsilon_f(\delta),
 \end{equation}
 where $\epsilon_f(\delta)$ is a suitable concentration functions, remaining very small (i.e., $\ll 1$) even for very small $\delta\ll1$. Of course, for $\delta\rightarrow0$ we must eventually have $\epsilon_f(\delta)\rightarrow1$. Similarly, we define concentration of measure for $g$.

 We now consider
 \begin{widetext}
 \begin{equation}\label{S:eq:help concentration}
  \mbox{Pr}\left[\left|\frac{f(x)g(x)}{\mu_f\mu_g}-1\right|>\delta\right] =
  \mbox{Pr}\left[\left|\left(\frac{f(x)}{\mu_f}-1\right)\left(\frac{g(x)}{\mu_g}-1\right) + \frac{f(x)}{\mu_f}-1 + \frac{g(x)}{\mu_g}-1\right|>\delta\right].
 \end{equation}
 \end{widetext}
 Since $\mbox{Pr}[|A+B+C|>\delta] \le \mbox{Pr}[|A| > \delta/3] + \mbox{Pr}[|B| > \delta/3] + \mbox{Pr}[|C| > \delta/3]$, we have
 \begin{equation}
  \begin{split}
   \text{eqn~(\ref{S:eq:help concentration})} \le&~
 	\mbox{Pr}\left[\left|\left(\frac{f(x)}{\mu_f}-1\right)\left(\frac{g(x)}{\mu_g}-1\right)\right|>\frac{\delta}{3}\right] \\
 	&+ \epsilon_f(\delta/3) + \epsilon_g(\delta/3).
  \end{split}
 \end{equation}
 Next, we use $\mbox{Pr}[|AB|>\delta] \le \mbox{Pr}[|A|>\sqrt{\delta}] + \mbox{Pr}[|B|>\sqrt{\delta}]$ to find
 \begin{equation}
  \begin{split}
   \text{eqn~(\ref{S:eq:help concentration})} \le&~ \epsilon_f(\sqrt{\delta/3}) + \epsilon_g(\sqrt{\delta/3}) \\
   &+ \epsilon_f(\delta/3) + \epsilon_g(\delta/3) \equiv \epsilon_{fg}(\delta),
  \end{split}
 \end{equation}
 where $\epsilon_{fg}(\delta)$ is the concentration function of $f(x)g(x) \simeq \mu_f\mu_g$.
\end{proof}

To estimate $S$ we now proceed similarly to Sec.~\ref{S:sec:av succ prob sqrt meas}. First, with $W = U\mbox{diag}(\lambda_1,\dots,\lambda_N)U^\dagger$ we write
\begin{equation}
	S = \sum_{i,j} \sqrt{q_iq_j} \sum_k \sqrt{\lambda_k} U_{ik} U^*_{jk}.
\end{equation}
Since $\mathbb{E}_U[U_{ik} U^*_{jk}] = \delta_{ij}/N$ (where $\mathbb{E}_U[\dots]$ denotes the expectation value with respect to the Haar measure for $U$), we find
\begin{equation}\label{S:eq:expectation U}
	\mathbb{E}_U[S] = \frac{1}{\|\Psi'\|} \frac{\sum_k\sqrt{\lambda_k}}{N} \longrightarrow \frac{\mu_\text{sqrt}(\gamma)}{\|\Psi'\|},
\end{equation}
where we used the CLT of Lemma~\ref{lemma CLT}.

To show concentration of measure, we also need to consider the variance. We start from
\begin{equation}
	|S|^2 = \frac{1}{\|\Psi'\|^2} \sum_{ijkl} \sqrt{q_iq_jq_lq_m} \sum_{kn} \sqrt{\lambda_k\lambda_n} U_{ik} U^*_{jk} U_{ln} U^*_{mn}.
\end{equation}
The expectation value over the product of four unitaries can be conveniently calculated using Weingarten calculus~\cite{CollinsMatsumotoNovak2022}:
\begin{equation}
 \begin{split}
  \mathbb{E}_U [U_{ik} U^*_{jk} U_{ln} U^*_{mn}] &= \frac{\delta_{ij}\delta_{lm}}{N^2-1} - \frac{\delta_{ij}\delta_{lm}\delta_{kn}}{N(N^2-1)} \\
  & -\frac{\delta_{im}\delta_{jl}}{N(N^2-1)} + \frac{\delta_{im}\delta_{jl}\delta_{kn}}{N^2-1}.
 \end{split}
\end{equation}
It can be quickly noticed that only the first term is responsible for the leading order contribution to $|S|^2$, i.e.,
\begin{equation}
 \mathbb{E}_U\left[|S|^2\right] = \frac{1}{\|\Psi'\|^2} \frac{N^2}{N-1}\left(\frac{\sum_k \sqrt{\lambda_k}}{N}\right)^2
\end{equation}
Subtracting the square of eqn~(\ref{S:eq:expectation U}) yields the variance
\begin{equation}
 \begin{split}
  \mathbb{V}_U[S] &= \frac{1}{\|\Psi'\|^2} \frac{1}{N^2-1}\left(\frac{\sum_k \sqrt{\lambda_k}}{N}\right)^2 \\
  &\longrightarrow \frac{1}{\|\Psi'\|^2} \frac{\mu_\text{sqrt}(\gamma)^2}{N^2-1},
 \end{split}
\end{equation}
where we again used the CLT. Note that the $1/N^2$ scaling leads to a strong suppression with respect to the mean value such that concentration of measure follows.

It is worth to remark that this result confirms our idea in Sec.~\ref{S:sec:random phase} for a concrete example: a high average success probability implies a high fidelity given enough (pseudo)randomness in the history states.

%%%%%%%%%%%%%%%%%%%%%%%%%%%%%%%%%%%%%%%%%%%%%%%%%%%%%%%%%%%%%%%%%%%%%%%%%%%%%%%%%%%%%%%%%%%%%%%%%%%%%%%%%%%%%%%%%%%%%%%%
\subsection{Mutual information and $N>D$}\label{S:sec:mutual info}
%%%%%%%%%%%%%%%%%%%%%%%%%%%%%%%%%%%%%%%%%%%%%%%%%%%%%%%%%%%%%%%%%%%%%%%%%%%%%%%%%%%%%%%%%%%%%%%%%%%%%%%%%%%%%%%%%%%%%%%%

Besides the average success probability, another useful figure of merit is the classical mutual information between histories $i$ and records $j$ based on the weights $q_{ij} = |\lr{r_j|\psi_i}|^2q_i$, which become $q_{ij} = |\sqrt{G}_{ji}|^2q_i$ for the sqrt measurement.\footnote{Note that, despite the symmetry $|\sqrt{G}_{ji}|^2 = |\sqrt{G}_{ij}|^2$, the conditional probability of having history $i$ given record $j$ is by Bayes' rule $q'_{i|j} = q_{j|i}q_i/q'_j \neq |\sqrt{G}_{ij}|^2$. However, for the uniform case $q_i=1/N$ we find $q'_j = 1/N$ such that $q'_{i|j} = q_{j|i}$.} Its mutual information is
\begin{equation}\label{S:eq:MI def}
	I = D(q_{ij}|q_i q'_j) = \sum_{ij} q_{ij}\big[\ln q_{ij} - \ln(q_i q'_j)\big],
\end{equation}
where $D(\cdot|\cdot)$ denotes the relative entropy and $q_i = \sum_j q_{ij}$ and $q'_j = \sum_i q_{ij}$ the marginal distributions of history $i$ and record $j$, respectively. Even though the $q_{ij}$ are formally a probability distribution, we are careful to call them weights as the interpretation of probabilities is subject to a debate in the MWI, even more so when the branches are not exactly orthogonal.

%%%%%%%%%%%%%%%%%%%%%%%%%%%%%%%%%%%%%%%%%%%%%%%%%%%%%%%%%%%%%%%%%%%%%%%%%%%%%%%%%%%%%%%%%%%%%%%%%%%%%%%%%%%%%%%%%%%%%%%%
\subsubsection{The case $\gamma = N/d \le 1$}
%%%%%%%%%%%%%%%%%%%%%%%%%%%%%%%%%%%%%%%%%%%%%%%%%%%%%%%%%%%%%%%%%%%%%%%%%%%%%%%%%%%%%%%%%%%%%%%%%%%%%%%%%%%%%%%%%%%%%%%%

\begin{figure}[t]
	\centering\includegraphics[width=0.48\textwidth,clip=true]{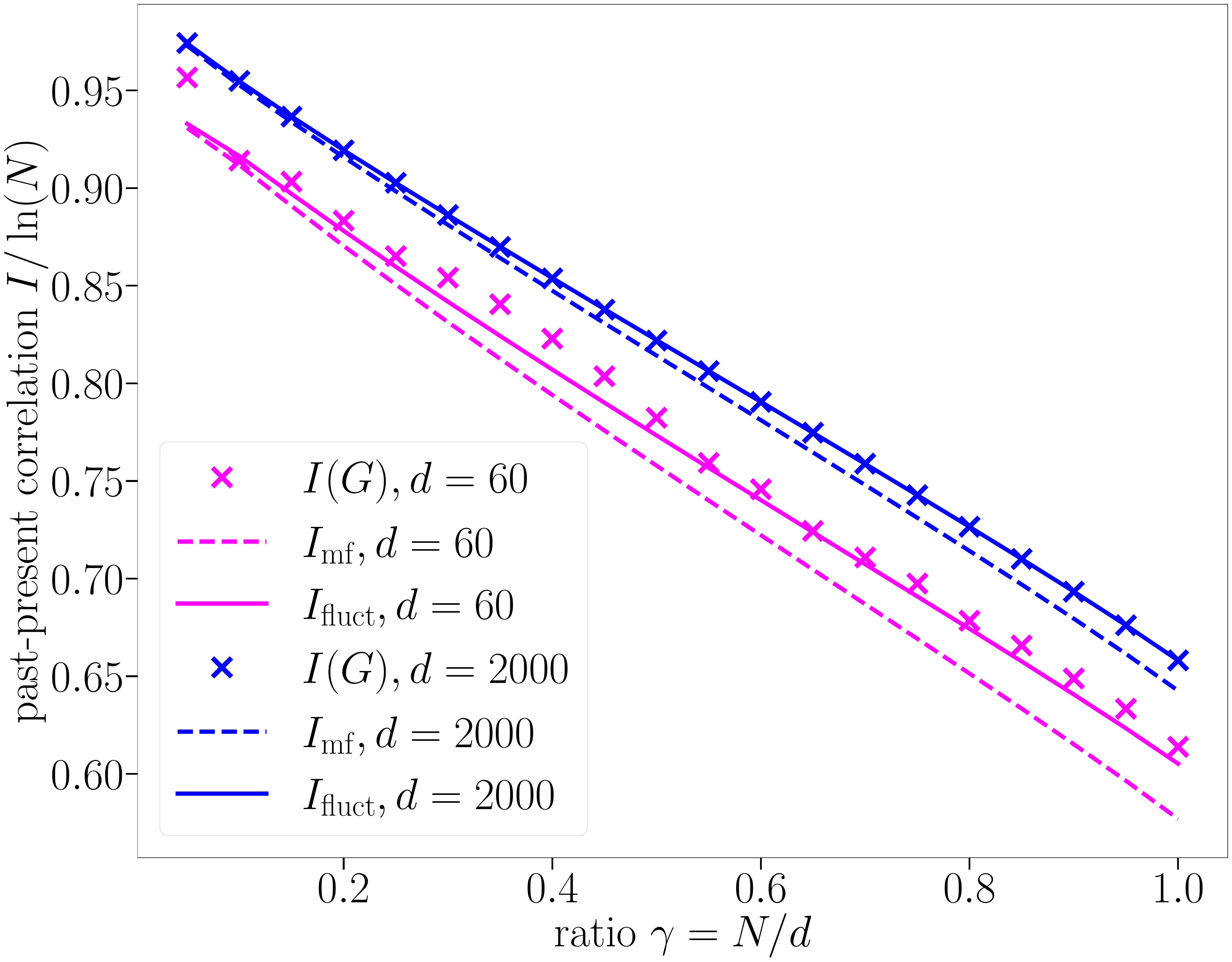}
	\caption{Plot of the normalized mutual information $I/\ln(N)$ for $d=60$ (in magenta) and $d=2000$ (in blue) as a function of $\gamma$. $I(G)$ is the exact mutual information for a single realization of the Gram matrix $G$ (crosses), the dashed line displays the mean field approximation [eqn~(\ref{S:eq:MI mean field})] and the solid line takes off-diagonal fluctuations into account [eqn~(\ref{S:eq:MI fluct})].}
	\label{fig mutual info fits}
\end{figure}

Calculating the mutual information requires to know $q_{ij}$ for $i\neq j$ for which we have no analytical result so far. We thus focus first on ``mean field probabilities'' and write $\mu_{ij}$ instead of $q_{ij}$. The mean field model is defined by taking the ensemble average over $G$, $\mu_{j|i} = \mathbb{E}_G[|\sqrt{G}_{ji}|^2]$. We know that (asymptotically for large $N$) $\mu_{i|i} = \overline P_S$ and by symmetry the $\mu_{j|i}$ for $j\neq i$ must be identical, hence $\mu_{j|i} = (1-\overline P_S)/(N-1)$. We then find
\begin{equation}
	\mu_{ij} = \Theta_N(j)\left[\delta_{ij}\overline{P}_S q_i + (1-\delta_{ij})\frac{1-\overline{P}_S}{N-1}q_i\right].
\end{equation}
Here, $\Theta_N(j)$ denotes a discrete Heavi-side step function: $\Theta_N(j) = 1$ if $j\le N$ and zero elsewise. It ensures that the probabilities for records $|r_j\rangle$ with $j>N$ are zero.

To estimate the mutual information we now consider for simplicity the case of uniform weights $q_i=1/N$. An elementary calculation then reveals
\begin{equation}\label{S:eq:MI mean field}
	I_\text{mf} = \ln(N) - H_2(\overline P_S) - (1-\overline P_S)\ln(N-1),
\end{equation}
where $H_2(p) = -p\ln p -(1-p)\ln(1-p)$ denotes the binary Shannon entropy. Asymptotically, for large $N\le d$, we obtain the simple result
\begin{equation}
	I_\text{mf} \approx \overline P_S(\gamma)\ln N,
\end{equation}
where we made the dependence of $\overline P_S$ on $\gamma$ clear.

While the simple mean field approximation might appear naive, it actually approximates the true value of the mutual information quite well as shown in Fig.~\ref{fig mutual info fits}. However, we can obtain an improved estimate by taking fluctuations for the off-diagonal elements $|\sqrt{G}_{ji}|^2$ ($j\neq i$) into account. To this end, we write
\begin{equation}
	I = \ln N - \frac{1}{N}\sum_i H(q_{j|i})
\end{equation}
with $H(q_{j|i}) = -\sum_j q_{j|i}\ln q_{j|i}$. We again approximate $q_{i|i} = \overline P_S$, i.e., we neglect fluctuations of the diagonal components, but we will take fluctuations of the off-diagonal elements into account. We write
\begin{equation}
	H(q_{j|i}) = H_2(\overline P_S) - (1-\overline P_S)\sum_{j(\neq i)} w_{j|i} \ln w_{j|i}
\end{equation}
with $w_{j|i} = q_{j|i}/(1-\overline P_S)$. Now, the central simplification arises by assuming that the $w_{j|i}$ are uniformly distributed over the $(N-2)$-probability simplex, which is justified by symmetry reasons: $\mathbb{E}[q_{j|i}]$ is independent of $j$ and $i$ for $j\neq i$. The distribution for $w = w_{j|i}$ is the beta distribution $B(1,N-2)$ [compare with eqn~(\ref{S:eq:fidelity pd})], which follows from marginalizing the Dirichlet distribution. It satisfies $\mathbb{E}[w\ln w] = (1-H_{N-1})/(N-1)$ with $H_N = \sum_{k=1}^N 1/k$ the $N$'th harmonic number.\footnote{Asymptotically, $H_N \approx \ln N + \Gamma$ with $\Gamma \approx 0.577$ the Euler-Mascheroni constant.} Thus, we arrive at the following estimate of the mutual information which includes fluctuations:
\begin{equation}\label{S:eq:MI fluct}
 I_\text{fluct} = \ln(N) - H_2(\overline P_S) - (1-\overline P_S)(H_{N-1} - 1).
\end{equation}
As we can see in Fig.~\ref{fig mutual info fits}, differences between $I_\text{fluct}$ and the true $I$ become invisible for large $d$. Note that $I(G) \simeq \mathbb{E}_G[I(G)]$ again concentrates around the mean (not explicitly shown here).

Moreover, Fig.~\ref{fig mutual info fits} also reveals that $I_\text{fluct} \ge I_\text{mf}$. Indeed, this holds in general since
\begin{equation}
 \begin{split}
  I_\text{fluct} &\approx \mathbb{E}_G[I(G)] = \mathbb{E}_G\left[D(q_{ij}|q_i q'_j) \right] \\
  &\ge D(\mu_{ij}|q_i q'_j) = I_\text{mf}.
 \end{split}
\end{equation}
Here, we used the convexity of relative entropy with respect to the first argument, $\mu_{ij} = \mathbb{E}[q_{ij}]$ and that the marginal distributions $q_i = 1/N = q'_j$ are independent of $G$.

%%%%%%%%%%%%%%%%%%%%%%%%%%%%%%%%%%%%%%%%%%%%%%%%%%%%%%%%%%%%%%%%%%%%%%%%%%%%%%%%%%%%%%%%%%%%%%%%%%%%%%%%%%%%%%%%%%%%%%%%
\subsubsection{The case $\gamma = N/d > 1$ for $d=D$}
%%%%%%%%%%%%%%%%%%%%%%%%%%%%%%%%%%%%%%%%%%%%%%%%%%%%%%%%%%%%%%%%%%%%%%%%%%%%%%%%%%%%%%%%%%%%%%%%%%%%%%%%%%%%%%%%%%%%%%%%

Suppose $|\Psi\rangle = \sum_{j=1}^N |\psi'_j\rangle$ contains $N>D$ many histories, which---as we have seen in Sec.~\ref{sec:end}---can well be the case without jeopardizing approximate decoherence. Since there can be at most $D$ many orthogonal record states in a $D$-dimensional universe, it is clear that we can no longer distinguish all histories, not even approximately. While the basic definition of the SLP remains the same, it is useful to consider this case separately. There are essentially two options.

First, within the decoherent histories formalism it is possible to coarse-grain histories further. We can define a function $f: \{1,\dots,N\} \rightarrow \{1,\dots,D\}$ and new history states $|\chi'_i\rangle = \sum_{j=1}^N \delta_{i,f(j)} |\psi'_j\rangle$ such that $|\Psi\rangle = \sum_{i=1}^D |\chi'_i\rangle$ contains only $D$ many history states (such a coarse-graining gives rise to \emph{inhomogenous} histories~\cite{IshamJMP1994}, an example of which we consider in Sec.~\ref{sec:Born}). An interesting question would then be how does $f$ need to be chosen in order to make the $|\chi'_i\rangle$ as distinguishable as possible? In essence, once such an $f$ has been found, this maps the problem back to the case $\gamma\le1$. However, it is far from clear which $f$ are physically legitimate, see also the discussion in Sec.~\ref{sec discussion part 1}. Thus, we here focus on the case where we have to deal with $N>D$ many histories.

In that case, one has to choose $D-1$ many histories that are in some sense optimally distinguishable, construct $D-1$ optimal record states for them, and leave the last one $|r_D\rangle$ as some sort of ``dull record''. For instance, assuming a suitable choice of labeling, one could set $|r'_D\rangle = \sum_{j\ge D} |\psi'_j\rangle$ where the histories $j\in\{D,\dots,N\}$ ideally have low weight, i.e., $\sum_{j\ge D} q_j$ is as small as possible, and are as orthogonal as possible to the previous histories $j<D$. Unfortunately, the precise formulation of these conditions will depend on the chosen figure of merit. Moreover, any figure of merit can only decrease for $N>D$ compared to the case $N = D$.

To give an illustration, we consider only the mutual information, again for the Haar random model of Sec.~\ref{sec:sub-Haar} and we assume for simplicity $d=D$, leaving possible generalizations to future work. Moreover, we restrict the analysis to the mean field case, which well approximates the true mutual information (see Fig.~\ref{fig mutual info fits}). Even then there are different ways to choose the dull record $|r_d\rangle$. Given that the Haar random history states $|\psi_i\rangle$ are homogeneously distributed, we first consider the simple choice that we use the sqrt measurement for the first $d-1$ of them, whereas the final dull record $|r_d\rangle$ is chosen orthogonal to $\mbox{span}\{|r_j\rangle\}_{j=1}^{d-1}$. Within the mean field model we have $|\lr{r_d|\psi_i}|^2 = 1/d$ for all $i\in\{1,\dots,N\}$ and $|\lr{r_j|\psi_i}|^2 = 1/d$ for all $j<d$ and $i\ge d$, i.e., we take the expected value of eqn~(\ref{S:eq:Haar overlap}) for all records and histories which we did not tune with the sqrt measurement. Then, the conditional mean field probability for $N>d$ is
\begin{equation}
	\begin{split}\label{S:eq:muij full}
		\mu_{j|i} =&~ \Theta_{d-1}(i)\Theta_{d-1}(j) \left[\delta_{ij}\overline P_S + (1-\delta_{ij})\frac{\overline P_S}{d-2}\right] \\
		&+ [1-\Theta_{d-1}(i)]\frac{1}{d}.
	\end{split}
\end{equation}
Here, the first line is the conditional probability for the first $d-1$ histories and records, which got tuned to each other by the sqrt measurement, whereas the second line describes how the records overlap with the histories $i\ge d$, and by construction $\lr{r_d|\psi_i} = 0$ for $i<d$. For the computation of the mutual information we first note that $q'_j \approx 1/d$. Thus, we have $I_\text{mf} \approx \ln d - \frac{1}{N}\sum_i H(q_{j|i})$. Distinguishing the cases for $i<d$ and $i\ge d$ carefully, we find for $\gamma>1$
\begin{equation}
	I_\text{mf} \approx \frac{d-1}{N}\left[\ln d - H_2(\overline P_S) - (1-\overline P_S)\ln(d-2)\right]
\end{equation}
with $\overline P_S = \overline P_S(\gamma = 1)$. Note that $I_\text{mf} \rightarrow0$ for $N\rightarrow\infty$.

\begin{figure}[t]
	\centering\includegraphics[width=0.44\textwidth,clip=true]{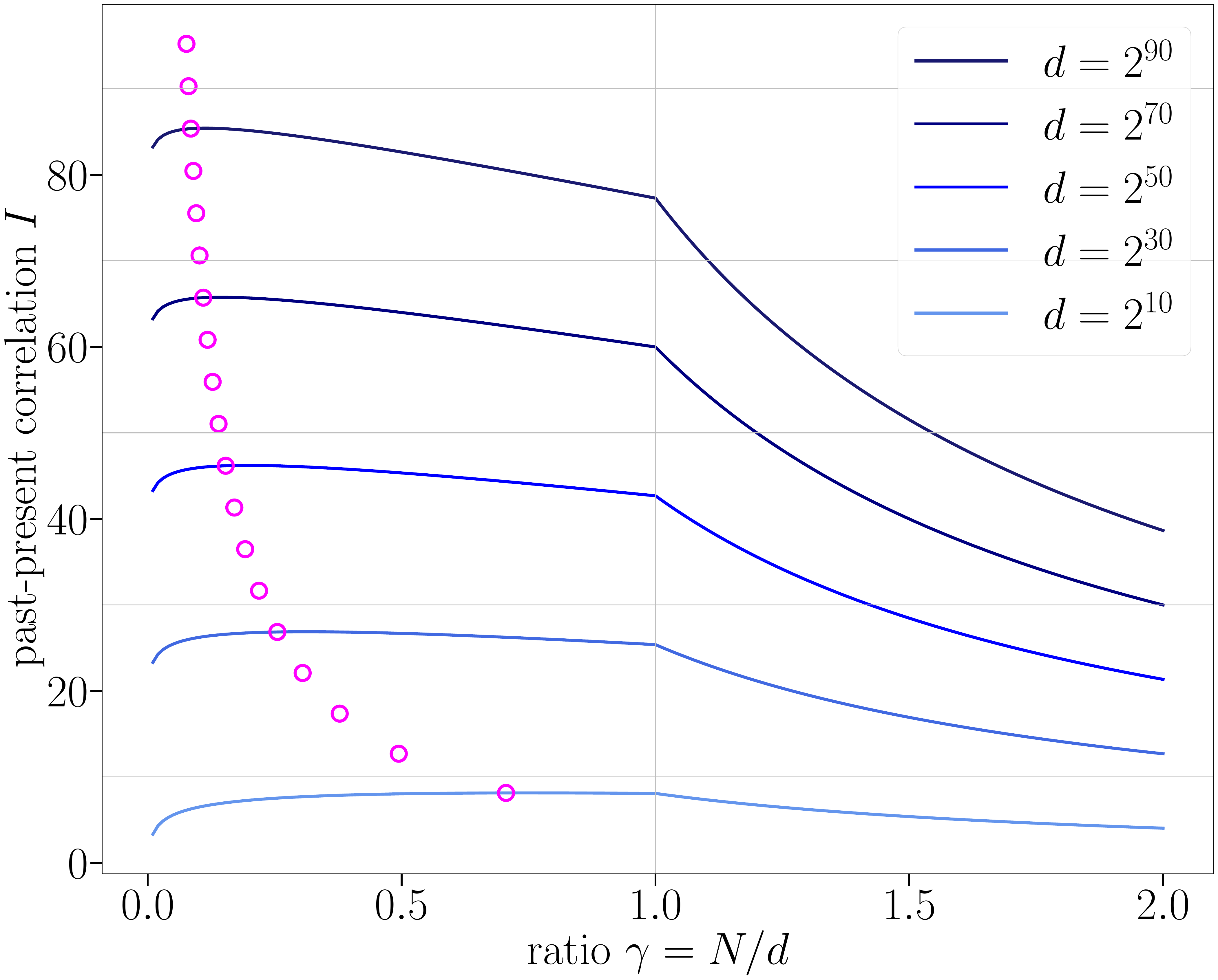}
	\caption{Plot of the mean field mutual information between histories and records (i.e., between past and present) as a function of $N/d$ for increasing dimensions $d=2^K$ from bottom (light blue) to top (dark blue). The horizontal gray lines correspond to the maximum possible mutual information $\log(d) = K$. The pink circles plot the position of the maximum mutual information as a function of $N/d$. }
	\label{fig mutual info}
\end{figure}

$I_\text{mf}$ is plotted in Fig.~\ref{fig mutual info}, which shows, as expected, a rapid decline of the mutual information for $\gamma>1$. Moreover, the pink circles in Fig.~\ref{fig mutual info} indicate the maximum of the mutual information, which must occur for $\gamma\le1$. Since $\overline P_S(\gamma)$ is almost a straight line for $\gamma\le1$ (cf.~with Fig.~\ref{fig sqrt meas}), we approximate $\overline P_S \approx 1-c\gamma$ with $c = 1-\overline P_S(1) \approx 0.28$. For large $N,d$ we then obtain from eqn~(\ref{S:eq:MI mean field}) the approximate form $I_\text{mf} \approx (1-c\gamma)\ln(N)$. Interestingly, the maximum of $I_\text{mf}$ is not obtained for $N=d$ but for a value that is approximately logarithmically smaller ($\sim d/\log d$).\footnote{The true maximum can be expressed in terms of the Lambert function $W_0(x)$. We use the simple approximation $W_0(x) \approx \ln(x)$ to investigate the asymptotic regime for large $x$.} For presentational purposes, let us now set $d=2^K$ and consider the mutual information with respect to the base 2 logarithm. Evaluation of the maximum then gives asymptotically for large $K$
\begin{equation}
	\max I \approx \left(1-\frac{1}{K+\log\beta}\right)\left(K-\log[c(K+\log\beta)]\right)
\end{equation}
with $\beta\equiv2\ln(2)/c$. Thus, the strongest past-present correlations are reached for a multiverse with $N\sim d/\log d$ many histories or branches, which allows, in principle, to decode the state of $\sim K-\log K$ qubits of the universe. Of course, in the multiverse observers on different branches cannot share their records with each other. Hence, it is more correct to say that each observer can \emph{expect} to decode $\sim K-\log K$ qubits, but the actual value given a certain record fluctuates around that average. Finally, recall that we assumed $d=D$ but in general $d=D^{2\alpha}$. Thus, if we extrapolate our results to the general case, then we expect to decode approximately $K = 2\alpha K_\text{tot}$ many qubits (minus logarithmic corrections), where  $K_\text{tot}$ is the total number of qubits in a $D$-dimensional universe.

We conclude by mentioning that the above procedure to define record states in case of $N>D$ is not the only one. For instance, one could also set $|r'_D\rangle = \sum_{j\ge D} |\psi'_j\rangle$ and implement the sqrt measurement on the remaining subspace. This procedure is not identical to the previously considered choice, but for large $d$ and $N-d$ the difference becomes negligible. Another strategy is to choose a subset $\C F$ of histories whose elements have a significantly smaller pairwise fidelity than the average $1/d$. However, while this guarantees a high figure of merit within $\C F$, it does not necessarily imply a high figure of merit in total. Moreover, the probability that the set $\C F$ contains many states is small unless $N\gg d$ for which any figure of merit is anyway close to zero.

%%%%%%%%%%%%%%%%%%%%%%%%%%%%%%%%%%%%%%%%%%%%%%%%%%%%%%%%%%%%%%%%%%%%%%%%%%%%%%%%%%%%%%%%%%%%%%%%%%%%%%%%%%%%%%%%%%%%%%%%
\section{Supplement to Sec.~\ref{sec:structure}}
%%%%%%%%%%%%%%%%%%%%%%%%%%%%%%%%%%%%%%%%%%%%%%%%%%%%%%%%%%%%%%%%%%%%%%%%%%%%%%%%%%%%%%%%%%%%%%%%%%%%%%%%%%%%%%%%%%%%%%%%

\emph{Outline.} We here present additional numerical results about the structure of decoherence for very long histories such that $N=M^L\gg D$ based on the random matrix model introduced in Sec.~\ref{sec:model}.\footnote{We recall that we condition all histories states to end up in subspace $\C H_0$ such that for a given $L$ we only consider $N=M^{L-1}$ histories in our plots.} We will find evidence for an emergent structure of decoherence with respect to various different quantities. Besides Hamming distance and Born's rule, which were studied also in the main text, this includes the localization of history states and the purity of the Petz recovery state for a given history. The strongest structure is observed for inhomogeneous histories with respect to the question whether they are Born-typical, as studied in Sec.~\ref{sec:Born}.

Numerical calculations were either performed for a Haar random initial state confined to subspace $\C H_1$ and for $\Delta t = 8\tau$ (which we call ``equilibrium'' histories) or for a random energy eigenstate and $\Delta t =\tau/2$ (``nonequilibrium'' histories). All results are generated with respect to a single realization of the random matrix Hamiltonian and the random initial state (the Hamiltonian is kept \emph{fixed} for the two initial states), even though by re-running the calculations we found that the results are quite typical. We did notice, however, that energy eigenstates in the bulk of the spectrum behave more like a Haar random state and give rise to stronger decoherence than energy eigenstates at the edges. While a more detailed investigation is desirable, this was numerically out of reach, and we believe the here displayed results nevertheless convey a clear trend.

%%%%%%%%%%%%%%%%%%%%%%%%%%%%%%%%%%%%%%%%%%%%%%%%%%%%%%%%%%%%%%%%%%%%%%%%%%%%%%%%%%%%%%%%%%%%%%%%%%%%%%%%%%%%%%%%%%%%%%%%
\subsection{Decoherence: the overall picture}\label{S:sec:decoherence overall}
%%%%%%%%%%%%%%%%%%%%%%%%%%%%%%%%%%%%%%%%%%%%%%%%%%%%%%%%%%%%%%%%%%%%%%%%%%%%%%%%%%%%%%%%%%%%%%%%%%%%%%%%%%%%%%%%%%%%%%%%

\begin{figure*}[tb]
 \centering\includegraphics[width=0.99\textwidth,clip=true]{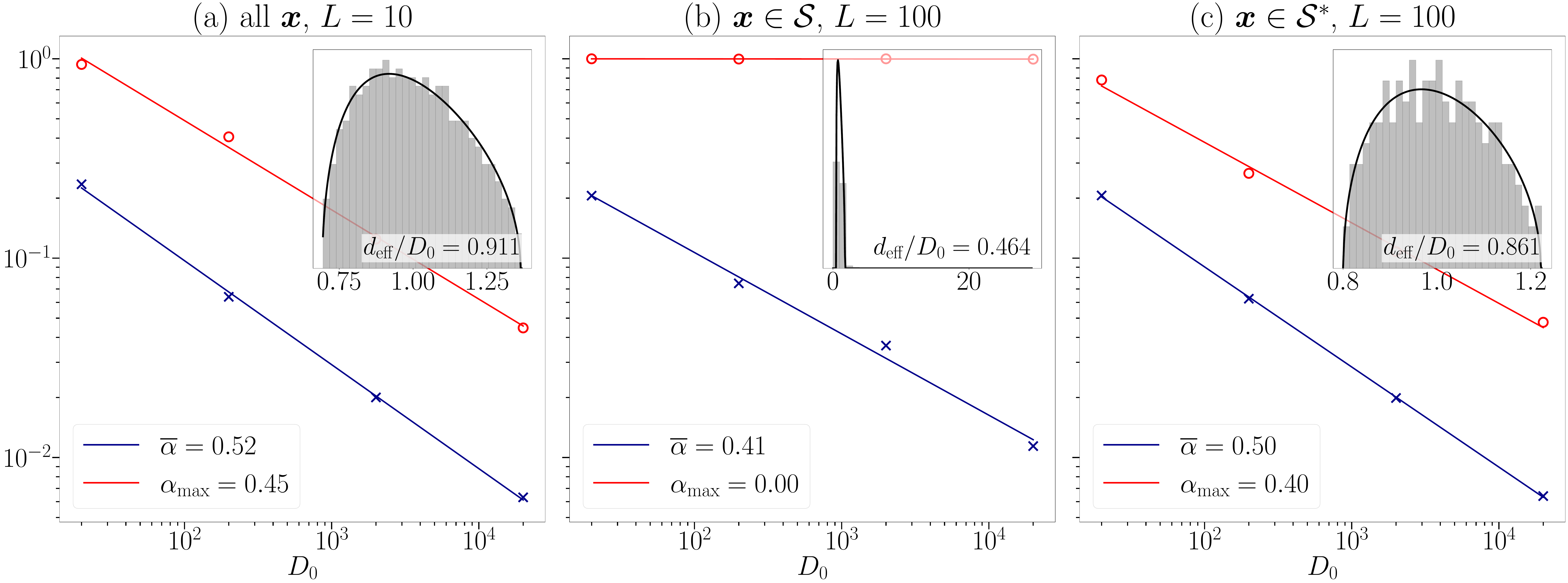}
 \caption{\bb{Equilibrium case.} Plot of the average decoherence $\overline G$ (blue crosses) and maximum coherence $G_\text{max}$ (red circles) on a double logarithmic scale. The solid blue and red lines are a fit from which the exponents $\overline\alpha$ and $\alpha_\text{max}$ are extracted. Insets display a histogram of the eigenvalue distribution of $G$ fitted to the Marchenko-Pastur distribution with effective dimension $d_\text{eff}$ (solid black line). (a) Plot for all histories $\bs x$ of length $L=10$. (b) Plot for a randomly chosen subset $\C S$ of histories for $L=100$. (d) Plot for a specific subset $\C S^*\subset\C S$ for $L=100$. }
 \label{fig dec intro eq}
\end{figure*}

We start by explaining the big picture with Fig.~\ref{fig dec intro eq} for the equilibrium case.

In Fig.~\ref{fig dec intro eq}(a), which is identical to Fig.~\ref{fig:dec intro} in the main text and included here for convenience, we compute the NDF $G$ for $L=10$, from which we extract the quantities
\begin{equation}\label{S:eq:decoherence to plot}
 \overline G \equiv \frac{\sum_{\bs x\neq\bs x'}|G_{\bs x,\bs x'}|}{N^2-N}, ~~~ G_\text{max} \equiv \max_{\bs x\neq\bs x'} |G_{\bs x,\bs x'}|.
\end{equation}
Here, $\overline G$ is the average decoherence and $G_\text{max}$ measures the maximum coherence between two histories, i.e., it measures the significance of statistical outliers with respect to the average $\overline G$. By varying the Hilbert space dimension $D$ and fitting to $1/D_0^\alpha$, we extract the characteristic exponent $\overline\alpha\approx0.52$ for $\overline G$ and $\alpha_\text{max} \approx0.45$ for $G_\text{max}$. The value $\overline\alpha\approx0.52$ is compatible with eqn~(\ref{eq:DF scaling form}) and it matches the $1/\sqrt{D_0}$ equilibrium behavior mentioned below eqn~(\ref{eq:DF scaling form}). The discrepancy of $\overline\alpha$ and $\alpha_\text{max}$ might hint at the fact that $L=10$ is quite large and the elements $G_{\bs x,\bs x'}$ do not quite behave like independent (pseudo-)random numbers. Overall, however, this confirms the basic picture of Refs.~\cite{StrasbergReinhardSchindlerPRX2024, WangStrasbergPRL2025}.

To get further insights, we check the (pseudo-)random character of $G_{\bs x,\bs x'}$ for the largest system size $D=60000$. We plot the histogram of its eigenvalues $\{\lambda_i\}$ as an inset in Fig.~\ref{fig dec intro eq}(a) and compare it with the predicted Marchenko-Pastur distribution of eqn~(\ref{S:eq:MP measure}) for the complex Wishart ensemble (solid black line). To this end, notice that we can write $G=X^\dagger X$, where $X$ is a matrix whose columns are the history states $|\psi(\bs x)\rangle$. We then find the best matching Marchenko-Pastur distribution (using a least square optimization for the histograms) as a function of the ratio $\gamma_\text{eff} = N/d_\text{eff}$ with $d_\text{eff}$ variable. Figure~\ref{fig dec intro eq}(a) shows that there is a well fitting Marchenko-Pastur distribution with $d_\text{eff}$ roughly matching $D_0$. This indicates that the $|\psi(\bs x)\rangle$ behave like independent random vectors with statistical properties in unison with eqn~(\ref{eq:DF scaling form}). Thus, for $L\approx10$ our results are in unison with our basic assumptions from Sec.~\ref{sec:approximate decoherence}.

The situation changes drastically when considering very long histories such that $N=2^{L-1}\gg D$. Figure~\ref{fig dec intro eq}(b) displays the case $L=100$ for a small subset $\C S$ of histories because computing the full $N\times N$ NDF is numerically impossible. Specifically, we construct a sub-NDF by sampling histories $\bs x\in\C S$, where the set $\C S$ is constructed as follows. First, let $L_\text{max}$ be the longest histories we are interested in. Second, we randomly sample $N=1000$ numbers $n\in\{0,1,\dots,L_\text{max}-1\}$, and then randomly sample a history with $n$ many 1s (e.g., $\bs x = 01001$ for $L_\text{max}=5$ and $n=2$). Recall that we require all history states to end up in $\C H_0$, so all histories end with $x_{L_\text{max}} = 0$. Moreover, we impose that the histories with no 1s ($\bs x=000\dots0$) and only 1s except for the last element ($\bs x=011\dots1$) are always elements of $\C S$. Finally, for any value $L<L_\text{max}$ we use the same set $\C S$ but cut the histories after $L$ steps and set $x_L=0$. The so generated sub-NDF is denoted $G_{\bs x,\bs x'}(\C S) = \lr{\psi(\bs x)|\psi(\bs x')}$ with $\bs x,\bs x'\in\C S$.

Figure~\ref{fig dec intro eq}(b) then shows the same quantities as Fig.~\ref{fig dec intro eq}(a) but evaluated for $G(\C S)$. We observe that $\overline{G(\C S)}$ still decays quickly but with a smaller exponent $\overline\alpha\approx0.41$, whereas $G(\C S)_\text{max}$ stopped decaying ($\alpha_\text{max} \approx 0$), indicating pairs of histories with close to maximal coherence for all $D$. Moreover, the inset shows an eigenvalue distribution with long but weak tails (too small to be visible to the naked eye) that can not be fitted in any reasonable way to the Marchenko-Pastur distribution. This indicates that $G(\C S)$ strongly deviates from the behavior of the complex Wishart ensemble.

Remarkably, Fig.~\ref{fig dec intro eq}(c) shows that a behavior similar to the short histories case of Fig.~\ref{fig dec intro eq}(a) can be \emph{restored} by focusing on a suitably chosen subset $\C S^*\subset\C S$ for $L=100$. While Fig.~\ref{fig dec intro eq}(c) is not identical to Fig.~\ref{fig dec intro eq}(a), this could be a consequence of being restricted to a small subset $\C S$ of history samples, and consequently an even smaller set $\C S^*$. The study of $\C S^*$ or, more generally, the question which histories remain decoherent and which become recoherent for large $L$ is the main content in the following. In the next subsection, we start our exploration of this question by looking at the localization of history states, where we also explain how Fig.~\ref{fig dec intro eq}(c) was generated.

\begin{figure*}[ht]
 \centering\includegraphics[width=0.99\textwidth,clip=true]{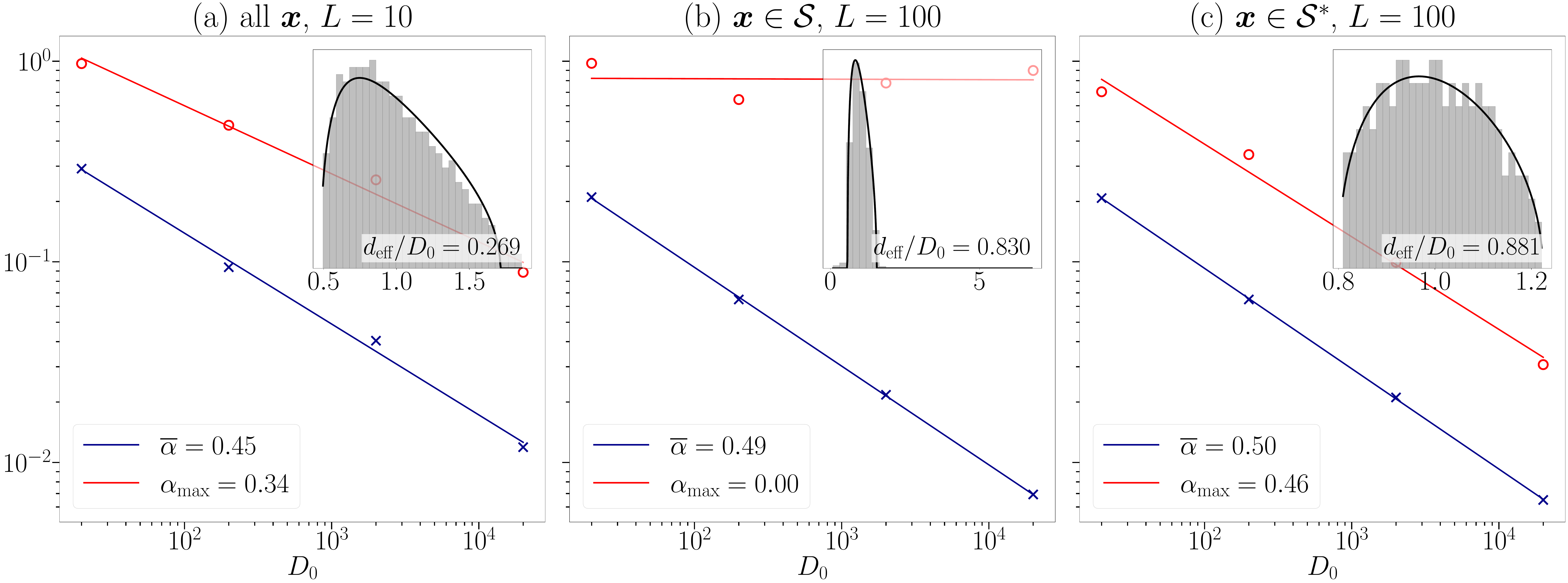}
 \caption{\bb{Nonequilibrium case} (otherwise as in Fig.~\ref{fig dec intro eq})\bb{.}%Plot of the average decoherence $\overline G$ (blue crosses) and maximum coherence $G_\text{max}$ (red circles) on a double logarithmic scale. The solid blue and red lines are a fit from which the exponents $\overline\alpha$ and $\alpha_\text{max}$ are extracted. Insets display a histogram of the eigenvalue distribution of $G$ fitted to the Marchenko-Pastur distribution with effective dimension $d_\text{eff}$ (solid black line). (a) Plot for all histories $\bs x$ of length $L=10$. (b) Plot for a randomly chosen subset $\C S$ of histories for $L=100$. (d) Plot for a specific subset $\C S^*\subset\C S$ for $L=100$. All plots are generated for a single realization of the random matrix Hamiltionian and Haar random initial state (as in all other plots shown below).
 }
 \label{fig dec intro neq}
\end{figure*}

Before doing so, however, we quickly consider the analogous plot to Fig.~\ref{fig dec intro eq} for the nonequilibrium case in Fig.~\ref{fig dec intro neq}. Noteworthy \emph{different} observations are the following:
\begin{itemize}
 \item For $L=10$ we find an exponent $\overline\alpha\approx0.45$ that is smaller than $1/2$ in unison with what we claimed below eqn~(\ref{eq:DF scaling form}): for short nonequilibrium times histories tend to decohere with an $\alpha<1/2$. It should be noted, however, that for the present random matrix example this value does not strongly differ from $1/2$ (unless one considers very short times $\Delta t\ll\tau$, not shown here for brevity). In more realistic spin chains one generally expects smaller values for $\alpha$ for nonequilibrium time scales~\cite{WangStrasbergPRL2025}. In addition, also the value $\alpha_\text{max}\approx0.34$ differs notably.
 \item While the fit to the Marchenko-Pastur distribution looks well, $d_\text{eff} \approx D_0/4$ is notably different from $D_0$. This makes sense because the average decoherence scales like $D_0^{-0.45}$, suggesting that the history states look Haar random in an effective Hilbert space of dimension $D_0^{0.9}$. For $D_0=20000$ we find $D_0^{0.9} \approx 0.37 D_0 \approx d_\text{eff}$ in unison with the expected result from the Wishart ensemble.
 \item For $L=100$ Fig.~\ref{fig dec intro neq}(b) shows that the average decoherence is actually stronger than in the equilibrium case (with $\overline\alpha=0.49$ instead of $\overline\alpha=0.41$). We believe this is, however, not a general effect, but related to the specific sample.
 \item The restoration of decoherence by focusing on a specific subset $\C S^*$ of history states $|\psi(\bs x)\rangle$ (whose construction is explained in the next section) works very well as shown in Fig.~\ref{fig dec intro neq}(c). Indeed, we even restore decoherence to higher values of $\overline\alpha$ and $\alpha_\text{max}$ than for the short histories with $L=10$ in Fig.~\ref{fig dec intro neq}(a).
\end{itemize}

%%%%%%%%%%%%%%%%%%%%%%%%%%%%%%%%%%%%%%%%%%%%%%%%%%%%%%%%%%%%%%%%%%%%%%%%%%%%%%%%%%%%%%%%%%%%%%%%%%%%%%%%%%%%%%%%%%%%%%%%
\subsection{Localization and decoherence}\label{sec localization}
%%%%%%%%%%%%%%%%%%%%%%%%%%%%%%%%%%%%%%%%%%%%%%%%%%%%%%%%%%%%%%%%%%%%%%%%%%%%%%%%%%%%%%%%%%%%%%%%%%%%%%%%%%%%%%%%%%%%%%%%

\begin{figure*}
\begin{subfigure}{.5\textwidth}
  \centering\includegraphics[width=\linewidth,clip=true]{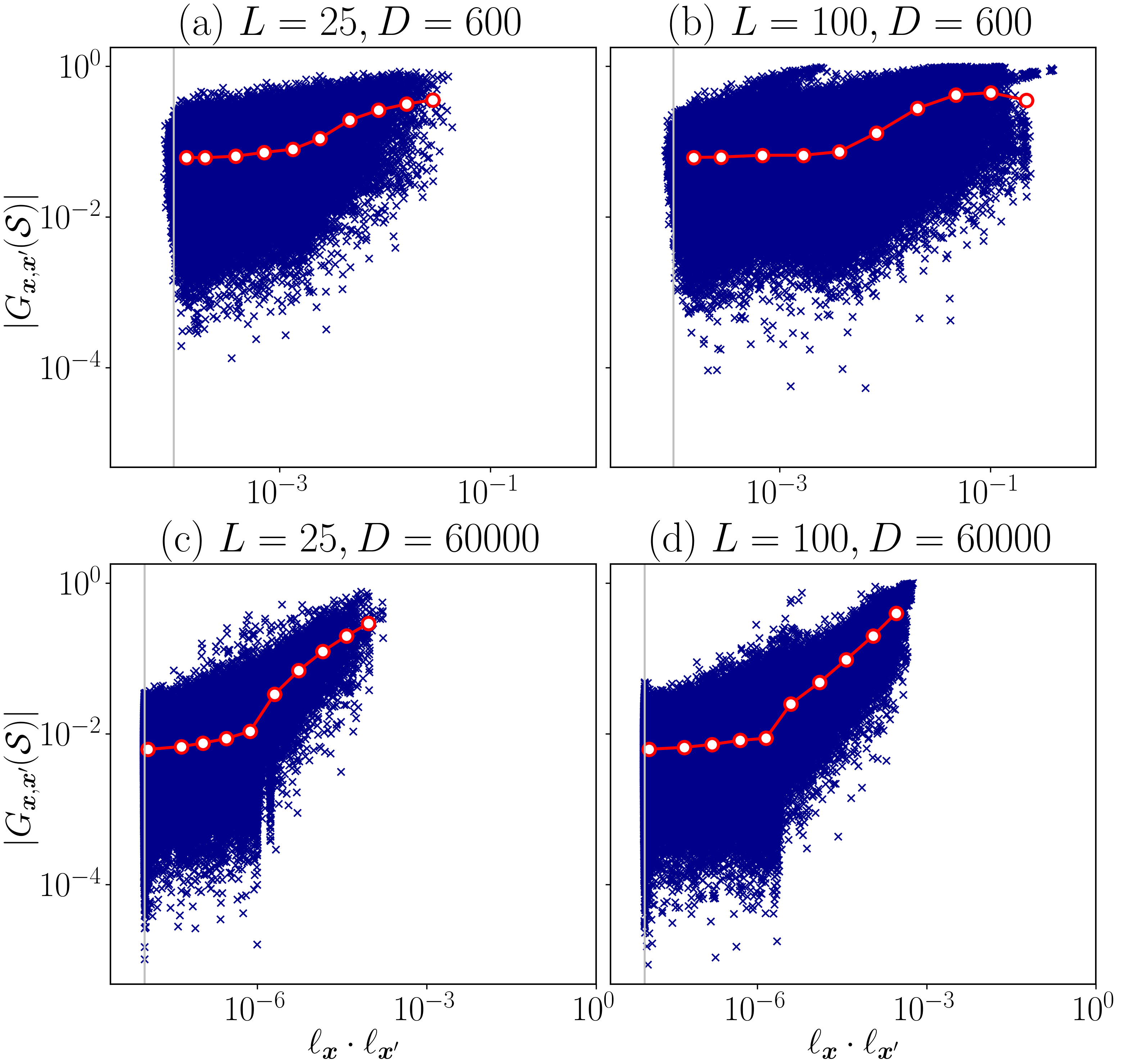}
  \caption{Equilibrium case.}
\end{subfigure}%
\begin{subfigure}{.5\textwidth}
  \centering\includegraphics[width=\linewidth,clip=true]{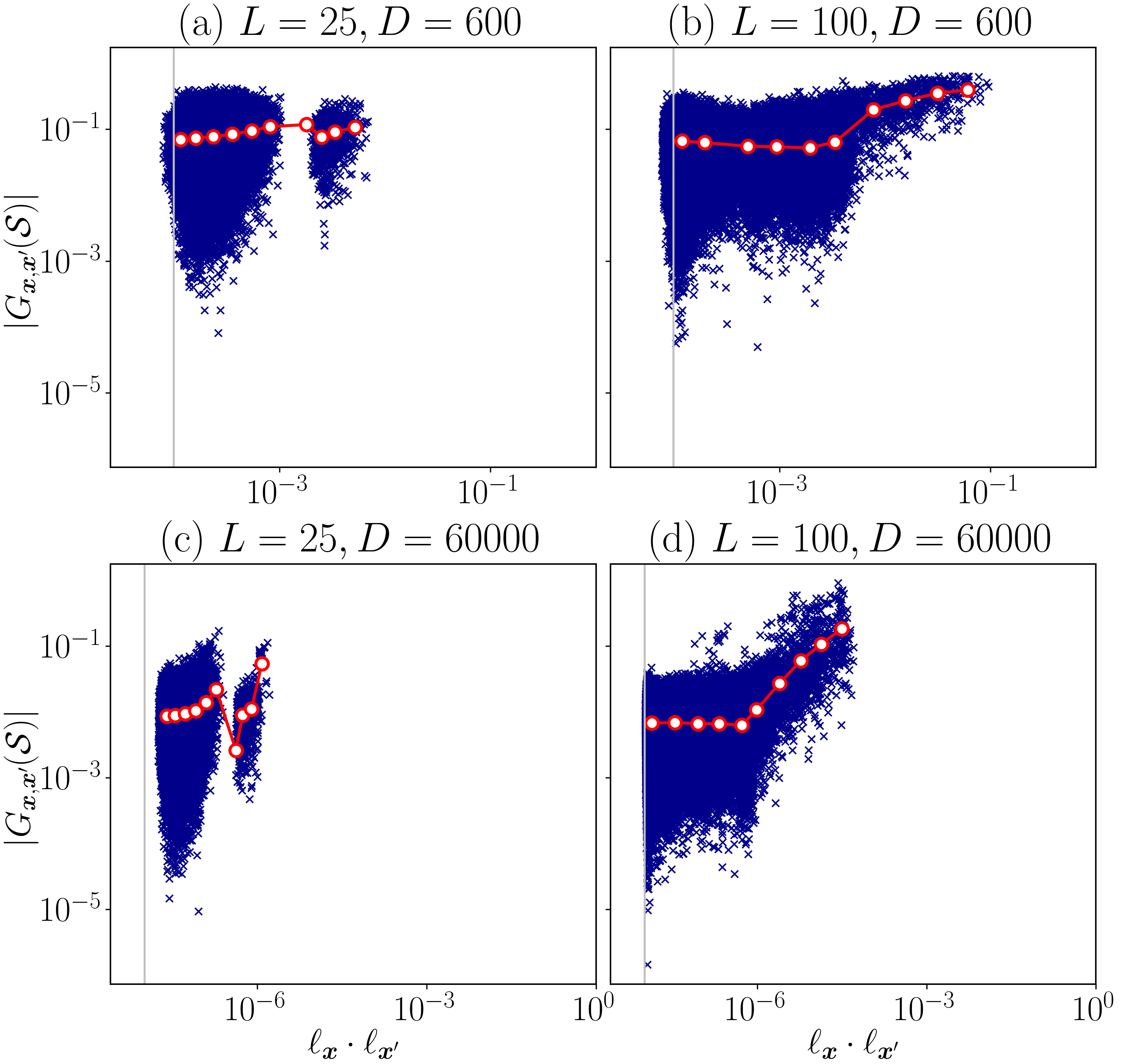}
  \caption{Nonequilibrium case.}
\end{subfigure}
\caption{Plot of the off-diagonal elements of the NDF as a function of localization on a double logarithmic scale to better display the four-orders-of-magnitude variations. Note that the scale of the $y$-axes are shared. The gray lines indicate the Haar average. The white disks are generated by dividing the range of $\ell_{\bs x}\ell_{\bs x'}$ (on a logarithmic scale) into ten equal bins and computing the respective averages of $\ell_{\bs x}\ell_{\bs x'}$ and $|G_{\bs x,\bs x'}|$ in each bin (the red line connecting the disks is a guide for the eye).}
\label{fig:dec loc}
\end{figure*}

Let $\ell_{\bs x} \equiv \sum_i |\lr{i|\psi(\bs x)}|^4$ be the localization of the history state $|\psi(\bs x)\rangle$ ($1/\ell_{\bs x}$ is also called the inverse participation ratio). Here, $\{|i\rangle\}_D$ is the computational basis defined to set up the Hamiltonian in eqn~(\ref{eq H blocks}), i.e., $|i\rl i|$ commutes with $H_{00}, H_{11}$ as well as $\Pi_0$ and $\Pi_1$. Then, $\ell_{\bs x}\in[1/D_0,1]$ quantifies how spread out $|\psi(\bs x)\rangle$ is over the Hilbert space with the limiting cases $\ell_{\bs x} = 1/D_0$ for a maximally delocalized state with $|\lr{i|\psi(\bs x)}|^2 = 1/D_0$ and $\ell_{\bs x} = 1$ for a fully localized state with $|\lr{i|\psi(\bs x)}|^2 = \delta_{i,i^*}$ for some $i^*$. Moreover, the average localization of a Haar random state is $\lr{\ell}_\text{Haar} = 2/(D_0+1)$, which follows from eqn~(\ref{S:eq:Dirichlet}). Now, the conjecture is that the overlap $\lr{\psi(\bs x)|\psi(\bs x')}$ remains small if both history states are delocalized, assuming that the $\lr{i|\psi(\bs x)}$ and $\lr{i|\psi(\bs x')}$ are approximately uncorrelated. Instead, if some $|\psi(\bs x)\rangle$ is strongly localized in some part of $\C H$, it tends to have a large overlap with other states localized in the same part---though its overlap will be very small with states localized in different parts.

This basic idea is confirmed in Figs.~\ref{fig dec intro eq}(c) and~\ref{fig dec intro neq}(c), where we generated $\C S^*$ by considering the 20\% of histories $\bs x$ with the smallest localization $\ell_{\bs x}$. We see that we obtain exponents $\overline\alpha$ and $\alpha_\text{max}$ close to the short histories case of Figs.~\ref{fig dec intro eq}(a) and~\ref{fig dec intro neq}(a). Moreover, the eigenvalue distribution is reasonably well captured by the Marchenko-Pastur distribution.

The conjectured correlation between coherence and localization is further supported by the scatter plot of Fig.~\ref{fig:dec loc}, where we plot $G_{\bs x,\bs x'}(\C S)$ for all $\bs x\neq\bs x'\in\C S$ as a function of their product of localizations $\ell_{\bs x}\cdot\ell_{\bs x'}$. To get a feeling for the absolute scale, the $x$ axes cover the interval $[1/D_0^2,1]$ from the minimum to the maximum of $\ell_{\bs x}\cdot\ell_{\bs x'}$, and the vertical gray line indicates $\lr{\ell}_\text{Haar}^2$. We see that some correlation between localization and decoherence is already visible for relatively short histories of length $L=25$ and the effect becomes more pronounced for $L=100$. We emphasize that the conjectured correlation between decoherence and localization is not one-to-one: as remarked above, it can happen that two history states are localized in different regions of the Hilbert space, thus having very little overlap.

The nonequilibrium case differs from the equilibrium case by having some additional structure for $L=25$ (for which we have no explanation) and the effect for $L=100$ is slightly weaker, but the trend is the same.

%%%%%%%%%%%%%%%%%%%%%%%%%%%%%%%%%%%%%%%%%%%%%%%%%%%%%%%%%%%%%%%%%%%%%%%%%%%%%%%%%%%%%%%%%%%%%%%%%%%%%%%%%%%%%%%%%%%%%%%%
\subsection{Petz purity and decoherence}\label{sec Petz}
%%%%%%%%%%%%%%%%%%%%%%%%%%%%%%%%%%%%%%%%%%%%%%%%%%%%%%%%%%%%%%%%%%%%%%%%%%%%%%%%%%%%%%%%%%%%%%%%%%%%%%%%%%%%%%%%%%%%%%%%

\begin{figure*}
\begin{subfigure}{.38\textwidth}
  \centering\includegraphics[width=\linewidth,clip=true]{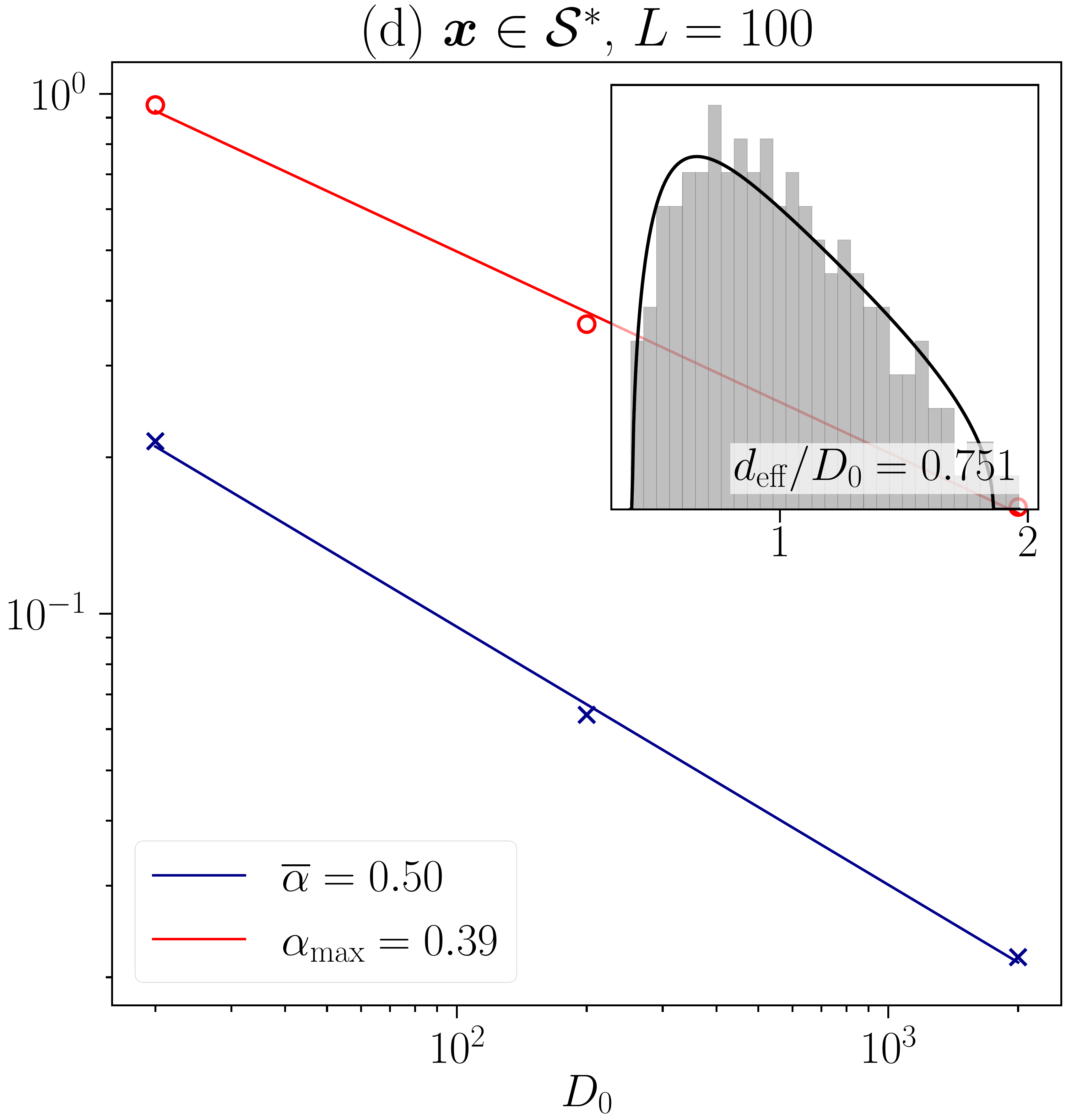}
  \caption{Equilibrium case.}
\end{subfigure}%
\begin{subfigure}{.38\textwidth}
  \centering\includegraphics[width=\linewidth]{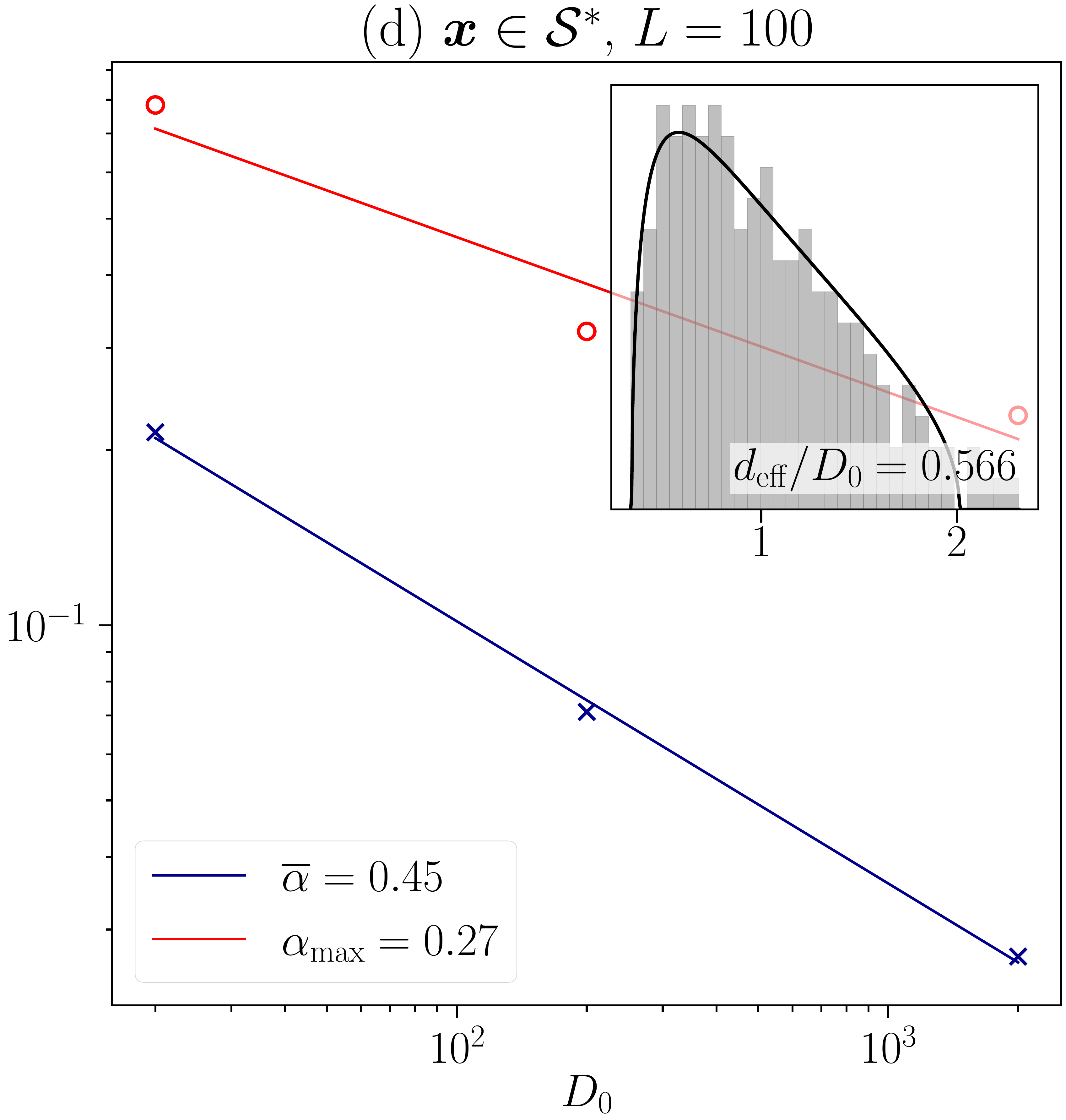}
  \caption{Nonequilibrium case.}
\end{subfigure}
\caption{Average decoherence $\overline G(\C S^*)$ (blue crosses) and maximum decoherence $G_\text{max}(\C S^*)$ (red circles) for the 20\% of histories with the lowest purity. Exponents are extracted from the fitted line. Note the double logarithmic scale. Inset: histogram matched to the Marchenko-Pastur distribution by varying $d_\text{eff}$.}
\label{fig:sub dec Petz}
\end{figure*}

Since localized states are restricted to a small ``corner'' of the Hilbert space, they seem more predictable (in the sense of having an easier description) than delocalized states, which are smeared out over a large part of the Hilbert space. The results in the previous subsection indicate that such ``easy'' states tend to show larger coherences, but the notion of localization is obviously basis dependent and this sketchy idea lacks full rigor.

A more rigorous formulation of this idea can be based on the Petz recovery map, which generalizes Bayes' rule to the quantum domain~\cite{PetzCMP1986, PetzQJM1988, JungeEtAlAHP2018, ParzygnatBuscemiQuantum2023, BaiBuscemiScaraniPRL2025}. To formulate it, we introduce the Kraus operator $K_{\bs x} = \Pi_{x_L} U_{\Delta t}\cdots\Pi_{x_1}U_{\Delta t}$ (also called a class operator in the context of decoherent histories) such that $|\psi'(\bs x)\rangle = K_{\bs x}|\Psi_0\rangle$. Now, we take the perspective of an observer inside the multiverse who has no prior knowledge about the initial state $|\Psi_0\rangle$ apart from the fact that all histories start in the subspace $\C H_1$. Thus, the associated maximum entropy state is $\rho_0 = \Pi_1/D_1$. We now ask whether knowledge of a particular history $\bs x$ allows its observer to reduce the uncertainty about the initial state by considering the initial Petz-recovered state (akin of a Bayesian retrodicted state), which is
\begin{equation}\label{S:eq:Petz}
 \tau_0(\bs x) = \frac{\sqrt{\rho_0}K_{\bs x}^\dagger K_{\bs x}\sqrt{\rho_0}}{\mbox{tr}\{\rho_0K_{\bs x}^\dagger K_{\bs x}\}} = \frac{\Pi_1 K_{\bs x}^\dagger K_{\bs x}\Pi_1}{\mbox{tr}\{\Pi_1 K_{\bs x}^\dagger K_{\bs x}\}}.
\end{equation}
The purity of that state,
\begin{equation}\label{S:eq:Petz purity}
 P_{\bs x} \equiv \mbox{tr}\{\tau_0(\bs x)^2\} = \frac{\mbox{tr}\{(K_{\bs x}^\dagger K_{\bs x}\Pi_1)^2\}}{\mbox{tr}\{K_{\bs x}^\dagger K_{\bs x}\Pi_1\}^2} \in [D_1^{-1},1],
\end{equation}
then provides another way to judge how informative a history $\bs x$ is.\footnote{We note that there is also a Petz recovered state at the final time, $\tau_L(\bs x) \sim K_{\bs x}\Pi_1 K_{\bs x}^\dagger$, which equals the time evolved state of the guess $\rho_0 = \Pi_1/D_1$. It turns out that it has the same purity as the state of eqn~(\ref{S:eq:Petz}).} High purities indicate that the history $\bs x$ is very informative in the sense that only a few members of the initial ensemble $\rho_0 = \Pi_1/D_1$ are likely compatible with the observed history, whereas low purities indicate that not much can be learned from the history about the initial state. Notice that the Petz purity is a measure \emph{independent} of the initial state $|\Psi_0\rangle$.

\begin{figure*}
\begin{subfigure}{.5\textwidth}
  \centering\includegraphics[width=\linewidth,clip=true]{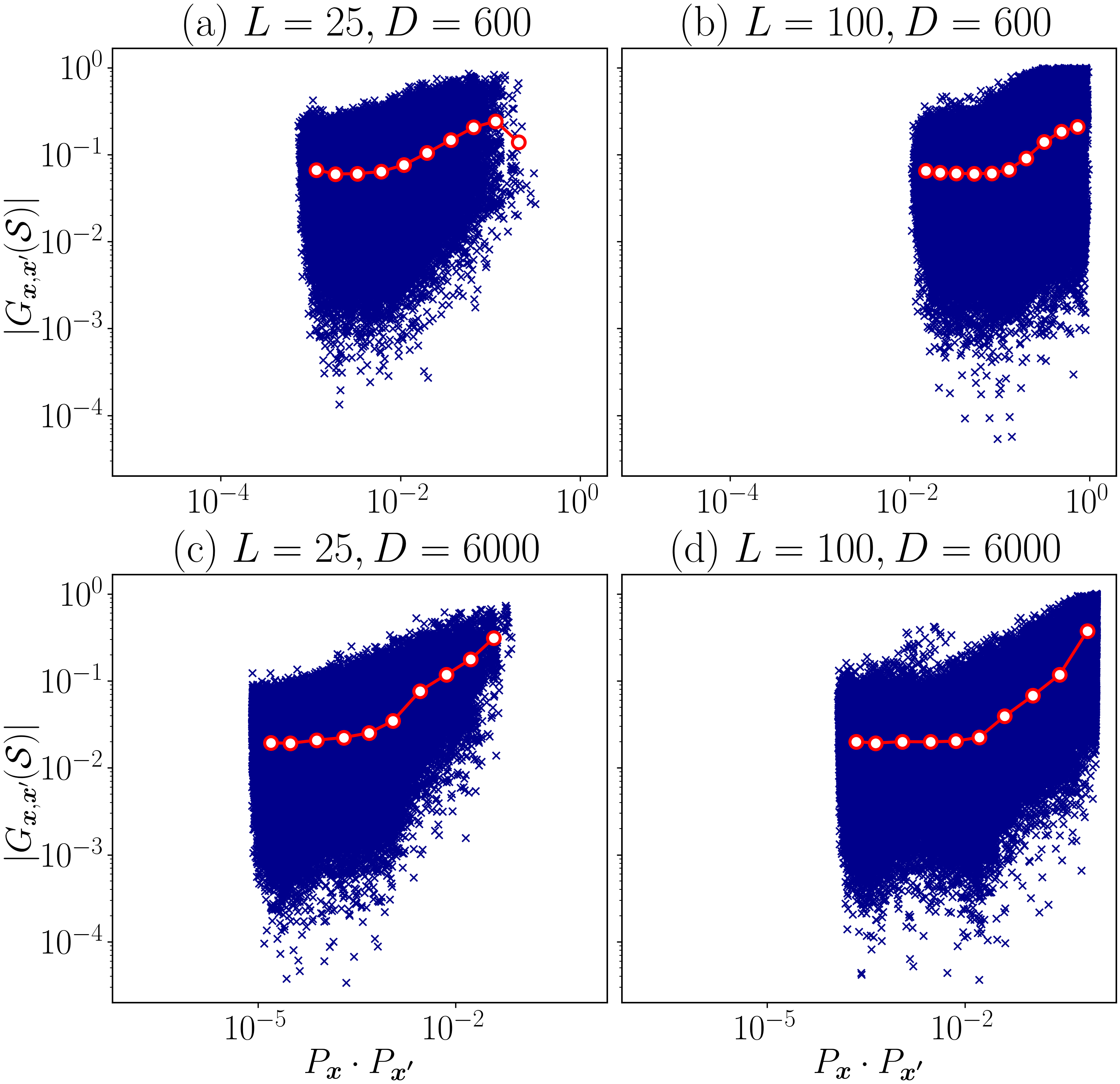}
  \caption{Equilibrium case.}
\end{subfigure}%
\begin{subfigure}{.5\textwidth}
  \centering\includegraphics[width=\linewidth]{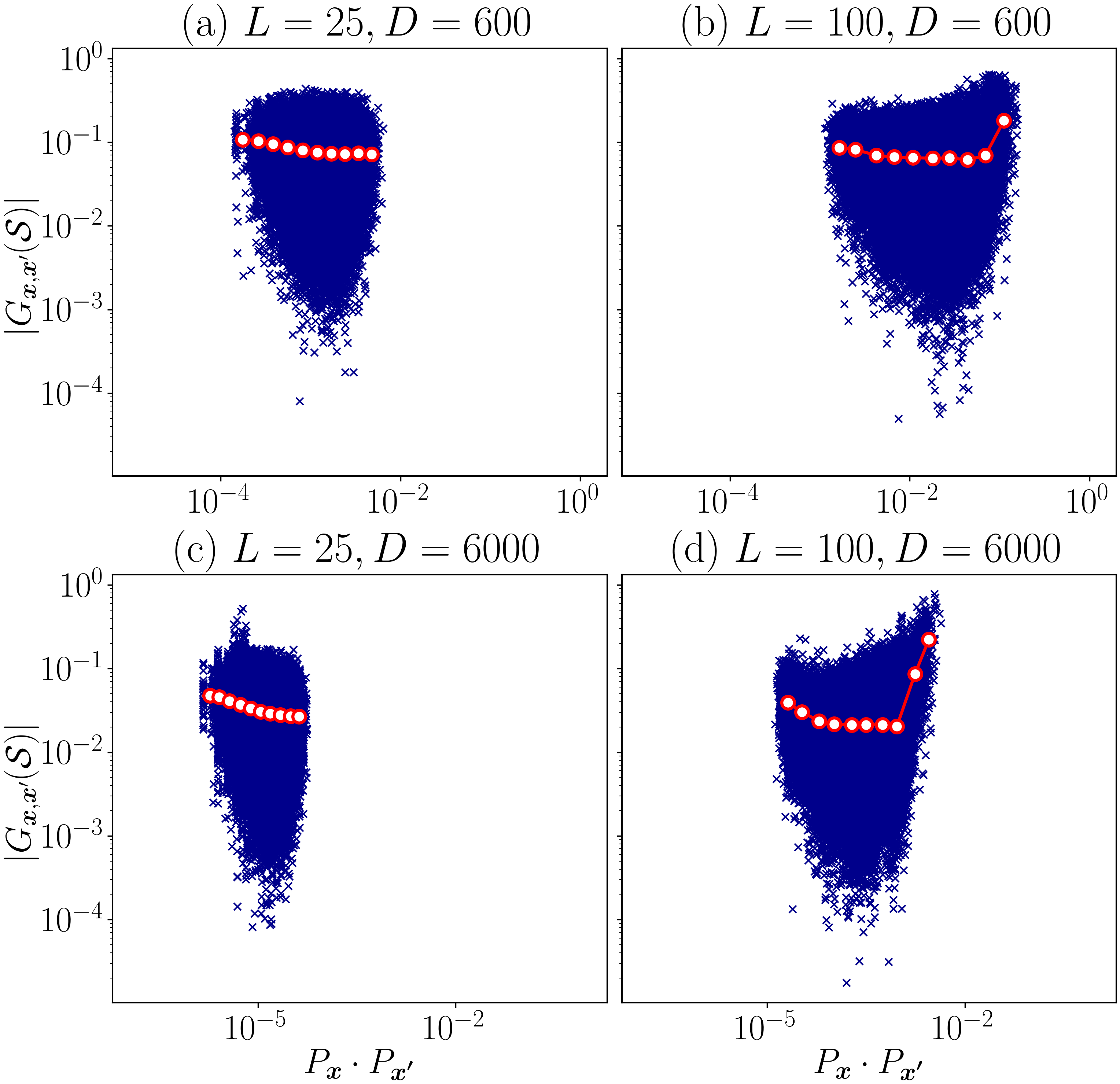}
  \caption{Nonequilibrium case.}
\end{subfigure}
\caption{Plot of the off-diagonal elements of the NDF as a function of Petz purity on a double logarithmic scale. The scale of the $y$-axes are shared. The white disks are generated by dividing the range of $P_{\bs x}P_{\bs x'}$ (on a logarithmic scale) into ten equal bins and computing the respective averages of $P_{\bs x}P_{\bs x'}$ and $G_{\bs x,\bs x'}$ in each bin (the red line connecting the disks is a guide for the eye).}
\label{fig:dec pur}
\end{figure*}

As in Figs.~\ref{fig dec intro eq}(c) and~\ref{fig dec intro neq}(c), the idea above is tested by considering in Fig.~\ref{fig:sub dec Petz} the NDF for the subset $\C S^*$ containing the 20\% of histories $\bs x$ with the lowest Petz purity $P_{\bs x}$. We find a similar behavior as in Figs.~\ref{fig dec intro eq}(c) and~\ref{fig dec intro neq}(c): average decoherence and maximum coherence are restored to a value close to the $L=10$ case, and the eigenvalue distribution is reasonably well captured by the Marchenko-Pastur distribution. This restoration works slightly better for the equilibrium case than for nonequilibrium case.

Note that the biggest Hilbert space dimension in Fig.~\ref{fig:sub dec Petz} is $D=6000$. The reason is that the computation of eqn~(\ref{S:eq:Petz purity}) requires to propagate the matrix $\Pi_1$ for every history $\bs x$, which is numerically very costly compared to propagating a pure state vector like $|\Psi_0\rangle$. Even for $D=6000$ computing the Petz purity is one of the major bottlenecks in terms of the required simulation time.

To further test the idea that Petz purity is correlated with decoherence, we use a scatter plot in Fig.~\ref{fig:dec pur}, similar to Fig.~\ref{fig:dec loc}. We see again a correlation between coherence and Petz purity, but this is again not one-to-one: also pairs of histories with high Petz purities can stay decoherent. Moreover, the tendency for the nonequilibrium case is now very weak compared to the equilibrium case. Note that, as in Fig.~\ref{fig:dec loc}, the $x$ axes range over the entire range $[1/D_1^{2},1]$ to get a feeling for the absolute scale. In particular, we observe that for $L=100$ a large fraction of Petz recovery states got almost purified, confirming the intuitive idea that longer histories contain more information.

%%%%%%%%%%%%%%%%%%%%%%%%%%%%%%%%%%%%%%%%%%%%%%%%%%%%%%%%%%%%%%%%%%%%%%%%%%%%%%%%%%%%%%%%%%%%%%%%%%%%%%%%%%%%%%%%%%%%%%%%
\subsection{Hamming distance and decoherence}\label{sec Hamming}
%%%%%%%%%%%%%%%%%%%%%%%%%%%%%%%%%%%%%%%%%%%%%%%%%%%%%%%%%%%%%%%%%%%%%%%%%%%%%%%%%%%%%%%%%%%%%%%%%%%%%%%%%%%%%%%%%%%%%%%%

\begin{figure*}
\begin{subfigure}{.5\textwidth}
  \centering\includegraphics[width=\linewidth,clip=true]{DECvsHam.png}
  \caption{Equilibrium case.}
\end{subfigure}%
\begin{subfigure}{.5\textwidth}
  \centering\includegraphics[width=\linewidth,clip=true]{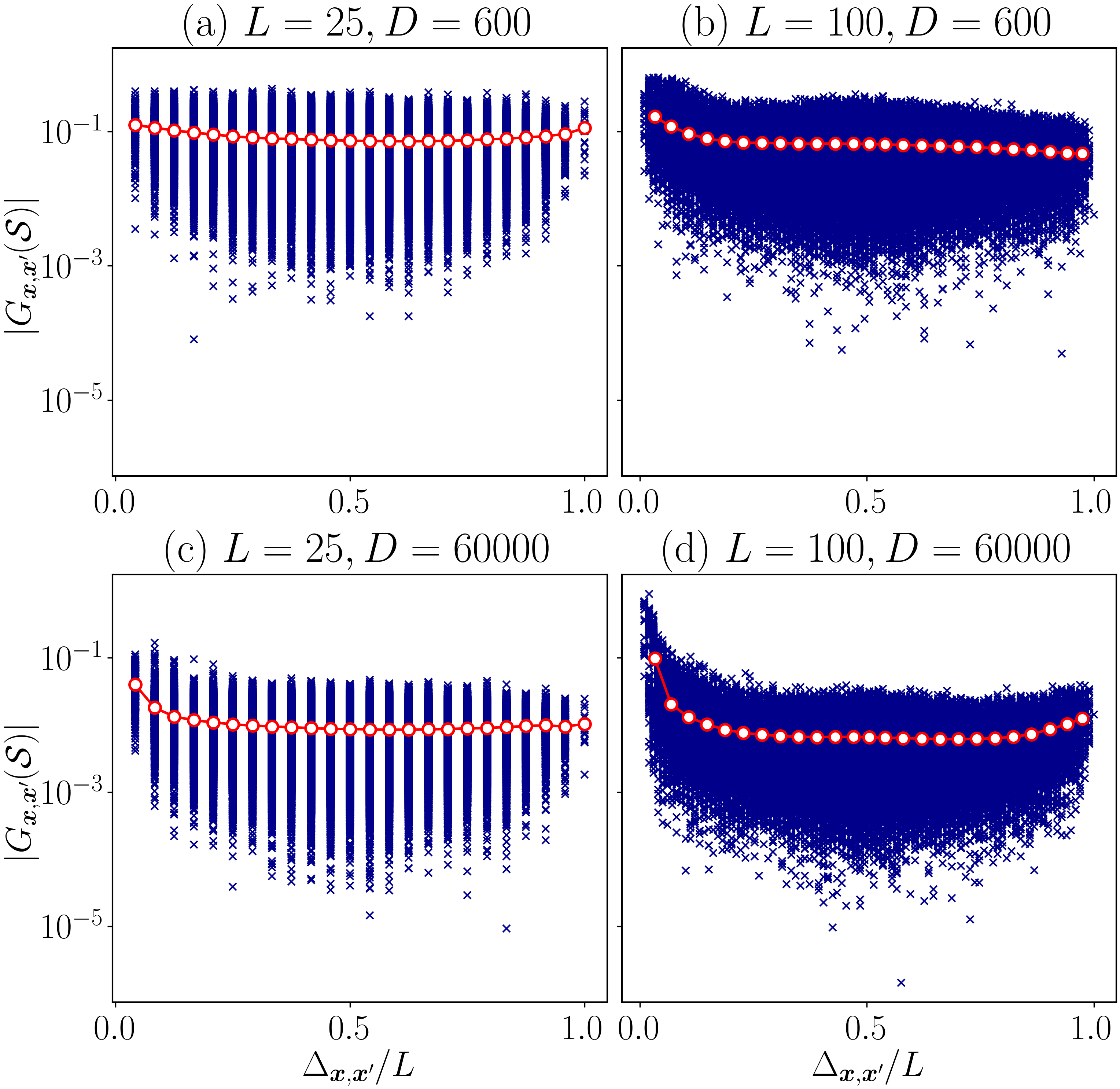}
  \caption{Nonequilibrium case.}
\end{subfigure}
\caption{Plot of the off-diagonal elements of the NDF (on a shared logarithmic scale) as a function of the (normalized) Hamming distance for different $L$ and $D$. The white disks again correspond to a sample average by dividing the $x$-axis into 25 bins (red lines are a guide for the eye). }
\label{fig:dec Ham}
\end{figure*}

\begin{figure*}[tb]
\begin{subfigure}{.5\textwidth}
  \centering\includegraphics[width=0.80\linewidth,clip=true]{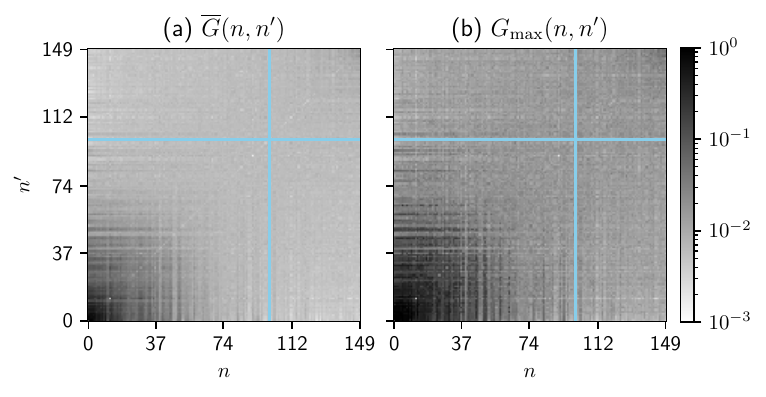}
  \caption{Equilibrium case.}
\end{subfigure}%
\begin{subfigure}{.5\textwidth}
  \centering\includegraphics[width=0.80\linewidth,clip=true]{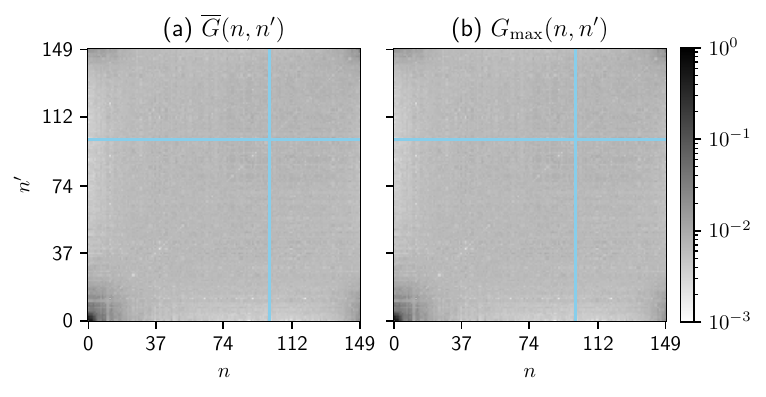}
  \caption{Nonequilibrium case.}
\end{subfigure}
\caption{Heat map plot of the average decoherence (a) and maximum coherence (b) as a function of $n = \#_1(\bs x)$ for all $\bs x\in\C S$ and $L=150$. The horizontal and vertical blue lines are at $\overline n$. If there is accidentally a pair $(n,n'$) without a corresponding $(\bs x,\bs x') \in \C S\times\C S$, we set $\overline G(n,n')=G_\text{max}(n,n')=0$. }
\label{fig:heat map}
\end{figure*}

So far, we identified that localized states, either in the sense of being localized in Hilbert space or having a high purity recovery state, can give rise to large coherences with other states. However, as remarked above, it is not the case that localized states necessarily share large coherences with any other localized state, as they might be localized in different parts of the Hilbert space. Thus, it would be desirable to have a kind of metric on the space of histories $\bs x$ that correlates with the amount of (de)coherence between the states $|\psi(\bs x)\rangle$.

Numerically, we found that the Hamming distance
\begin{equation}
 \Delta_{\bs x,\bs x'} = \sum_{i=1}^{L-1} |x_i-x'_i|
\end{equation}
provides a reasonable candidate for such a correlation. Namely, the closer two histories are in terms of the Hamming distance, the more likely it is that they are coherent, as illustrated in Fig.~\ref{fig:dec Ham} [Fig.~\ref{fig:dec Ham}(a) is the same as in the main manuscript]. This correlation is, however, again only statistical in nature: while there is a clear tendency, a short Hamming distance does not necessarily imply coherence. Moreover, the correlation is again stronger in the equilibrium case compared to the nonequilibrium case, but we currently have no specific explanation for this.

%%%%%%%%%%%%%%%%%%%%%%%%%%%%%%%%%%%%%%%%%%%%%%%%%%%%%%%%%%%%%%%%%%%%%%%%%%%%%%%%%%%%%%%%%%%%%%%%%%%%%%%%%%%%%%%%%%%%%%%%
\subsection{$n$-dependence, decoherence and weights}\label{sec n dependence}
%%%%%%%%%%%%%%%%%%%%%%%%%%%%%%%%%%%%%%%%%%%%%%%%%%%%%%%%%%%%%%%%%%%%%%%%%%%%%%%%%%%%%%%%%%%%%%%%%%%%%%%%%%%%%%%%%%%%%%%%

In preparation of the next (and penultimate) section, it is useful to consider decoherence as a function of the number of ``1s'' in a history $\bs x$, denoted $n=\#_1(\bs x)$. For example, for $\bs x = 0110010$ one gets $n = \#_1(\bs x) = 3$. Note that the expected number of ``1s'' for histories of length $L$ is $\overline n = (2/3)(L-1)$ for \emph{both} the equilibrium \emph{and} nonequilibrium case. In the equilibrium case the factor 2/3 arises from the fact that we consider time scales for which the system has time to relax back to equilibrium at which $\lr{\Pi_1}_\text{Eq} = D_1/D = 2/3$ for our choice of parameters (recall Fig.~\ref{fig averages}). In the nonequilibrium histories the factor 2/3 arises because the ratio to jump from $\C H_0$ to $\C H_1$ over the inverse process is $T_{1|0}/T_{0|1} \approx 2/3$, where $T_{x|x'}$ is the sampled transition matrix for the nonequilibrium time scale, generated in the same way as explained in Sec.~\ref{sec:Born} of the main manuscript (of course, this result assumes Markovianity, which we found to be justified in the decoherent subspace, see also Sec.~\ref{sec:Born}). In both cases the factor $L-1$ (instead of $L$) arises from the fact that the histories are conditioned to end up in subspace $\C H_0$.

The dependence of decoherence on $n$ is shown in the ``heat map'' plots of Fig.~\ref{fig:heat map} for $L=150$, which is created by changing the dependence of the NDF $G(\C S)$ from $\bs x$ to $n=\#_1(\bs x)$. Since multiple $\bs x$ can give rise to the same $n$, we consider the average decoherence
\begin{equation}
	\overline G(n,n') \equiv \frac{\sum_{\bs x\neq\bs x'} \delta_{n,\#_1(\bs x)} \delta_{n',\#_1(\bs x')} |G_{\bs x,\bs x'}|}{\sum_{\bs x\neq\bs x'} \delta_{n,\#_1(\bs x)} \delta_{n',\#_1(\bs x')}},
\end{equation}
obtained by summing all $|G_{\bs x,\bs x'}|$ with a given $n$ and $n'$ divided by the total number of summands, and the maximum coherence for a given $(n,n')$
\begin{equation}
	\begin{split}\label{S:eq:dec S n max}
		&G_\text{max}(n,n') \equiv \\
		& ~\max\{|G_{\bs x,\bs x'}|~|~\bs x\neq\bs x', n=\#_1(\bs x), n'=\#_1(\bs x')\}.
	\end{split}
\end{equation}
For the equilibrium case we can see a clear tendency for larger coherences when $n$ and $n'$ both deviate from the expected mean $\overline n$, whereas this tendency is almost invisble for the nonequilibrium case. This will be further investigated in the next subsection when we considered inhomogeneous histories.

\begin{figure*}[tb]
	\centering\includegraphics[width=0.99\textwidth,clip=true]{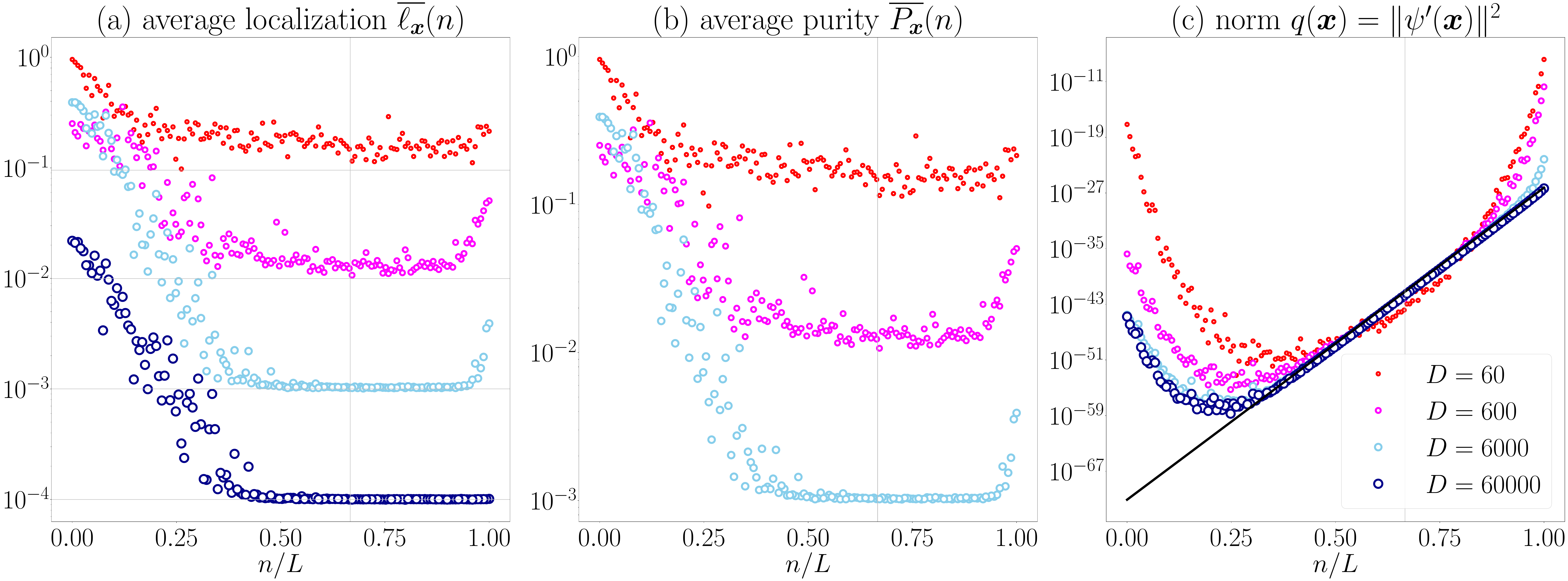}
	\caption{As a function of $n/L$ with $L=150$ we plot on a logarithmic scale the average localization (a), the average Petz purity (b) and the weights $q(\bs x)$ (c). Increasing Hilbert space dimensions from top to bottom are indicated by increasing circle size and different colors, see legend in (c). The vertical gray line indicates $\overline n/L$. The horizontal gray lines in (a) indicate the Haar average $2/(D_0+1)$. The black line in (c) is the probability $p(\bs x)$ obtained from a Bernoulli trial, see the main text. }
	\label{fig N1 eq}
\end{figure*}

\begin{figure*}[htb]
	\centering\includegraphics[width=0.99\textwidth,clip=true]{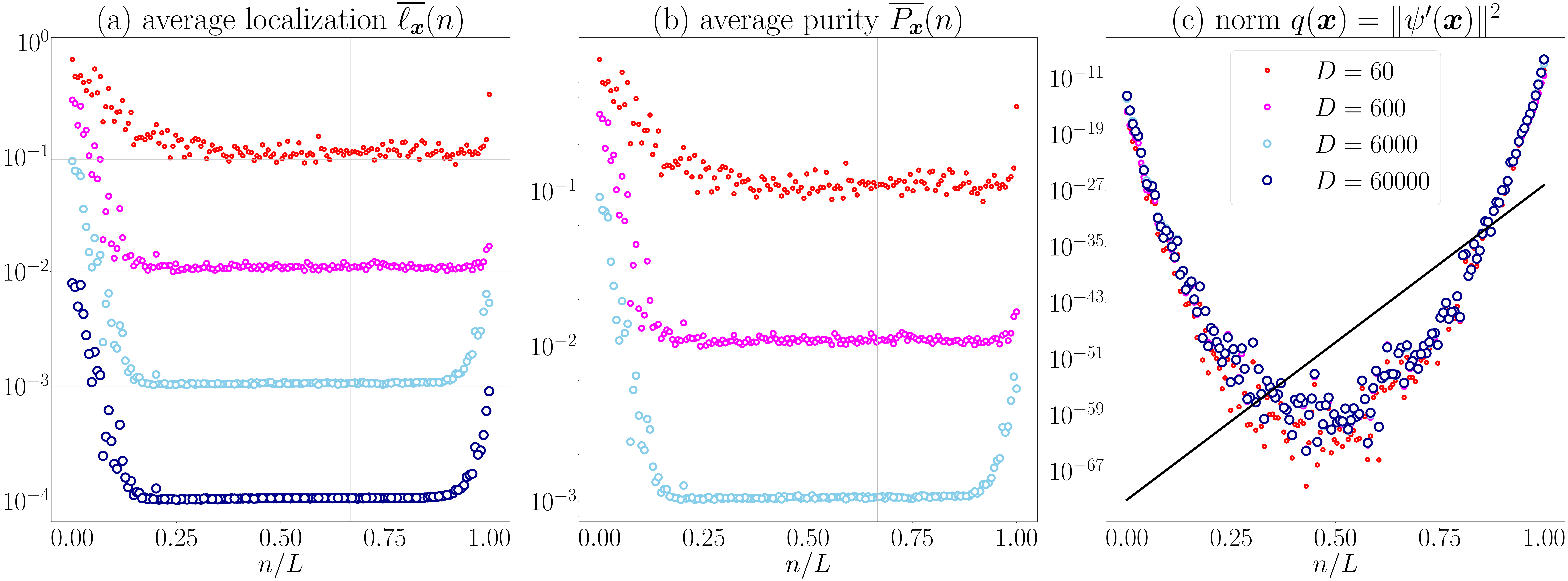}
	\caption{Nonequilibrium case (otherwise the same as in Fig.~\ref{fig N1 eq}).}
	\label{fig N1 neq}
\end{figure*}

Further details are revealed in Fig.~\ref{fig N1 eq} for the equilibrium case and in Fig.~\ref{fig N1 neq} for the nonequilibrium case. In Figs.~\ref{fig N1 eq}(a) and~\ref{fig N1 neq}(a) we plot the average localization
\begin{equation}
	\overline{\ell_{\bs x}}(n) = \frac{\sum_{\bs x}\delta_{n,\#_1(\bs x)}\ell_{\bs x}}{\sum_{\bs x}\delta_{n,\#_1(\bs x)}}
\end{equation}
(with $\ell_{\bs x}$ introduced in Sec.~\ref{sec localization}) as a function of $n$. We find that a higher localization correlates well with the deviation $n-\overline n$ from the mean. Of course, since we found that high localization correlates with high coherence, this finding can be anticipated from Fig.~\ref{fig:heat map}, but it is now also quite visible for the nonequilibrium case. Moreover, the horizontal gray lines indicate the Haar average $\lr{\ell}_\text{Haar}$, showing that histories with strong decoherence and $n\approx\overline n$ continue to behave Haar random.

In Figs.~\ref{fig N1 eq}(b) and~\ref{fig N1 neq}(b) we plot the average Petz purity
\begin{equation}
	\overline{P_{\bs x}}(n) = \frac{\sum_{\bs x}\delta_{n,\#_1(\bs x)}P_{\bs x}}{\sum_{\bs x}\delta_{n,\#_1(\bs x)}}
\end{equation}
(with $P_{\bs x}$ introduced in Sec.~\ref{sec Petz}) as a function of $n$. The conclusion is similar to before: high purity correlates with stronger deviation $n-\overline n$, which can be anticipated from Fig.~\ref{fig:heat map} as we already found that high Petz purities correlate with high coherence, but here we can also see a clear tendency for the nonequilibrium case.

\begin{figure*}[htb]
	\centering\includegraphics[width=0.99\textwidth,clip=true]{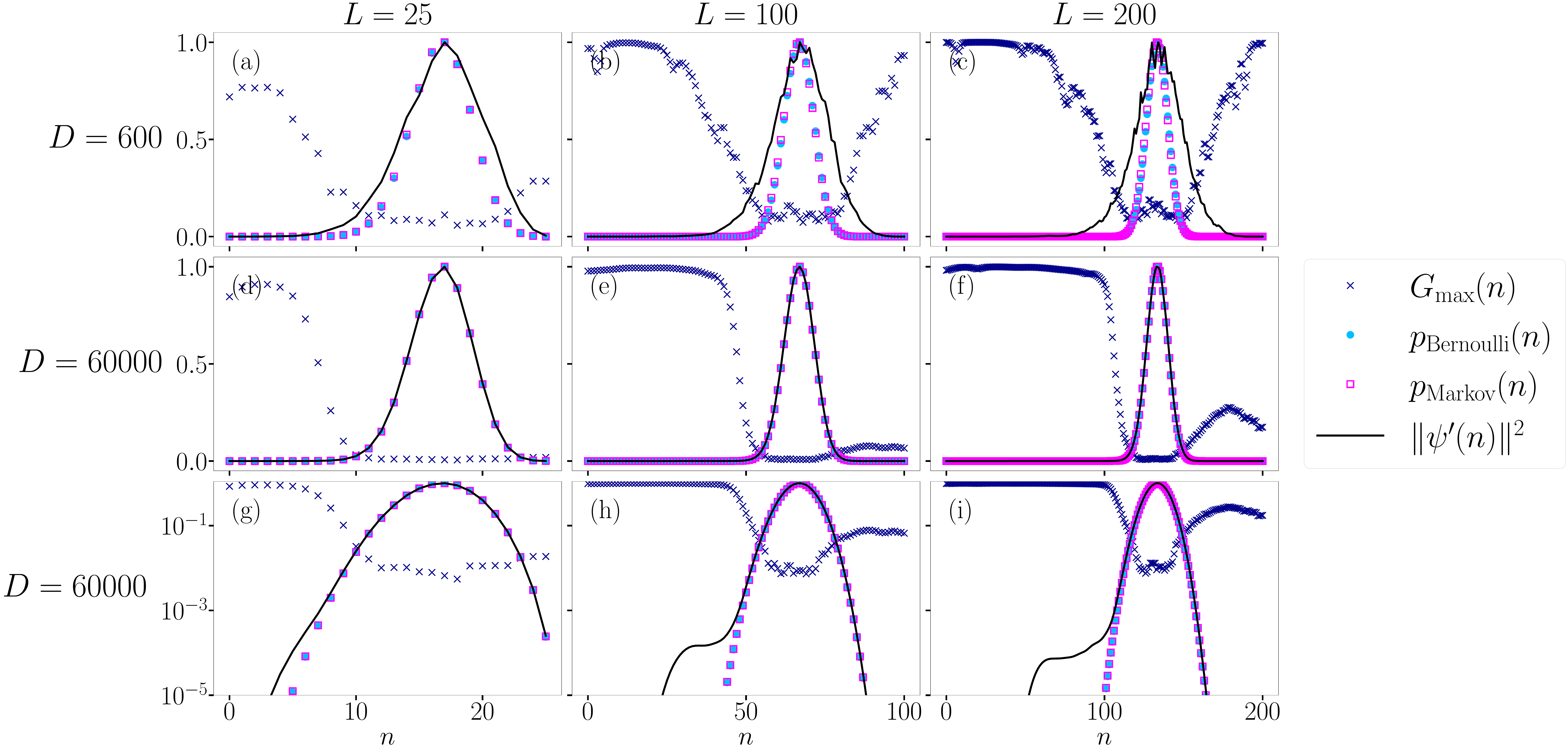}
	\caption{\bb{Equilibrium histories} (for the nonequilibrium case see Fig.~\ref{fig Born neq} in the main manuscript)\bb{.} Decoherence, probabilities and weights of inhomogeneous histories as a function of $n\in\{0,1,\dots,L\}$ for $L\in\{25,100,200\}$ and $D\in\{600,60000\}$ (the third row is identical to the second but displayed on a logarithmic scale). We plot the maximum coherence (dark blue crosses), the weights of the history states (displayed for better visibility as a black solid line instead of discrete points), the probabilities resulting from an idealized Markov process (magenta squares), and the probabilities of an idealized Bernoulli trial with the same average (sky blue disks). To plot the probabilistic quantities on the same scale as the (de)coherence, all of them got rescaled such that their maximum value equals one. In contrast to before, we here do not demand that the histories end in subspace $\C H_0$: we keep both outcomes $x_L\in\{0,1\}$ to construct the $|\psi(n)\rangle$.	}
	\label{fig Born eq}
\end{figure*}

Finally, in Figs.~\ref{fig N1 eq}(c) and~\ref{fig N1 neq}(c) we study another quantity, which becomes more important in the next subsection. This is the probability weight $q(\bs x) = \lr{\psi'(\bs x)|\psi'(\bs x)}$ of the individual history states. For equilibrium histories we expect that the system returns to its equilibrium state after each time step $\Delta t$; thus, we expect $q(\bs x)$ to coincide with the probability of a Bernoulli trial $p(\bs x) = (2/3)^n (1/3)^{L-1-n}$ for a given $n=\#_1(\bs x)$ (black line). That tendency is clearly visible for increasing $D$ in Fig.~\ref{fig N1 eq}(c), but convergence is slow for $n$ far from $\overline n$ (particularly visible for $n\rightarrow0)$. Whether $q(\bs x) \rightarrow p(\bs x)$ for $D\rightarrow\infty$ can thus not be answered unambiguously owing to numerical limitations. It is, however, interesting to observe that the maverick branches (histories whose statistics are atypical according to Born's rule) are characterized by weights significantly \emph{larger} than $p(\bs x)$, which indicates that cutting away low weight branches does not necessarily imply Born's rule as often stipulated. Moreover, in the nonequilibrium case [Fig.~\ref{fig N1 neq}(c)] the weights $q(\bs x)$ no longer tend to match the probabilities $p(\bs x) = (2/3)^n (1/3)^{L-1-n}$ of a memoryless Bernoulli trial (despite the same mean). Thus, correlations matter for the full probabilities of the nonequilibrium histories $\bs x$.

%%%%%%%%%%%%%%%%%%%%%%%%%%%%%%%%%%%%%%%%%%%%%%%%%%%%%%%%%%%%%%%%%%%%%%%%%%%%%%%%%%%%%%%%%%%%%%%%%%%%%%%%%%%%%%%%%%%%%%%%
\subsection{Inhomogeneous histories and Born's rule}\label{sec Born}
%%%%%%%%%%%%%%%%%%%%%%%%%%%%%%%%%%%%%%%%%%%%%%%%%%%%%%%%%%%%%%%%%%%%%%%%%%%%%%%%%%%%%%%%%%%%%%%%%%%%%%%%%%%%%%%%%%%%%%%%

For the investigation of inhomogeneous histories we now consider equilibrium time scales and a global Haar random initial state (\emph{not} confined to subspace $\C H_1$) as in Ref.~\cite{StrasbergSchindlerArXiv2023}. The results in Fig.~\ref{fig Born eq} show exactly the same tendency as in the main text: decoherence happens for histories that are Born-typical. As a consequence of the equilibrium time scale, the only difference is that the probabilities $p_\text{Markov}(n)$ (magenta squares) and $p_\text{Bernoulli}(n)$ (sky blue disks) now match perfectly. Moreover, we see that recoherent histories have weights $q(n)$ clearly deviating from the Bernoulli distribution, which becomes clear when looking at the logarithmic scale in the last row of Fig.~\ref{fig Born eq}.

%%%%%%%%%%%%%%%%%%%%%%%%%%%%%%%%%%%%%%%%%%%%%%%%%%%%%%%%%%%%%%%%%%%%%%%%%%%%%%%%%%%%%%%%%%%%%%%%%%%%%%%%%%%%%%%%%%%%%%%%
\subsection{Summary and discussion}\label{sec summary part 2}
%%%%%%%%%%%%%%%%%%%%%%%%%%%%%%%%%%%%%%%%%%%%%%%%%%%%%%%%%%%%%%%%%%%%%%%%%%%%%%%%%%%%%%%%%%%%%%%%%%%%%%%%%%%%%%%%%%%%%%%%

\begin{figure*}[tb]
\begin{subfigure}{.5\textwidth}
  \centering\includegraphics[width=0.90\textwidth,clip=true]{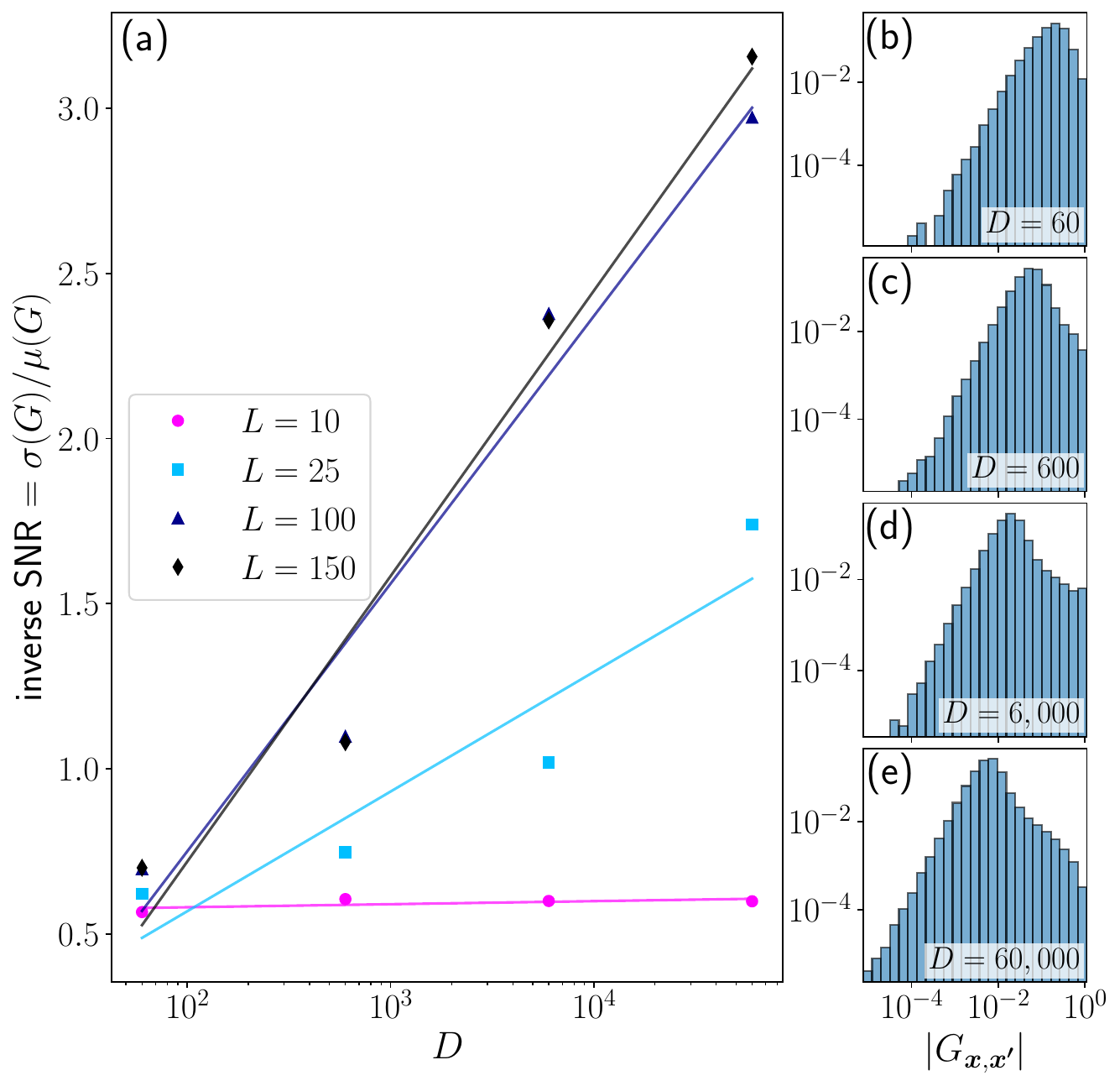}
  \caption{Equilibrium case.}
\end{subfigure}%
\begin{subfigure}{.5\textwidth}
  \centering\includegraphics[width=0.90\textwidth,clip=true]{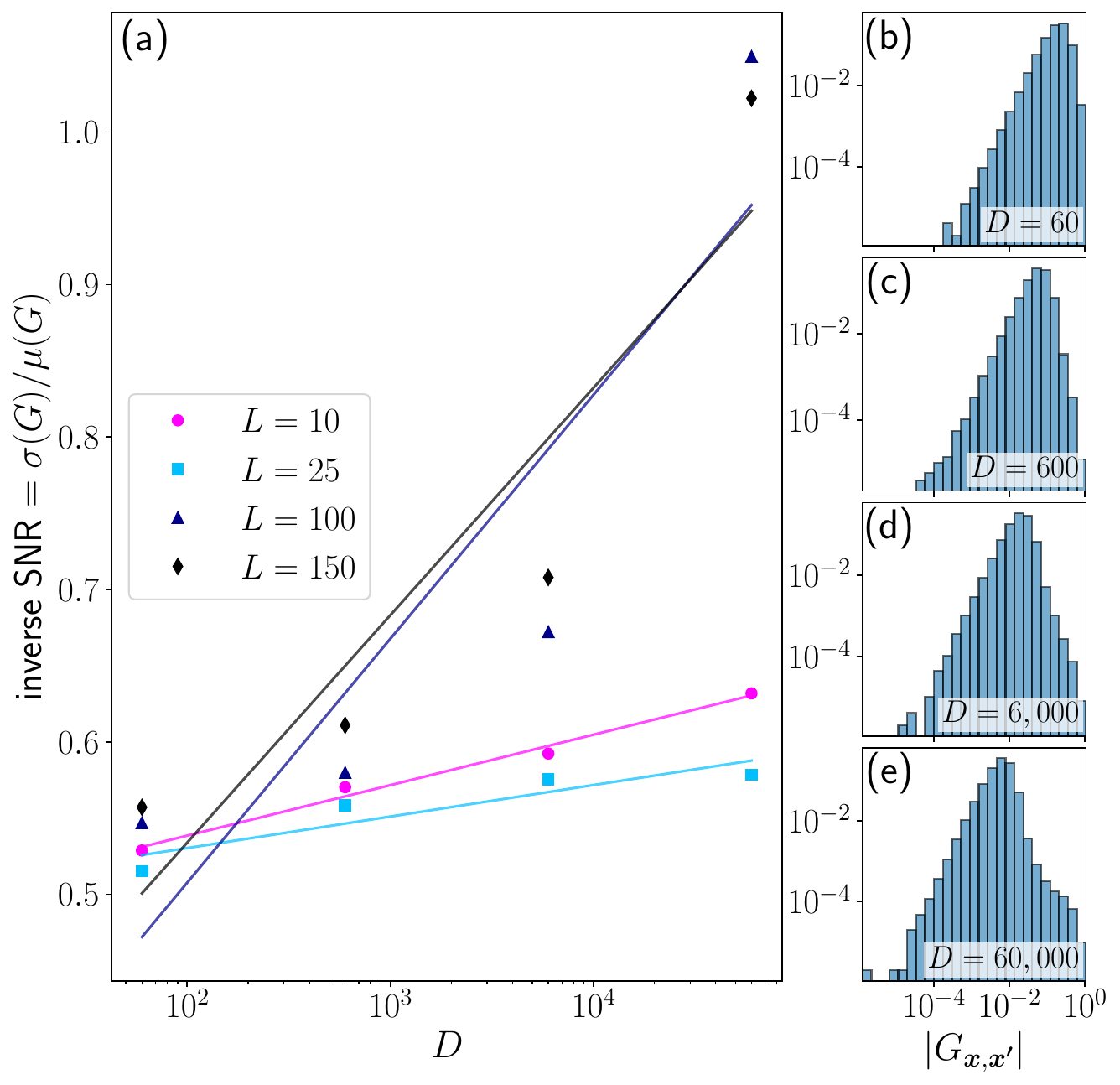}
  \caption{Nonequilibrium case.}
\end{subfigure}
\caption{Left: We plot the inverse SNR of $|G_{\bs x,\bs x'}|$ (excluding the diagonal elements, which are always one, and avoiding the double counting of $|G_{\bs x,\bs x'}| = |G_{\bs x',\bs x}|$) for various $L$ over $D$ on a logarithmic scale. The lines are a linear fit to visualize the trend. Right: Probability histograms (on a double logarithmic scale to better visualize outliers) of $|G_{\bs x,\bs x'}|$ for $L=100$ and various $D$ from top to bottom. }
\label{fig:inverse SNR}
\end{figure*}

We numerically studied decoherent histories in a simple yet, as we argued, generic random matrix model based on exact unitary (Schr\"odinger) time evolution for pure states. In particular, we studied the fate of decoherence for very long histories when we ``run out'' of Hilbert space to accommodate all of them.

We then found that the structure of decoherence changes drastically. Whereas for small $L$ all histories look equally decoherent [see Figs.~\ref{fig dec intro eq}(a) and~\ref{fig dec intro neq}(a) for $L=10$], for large $L$ there is a clear structure of decoherence: some histories remain decoherent from all other histories as before and some become (sometimes maximally) coherent with respect to some other histories. We call this change from initially decoherent histories for small $L$ to coherent histories for large $L$ \emph{recoherence}. Note that $L$ does not need to be very large: already for $L=25$ we see an emerging pattern.

This recoherence does not happen randomly. Instead, the dynamics of the system filters out a specific subset of histories, characterized by certain physical properties, that recohere. Recoherent histories $\bs x$ tend to have a high purity Petz recovery state and their states $|\psi(\bs x)\rangle$ tend to show strong localization. Moreover, histories which are not ``far'' from each other in the sense of a small Hamming distance $\Delta_{\bs x,\bs x'}$ are more likely to show recoherence. While we believe that these properties make sense intuitively, we have no elaborate theoretical explanation for it, and further research is necessary. Many more properties besides localization, purity and Hamming distance could also be studied---for example, how does decoherence correlate with the algorithmic complexity of $\bs x$?

In addition, we saw that recoherence is strongly correlated with statistical outliers in a probabilistic sense (as first found in Ref.~\cite{StrasbergSchindlerArXiv2023}). Histories whose frequencies deviate significantly from Born's rule are strongly recoherent and non-Markovian. Note that the focus on frequencies seems essential: it is not the case that elementary histories $\bs x$ with small weight $q(\bs x)$ are recoherent [see Figs.~\ref{fig:heat map},~\ref{fig N1 eq}(c) and~\ref{fig N1 neq}(c)], but inhomogeneous history states $|\psi(n)\rangle$ with a frequency $n/L$ deviating from Born's rule are. Remarkably, this behavior is true in and out of equilibrium, but we have again no elaborate theoretical explanation for this observation.

We finish with a final investigation of how the above structure of decoherence scales with $D$ for fixed $L$. After all, one would expect for fixed $L$ and $D\rightarrow\infty$ that all structure vanishes. While that might be true (it is not accessible to numerical analysis), it seems that the observed structure is quite \emph{persistent}. This becomes clear from several plots above where a larger Hilbert space dimension actually resulted in a clearer pattern, but we investigate it in detail in Fig.~\ref{fig:inverse SNR}. In the large plot we consider the inverse signal-to-noise ratio (SNR) equal to the standard deviation divided by the mean of $|G_{\bs x,\bs x'}|$ for different $L$ and $D$ (excluding the diagonal elements $G_{\bs x,\bs x}=1$ and avoiding the double counting of the off-diagonals $|G_{\bs x,\bs x'}| = |G_{\bs x',\bs x}|$). We see that for small $L$ the inverse SNR is fixed, as also suggested by Fig.~\ref{fig dec intro eq}(a), where statistical outliers behave similar to the mean. However, and somewhat suprisingly, for larger $L$ we find an \emph{increasing} inverse SNR as a function of $D$, indicating more structure (on a relative scale) in $G_{\bs x,\bs x'}$ for larger $D$ and fixed $L$. This can be also inferred from the histograms on the right of Fig.~\ref{fig:inverse SNR}: the mean shifts towards smaller values for larger $D$, but the support of the histogram grows. However, the histograms also indicate that perhaps for very large $D$ the support no longer grows and only the mean shifts. This would then result in a constant instead of growing inverse SNR, but it would not imply a vanishing of the decoherence structure for $D\rightarrow\infty$ on a relative scale. It would be certainly good to get more insights here. Finally, note that this conclusion is valid for both equilibrium and nonequilibrium histories.

\end{document}